\documentclass[12pt]{article}
\usepackage{graphicx}
\usepackage[utf8]{inputenc}
\usepackage{textcomp}
\usepackage{amsmath}
\usepackage[shortlabels]{enumitem}
\usepackage{chetnew}
\usepackage{tikz}
\usepackage{amssymb,amsfonts,mathrsfs,dsfont,yfonts,bbm}
\usepackage{extarrows}
\usepackage{xcolor}
\usepackage{braket}
\usepackage[normalem]{ulem}
\usepackage{import}
\usepackage{booktabs}
\usepackage{varioref}
\usepackage{bm}
\usetikzlibrary{arrows.meta,decorations.markings}
\definecolor{slabedge}{HTML}{2B3A55}
\definecolor{slabface}{HTML}{AAC4E2}
\definecolor{surf}{HTML}{12B886}
\definecolor{anyon}{HTML}{E03131}

\newcommand{\refB}{\reflectbox{$B$}}

\newcommand{\cD}{\mathcal{D}}

\newcommand{\cU}{\mathcal{U}}

\newcommand{\CR}{\mathcal{R}}

\newcommand{\beqn}{\begin{eqnarray}}
\newcommand{\eeqn}{\end{eqnarray}}

\newcommand{\be}{\begin{equation}}
\newcommand{\ee}{\end{equation}}
\newcommand{\bea}{\begin{eqnarray}}
\newcommand{\eea}{\end{eqnarray}}
\newcommand{\Rep}{{\rm Rep}}

\newcommand{\CL}{\mathcal{L}}
\newcommand{\CA}{\mathcal{A}}
\newcommand{\CF}{\mathcal{F}}

\newcommand{\CD}{\mathcal{D}}

\newcommand{\CC}{\mathcal{C}}
\newcommand{\cC}{\mathcal{C}}
\newcommand{\CZ}{\mathcal{Z}}

\newcommand{\CT}{\mathcal{T}}

\newcommand{\CM}{\mathcal{M}}

\newcommand{\Pic}{\mathsf{Pic}}

\newcommand{\URC}{\mathsf{URC}}

\newcommand{\based}{\text{based}}

\usepackage{amsmath}	% required for `\align' (yatex added)
\usepackage{tikz}
\usepackage{tikz-cd}
\newcommand{\tikzmath}[2][]
     {\vcenter{\hbox{\begin{tikzpicture}[#1]#2
                     \end{tikzpicture}}}
     }
\usetikzlibrary{calc,decorations.markings,decorations.pathmorphing,arrows.meta}

\tikzset{
	super thick/.style={line width=3pt}
}
\tikzstyle{knot}=[preaction={super thick, white, draw}]
\tikzstyle{mid>}=[decoration={markings, mark=at position 0.5 with {\arrow{>}}}, postaction={decorate}]
\tikzstyle{far>}=[decoration={markings, mark=at position 0.65 with {\arrow{>}}}, postaction={decorate}]
\tikzstyle{mid<}=[decoration={markings, mark=at position 0.5 with {\arrow{<}}}, postaction={decorate}]
\tikzstyle{far<}=[decoration={markings, mark=at position 0.35 with {\arrow{<}}}, postaction={decorate}]

\usetikzlibrary{calc}
\newcommand{\roundNbox}[6]{
	\draw[rounded corners=5pt, very thick, #1] ($#2+(-#3,-#3)+(-#4,0)$) rectangle ($#2+(#3,#3)+(#5,0)$);
	\coordinate (ZZa) at ($#2+(-#4,0)$);
	\coordinate (ZZb) at ($#2+(#5,0)$);
	\node at ($1/2*(ZZa)+1/2*(ZZb)$) {#6};
}

\renewcommand{\ket}[1]{\left\lvert#1\right\rangle}
\renewcommand{\bra}[1]{\left\langle#1\right\lvert}
\newcommand{\iprod}[2]{\left\langle#1\middle\lvert#2\right\rangle}

\newcommand{\AT}{\mathsf{AnyonTheory}}
\newcommand{\Ex}{\mathsf{Ex}}

\DeclareMathOperator{\Aut}{Aut}

\DeclareMathOperator{\End}{End}
\DeclareMathOperator{\ev}{ev}
\DeclareMathOperator{\Hom}{Hom}
\DeclareMathOperator{\Id}{Id}
\DeclareMathOperator{\id}{id}
\DeclareMathOperator{\Irr}{Irr}
\DeclareMathOperator{\Kar}{Kar}

\DeclareMathOperator{\mop}{mop}
\DeclareMathOperator{\Obj}{Obj}
\DeclareMathOperator{\op}{op}
\DeclareMathOperator{\rev}{rev}
\DeclareMathOperator{\tr}{tr}
\DeclareMathOperator{\Tr}{Tr}

\usepackage{amsmath}
\usepackage{amsthm}

\newtheorem{thm}{Theorem}
\newtheorem{lem}{Lemma}
\newtheorem{prop}{Proposition}
\newtheorem{cor}{Corollary}

\theoremstyle{definition}
\newtheorem{defn}{Definition}

\newtheorem{remark}{Remark}
\newtheorem{nota}{Notation}
\newtheorem{const}{Construction}
\newtheorem{claim}{Claim}

\newcommand{\doi}[1]{\href{https://dx.doi.org/#1}{{\small DOI:#1}}}

\title{Reality and Complexity of \texorpdfstring{$F$}{F}-symbols \texorpdfstring{\\[2mm]}{}  in~\texorpdfstring{$2+1$}{2+1}d Topological Phases}
\author{Matthew Buican\email{m.buican@qmul.ac.uk},$^{\diamondsuit}$ Peter Huston\email{P.E.Huston@leeds.ac.uk},$^{\spadesuit}$ and Jiannis K. Pachos\email{J.K.Pachos@leeds.ac.uk}$^{\heartsuit}$}

\affiliation{$^{\diamondsuit}$\smallskip Centre for Theoretical Physics and Astronomy\\
Queen Mary University of London, London E1 4NS, UK\\[1mm]$^{\spadesuit}$\smallskip School of Mathematics, University of Leeds, Leeds LS2 9JT, UK \\[1mm] $^{\heartsuit}$\smallskip School of Physics and Astronomy, University of Leeds, Leeds LS2 9JT, UK}

\abstract{The $F$-symbols of an anyon theory encode the associativity of fusion and constitute some of the theory's most fundamental and, simultaneously, subtle data. Much of the subtlety lies in the gauge-dependence of the $F$-symbols. Despite their generic complexity, many braided anyon theories admit gauges in which all $F$-symbols are real, a phenomenon for which no general organising principle has been known. We identify a physical mechanism underlying this reality. For a unitary ribbon fusion category admitting an appropriate braided charge-conjugation symmetry, we show that the complex-conjugated $F$-symbols are related to the original ones by a gauge transformation. Finding a real gauge is thereby reduced to a condition on these transformations. When the charge-conjugation symmetry is suitably ``flat,'' or equivalently when the associated ``twisted'' Frobenius-Schur (or ``generalized'' Kawanaka-Matsuyama) data respects a grading, these local transformations can be trivialised and a real gauge exists. This framework unifies a broad range of previously disparate examples, including families of Chern-Simons theories whose $F$-symbols are difficult to directly compute.
Finally, we exhibit a unitary ribbon category that realizes a novel obstruction to the existence of a real gauge and therefore has inherently complex $F$-symbols.}

\begin{document}

\setcounter{tocdepth}{2}
\maketitle
\toc

\section{Introduction}

Topological phases in $2+1$ dimensions exhibit a rich spectrum of anyonic excitations that are famously neither bosonic nor fermionic \cite{Leinaas:1977fm,Wilczek:1982wy}. The generalized statistics of these excitations can be computed from braiding processes involving anyon worldlines. The simplest such amplitudes give rise to projective representations of the modular group: the self-statistics ($T$ matrix) and the double braiding ($S$ matrix). It is natural to wonder whether $S$ and $T$, perhaps supplemented by additional braiding data, suffice to classify theories of anyons.

Recently, it has become clear that general systems of anyons cannot be characterized by braiding alone. This statement can be understood by studying properties of unitary modular tensor categories (UMTCs), the mathematical structures that capture the algebraic data of anyons \cite{Moore:1991ks,Kitaev:2005hzj}. In this language, one can show by counterexample that $S$ and $T$ do not suffice to specify a UMTC \cite{mignard2021modular}. More generally, the full braiding data of a UMTC, in the form of the $R$-symbols, is determined by the modular data \cite{ng2024recovering}.\footnote{This statement holds in the NRW gauge discussed in \cite{ng2024recovering}. Technically the NRW gauge also depends on an ordering of anyons. However, after performing a change of basis by $(\Gamma^{ab}_c)_{\mu\nu}=\sqrt{(R^{ba}_c)}_{\mu\nu}$, the $R$-symbols no longer depend on the ordering. We thank Z.~Wang for a discussion of this point.} Therefore, braiding alone is not enough to determine an anyon theory (although braiding suffices to specify UMTCs having small numbers of simple objects).

By construction, a UMTC is specified by a set of $R$-symbols describing braiding accompanied by a set of $F$-symbols that describe associativity of the fusion of anyon world lines. Physically, the $F$-symbols encode the topological operator product expansion data of the anyon theory. The above results suggest that there may be certain UMTC invariants intrinsic to the $F$-symbols. However, finding the $F$-symbols explicitly is generally rather difficult, and they are subject to complicated (generalized) cohomological equivalences arising from gauge transformations (which change the basis choice for junction spaces where anyons collide).

Therefore, to get a handle on the wild \lq\lq theory space" of $2+1$d TQFTs and their anyons,  it is important to understand if the $F$-symbols manifest a hidden simplicity. One of the most basic questions one can ask is: are there presentations (i.e., gauge choices) where the $F$-symbols take values in simple number fields (see also \cite{MR2312110,morrison2012non})?

In this paper, our particular interest is to understand when there exists a gauge such that the $F$-symbols are real. While the particular choice of such a gauge (when it exists!) amounts to fixing a basis and is therefore unphysical, one of the main upshots of our paper is that {\it the existence of a gauge in which the $F$-symbols are real is deeply physical.} Indeed, we will see this phenomenon is related to categorical properties of anyons and their symmetries. In order to make these categorical properties more intuitive, we will furnish corresponding spacetime interpretations (see Table \ref{CatST}).

\begin{table}
\begin{center}
\begin{tabular}{l|l}
\textbf{Categorical notion} & \textbf{Spacetime interpretation (if applicable)} \\\hline
$(\cdot)^{\rm op}$ & $t \to -t$ (time reflection) \\
$(\cdot)^{\rm mop}$ & $x \to -x$ (horizontal spatial reflection) \\
$(\cdot)^{\rm rev}$ & $y \to -y$ (reflection normal to $\Sigma$) \\
$(\cdot)^{\rm op,mop}$ &  $\pi$-rotation in $\Sigma$ \\
\end{tabular}
\caption{Spacetime interpretations for certain categorical notions that will be useful in the construction of real $F$-symbols. For simplicity, we will imagine a spacetime consisting of a 3-manifold $\CM=\Sigma\times I$, where $\Sigma$ is a Riemann surface, and $I$ is an interval. Reflections should be understood at the level of local coordinates (i.e., within a coordinate patch). A more detailed discussion of this table appears in subsequent sections.
}
\end{center}
\end{table}\label{CatST}

Beyond its relation to symmetries, the potential existence of a real gauge for the $F$-symbols has practical consequences for quantum simulation and experimental implementations of anyonic systems.  In this context, the lack of gauge invariance of the $F$-symbols is related to the fact that an experimenter accesses these quantities after making certain choices: once one specifies the physical operators used to create, move, fuse, and measure anyons, one has also fixed a laboratory basis for the relevant fusion spaces.  While certain closed amplitudes are gauge invariant, their decomposition into $F$-moves depends on the chosen fusion bases. In lattice realisations of quantum double or string-net models, this basis is fixed concretely by the choice of ribbon (or string) operators, which are equivalent to Wilson-lines. Thus, choosing a ribbon-operator convention is not merely notation: it determines the gauge in which the experimentally reconstructed $F$-matrix elements are expressed.

This observation makes a potential real gauge experimentally useful. If the implemented ribbon or string operators realise a gauge in which the relevant $F$-symbols are real, then the corresponding amplitudes have no imaginary components. The task of verifying an $F$-move is then reduced from reconstructing a complex matrix element to determining a real amplitude, together with its sign relative to a chosen reference convention.  In particular, one need not perform full complex-amplitude tomography for those $F$-symbols, and, in favourable protocols, the imaginary components that would otherwise require Hadamard-type interferometric measurements are absent by construction. The mathematical question of whether a real gauge exists therefore has a direct operational meaning: it tells us whether an experimental gauge can be chosen so that the associativity data of the anyons is represented by real, rather than inherently complex, amplitudes.

A concrete illustration of this operational point appears in recent work on the $D(S_3)$ quantum double model, where sequences of ribbon and projection operators were used to reconstruct the exchange and fusion-recombination data of non-Abelian $G$ anyons~\cite{BylesForbesPachos2024,byles2026}.  In that protocol, the relevant $F$-symbol data is extracted from overlaps of physically prepared ribbon states, so the choice of ribbon operators fixes the laboratory fusion basis in which the matrix elements are represented.  The construction yields norm-squared $F$-matrix elements, and, in the $D(S_3)$ example, the corresponding $F$-matrix is real, so the remaining ambiguity is reduced to a sign choice.

Our starting point for understanding conditions leading to real $F$-symbols is the following phenomenological observation: although UMTCs are inherently complex-valued (e.g., anyon statistics are complex), and the $F$-symbols are generically complex unitary matrices, braided fusion categories often admit gauges in which the $F$-symbols are real (i.e., orthogonal matrices). Indeed, a large body of examples suggests that the existence of real gauges for $F$-symbols is not accidental but rather is tied to the presence of braiding, together with an appropriate charge-conjugation symmetry. 

Abelian UMTCs provide the simplest illustration: in the Quinn gauge \cite{Quinn:1998un}, their $F$-symbols take values in $O(1)\cong\{\pm1\}$, and the same conclusion extends to braided Abelian categories with degenerate $S$-matrix. By contrast, if an Abelian fusion category has inherently complex $F$-symbols (i.e., $F$-symbols that cannot be made real in any gauge), then braiding is obstructed \cite{Paper1}. The category ${\rm Vec}_{\mathbb Z_3}^{\omega}$ with non-trivial associator class is the simplest example: its $F$-symbols cannot be made real, and consequently it admits no braiding (see Sec.~\ref{Complex}). Real gauges are also known, or strongly suggested, in several non-Abelian braided families, including $SU(2)_k$, metaplectic categories, many Chern-Simons theories, braided Tambara-Yamagami (TY) categories, and the non-symmetric braided near-group categories.

One of the main motivations of this work is to find a principle that explains and unifies the above examples of real $F$-symbols. The construction of the Quinn gauge, which gives $F\in O(1)$ for Abelian UMTCs, does not immediately generalise to non-Abelian theories. Nevertheless, the examples above suggest that two ingredients play a distinguished role: braiding and charge conjugation.

Our approach separates the problem into two logically distinct questions. First, when are the $F$-symbols gauge-equivalent to their complex conjugates? Second, when can this gauge transformation itself be trivialised by a single consistent choice of bases for all trivalent fusion spaces?

The first question has a remarkably general answer in the presence of braided charge conjugation. Indeed, suppose that $\cC$ is a unitary ribbon category admitting a unitary braided autoequivalence
\begin{equation}
    J:\cC\longrightarrow\cC~,
    \qquad
    J(a)\cong a^\vee~,
\end{equation}
on simple objects. In Sec.~\ref{sec:mechanism} we construct, on each splitting space $V_a^{bc}$, a local\footnote{Here, by local we mean that the transformation applies to the space $V_a^{bc}$ for a specific triple $a,b,c\in\Irr(\cC)$, rather than any notion of locality in spacetime.} antiunitary transformation
\begin{equation}
    A_a^{bc}=\mathsf M_a^{bc}K~,
\end{equation}
where $K$ is complex conjugation in the chosen trivalent basis and $\mathsf M_a^{bc}$ is unitary. Applying these transformations to an $F$-move gives
\begin{equation}
\label{FFbarEquiv}
    \overline{F_d^{abc}}
    =
    \mathsf M_L\,F_d^{abc}\,\mathsf M_R^\dagger~,
\end{equation}
where $\mathsf M_L$ and $\mathsf M_R$ are certain products of $\mathsf M_x^{yz}$ matrices. Thus $F$ and $\overline F$ are related by an ordinary vertex-basis gauge transformation. Importantly, this statement does not yet require the charge-conjugation operation to square to the identity.

Equation~\eqref{FFbarEquiv} reduces the problem of reality to a local one. If every matrix $\mathsf M_a^{bc}$ is symmetric,
\begin{equation}
    \bigl(\mathsf M_a^{bc}\bigr)^T
    =
    \mathsf M_a^{bc}~,
\end{equation}
then its unitarity allows a Takagi factorisation
\begin{equation}
    \mathsf M_a^{bc}
    =
    \mathsf W_a^{bc}
    \bigl(\mathsf W_a^{bc}\bigr)^T~.
\end{equation}
The corresponding local change of trivalent basis, $\Gamma_a^{bc}=(\mathsf W_a^{bc})^T$, simultaneously trivialises the gauge transformation in~\eqref{FFbarEquiv} for every $F$-move. Consequently, in the transformed gauge,
\begin{equation}
    \overline{F_d^{abc}}=F_d^{abc}~.
\end{equation}
The central question is therefore reduced to determining when the local matrices $\mathsf M_a^{bc}$ are symmetric.

Sections~\ref{DefsPrelim} and~\ref{sec:sufficient} provide the categorical explanation of when this symmetry condition occurs. The structural involutions of a unitary ribbon category give a canonical braided equivalence
\begin{equation}
    \cC\longrightarrow\overline{\cC}^{\rev}~,
\end{equation}
that acts as duality on simple objects. 
 
Composing this structural equivalence with a charge conjugation produces a comparison equivalence, $R:\cC\to\overline{\cC}^{\rev}$, that can be chosen to preserve the preferred simple-object labels. Its gauge symbol is the categorical counterpart of the $\mathsf M_a^{bc}$ matrices appearing in the argument above. Whether that symbol can be trivialised depends on the \lq\lq flatness" of charge conjugation.

This flatness property is defined in Sec.~\ref{sec:sufficient}. It requires that, for some choice of unitary dual functor, the squared comparison, $\overline{R}\circ R$, has trivial gauge symbol. The upshot is that, at the level of the local matrices, flatness gives
\begin{equation}
    \overline{\mathsf M_a^{bc}}\,
    \mathsf M_a^{bc}
    =
    \mathbbm 1~\ \Rightarrow\ (\mathsf M^{bc}_a)^T=\mathsf{M}_a^{bc}~,
\end{equation}
where the final implication follows from unitarity. Flatness therefore supplies precisely the condition required by the Takagi construction. 

The reason we have chosen to call such a charge conjugation flat is that flatness amounts to the existence of a choice of unitary dual functor for which all
\begin{equation}\label{Flatness}
    \lambda_a^{bc} =\frac{\sigma_b\sigma_c}{\sigma_a}=1~, 
    \ \ \  N_a^{bc}\neq0~,
\end{equation}
where the various $\sigma_i:=\varkappa_i/t_i$, defined precisely in Sec. \ref{sec:sufficient}, are phases built from the interaction between the charge conjugation symmetry and the pivotal data. This condition can be understood pictorially as
in Fig.~\ref{fig:flat}.
\begin{figure}[htbp]
\centering
\begin{tikzpicture}[scale=0.92,
  ed/.style  ={-{Stealth[length=1.6mm]}, line width=0.5pt, draw=black!62},
  nod/.style ={circle, fill=black!62, inner sep=0pt, minimum size=2.7pt},
  fac/.style ={fill=black!7},
  lab/.style ={font=\scriptsize, text=black!78, fill=white, inner sep=1.1pt},
  hol/.style ={font=\footnotesize, text=green!40!black, fill=white, inner sep=1.2pt},
  ttl/.style ={font=\footnotesize\itshape, text=black!70},
  capt/.style={font=\scriptsize, align=center}]

%% Left panel: a general allowed fusion vertex and its holonomy

\node[ttl] at (1.90,0.55) {one fusion vertex};

\coordinate (P) at (0.35,-2.35);
\coordinate (Q) at (1.90,-0.35);
\coordinate (R) at (3.45,-2.35);

\fill[fac] (P) -- (Q) -- (R) -- cycle;

\draw[ed] (P) -- (Q);
\draw[ed] (Q) -- (R);
\draw[ed] (P) -- (R);

\node[nod] at (P) {};
\node[nod] at (Q) {};
\node[nod] at (R) {};

\node[lab] at (1.125,-1.35) {$\sigma_b$};
\node[lab] at (2.675,-1.35) {$\sigma_c$};
\node[lab] at (1.900,-2.35) {$\sigma_a$};

\node[hol] at (1.90,-1.80) {$\lambda^{bc}_a$};

\node[capt] at (1.90,-3.24) {$b\otimes c\ni a$, \ i.e. \ $N^{bc}_a\neq0$};
\node[capt] at (1.90,-3.98) {$\lambda^{bc}_a=\dfrac{\sigma_b\,\sigma_c}{\sigma_a}$};

%%%%%%%%%%%%%%%%%%%%%%%%%%%%%%%%%%%%%%%%%%%%%%%%%%%%%%%%%%%%%%%%%%%%%%%%%%%%%%%
\draw[densely dashed, draw=black!25, line width=0.4pt] (4.75,0.55) -- (4.75,-4.10);

%% Right panel: the same vertex in the flat case.
\node[ttl] at (7.60,0.55) {flat: every fusion vertex closes};

\coordinate (P2) at (6.05,-2.35);
\coordinate (Q2) at (7.60,-0.35);
\coordinate (R2) at (9.15,-2.35);

\fill[fac] (P2) -- (Q2) -- (R2) -- cycle;

\draw[ed] (P2) -- (Q2);
\draw[ed] (Q2) -- (R2);
\draw[ed] (P2) -- (R2);

\node[nod] at (P2) {};
\node[nod] at (Q2) {};
\node[nod] at (R2) {};

\node[lab] at (6.825,-1.35) {$\sigma_b$};
\node[lab] at (8.375,-1.35) {$\sigma_c$};
\node[lab] at (7.600,-2.35) {$\sigma_a$};

\node[hol] at (7.60,-1.80) {$+1$};

\node[capt] at (7.60,-3.24) {$\sigma_b\,\sigma_c=\sigma_a$ \ whenever \ $N^{bc}_a\neq0$};
\node[capt] at (7.60,-3.98) {$\lambda\equiv1$: \ $\sigma$ grades the fusion rules};

\end{tikzpicture}
\caption{A schematic picture of flatness in terms of the $\sigma_a$ phases that describe the interaction between charge conjugation and pivotal data. Left: an allowed fusion vertex $(a;b,c)$, i.e. $b\otimes c\ni a$, drawn as a triangle whose legs carry $\sigma_b$, $\sigma_c$ and $\sigma_a$, and whose holonomy is $\lambda_a^{bc}=\sigma_b\sigma_c/\sigma_a$. Right: flatness in this context is the requirement that this holonomy be trivial at every allowed fusion vertex, which is precisely the statement that $\sigma$ grades the fusion rules.}
\label{fig:flat}
\end{figure}
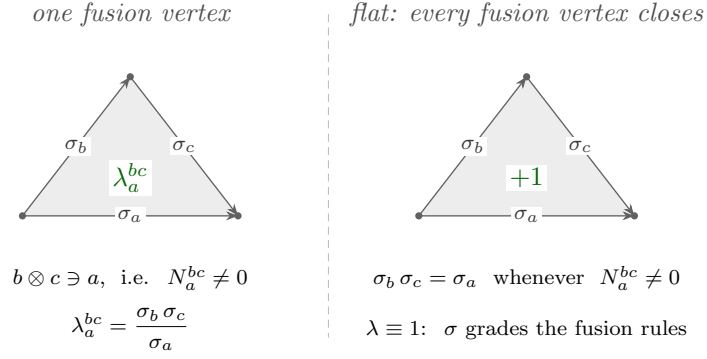
Our main result can therefore be stated as follows:
\begin{claim}
\label{claim:main}
Let $\cC$ be a unitary ribbon category with a flat charge-conjugation symmetry. Then there exists a gauge in which the $F$-symbols of $\cC$ are real and its $R$-matrices are symmetric.
\end{claim}

Let us make a few comments about this claim. First, since we work throughout with orthonormal fusion bases and unitary gauge transformations, the $F$-matrices in such a gauge are automatically real orthogonal. Second, we formulate the result for unitary ribbon categories rather than UMTCs because none of our arguments require non-degeneracy of the braiding. Third, the abstract categorical definition of flatness has particularly transparent forms in important classes of examples. If every simple object of $\cC$ is self-dual,\footnote{We will sometimes refer to such a $\CC$ as being self-dual. When saying a category is self-dual, we refer to the fact that all simple objects are self-dual rather than referring to self-duality of the theory as an object in some monoidal higher category.} we may take $J=\Id_{\cC}$, for which $t_i=1$, and the flatness condition reduces to
\begin{equation}
    \frac{\varkappa_b\varkappa_c}{\varkappa_a}=1~,
    \qquad
    \text{whenever }N_a^{bc}\neq0~,
\end{equation}
where $\varkappa_a$ is the second Frobenius-Schur indicator. Thus, in the self-dual case, flatness is precisely the statement that the Frobenius-Schur indicator defines a $\mathbb Z_2$ grading of the fusion rules.

For representation categories of finite groups, the corresponding condition is naturally related to the twisted Frobenius-Schur indicator of Kawanaka and Matsuyama \cite{kawanaka1990twisted}. We make this relation explicit in Appendix~\ref{KMAppendix}. This motivates viewing the general flatness condition as a categorical analogue of a twisted Frobenius-Schur grading or, equivalently, of a generalized Kawanaka-Matsuyama grading.  A more systematic study of these quantities for arbitrary non-Tannakian unitary ribbon categories is left for future work.

In cases without fusion multiplicity (e.g., as in the case of Abelian anyon models or various non-Abelian cases like Ising or Fib$\cong (G_2)_1$ Chern-Simons), the condition for real $F$ symbols simplifies:
\begin{claim}
    \label{claim:trivMult}
    In a unitary ribbon category with no fusion multiplicity and a braided $\mathbb{Z}_2$ charge conjugation symmetry, there exists a gauge in which the $F$-symbols are real.
\end{claim}

As a simple corollary of Claim \ref{claim:main}, we have the following special case, which already covers a great many examples:\footnote{Note that, as in Claim \ref{claim:main}, the $F$-symbols are also orthogonal.}
\begin{cor}
\label{cor:selfDual}
In a unitary ribbon category in which every object is self-dual and the second Frobenius-Schur indicator defines a $\mathbb Z_2$ grading, there exists a gauge in which the $F$-symbols are real and the $R$-matrices are symmetric.
\end{cor}

This discussion also raises the question of whether there are braided fusion categories with inherently complex $F$-symbols. In fact, we answered this question recently in \cite{Paper1}. There the idea was to note that, from the results in \cite{MR3421083}, we know there are braided categories where charge conjugation is not a genuine symmetry (i.e., it is not a braided autoequivalence of the category in question) but only a symmetry of the modular data (in other words, exchanging $a\leftrightarrow a^{\vee}$ leaves $S$ and $T$ invariant; this invariance is always achieved since $S^2=C$, where $C$ is a matrix implementing charge conjugation).\footnote{Such examples of charge conjugation were termed \lq\lq quasi" zero-form symmetries in \cite{Buican:2020who}.} Studying such categories, we found the following minimal (in terms of rank and total categorical dimension) ${\rm Rep}(G)$ examples with inherently complex $F$-symbols \cite{Paper1}: ${\rm Rep}(\mathbb{Z}_7\rtimes\mathbb{Z}_3)$ and ${\rm Rep}(\mathbb{Z}_5\rtimes\mathbb{Z}_4)$.

In this paper, we find a more subtle example of a category with inherently complex $F$-symbols: 
${\rm Rep}(\mathbb{F}_3^2\rtimes Q_8)$. Unlike the examples in \cite{Paper1}, this case has a braided charge conjugation: the trivial one, since all objects are self-dual. However, its Frobenius-Schur indicator does not define a $\mathbb{Z}_2$ grading, so the trivial charge-conjugation symmetry is not flat. We then show that this category has inherently complex $F$-symbols. As such, it is a subtle probe of the sharpness of our above claims and corollaries.

We thank Cesar Galindo for sharing a preliminary version of his independent work~\cite{Cesar} while we were completing this paper. \cite{Cesar} investigates closely related questions concerning the reality of $F$-symbols. The two works were developed independently and provide complementary approaches to understanding the structures underlying real $F$-symbol gauges for anyon theories. Our results agree on their common overlap and are often complementary in scope. For example, in Sec. \ref{ChernSimons} we find a gauge with real $F$-symbols and symmetric $R$-matrices for all self-dual Chern-Simons theories with compact and simply connected gauge groups at any level. On the other hand, \cite{Cesar} finds a gauge with real $F$-symbols in the examples of Sec. \ref{ChernSimons} as well as for general $\CC(\mathfrak g,k)$ for every complex simple $\mathfrak g$ and every $k$.

The plan of the paper is as follows. In Sec.~\ref{folding}, we first develop the physical intuition behind the construction using folding and Lagrangian algebras. In the UMTC setting this gives a spacetime picture of the gauge equivalence between $F$ and $\overline F$, and in the multiplicity-free case already leads directly to a real gauge. In Sec.~\ref{sec:mechanism}, we give a general diagrammatic derivation that does not require modularity. We construct the local antiunitaries $A_a^{bc}=\mathsf M_a^{bc}K$, prove $\overline F=\mathsf M_LF\mathsf M_R^\dagger$, and show that symmetry of the local $\mathsf M_a^{bc}$ matrices is sufficient for reality of the $F$-symbols. For arbitrary fusion multiplicities the required change of basis is obtained by Takagi factorisation. Section~\ref{DefsPrelim} develops the categorical framework needed to identify the origin of these gauge transformations. We introduce based unitary ribbon categories in order to keep track of trivalent basis choices, relate them explicitly to the usual anyon-theory description, and study the structural involutions $(\cdot)^{\op}$, $(\cdot)^{\mop}$, $(\cdot)^{\rev}$, and complex conjugation. These constructions provide the categorical counterpart of the diagrammatic gauge equivalence between $F$ and $\overline F$ and the spacetime maneuvers in Sec. \ref{folding}.
In Sec.~\ref{sec:sufficient}, we define flat charge-conjugation symmetries and prove our main theorem. Flatness makes the relevant comparison equivalence involutive at the level of its gauge symbol, forcing the local gauge matrices to be symmetric; Takagi factorisation then gives a gauge with real $F$-symbols and symmetric $R$-matrices. We also give a complementary diagrammatic interpretation of this symmetry in the self-dual case with trivial second Frobenius-Schur indicator. 
We then apply the theorem in Sec.~\ref{implications} to a variety of anyon theories, including families of Chern-Simons theories for which the $F$-symbols are difficult to compute directly. Finally, in Sec.~\ref{Complex} we investigate obstructions to reality and exhibit the self-dual category ${\rm Rep}(\mathbb F_3^2\rtimes Q_8)$, whose Frobenius-Schur indicator fails to define a grading and whose $F$-symbols are inherently complex. The appendices establish the connection of Sec. \ref{sec:sufficient} with the Kawanaka-Matsuyama indicator and provide a structural theorem used in our non-Abelian Chern-Simons examples.

\section{An intuitive picture from folding and some categorical notions in spacetime}
\label{folding}

This section provides a gentle warm-up and foundation geared toward the high-energy theory community. In particular, we give an intuitive physical explanation, from the perspective of folding, for two special cases of more general phenomena discussed in this paper:
\begin{enumerate}
\item UMTCs with a braided $\mathbb{Z}_2$ charge conjugation symmetry exhibit a gauge equivalence between their $F$-symbols and the complex conjugate $F$-symbols
\begin{eqnarray}\label{Freal}
\overline{\left[{F}^{abc}_{d}\right]}_{(e,\alpha,\beta)(f,\mu,\nu)} 
&=& \sum_{\alpha',\beta',\mu',\nu'} 
\left[\Gamma^{ab}_{e}\right]_{\alpha\alpha'}
\left[\Gamma^{ec}_{d}\right]_{\beta\beta'}
\left[F^{abc}_{d}\right]_{(e,\alpha',\beta')(f,\mu',\nu')}\left[\left(\Gamma^{bc}_{f}\right)^{-1}\right]_{\mu'\mu}
\Big[\big(\Gamma^{af}_{d}\big)^{-1}\Big]_{\nu'\nu}\cr&=:& \left[F^{abc}_{d}\right]_{(e,\alpha,\beta)(f,\mu,\nu)}^{\Gamma}~.
\end{eqnarray}
\item Specializing to the case of no fusion multiplicity (e.g., Abelian models or non-Abelian TQFTs like ${\rm Fib}\cong (G_2)_1$ Chern-Simons theory and Ising), there is a choice of fusion bases with real $F$-symbols
\begin{equation}
\overline{\left[{F}^{abc}_d\right]^{\widehat\Gamma}}_{ef}=\left[F^{abc}_d\right]^{\widehat\Gamma}_{ef}\in\mathbb{R}~.
\label{gaugereality}
\end{equation}
\end{enumerate}
As a bonus, our discussion also invites us to consider the spacetime interpretations of certain categorical notions described in Table \ref{CatST} in more detail.

Before continuing, let us note the limitations of this section and advertise what is to come in the rest of this paper. Throughout this section, we will assume that the underlying category is a UMTC. Later sections will describe results on gauge equivalence of $F$ and $\overline F$ that apply more broadly to unitary ribbon categories with a $\mathbb{Z}_2$ braided charge conjugation symmetry. In fact, our results on reality of $F$-symbols also generalize to unitary ribbon categories with fusion multiplicity subject to further constraints on the charge conjugation symmetry. We will discuss these generalizations in more detail in Secs. \ref{sec:mechanism}, \ref{DefsPrelim}, and \ref{sec:sufficient}. The main purpose of this section is to develop intuition for what is to come using light notation familiar to high-energy theorists.

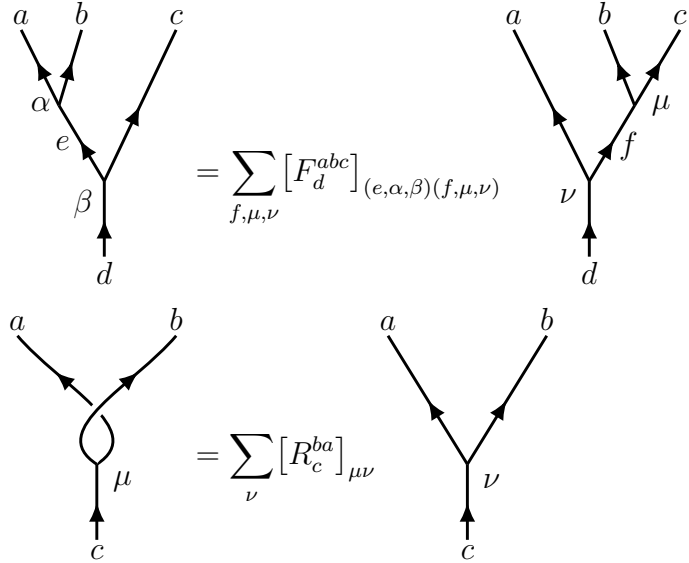
\begin{figure}[!ht]
\begin{align*}
\begin{tikzpicture}[baseline,
    line width=1.1pt,
    >={Latex[length=2.6mm,width=2.1mm]},
    mid/.style={
      postaction={decorate},
      decoration={
        markings,
        mark=at position #1 with {\arrow{Latex[length=2.6mm,width=2.1mm]}}
      }
    },
    mid/.default=0.5,
    every node/.style={inner sep=1.5pt}
 ]
  %F-MOVE (top row)
  \begin{scope}[shift={(0,0)}]
    \coordinate (Fd)  at (1.0,-1.0);
    \coordinate (Fb)  at (1.0, 0.0);
    \coordinate (Fa)  at (0.4, 1.0);
    \coordinate (Fla) at (-0.10,2.0);
    \coordinate (Flb) at (0.70, 2.0);
    \coordinate (Flc) at (1.95, 2.0);
    \draw[mid]      (Fd) -- (Fb);
    \draw[mid=0.55] (Fb) -- (Fa);   % e
    \draw[mid=0.6]  (Fa) -- (Fla);  % a
    \draw[mid=0.6]  (Fa) -- (Flb);  % b
    \draw[mid]      (Fb) -- (Flc);  % c
    \node[above] at (Fla) {$a$};
    \node[above] at (Flb) {$b$};
    \node[above] at (Flc) {$c$};
    \node[below] at (Fd)  {$d$};
    \node at (0.16,1.02)  {$\alpha$};
    \node at (0.45,0.50)  {$e$};
    \node at (0.72,-0.30) {$\beta$};
  \end{scope}
  \end{tikzpicture}
  &=
    \displaystyle\sum_{f,\mu,\nu}\bigl[F^{abc}_d\bigr]_{(e,\alpha,\beta)(f,\mu,\nu)}
  \begin{tikzpicture}[baseline,
    line width=1.1pt,
    >={Latex[length=2.6mm,width=2.1mm]},
    mid/.style={
      postaction={decorate},
      decoration={
        markings,
        mark=at position #1 with {\arrow{Latex[length=2.6mm,width=2.1mm]}}
      }
    },
    mid/.default=0.5,
    every node/.style={inner sep=1.5pt}
  ]
    \coordinate (Gd)  at (1.0,-1.0);
    \coordinate (Gn)  at (1.0, 0.0);
    \coordinate (Gm)  at (1.6, 1.0);
    \coordinate (Gla) at (0.00,2.0);
    \coordinate (Glb) at (1.20,2.0);
    \coordinate (Glc) at (2.20,2.0);
    \draw[mid]      (Gd) -- (Gn);
    \draw[mid]      (Gn) -- (Gla);  % a
    \draw[mid=0.55] (Gn) -- (Gm);   % f
    \draw[mid=0.6]  (Gm) -- (Glb);  % b
    \draw[mid=0.6]  (Gm) -- (Glc);  % c
    \node[above] at (Gla) {$a$};
    \node[above] at (Glb) {$b$};
    \node[above] at (Glc) {$c$};
    \node[below] at (Gd)  {$d$};
    \node at (1.95,1.00)  {$\mu$};
    \node at (1.52,0.46)  {$f$};
    \node at (0.72,-0.20) {$\nu$};
  \end{tikzpicture}\\
  %R-MOVE (bottom row)
  \begin{tikzpicture}[baseline,
    line width=1.1pt,
    >={Latex[length=2.6mm,width=2.1mm]},
    mid/.style={
      postaction={decorate},
      decoration={
        markings,
        mark=at position #1 with {\arrow{Latex[length=2.6mm,width=2.1mm]}}
      }
    },
    mid/.default=0.5,
    every node/.style={inner sep=1.5pt}
  ] % [shift={(0,-4.7)}]
    \coordinate (Rc)  at (1.0,-1.0);
    \coordinate (Rm)  at (1.0, 0.0);
    \coordinate (Rla) at (-0.05,1.7);
    \coordinate (Rlb) at (2.05,1.7);
    \draw[mid=0.7] (Rm) .. controls (1.7,0.45) and (0.5,1.05) .. (Rla);
    \draw[mid=0.7,preaction={draw,line width=5pt,white}]
          (Rm) .. controls (0.3,0.45) and (1.5,1.05) .. (Rlb);
    \draw[mid] (Rc) -- (Rm);
    \node[above] at (Rla) {$a$};
    \node[above] at (Rlb) {$b$};
    \node[below] at (Rc)  {$c$};
    \node at (1.34,-0.22) {$\mu$};
  \end{tikzpicture}
  % \node at (3.3,-4.0) {$=$};
  &=
  % \node[right] at (3.7,-4.0)
    \displaystyle\sum_{\nu}\bigl[R^{ba}_c\bigr]_{\mu\nu}
  \begin{tikzpicture}[baseline,
    line width=1.1pt,
    >={Latex[length=2.6mm,width=2.1mm]},
    mid/.style={
      postaction={decorate},
      decoration={
        markings,
        mark=at position #1 with {\arrow{Latex[length=2.6mm,width=2.1mm]}}
      }
    },
    mid/.default=0.5,
    every node/.style={inner sep=1.5pt}
  ] % [shift={(7.4,-4.7)}]
    \coordinate (Sc)  at (1.0,-1.0);
    \coordinate (Sn)  at (1.0, 0.0);
    \coordinate (Sla) at (-0.05,1.7);
    \coordinate (Slb) at (2.05,1.7);
    \draw[mid] (Sn) -- (Sla);
    \draw[mid] (Sn) -- (Slb);
    \draw[mid] (Sc) -- (Sn);
    \node[above] at (Sla) {$a$};
    \node[above] at (Slb) {$b$};
    \node[below] at (Sc)  {$c$};
    \node at (1.32,-0.22) {$\nu$};
  % \end{scope}
\end{tikzpicture}
% }
\end{align*}
  \caption{The $F$- and $R$-symbols of a braided fusion category. The $F$-symbols govern fusion associativity, while the $R$-symbols govern braiding. Splitting spaces, $V_x^{yz}\cong{\rm Hom}(x\to y\otimes z)$, are represented diagrammatically by trivalent vertices with $x$ splitting into $y$ and $z$. We have ${\rm dim}_{\mathbb{C}}V_x^{yz}=N_{yz}^x$. Note that the $F$- and $R$-symbols depend on the choice of bases for these fusion spaces and therefore transform as in \eqref{eq:gaugeTransformation} and \eqref{RGT}.
  \label{fig:fr-moves}}
\end{figure}

We begin by considering the 3-manifold $\CM\cong\Sigma\times I$, where $\Sigma$ is a Riemann surface, and $I$ is a finite interval. For concreteness we can take $t,x$ as local coordinates on $\Sigma$ and $-1\le y\le 1$ as a coordinate on $I$. We place a non-spin $2+1$d TQFT, $\CT$, on $\CM$. We denote the corresponding UMTC as $\CC$ and sometimes use $\CT$ and $\CC$ interchangeably. The anyon world lines of $\CT$ correspond to objects in $\CC$, and we will use both terminologies interchangeably as well. Any line operator in $\CT$ or object in $\CC$ can be built up from a sum of simple lines / objects, the number of which, $|\CC|$, is called the rank of $\CC$ (we denote ${\rm Irr}(C)$ as the corresponding set of simples). We can define braiding data, in the form of $R$-symbols, and associativity data, in the form of $F$ symbols, as in Fig. \ref{fig:fr-moves}. Note that these quantities involve fusion spaces (vector spaces of three-punctured spheres) and a change of basis of these spaces results in the transformation
\begin{eqnarray}
\label{eq:gaugeTransformation}
\left[F^{abc}_{d}\right]^{\Gamma}_{(e,\alpha,\beta)(f,\mu,\nu)} 
&:=& \sum_{\alpha',\beta',\mu',\nu'} 
\left[\Gamma^{ab}_{e}\right]_{\alpha\alpha'}
\left[\Gamma^{ec}_{d}\right]_{\beta\beta'}
\left[F^{abc}_{d}\right]_{(e,\alpha',\beta')(f,\mu',\nu')}\left[\left(\Gamma^{bc}_{f}\right)^{-1}\right]_{\mu'\mu}
\left[\left(\Gamma^{af}_{d}\right)^{-1}\right]_{\nu'\nu}~,\ \ \ \ \ \ 
\end{eqnarray}
where $[\Gamma^{xy}_z]\in U(N^{xy}_z)$ is a unitary change of basis for the Hilbert space corresponding to the sphere punctured by lines $x$, $y$, and $z$ (more mathematically, we write this space as ${\rm Hom}(z,x\otimes y)$ or, equivalently, ${\rm Hom}(z\to x\otimes y)$). Similarly, for $R$ we have
\begin{align}\label{RGT}
\left[R^{ba}_{c}\right]^{\Gamma}_{\mu\nu}
&:= \sum_{\mu',\nu'} 
\left[\Gamma^{ba}_c\right]_{\mu\mu'}
\left[R^{ba}_c\right]_{\mu'\nu'}
\left[(\Gamma^{ab}_c)^{-1}\right]_{\nu'\nu}~.
\end{align}
Note that since $\CC$ is modular, its braiding satisfies a non-degeneracy condition: the modular $S$ matrix / Hopf link is a non-degenerate $|\CC|\times|\CC|$ matrix measuring the mutual Aharonov-Bohm phases of the simple lines.

\begin{figure}[ht]
  \centering
  \begin{tikzpicture}[
      x={(1cm,0cm)}, y={(0.5cm,0.3cm)}, z={(0cm,1cm)},
      line cap=round, line join=round,
      >={Stealth[length=2.4mm]},
      sccsurf/.style={fill=surf, fill opacity=0.28, draw=surf!70!black, line width=0.7pt},
      bdry/.style={fill=surf, fill opacity=0.5, draw=surf!65!black, line width=0.8pt},
      hid/.style={draw=slabedge!55, line width=0.6pt, dash pattern=on 2pt off 2pt},
      vis/.style={draw=slabedge, line width=0.95pt},
      aline/.style={draw=anyon, line width=1.5pt},
      albl/.style={fill=white, fill opacity=0.82, text opacity=1, inner sep=1.2pt},
    ]

    %LEFT SLAB :  M = Sigma x 
    \coordinate (LA) at (0,0,0);   \coordinate (LB) at (5,0,0);
    \coordinate (LC) at (5,2,0);   \coordinate (LD) at (0,2,0);
    \coordinate (LE) at (0,0,1.5); \coordinate (LF) at (5,0,1.5);
    \coordinate (LG) at (5,2,1.5); \coordinate (LH) at (0,2,1.5);

    \draw[hid] (LD)--(LA);  \draw[hid] (LD)--(LC);  \draw[hid] (LD)--(LH);

    \draw[sccsurf] (2.5,0,0)--(2.5,2,0)--(2.5,2,1.5)--(2.5,0,1.5)--cycle;

    % anyon line: two mid-line arrows (one each side of S_CC), pointing left
    \draw[aline, postaction={decorate}, decoration={markings,
          mark=at position 0.27 with {\arrow{Stealth[length=2.6mm]}},
          mark=at position 0.73 with {\arrow{Stealth[length=2.6mm]}}}]
          (5,1,0.75) -- (0,1,0.75);
    \fill[anyon] (2.5,1,0.75) circle (2.3pt);
    \draw[white,line width=0.4pt] (2.5,1,0.75) circle (2.3pt);

    \fill[slabface,opacity=0.10] (LE)--(LF)--(LG)--(LH)--cycle;
    \fill[slabface,opacity=0.14] (LA)--(LB)--(LF)--(LE)--cycle;
    \fill[slabface,opacity=0.18] (LB)--(LC)--(LG)--(LF)--cycle;

    \draw[vis] (LA)--(LB)--(LC); \draw[vis] (LA)--(LE);
    \draw[vis] (LB)--(LF); \draw[vis] (LC)--(LG);
    \draw[vis] (LE)--(LF)--(LG)--(LH)--(LE);

    \node[albl] at (0,1,1.08) {$a^{\vee}$};
    \node[albl] at (5,1,1.08) {$a$};
    \node at (2.4,1,2.75) {$S_{CC}$};
    \draw[->,slabedge,line width=0.7pt] (2.5,1,2.55) to[bend right=18] (2.5,1,1.62);
    \node[anchor=north,font=\footnotesize] at (2.62,0,-0.02) {$y=0$};
    \node[anchor=north] at (2.5,0,-0.55) {$\mathcal{M}\cong\Sigma\times I$};

    %FOLD ARROW 
    \draw[-{Stealth[length=3mm]},line width=1.1pt,slabedge]
          (7.4,1,0.9) -- (9.2,1,0.9) node[midway,above=3pt]{fold};

    %RIGHT SLAB :  C box Cbar 
    \coordinate (RA) at (10,0,0);    \coordinate (RB) at (12.4,0,0);
    \coordinate (RC) at (12.4,2,0);  \coordinate (RD) at (10,2,0);
    \coordinate (RE) at (10,0,1.5);  \coordinate (RF) at (12.4,0,1.5);
    \coordinate (RG) at (12.4,2,1.5);\coordinate (RH) at (10,2,1.5);

    \draw[hid] (RD)--(RA);  \draw[hid] (RD)--(RC);  \draw[hid] (RD)--(RH);

    \draw[bdry] (RA)--(RD)--(RH)--(RE)--cycle;

    % anyon line (a,a): single mid-line arrow, pointing left toward the boundary
    \draw[aline, postaction={decorate}, decoration={markings,
          mark=at position 0.5 with {\arrow{Stealth[length=2.6mm]}}}]
          (12.4,1,0.75) -- (10,1,0.75);
    \fill[anyon] (10,1,0.75) circle (2.3pt);
    \draw[white,line width=0.4pt] (10,1,0.75) circle (2.3pt);

    \fill[slabface,opacity=0.10] (RE)--(RF)--(RG)--(RH)--cycle;
    \fill[slabface,opacity=0.14] (RA)--(RB)--(RF)--(RE)--cycle;
    \fill[slabface,opacity=0.18] (RB)--(RC)--(RG)--(RF)--cycle;

    \draw[vis] (RA)--(RB)--(RC); \draw[vis] (RA)--(RE);
    \draw[vis] (RB)--(RF); \draw[vis] (RC)--(RG);
    \draw[vis] (RE)--(RF)--(RG)--(RH)--(RE);

    \node[albl,anchor=west] at (12.55,1,0.95) {$(a,a)$};
    \node at (9.55,1,2.7) {$\CA_{\CL}$};
    \draw[->,slabedge,line width=0.7pt] (9.7,1,2.5) to[bend left=18] (10,1,1.2);
    \node[anchor=north] at (11.2,0,-0.55) {$\mathcal{C}\boxtimes\overline{\mathcal{C}}$};

  \end{tikzpicture}
  \caption{On the left, we have the 3-manifold, $\CM\cong\Sigma\times I$. The surface at $y=0$ acts on the simple line $a$ to produce $S_{CC}(a)=a^{\vee}$. To go to the righthand side, we fold across $S_{CC}$ and produce $\CC\boxtimes\CC^{\rm rev}$. The resulting gapped boundary at $y=0$ is described by the Lagrangian algebra, $\CA_{\CL}$. Only the lines in $\CA_{\CL}$ can end on the gapped boundary.}
  \label{fig:Fold}
\end{figure}
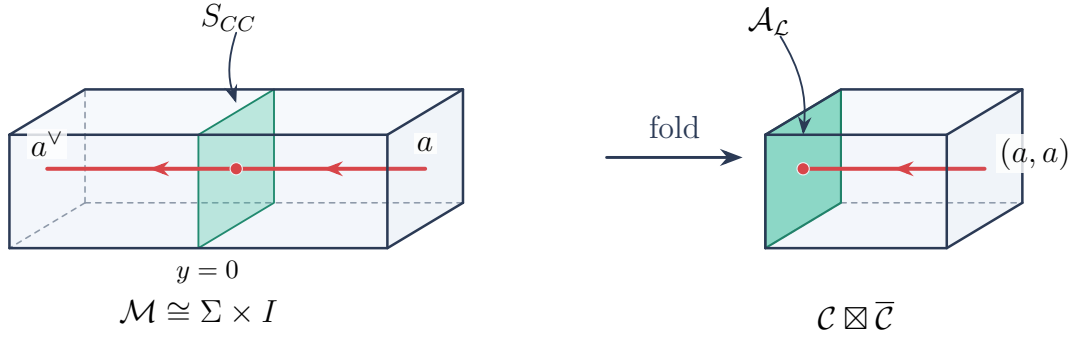

At $y=0\in I$ we place a surface, $S_{CC}$, implementing the charge conjugation symmetry. In other words, the action of this surface on lines / objects is
\begin{equation}
S_{CC}(a)=a^{\vee}~.
\end{equation}
Here $a^{\vee}$ is the dual or orientation reversal of $a$. For simple objects, $a\otimes a^{\vee}=1+\cdots$, where $1$ is the trivial object. The ellipses contain various potential non-trivial objects in the non-Abelian case. If the theory is self-dual, then $a=a^{\vee}$ for all simple objects, and $S_{CC}\cong S_1$ is the trivial surface. If the theory is not self-dual, we assume $S_{CC}$ to be a non-trivial topological surface satisfying $S_{CC}\boxtimes S_{CC}=S_1$ (we do not study cases that square to a non-trivial soft braided automorphism\footnote{Although this section applies also to non-Abelian theories, we note, for the orientation of the reader, that, in the case of Abelian theories, such non-trivial soft braided autoequivalences do not exist. The reason is that the action of soft-braided autoequivalences on $F$ force them to correspond to $U\in H^2(G,U(1))$ while the action of soft-braided autoequivalences on $R$ force $U$ to be trivial in $H^2(G,U(1))$. Here $G$ is the fusion group of the Abelian theory in question.}). In the language of \cite{Gaiotto:2014kfa}, $S_{CC}$ is co-dimension one and therefore implements a 0-form symmetry of $\CT$. In the language of category theory, $S_{CC}$ corresponds to a braided autoequivalence of $\CC$.\footnote{Physically, $S_{CC}$ is a condensation surface (obtained by summing over an algebra of lines in $\CT$ at $y=0$). Mathematically, it is an object in the 2-category ${\rm Mod}(\CC)$ (e.g., see the recent discussion in \cite{Kong:2024ykr,KNBalasubramanian:2025vum}).} This construction is described on the lefthand side of Fig. \ref{fig:Fold}.

In what follows, it is useful to \lq\lq fold" the theory at the location of the $S_{CC}$ surface (as on the righthand side of Fig. \ref{fig:Fold}). The reason for doing so is that the theory then admits a gapped boundary at $y=0$ is equivalent to $\CC$ as a fusion category (since the boundary is 1+1 dimensional, we do not have a notion of braiding on the boundary itself and we can \lq\lq forget" the braiding). Folding amounts to taking $y\to-y$ for points in $\CM$ with $y<0$ and producing the manifold $\CM'\cong\Sigma\times I'$, where $I'$ is the interval with $0\le y\le1$. At the level of the UMTC, $\CC$, this procedure stacks a copy of the $y\le0$ UMTC data but with its braiding reversed (geometrically, this reversal is due to the reflection across the $y=0$ $S_{CC}$ surface which swaps under and over crossings). The category corresponding to this reflected data is called $\CC^{\rm rev}$ and has the same underlying objects as $\CC$. Our discussion is summarized in the corresponding entry in Table \ref{CatST}.

The UMTC that results from the above procedure is often written as $\CC\boxtimes\CC^{\rm rev}$ (recall that the points on $\CM$ with $y\ge0$ are untouched and therefore furnish a factor of $\CC$). This UMTC is isomorphic to the so-called Drinfeld center of $\CC$, $\CZ(\CC)\cong\CC\boxtimes\CC^{\rm rev}$ (e.g., see \cite{etingof2015tensor}). In modern physics language, we have constructed the SymTFT for $\CC$ (see \cite{Schafer-Nameki:2023jdn} and references therein for a review of SymTFTs). Note that lines piercing $S_{CC}$ in the unfolded configuration of Fig. \ref{fig:Fold} are orientation reversed. Therefore, a line $a$ being mapped to $S_{CC}(a)=a^{\vee}$ on the lefthand side of Fig. \ref{fig:Fold} is mapped to a line $a\boxtimes a\in\CZ(\CC)$ on the righthand side of that figure.

For high-energy physicists, it may be more natural to think about an underlying action. For example, suppose that on $\CM$ we have the theory $\CT$ with
\begin{equation}
S_{\CT}[A]=\int_{\CM}d^3x\CL(A)~,
\end{equation}
where $A$ corresponds to a gauge field with, say, a Chern-Simons action. Then, folding and stacking amounts to doing the following
\begin{equation}\label{StackAct}
S_{\CT}[A]\longrightarrow S_{\CT\boxtimes\CT^{\rm rev}}[A,A']:=\int _{\CM'}\CL[A]-\int _{\CM'}\CL[A']~,
\end{equation}
where $\CL[A]$ and $\CL[A']$ are identical Lagrangians. The relative minus sign follows from the fact that we take $y\to-y$ and reverse orientation. In particular, the $A'$ gauge theory gives rise to $\CT^{\rm rev}$, which is the TQFT corresponding to $\CC^{\rm rev}$. Clearly, the line operators after performing this procedure take the form $a\boxtimes b$ where $a$ and $b$ are line operators built from holonomies of $A$ and $A'$ respectively.

To understand the relation between the $F$-symbols of $\CC$ and their complex conjugates, it is useful to study the corresponding Lagrangian algebra, $\CA_{\CL}$. This algebra has an underlying object\footnote{A useful recent physics reference for these algebras is \cite{Kaidi:2021gbs}. See also \cite{Kapustin:2010hk,Davydov:2010kfz}.
\label{LagRef}}
\begin{equation}\label{Lag}
\CL=\bigoplus_{a\in {\rm Irr}(\CC)}a\boxtimes a~,
\end{equation}
built from acting with the folding trick on all simple lines passing through $S_{CC}$. The lines appearing in \eqref{Lag} are precisely the lines that can end topologically on the resulting gapped boundary as on the righthand side of Fig. \ref{fig:Fold}. The Lagrangian algebra also has an associated multiplication with various properties that we will not describe in detail here (we instead refer the interested reader to the relevant literature in footnote \ref{LagRef}). This multiplication is necessary to construct a spacetime mesh from $\CL$ that gauges the corresponding (generally non-invertible) one-form symmetry and produces the gapped boundary. (This discussion explains why the lines in \eqref{Lag} can end on the gapped boundary.) Had the $S_{CC}$ surface not been present at $y=0$, folding across the trivial surface, $S_1$, would have instead produced a different algebra with a different underlying object
\begin{equation}\label{LagS1}
\CL'=\bigoplus_{a\in {\rm Irr}(\CC)}a\boxtimes a^\vee~.
\end{equation}

To make further progress, it is useful to analyse the symmetries of $\CZ(\CC)\simeq\CC\boxtimes\CC^{\rm rev}$. In particular, we claim it has an anti-unitary symmetry, $K$, that acts as follows
\begin{equation}\label{Kaction}
K:\ a\boxtimes b\ \longmapsto\ b^{\vee}\boxtimes a^{\vee}~.
\end{equation}
To understand this point note that, before folding, the reflection $y\to-y$ exchanges the lines in the region $y<0$ with those in the region $y>0$ and takes lines to duals. However, this action leaves the $y=0$ locus invariant and, crucially, it leaves the surface invariant since $S_{CC}$ (or $S_1$ in the case of no surface at $y=0$) is its own orientation reversal (it satisfies $S_{CC}\boxtimes S_{CC}= S_1$). This action emphatically does not generally give rise to a symmetry of the unfolded theory. However, after folding, this action descends to the symmetry $K$ that exchanges the two sheets (i.e., the degrees of freedom that were in the positive $y$ region with the degrees of freedom that were in the negative $y$ region before folding), applies a complex conjugation (due to the anti-unitary nature of the reflection), and takes objects to duals. In addition, $K$ fixes the boundary (due to the fact that $K$ preserves the surface in the unfolded picture). This statement is also clear from the invariance of the objects in \eqref{Lag} and \eqref{LagS1} under $K$.\footnote{Our logic also guarantees that the Lagrangian algebra multiplication is preserved as well.}

Let us examine what the above discussion tells us about the relation between $F$ and $\overline F$. We start by folding without $S_{CC}$ at $y=0$ so that we produce \eqref{LagS1}. The corresponding excitations living on the gapped boundary are $\CC$, and we have the following relation induced by $K$
\begin{eqnarray}\label{FrealS1}
\overline{\left[{F}^{a^\vee b^\vee c^\vee}_{d^\vee}\right]}_{(e^\vee,\alpha,\beta)(f^\vee,\mu,\nu)} 
&=& \sum_{\alpha',\beta',\mu',\nu'} 
\left[\Gamma^{ab}_{e}\right]_{\alpha\alpha'}
\left[\Gamma^{ec}_{d}\right]_{\beta\beta'}
\left[F^{abc}_{d}\right]_{(e,\alpha',\beta')(f,\mu',\nu')}\left[\left(\Gamma^{bc}_{f}\right)^{-1}\right]_{\mu'\mu}
\Big[\big(\Gamma^{af}_{d}\big)^{-1}\Big]_{\nu'\nu}~.\ \ \ \ \ \ 
\end{eqnarray}
In particular, before applying $K$, the boundary $F$-symbols are $F^{abc}_d$, and, after applying $K$, they become $\overline{F^{a^\vee b^\vee c^\vee}_{d^\vee}}$. Since $K$ is a symmetry, the boundary fusion category is mapped to itself, and the $F$-symbols before and after application of $K$ are related by a gauge transformation.

Now, consider taking a bulk $S_{CC}$ surface, initially at $y=y_0>0$, and fusing it with the gapped boundary corresponding to \eqref{LagS1} by taking $y_0\to0$. This maneuver produces a gapped boundary described by \eqref{Lag}, and it therefore modifies \eqref{FrealS1} to \eqref{Freal}, since $S_{CC}$ takes objects to duals. For ease of reference, we reproduce this latter equation below 
\begin{eqnarray}\label{Freal2}
\overline{\left[{F}^{a b c}_{d}\right]}_{(e,\alpha,\beta)(f,\mu,\nu)} 
&=& \sum_{\alpha',\beta',\mu',\nu'} 
\left[\Gamma^{ab}_{e}\right]_{\alpha\alpha'}
\left[\Gamma^{ec}_{d}\right]_{\beta\beta'}
\left[F^{abc}_{d}\right]_{(e,\alpha',\beta')(f,\mu',\nu')}\left[\left(\Gamma^{bc}_{f}\right)^{-1}\right]_{\mu'\mu}
\Big[\big(\Gamma^{af}_{d}\big)^{-1}\Big]_{\nu'\nu}\cr&=& \left[F^{abc}_{d}\right]_{(e,\alpha,\beta)(f,\mu,\nu)}^{\Gamma}~.
\end{eqnarray}

To reproduce the promised result in \eqref{gaugereality}, we note that when the theory has no fusion multiplicities, the above $\Gamma$ matrices are numbers, and we can take the square root
\begin{equation}
    \widehat{\Gamma}^{ab}_c:=\sqrt{\Gamma^{ab}_c}~,
\end{equation}
to arrive at a real gauge. Indeed, in this gauge, we have
\begin{equation}\label{GTFFstar}
    \left[F^{abc}_d\right]^{\widehat\Gamma}_{ef}:= \sqrt{\Gamma^{ab}_e}\sqrt{\Gamma^{ec}_{d}}\left[F^{abc}_d\right]_{ef}\sqrt{(\Gamma^{bc}_f)^{-1}}\sqrt{(\Gamma^{af}_d)^{-1}}~.
\end{equation}
Then it follows from \eqref{Freal} that \eqref{gaugereality} holds
\begin{equation}
\overline{\left[{F}^{abc}_d\right]^{\widehat\Gamma}}_{ef}=\left[F^{abc}_d\right]^{\widehat\Gamma}_{ef}\in\mathbb{R}~.
\label{gaugereality2}
\end{equation}
This is precisely the reality condition we were after, and it establishes Claim \ref{claim:trivMult} in the case of UMTCs (the discussion in Sec. \ref{sec:mechanism} implicitly establishes this claim for general unitary ribbon categories). Note that, in the case of higher fusion multiplicities, this process of choosing a real gauge is more complicated and involves Takagi factorization (see the discussion around \eqref{eq:localTakagi} for an explanation of Takagi factorization and the discussion in Sec. \ref{sec:sufficient} for an explanation of which conditions imply the existence of such a factorization).

Before closing, we note that our $K$ symmetry discussed above is in fact a braided equivalence between
\begin{equation}\label{Zequiv}
\CZ(\CC)\cong\CC\boxtimes\CC^{\rm rev}\leftrightarrow \overline\CC^{\rm rev}\boxtimes\overline\CC~.
\end{equation}
In Sec. \ref{DefsPrelim}, we will see that the action of $K$ can be read through the following braided equivalences
\begin{equation}\label{BraidedEquivs}
\overline\CC\leftrightarrow\CC^{\op}\leftrightarrow\CC^{\mop}\leftrightarrow\CC^{\rev}~,
\end{equation}
which act on objects by
\begin{equation}\label{BraidedEquivsObjects}
\overline\CC\ \underset{a\mapsto a}{\longleftrightarrow}\ \CC^{\op}\ \underset{a\mapsto a^{\vee}}{\longleftrightarrow}\ \CC^{\mop}\ \underset{a\mapsto a}{\longleftrightarrow}\ \CC^{\rev}~.
\end{equation}
Combining with the $S_{CC}$ charge conjugation is then necessary to have a map relating the categories in \eqref{Zequiv} that is the identity on objects. Note that we will use the equivalences in \eqref{BraidedEquivs} in our general proof in Sec. \ref{sec:sufficient}.

As a concluding comment, we note that the above folding construction is often written as $\CC\to\CC\boxtimes\overline\CC$ in the physics literature. Here, $\overline\CC$ is a UMTC with the same objects as $\CC$ but with all the $F$ and $R$ symbols complex conjugated. When the underlying QFT is a unitary TQFT (and $\CC$ is a UMTC), one is free to replace $\CC^{\rm rev}\leftrightarrow\overline\CC$ (after mapping objects to duals), since the effect on the path integral is the same: $Z(M)=\int DA e^{iS_{\CT}[A]}\to\int DA e^{-iS_{\CT}[A]}$. This is the physical avatar of the fact that we should get equivalent TQFTs/UMTCs from taking $x\to-x$ or $t\to-t$ instead of $y\to-y$ or complex conjugating our theory. After all, our choice of coordinates is arbitrary since the theory is Poincar\'e invariant. This discussion suggests that there should be a relation between the corresponding categorical operations $\CC\to\CC^{\rm mop}$ and $\CC\to\CC^{\rm op}$ described in Table \ref{CatST}. Indeed, in the next section, we will see that the freedom to perform coordinate transformations corresponds to the braided equivalences in \eqref{BraidedEquivs}.

\section{The mechanism for real F-symbols}
\label{sec:mechanism}
The purpose of this section is to isolate the general mechanism by which a real gauge for the $F$-symbols can arise, independently of the specific conditions that guarantee it. Throughout we work in the langauge of anyon theories and string diagrams. This section should be understood as a foundational section and warm-up for the main results in Sec. \ref{sec:sufficient} aimed at those readers more familiar with anyon models. Our results here generalize those in Sec. \ref{folding} in the sense that we do not assume modularity. We also work out explicitly what conditions are needed on gauge transformations to guarantee reality of the $F$-symbols in the higher-multiplicity case. However, we postpone a discussion of the categorical properties that deliver these conditions to section Sec. \ref{sec:sufficient}.

Assuming the existence of a braided charge-conjugation autoequivalence that sends anyons to their duals (i.e., $J(a)\cong a^\vee$), we first show that, in a unitary ribbon category, the complex-conjugated $F$-symbols are related to the original ones by a vertex basis gauge transformation (see also \eqref{Freal} for the special case of a UMTC)
\begin{equation}\label{FFbargauge}
    \overline{F_d^{abc}}
    =
    \mathsf M_L F_d^{abc}\mathsf M_R^{\dagger}~.
\end{equation}
We then show that, if the unitary matrices, $\mathsf M_a^{bc}$, are symmetric, this gauge transformation can be removed by a consistent change of trivalent fusion bases, and the corresponding $F$-symbols are real. For multiplicity-free fusion spaces this reduction simply amounts to taking square roots of phases (see \eqref{gaugereality} for a discussion of this special case), while for higher fusion multiplicities, it is implemented by Takagi factorisation.

Much of the remainder of the paper is concerned with determining when the $\mathsf M_a^{bc}$ matrices are symmetric, culminating in our results of Sec. \ref{sec:sufficient}. In later sections, we show examples of theories obeying and violating these conditions.

\subsection{Fusion-state conventions and the
\texorpdfstring{$F$}{F}-amplitude}
Throughout this section, $\cC$ is a unitary ribbon fusion category, and $\Irr(\cC)$ is a set, closed under duality, containing one representative of each isomorphism class of simple objects. For every allowed anyon fusion vertex, we consider the corresponding admissible triple of simple objects, $(a,b,c)$ and choose an ordered orthonormal basis
\begin{equation}
    \left\{
        \ket{a\to b\otimes c;\alpha}
    \right\}_{\alpha=1}^{N_{bc}^{a}}
    \subset
    V_a^{bc}~,
    \qquad
    V_a^{bc}:=
    \operatorname{Hom}_{\mathcal C}(a,b\otimes c)~.
\end{equation}
Diagrammatically, we represent the basis vector
$\ket{a\to b\otimes c;\alpha}$ by $\ket{
\begin{tikzpicture}[
    scale=0.35,
    baseline=-0.8ex,
    line cap=round,
    line join=round,
    >=stealth
]
% coordinates
\coordinate (O) at (0,0);
\coordinate (A) at (0,-1.3);
\coordinate (B) at (-1.1,1.1);
\coordinate (C) at (1.1,1.1);
% lines
\draw[thick] (A) -- (O);
\draw[thick] (O) -- (B);
\draw[thick] (O) -- (C);
% arrows in the middle
\draw[->,thick] (0,-0.85) -- (0,-0.45);
\draw[->,thick] (-0.45,0.45) -- (-0.7,0.7);
\draw[->,thick] (0.45,0.45) -- (0.7,0.7);
% labels
\node at (0.6,-0.8) {$a$};
\node at (-1.2,0.4) {$b$};
\node at (1.2,0.4) {$c$};
\node at (-0.55,-0.25) {$\alpha$};
\end{tikzpicture} 
}$. 
The dagger defines a conjugate-linear isometric identification
\begin{equation}
    V_a^{bc}
    \longrightarrow
    \operatorname{Hom}_{\mathcal C}(b\otimes c,a)~,
    \qquad
    \ket{a\to b\otimes c;\alpha}
    \longmapsto
    \ket{a\to b\otimes c;\alpha}^{\dagger}~.
\end{equation}
Using the corresponding Riesz identification with the Hilbert-space
dual, we define
\begin{equation}
    \bra{a\to b\otimes c;\alpha}
    :=
    \ket{a\to b\otimes c;\alpha}^{\dagger}~,
\end{equation}
denoted diagrammatically as
$\bra{
\begin{tikzpicture}[
    scale=0.35,
    baseline=-0.6ex,
    line cap=round,
    line join=round,
    >=stealth
]
% coordinates
\coordinate (O) at (0,0);
\coordinate (A) at (0,1.3);
\coordinate (B) at (-1.1,-1.1);
\coordinate (C) at (1.1,-1.1);
% lines
\draw[thick] (O) -- (A);
\draw[thick] (B) -- (O);
\draw[thick] (C) -- (O);
% arrows in the middle
\draw[->,thick] (0,0.45) -- (0,0.85);
\draw[->,thick] (-0.7,-0.7) -- (-0.45,-0.45);
\draw[->,thick] (0.7,-0.7) -- (0.45,-0.45);
% labels
\node at (0.6,0.9) {$a$};
\node at (-1.2,-0.3) {$b$};
\node at (1.2,-0.3) {$c$};
\node at (-0.55,0.15) {$\alpha$};
\end{tikzpicture}}$, so that $\langle a\to b\otimes c;\alpha| a\to b\otimes c;\alpha'\rangle = \delta_{\alpha\alpha'}$.

For fixed external labels $a$, $b$, $c$, and $d$, the two possible
ways of associating the three outgoing anyons give two distinct
channel decompositions of the four-point fusion Hilbert space.  We
denote them by
\begin{equation}
    \mathcal H_L^{abc;d}
    :=
    \bigoplus_e
    V_e^{ab}\otimes V_d^{ec},
    \qquad
    \mathcal H_R^{abc;d}
    :=
    \bigoplus_f
    V_f^{bc}\otimes V_d^{af}~.
\label{eq:HLHR-def}
\end{equation}
Equivalently,
$\mathcal H_L^{abc;d}
    \cong
    \operatorname{Hom}_{\mathcal C}
    \bigl(d,(a\otimes b)\otimes c\bigr)$ and
$\mathcal H_R^{abc;d}
    \cong
    \operatorname{Hom}_{\mathcal C} \bigl(d,a\otimes(b\otimes c)\bigr)$.
Thus $\mathcal H_L^{abc;d}$ and $\mathcal H_R^{abc;d}$ describe the
same physical four-anyon fusion space in the two different
bracketings.  They are canonically related by the associator.

The chosen orthonormal bases of the trivalent splitting spaces induce orthonormal fusion-tree bases
\begin{align}
    \ket{e,\gamma,\delta}_L
    &:=
    \bigl(
        \ket{e\to a\otimes b;\gamma}
        \otimes \id_c
    \bigr)
    \ket{d\to e\otimes c;\delta}~,
    \\
    \ket{f,\mu,\nu}_R
    &:=
    \bigl(
        \id_a\otimes
        \ket{f\to b\otimes c;\mu}
    \bigr)
    \ket{d\to a\otimes f;\nu}~.
\end{align}
Here $e$ and $f$ label the intermediate fusion channels, while
$\gamma,\delta,\mu,\nu$ label bases of the corresponding trivalent
multiplicity spaces.

With the associator convention used throughout this paper, let
\begin{equation}
    \alpha_{a,b,c}:
        \mathcal H_R^{abc;d}
    \longrightarrow
    \mathcal H_L^{abc;d}~,
\end{equation}
denote the canonical identification between the two bracketings. The $F$-matrix is the matrix representation of this associator in the two fusion-tree bases
\begin{equation}
\bigl[F_d^{abc}\bigr]_
    {(e,\gamma,\delta),(f,\mu,\nu)}
    =
        {}_{L}\!\braket{
        e,\gamma,\delta
        |
        \alpha_{a,b,c}
        |
        f,\mu,\nu
    }_R~.
\label{eq:F-as-associator}
\end{equation}
Equivalently,
\begin{equation}
    \alpha_{a,b,c}
    \ket{f,\mu,\nu}_{R}
    =
    \sum_{f,\mu,\nu}
    \bigl[F_d^{abc}\bigr]_
    {(e,\gamma,\delta),(f,\mu,\nu)}
    \ket{f,\mu,\nu}_{R}~.
\label{eq:F-basis-change}
\end{equation}
Since the category is unitary and the fusion-tree bases are orthonormal, $F_d^{abc}$ is a unitary matrix.

The $F$- and $R$-symbols depend on the chosen bases of the trivalent splitting spaces as discussed in \eqref{eq:gaugeTransformation} and \eqref{RGT}. Indeed, recall that a vertex basis gauge transformation is a choice of unitary $\Gamma_z^{xy}\in U(N_{xy}^{z})$ on each $V_z^{xy}$, under which
\begin{eqnarray}
\label{eq:gaugeTransformationAn}
\left[F^{abc}_{d}\right]^{\Gamma}_{(e,\alpha,\beta)(f,\mu,\nu)}
&:=& \sum_{\alpha',\beta',\mu',\nu'}
\left[\Gamma^{ab}_{e}\right]_{\alpha\alpha'}
\left[\Gamma^{ec}_{d}\right]_{\beta\beta'}
\left[F^{abc}_{d}\right]_{(e,\alpha',\beta')(f,\mu',\nu')}\left[\left(\Gamma^{bc}_{f}\right)^{-1}\right]_{\mu'\mu}
\Big[\left(\Gamma^{af}_{d}\right)^{-1}\Big]_{\nu'\nu}~,\ \ \ \ \ \
\end{eqnarray}
and, for the braiding,\footnote{Notice that our conventions for the $R$-symbols differ slightly from those in \cite{Barkeshli:2014cna}.}
\begin{align}\label{RGTAn}
\left[R^{ba}_{c}\right]^{\Gamma}_{\mu\nu}
&:= \sum_{\mu',\nu'}
\left[\Gamma^{ba}_c\right]_{\mu\mu'}
\left[R^{ba}_c\right]_{\mu'\nu'}
\left[\left(\Gamma^{ab}_c\right)^{-1}\right]_{\nu'\nu}~.
\end{align}

\subsection{Local transformations of a trivalent vertex}
We now introduce the local transformations that will be used in the
next subsection to relate the complex-conjugated $F$-symbol,
$\overline F$, to $F$.  We will define the action of these transformations on a trivalent splitting space and in the
chosen fusion bases.

\subsubsection{Hermitian conjugation}

We denote the adjoint fusion space by $V_{bc}^{a} := \operatorname{Hom}_{\mathcal C}(b\otimes c,a)$.
The dagger structure defines the conjugate-linear isometric map
\begin{equation}
    \mathsf D_a^{bc}:V_a^{bc}\longrightarrow V_{bc}^a,
    \qquad
    \mathsf D_a^{bc}(v)=v^\dagger~.
\end{equation}
Using the bra notation introduced above, for $\ket{\psi}
    =
    \sum_{\alpha}
    \psi_{\alpha}
    \ket{a\to b\otimes c;\alpha}$,
Hermitian conjugation acts as
\begin{equation}
    \mathsf D_a^{bc}\ket{\psi}
    =
    \sum_{\alpha}
    \overline{\psi}_{\alpha}\,
    \bra{a\to b\otimes c;\alpha}~.
\end{equation}
Diagrammatically, starting from $\ket{
\begin{tikzpicture}[
    scale=0.35,
    baseline=-0.8ex,
    line cap=round,
    line join=round,
    >=stealth
]
\coordinate (O) at (0,0);
\coordinate (A) at (0,-1.3);
\coordinate (B) at (-1.1,1.1);
\coordinate (C) at (1.1,1.1);

\draw[thick] (A) -- (O);
\draw[thick] (O) -- (B);
\draw[thick] (O) -- (C);

\draw[->,thick] (0,-0.85) -- (0,-0.45);
\draw[->,thick] (-0.45,0.45) -- (-0.7,0.7);
\draw[->,thick] (0.45,0.45) -- (0.7,0.7);

\node at (0.6,-0.8) {$a$};
\node at (-1.2,0.4) {$b$};
\node at (1.2,0.4) {$c$};
\node at (-0.55,-0.25) {$\alpha$};
\end{tikzpicture}
}
\in V_a^{bc}$,
the first operation is Hermitian conjugation
\begin{equation}
\ket{
\begin{tikzpicture}[
    scale=0.35,
    baseline=-0.8ex,
    line cap=round,
    line join=round,
    >=stealth
]
\coordinate (O) at (0,0);
\coordinate (A) at (0,-1.3);
\coordinate (B) at (-1.1,1.1);
\coordinate (C) at (1.1,1.1);

\draw[thick] (A) -- (O);
\draw[thick] (O) -- (B);
\draw[thick] (O) -- (C);

\draw[->,thick] (0,-0.85) -- (0,-0.45);
\draw[->,thick] (-0.45,0.45) -- (-0.7,0.7);
\draw[->,thick] (0.45,0.45) -- (0.7,0.7);

\node at (0.6,-0.8) {$a$};
\node at (-1.2,0.4) {$b$};
\node at (1.2,0.4) {$c$};
\node at (-0.55,-0.25) {$\alpha$};
\end{tikzpicture}
}
\xrightarrow{\ \mathsf D_a^{bc}\ }
\bra{
\begin{tikzpicture}[
    scale=0.35,
    baseline=-0.6ex,
    line cap=round,
    line join=round,
    >=stealth
]
\coordinate (O) at (0,0);
\coordinate (A) at (0,1.3);
\coordinate (B) at (-1.1,-1.1);
\coordinate (C) at (1.1,-1.1);

\draw[thick] (O) -- (A);
\draw[thick] (B) -- (O);
\draw[thick] (C) -- (O);

\draw[->,thick] (0,0.45) -- (0,0.85);
\draw[->,thick] (-0.7,-0.7) -- (-0.45,-0.45);
\draw[->,thick] (0.7,-0.7) -- (0.45,-0.45);

\node at (0.6,0.9) {$a$};
\node at (-1.2,-0.3) {$b$};
\node at (1.2,-0.3) {$c$};
\node at (-0.55,0.15) {$\alpha$};
\end{tikzpicture}
}
\in V_{bc}^{a}~.
\end{equation}
Thus, in the chosen basis, $\mathsf D_a^{bc}$ acts as complex
conjugation on the coefficient vector, together with the replacement
of each trivalent ket vertex by its adjoint bra vertex. This accompanying diagrammatic transformation will make the subsequent composition of these operations more transparent.

\subsubsection{Anticlockwise \texorpdfstring{$\pi$}{pi}-rotation}

Starting from an adjoint trivalent vertex, the counterclockwise half-turn
defines a linear map
\begin{equation}
    P_{\circlearrowleft,a}^{bc}:
    V_{bc}^{a}
    \longrightarrow
    V_{a^\vee}^{c^\vee b^\vee}~.
\end{equation}
Its action in the chosen bases is
\begin{equation}
    \bra{a\to b\otimes c;\beta}
    \longrightarrow
    \sum_{\gamma}
    \bigl(P_{\circlearrowleft,a}^{bc}\bigr)_{\beta\gamma}
    \ket{
        a^\vee\to c^\vee\otimes b^\vee;\gamma
    }~.
\end{equation}
Diagrammatically, the counterclockwise half-turn acts as
\begin{equation}
\bra{
\begin{tikzpicture}[
    scale=0.35,
    baseline=-0.8ex,
    line cap=round,
    line join=round,
    >=stealth
]
% coordinates
\coordinate (O) at (0,0);
\coordinate (A) at (0,-1.3);
\coordinate (B) at (-1.1,1.1);
\coordinate (C) at (1.1,1.1);
% lines
\draw[thick] (A) -- (O);
\draw[thick] (O) -- (B);
\draw[thick] (O) -- (C);
% arrows in the middle
\draw[->,thick] (0,-0.85) -- (0,-0.45);
\draw[->,thick] (-0.45,0.45) -- (-0.7,0.7);
\draw[->,thick] (0.45,0.45) -- (0.7,0.7);
% labels
\node at (0.6,-0.8) {$a$};
\node at (-1.2,0.4) {$b$};
\node at (1.2,0.4) {$c$};
\node at (-0.6,-0.3) {$\beta$};
\end{tikzpicture} 
}
\longrightarrow
\ket{
\begin{tikzpicture}[
    scale=0.21,
    baseline=-0.6ex,
    line width=0.8pt,
    line cap=round,
    line join=round,
    >=stealth,
    sstrand/.style={
        postaction={decorate},
        decoration={
            markings,
            mark=at position .45 with {\arrow{>}}
        }
    }
]

% Trivalent vertex
\coordinate (V) at (0,0);

% Left outgoing strand b
\draw[sstrand]
    (V)
    .. controls (-0.25,0.48) and (-0.55,1.42) ..
    (-1.12,1.46)
    .. controls (-1.75,1.50) and (-1.78,0.72) ..
    (-1.78,-2.35);

% Right outgoing strand c, continuing around the upper-left side
\draw[sstrand]
    (V)
    .. controls (0.27,0.50) and (0.62,0.65) ..
    (0.87,1.04)
    .. controls (1.35,1.82) and (1.02,2.55) ..
    (0.18,2.80)
    .. controls (-1.12,3.20) and (-2.72,2.72) ..
    (-2.72,0.95)
    -- (-2.72,-2.35);

% Incoming strand a, looping around the lower-right side
\draw[sstrand]
    (2.05,2.18)
    -- (2.05,-1.55)
    .. controls (2.05,-2.35) and (1.48,-2.72) ..
    (0.82,-2.72)
    .. controls (0.22,-2.72) and (0.03,-2.20) ..
    (0.03,-1.42)
    -- (V);

% Labels
\node at (-1,0.5) {$b$};
\node at (-.8,3.75) {$c$};
\node at (2.95,-1.5) {$a$};
\node at (-0.7,-1.15) {$\beta$};

\end{tikzpicture}}
=
\sum_{\beta}
\bigl(P_{\circlearrowleft,a}^{bc}\bigr)_{\beta\gamma}
\ket{
\begin{tikzpicture}[
    scale=0.35,
    baseline=-0.6ex,
    line cap=round,
    line join=round,
    >=stealth
]
% coordinates
\coordinate (O) at (0,0);
\coordinate (A) at (0,1.3);
\coordinate (B) at (-1.1,-1.1);
\coordinate (C) at (1.1,-1.1);
% lines
\draw[thick] (O) -- (A);
\draw[thick] (B) -- (O);
\draw[thick] (C) -- (O);
% arrows in the middle
\draw[->,thick] (0,0.45) -- (0,0.85);
\draw[->,thick] (-0.7,-0.7) -- (-0.45,-0.45);
\draw[->,thick] (0.7,-0.7) -- (0.45,-0.45);
% labels
\node at (1,1.1) {$a^\vee$};
\node at (-1.6,-0.5) {$c^\vee$};
\node at (1.8,-0.5) {$b^\vee$};
\node at (-0.6,0.2) {$\gamma$};
\end{tikzpicture}
}~.
\end{equation}
In the curved diagram, we retain the original labels $a,b,c$ and
encode dualisation by the orientation of the bent strands.  When the
rotated diagram is expanded in the standard fusion basis, all strands
are returned to the standard orientation and an orientation-reversed
strand labelled $x$ is consequently represented by a strand labelled
$x^\vee$. The reversal of the order of the two outgoing labels follows
from the canonical duality identification
\begin{equation}
    (b\otimes c)^\vee
    \cong
    c^\vee\otimes b^\vee~.
\end{equation}
Finally, the matrix $P_{\circlearrowleft,a}^{bc}$ is unitary with respect to
the corresponding spherical inner products.

\subsubsection{Braiding}
\label{sssec:diagrammaticBraidingMatrix}
The categorical braiding $\beta_{b,c}:b\otimes c\longrightarrow c\otimes b$ induces a unitary map between the corresponding splitting spaces,
\begin{equation}
    V_a^{bc}\longrightarrow V_a^{cb}~,
\end{equation}
where $V_a^{bc}
    =
    \operatorname{Hom}_{\mathcal C}(a,b\otimes c)$.
Thus, in the chosen orthonormal bases,
\begin{equation}
    \ket{a\to b\otimes c;\mu}
    \longrightarrow
    \sum_{\nu=1}^{N_{cb}^{a}}
    \bigl(R_a^{bc}\bigr)_{\mu\nu}
    \ket{a\to c\otimes b;\nu}~,
\end{equation}
denoted diagrammatically as (see also Fig. \ref{fig:fr-moves}; note that our conventions for the $R$-matrix differ slightly from those in \cite{Barkeshli:2014cna})
\begin{equation}\ket{\begin{tikzpicture}[
    baseline={(current bounding box.center)},
    scale=0.33,
        baseline=0.3ex,
    line width=0.8pt,
    >={Stealth[length=3pt,width=4pt]},
%    >={Latex[length=2.6mm,width=2.1mm]},
    mid/.style={
        postaction={decorate},
        decoration={
            markings,
            mark=at position #1 with {\arrow{>}}
        }
    },
    mid/.default=0.78,
    every node/.style={inner sep=1pt}
  ]
    \coordinate (Rc)  at (1.0,-1.0);
    \coordinate (Rm)  at (1.0,0.0);
    \coordinate (Rla) at (-0.05,1.7);
    \coordinate (Rlb) at (2.05,1.7);

    \draw[mid=0.85]
      (Rm) .. controls (1.7,0.45) and (0.5,1.05) .. (Rla);

    \draw[
      mid=0.85,
      preaction={draw=white,line width=4pt}
    ]
      (Rm) .. controls (0.3,0.45) and (1.5,1.05) .. (Rlb);

    \draw[mid=0.85] (Rc)--(Rm);

\node at (1.7,-0.58) {$a$};
\node at (-0.5,1.20) {$c$};
\node at (2.5,1.20) {$b$};
\node at (0.2,0) {$\mu$};
\end{tikzpicture}}
= \sum_{\nu=1}^{N_{cb}^{a}}
    \bigl(R_a^{bc}\bigr)_{\mu\nu} 
\ket{
\begin{tikzpicture}[
    scale=0.35,
    baseline=-0.8ex,
    line cap=round,
    line join=round,
    >=stealth
]
% coordinates
\coordinate (O) at (0,0);
\coordinate (A) at (0,-1.3);
\coordinate (B) at (-1.1,1.1);
\coordinate (C) at (1.1,1.1);
% lines
\draw[thick] (A) -- (O);
\draw[thick] (O) -- (B);
\draw[thick] (O) -- (C);
% arrows in the middle
\draw[->,thick] (0,-0.85) -- (0,-0.45);
\draw[->,thick] (-0.45,0.45) -- (-0.7,0.7);
\draw[->,thick] (0.45,0.45) -- (0.7,0.7);
% labels
\node at (0.6,-0.8) {$a$};
\node at (-1.2,0.4) {$c$};
\node at (1.2,0.4) {$b$};
\node at (-0.55,-0.25) {$\nu$};
\end{tikzpicture} 
}~.
\end{equation}
The matrix $R_a^{bc}$ is therefore the matrix of the braiding map
with respect to a basis of $V_a^{bc}$ in the domain and a basis of
$V_a^{cb}$ in the codomain.  Only after choosing an additional
identification $V_a^{cb}\cong V_a^{bc}$ may this local braiding map
be regarded as an endomorphism of a single splitting space.  In
representations of the braid group on a fixed multi-anyon Hilbert
space, such an identification of the final and initial puncture
orderings is usually implicit.
The exchange with the opposite crossing is induced by $\beta_{c,b}^{-1}:b\otimes c\longrightarrow c\otimes b$
and is represented by
$\bigl(R_a^{cb}\bigr)^{-1}$.

\subsubsection{Charge conjugation}
Let
\begin{equation}
    J:\mathcal C\longrightarrow\mathcal C~,
\end{equation}
be a unitary braided monoidal charge-conjugation autoequivalence,
mapping each simple object to its dual,
\begin{equation}
    J(a)\cong a^\vee~.
\end{equation}
Its action on the splitting space $V_a^{bc}$ is represented in the
chosen bases by a unitary matrix $U_a^{bc}$:
\begin{equation}
    J\ket{a\to b\otimes c;\alpha}
    =
    \sum_{\beta}
    \bigl(U_a^{bc}\bigr)_{\alpha\beta}
    \ket{a^\vee\to b^\vee\otimes c^\vee;\beta}~.
\end{equation}
Here the tensorator of $J$ and the chosen identifications
$J(x)\cong x^\vee$ are absorbed into the definition of $U_a^{bc}$.
Diagrammatically, charge conjugation preserves the standard orientation of the splitting vertex while replacing each object by
its charge conjugate:
\[
J\ket{
\begin{tikzpicture}[
    scale=0.35,
    baseline=-0.8ex,
    line cap=round,
    line join=round,
    >=stealth
]
\coordinate (O) at (0,0);
\coordinate (A) at (0,-1.3);
\coordinate (B) at (-1.1,1.1);
\coordinate (C) at (1.1,1.1);

\draw[thick] (A) -- (O);
\draw[thick] (O) -- (B);
\draw[thick] (O) -- (C);

\draw[->,thick] (0,-0.85) -- (0,-0.45);
\draw[->,thick] (-0.45,0.45) -- (-0.7,0.7);
\draw[->,thick] (0.45,0.45) -- (0.7,0.7);

\node at (0.6,-0.8) {$a$};
\node at (-1.2,0.4) {$b$};
\node at (1.2,0.4) {$c$};
\node at (-0.55,-0.25) {$\alpha$};
\end{tikzpicture}
}
=
\sum_{\beta}
\bigl(U_a^{bc}\bigr)_{\alpha\beta}
\ket{
\begin{tikzpicture}[
    scale=0.35,
    baseline=-0.8ex,
    line cap=round,
    line join=round,
    >=stealth
]
\coordinate (O) at (0,0);
\coordinate (A) at (0,-1.3);
\coordinate (B) at (-1.1,1.1);
\coordinate (C) at (1.1,1.1);

\draw[thick] (A) -- (O);
\draw[thick] (O) -- (B);
\draw[thick] (O) -- (C);

\draw[->,thick] (0,-0.85) -- (0,-0.45);
\draw[->,thick] (-0.45,0.45) -- (-0.7,0.7);
\draw[->,thick] (0.45,0.45) -- (0.7,0.7);

\node at (0.75,-0.8) {$a^\vee$};
\node at (-1.5,0.4) {$b^\vee$};
\node at (1.35,0.4) {$c^\vee$};
\node at (-0.7,-0.25) {$\beta$};
\end{tikzpicture}
}~.
\]

\subsubsection{Pivotal identification}
The charge-conjugation step applied to a vertex carrying dual labels
produces double-dual objects.  We identify these with the original
objects using the chosen pivotal structure
\begin{equation}
    \mathsf p:
    (-)^{\vee\vee}
    \Longrightarrow
    \operatorname{Id}_{\mathcal C}~,
    \qquad
    p_x:x^{\vee\vee}\longrightarrow x~,
\end{equation}
which a ribbon category already carries.  We take $\Irr(\cC)$ to be a skeleton in which $x^{\vee\vee}=x$ exactly so that each $p_x$ is an endomorphism of a simple object. Since $\End(x)=\mathbb C\,\id_x$, we can take $p_x$ to be a number. Unitarity then implies this number is a phase. As we will discuss in more detail in the later category theory sections, in the case of self-dual theories, we have $p_x=\varkappa_x$ (i.e., $p_x$ is the second Frobenius-Schur indicator).

The pivotal structure induces the following linear map on vertices
\begin{equation}
    \label{eq:mechanismPivotalCoboundary}
    \mathsf p_a^{bc}:
    V_{a^{\vee\vee}}^{b^{\vee\vee}c^{\vee\vee}}
    \longrightarrow
    V_a^{bc}~,
\end{equation}
defined by
\begin{equation}
    \mathsf p_a^{bc}(v)
    :=
    (p_b\otimes p_c)
    \circ v\circ p_a^{-1}~,
\end{equation}
and act as
\begin{equation}
    \ket{a^{
    \vee\vee}\to b^{
    \vee\vee}\otimes c^{
    \vee\vee};\alpha}
    =
    \sum_{\beta}
    [\mathsf p^{ab}_c]_{\alpha\beta}
    \ket{a\to b\otimes c;\beta}~.
\end{equation}
Note that we explicitly have $[\mathsf p^{ab}_c]_{\alpha\beta}= (p_b p _c/p_a)\delta_{\alpha\beta}$.

\subsubsection{Composition of transformations}
\label{sssec:mechanismR}
Having defined these local maps, we combine them into a single
antiunitary endomorphism of each splitting space. With the orientation conventions introduced above, their successive action is
\begin{equation}
\begin{aligned}
    V_a^{bc}
    &\xrightarrow{\ \mathsf D_a^{bc}\ }
    V_{bc}^{a}
    \xrightarrow{\ P_{\circlearrowleft,a}^{bc}\ }
    V_{a^\vee}^{c^\vee b^\vee}
    \xrightarrow{\ R_{a^\vee}^{c^\vee b^\vee}\ }
    V_{a^\vee}^{b^\vee c^\vee}
    \xrightarrow{\ U_{a^\vee}^{b^\vee c^\vee}\ }
    V_{a^{\vee\vee}}^{b^{\vee\vee}c^{\vee\vee}}
    \xrightarrow{\ \mathsf p_a^{bc}\ }
    V_a^{bc}~.
\end{aligned}
\label{eq:A-local-sequence}
\end{equation}
Here charge conjugation sends
$J(x^\vee)\cong x^{\vee\vee}$, while the final identification of the
double-dual objects with the original ones is supplied by the chosen
pivotal structure.

We therefore define
\begin{equation}
    A_a^{bc}
    :=
    \mathsf p_a^{bc}
    \circ
    J_{a^\vee}^{b^\vee c^\vee}
    \circ
    R_{a^\vee}^{c^\vee b^\vee}
    \circ
    P_{\circlearrowleft,a}^{bc}
    \circ
    \mathsf D_a^{bc}
    :
    V_a^{bc}\longrightarrow V_a^{bc}~.
\label{eq:A-local-definition}
\end{equation}
The map $A_a^{bc}$ is antiunitary because
$\mathsf D_a^{bc}$ is conjugate-linear, whereas the half-turn,
braiding, charge-conjugation and pivotal maps are linear and unitary. In the basis chosen above,
\begin{equation}
    A_a^{bc}
    =
    \mathsf M_a^{bc}K~,
    \qquad
    \mathsf M_a^{bc}
    = {\mathsf p_a^{bc}}
    U_{a^\vee}^{b^\vee c^\vee}
    R_{a^\vee}^{c^\vee b^\vee}
    P_{\circlearrowleft,a}^{bc} ~,
\label{eq:A-MK}
\end{equation}
where $K$ denotes complex conjugation of the coefficients with respect to a chosen basis and $\mathsf M_a^{bc}$ is unitary.

At this point we make no assumption about the square of $A_a^{bc}$.
The role of $A_a^{bc}$ here is simply to package the one-pass
diagrammatic transformation that takes the complex-conjugated
$F$-symbol network back to the original fusion-tree structure.  We
now apply this transformation to the two fusion-tree bases to obtain
the gauge relation between $\overline F$ and $F$.

\subsection{Application to the \texorpdfstring{$F$}{F}-symbol amplitude}

We now apply the local antiunitary transformation introduced above to
the closed diagram representing an $F$-symbol matrix element.
With
the orientation and multiplicity conventions of the preceding
subsection, the diagram displayed below represents
\begin{equation}
\bigl[F_d^{abc}\bigr]_{(i,\nu,\mu),(j,\alpha,\beta)}
=
\Biggl\langle
\begin{tikzpicture}[
    scale=0.95,
    baseline=4.5ex,
    line cap=round,
    line join=round,
    >={Stealth[length=3pt,width=4pt]},
]

\tikzset{
    midarrow/.style={
        thick,
        postaction={decorate},
        decoration={
            markings,
            mark=at position 0.53 with {\arrow{>}}
        }
    },
    midarrowrev/.style={
        thick,
        postaction={decorate},
        decoration={
            markings,
            mark=at position 0.53 with {\arrow{>}}
        }
    }
}

\begin{scope}[rotate around={180:(0,1)}]

% vertices
\coordinate (beta)  at (0,2.00);
\coordinate (alpha) at (-0.42,1.00);
\coordinate (nu)    at (0.42,1.00);
\coordinate (mu)    at (0,0);

% outer curve d, moved to the other side by sphericality
\draw[midarrow] (mu) .. controls (1.5,0.42) and (1.5,1.58) .. (beta);

% inner curved edges
\draw[midarrow] (beta) .. controls (-0.20,1.82) and (-0.44,1.58) .. (alpha); % i
\draw[midarrow] (beta) .. controls (0.20,1.82) and (0.44,1.58) .. (nu);      % c
\draw[midarrow] (alpha) .. controls (-0.44,0.42) and (-0.20,0.18) .. (mu);   % a
\draw[midarrow] (nu)    .. controls (0.44,0.42) and (0.20,0.18) .. (mu);     % j

% short central arc b with reversed arrow
\draw[midarrow] (alpha) .. controls (-0.10,1.17) and (0.10,1.17) .. (nu);

% dots
\fill (beta)  circle (1.8pt);
\fill (alpha) circle (1.8pt);
\fill (nu)    circle (1.8pt);
\fill (mu)    circle (1.8pt);

% vertex labels
\node[above=-18pt] at (beta) {$\mu$};
\node[left=-15pt] at (alpha) {$\nu$};
\node[right=-15pt] at (nu) {$\alpha$};
\node[below=-18pt] at (mu) {$\beta$};

% edge labels
\node at (1.4,1.00) {$d$};
\node at (-0.6,1.50) {$i$};
\node at (0.6,1.50) {$c$};
\node at (-0.55,0.50) {$a$};
\node at (0,1.4) {$b$};
\node at (0.6,0.50) {$j$};

\end{scope}
\end{tikzpicture} 
\Biggr\rangle~.
\label{eq:F-closed-amplitude}
\end{equation}
Here $i$ and $j$ are the intermediate channels of the left- and
right-associated fusion trees, respectively.  More precisely,
$\nu$ labels $V_i^{ab}$ and $\mu$ labels $V_d^{ic}$, while
$\alpha$ labels $V_j^{bc}$ and $\beta$ labels $V_d^{aj}$.
The brackets in Eq.~\eqref{eq:F-closed-amplitude} denote the
normalised spherical evaluation of the closed network and should not
be confused with a Hilbert-space bra.

We now follow the transformation of this amplitude step by step.
Hermitian conjugation first produces its complex conjugate.  The
remaining operations, anticlockwise half-turn, braiding, charge
conjugation and pivotal identification, are linear and return the
reflected network to the same object labels and fusion-tree topology
as the original one.  During this process they induce unitary changes
of basis on the trivalent multiplicity spaces.

To keep the graphical equations readable, we denote by
$\mathcal P$, $\mathcal R$ and $\mathcal U$ the accumulated actions
of the half-turn, braiding and charge-conjugation transformations on
the four trivalent multiplicity spaces.  When fusion multiplicities
are present these symbols represent matrices together with the
appropriate contractions over intermediate multiplicity indices, not
independent scalar factors.  The precise matrix relation will be
given at the end of the subsection.

\begin{itemize}

\item
\textbf{Hermitian conjugation.}

We first take the Hermitian conjugate of the closed diagram in
Eq.~\eqref{eq:F-closed-amplitude}.  Since the spherical evaluation of
a closed network is a complex scalar, this gives
\begin{equation}
\overline{
\bigl[F_d^{abc}\bigr]}_{(i,\nu,\mu),(j,\alpha,\beta)
}
=
\Biggl\langle
\begin{tikzpicture}[
    scale=0.95,
    baseline=4.5ex,
    line cap=round,
    line join=round,
    >={Stealth[length=3pt,width=4pt]},
]

\tikzset{
    midarrow/.style={
        thick,
        postaction={decorate},
        decoration={
            markings,
            mark=at position 0.53 with {\arrow{<}}
        }
    },
    midarrowrev/.style={
        thick,
        postaction={decorate},
        decoration={
            markings,
            mark=at position 0.53 with {\arrow{>}}
        }
    }
}

% vertices (tight, symmetric layout)
\coordinate (beta)  at (0,2.00);
\coordinate (alpha) at (-0.42,1.00);
\coordinate (nu)    at (0.42,1.00);
\coordinate (mu)    at (0,0);

% outer left curve f
\draw[midarrow] (mu) .. controls (-1.5,0.42) and (-1.5,1.58) .. (beta);

% inner curved edges (curving outward)
\draw[midarrow] (beta) .. controls (-0.20,1.82) and (-0.44,1.58) .. (alpha); % c
\draw[midarrow] (beta) .. controls (0.20,1.82) and (0.44,1.58) .. (nu);      % i
\draw[midarrow] (alpha) .. controls (-0.44,0.42) and (-0.20,0.18) .. (mu);   % j
\draw[midarrow] (nu)    .. controls (0.44,0.42) and (0.20,0.18) .. (mu);     % a

% central b-line only: reversed arrow
\draw[midarrowrev] (alpha) .. controls (-0.10,1.17) and (0.10,1.17) .. (nu);

% dots
\fill (beta)  circle (1.8pt);
\fill (alpha) circle (1.8pt);
\fill (nu)    circle (1.8pt);
\fill (mu)    circle (1.8pt);

% vertex labels
\node[above=2pt] at (beta) {$\mu$};
\node[left=-2pt] at (alpha) {$\alpha$};
\node[right=0pt] at (nu) {$\nu$};
\node[below=2pt] at (mu) {$\beta$};

% edge labels
\node at (-1.4,1.00) {$d$};
\node at (-0.6,1.50) {$c$};
\node at (0.6,1.50) {$i$};
\node at (-0.55,0.50) {$j$};
\node at (0,1.4) {$b$};
\node at (0.6,0.50) {$a$};

\end{tikzpicture}
\Biggl\rangle~.
\end{equation}
Diagrammatically, Hermitian conjugation reflects the entire oriented
network and replaces every trivalent morphism by its adjoint.  The
orientation arrows are carried along by the reflection, while the
object labels themselves remain unchanged.  In particular, no
dualisation of the labels is performed at this stage.

The reflection also interchanges the relative positions of the
vertices of the network.  Any apparent change in the direction of an
individual arrow in the planar drawing is therefore simply the image
of that oriented strand under the reflection, rather than an
additional dualisation operation.

\item
\textbf{Anticlockwise $\pi$-rotation.}

We next apply an anticlockwise half-turn to each of the four trivalent vertices of the conjugated network. For the closed
$F$-symbol diagram we immediately express the rotated vertices in the
standard splitting-space bases.  Consequently, all strands are
returned to the standard graphical orientation and a physical strand
whose orientation has been reversed by the half-turn is represented
by the corresponding dual object $x^\vee$.

The geometrical transformation is

\begin{equation}
    \Biggl\langle
\begin{tikzpicture}[
    scale=0.95,
    baseline=4.5ex,
    line cap=round,
    line join=round,
    >={Stealth[length=3pt,width=4pt]},
]

\tikzset{
    midarrow/.style={
        thick,
        postaction={decorate},
        decoration={
            markings,
            mark=at position 0.53 with {\arrow{<}}
        }
    },
    midarrowrev/.style={
        thick,
        postaction={decorate},
        decoration={
            markings,
            mark=at position 0.53 with {\arrow{>}}
        }
    }
}

% vertices (tight, symmetric layout)
\coordinate (beta)  at (0,2.00);
\coordinate (alpha) at (-0.42,1.00);
\coordinate (nu)    at (0.42,1.00);
\coordinate (mu)    at (0,0);

% outer left curve f
\draw[midarrow] (mu) .. controls (-1.5,0.42) and (-1.5,1.58) .. (beta);

% inner curved edges (curving outward)
\draw[midarrow] (beta) .. controls (-0.20,1.82) and (-0.44,1.58) .. (alpha); % c
\draw[midarrow] (beta) .. controls (0.20,1.82) and (0.44,1.58) .. (nu);      % i
\draw[midarrow] (alpha) .. controls (-0.44,0.42) and (-0.20,0.18) .. (mu);   % j
\draw[midarrow] (nu)    .. controls (0.44,0.42) and (0.20,0.18) .. (mu);     % a

% central b-line only: reversed arrow
\draw[midarrowrev] (alpha) .. controls (-0.10,1.17) and (0.10,1.17) .. (nu);

% dots
\fill (beta)  circle (1.8pt);
\fill (alpha) circle (1.8pt);
\fill (nu)    circle (1.8pt);
\fill (mu)    circle (1.8pt);

% vertex labels
\node[above=2pt] at (beta) {$\mu$};
\node[left=-2pt] at (alpha) {$\alpha$};
\node[right=0pt] at (nu) {$\nu$};
\node[below=2pt] at (mu) {$\beta$};

% edge labels
\node at (-1.4,1.00) {$d$};
\node at (-0.6,1.50) {$c$};
\node at (0.6,1.50) {$i$};
\node at (-0.55,0.50) {$j$};
\node at (0,1.4) {$b$};
\node at (0.6,0.50) {$a$};

\end{tikzpicture}
    \Biggr\rangle
    \xrightarrow{\;\mathcal P\;}
    \Biggl\langle
\begin{tikzpicture}[
    scale=0.95,
    baseline=4.5ex,
    line cap=round,
    line join=round,
    >={Stealth[length=3pt,width=4pt]},
]

\tikzset{
    midarrow/.style={
        thick,
        postaction={decorate},
        decoration={
            markings,
            mark=at position 0.53 with {\arrow{>}}
        }
    },
    midarrowrev/.style={
        thick,
        postaction={decorate},
        decoration={
            markings,
            mark=at position 0.53 with {\arrow{<}}
        }
    }
}

\begin{scope}[rotate around={180:(0,1)}]

% vertices
\coordinate (beta)  at (0,2.00);
\coordinate (alpha) at (-0.42,1.00);
\coordinate (nu)    at (0.42,1.00);
\coordinate (mu)    at (0,0);

% outer left curve f
\draw[midarrow] (mu) .. controls (-1.5,0.42) and (-1.5,1.58) .. (beta);

% inner curved edges
\draw[midarrow] (beta) .. controls (-0.20,1.82) and (-0.44,1.58) .. (alpha); % b
\draw[midarrow] (beta) .. controls (0.20,1.82) and (0.44,1.58) .. (nu);      % c
\draw[midarrow] (alpha) .. controls (-0.44,0.42) and (-0.20,0.18) .. (mu);   % a
\draw[midarrow] (nu)    .. controls (0.44,0.42) and (0.20,0.18) .. (mu);     % e

% short central arc b (reversed orientation)
\draw[midarrowrev] (alpha) .. controls (-0.10,1.17) and (0.10,1.17) .. (nu);

% dots
\fill (beta)  circle (1.8pt);
\fill (alpha) circle (1.8pt);
\fill (nu)    circle (1.8pt);
\fill (mu)    circle (1.8pt);

% vertex labels
\node[above=-18pt] at (beta) {$\mu$};
\node[left=-15pt] at (alpha) {$\alpha$};
\node[right=-15pt] at (nu) {$\nu$};
\node[below=-18pt] at (mu) {$\beta$};

% edge labels
\node at (-1.45,1.00) {$d^\vee$};
\node at (-0.6,1.50) {$c^\vee$};
\node at (0.6,1.50) {$i^\vee$};
\node at (-0.65,0.65) {$j^\vee$};
\node at (0,1.4) {$b^\vee$};
\node at (0.6,0.50) {$a^\vee$};

\end{scope}
\end{tikzpicture}
    \Biggr\rangle   \overset{\mathrm{sph}}{\sim}
    \Biggl\langle
\begin{tikzpicture}[
    scale=0.95,
    baseline=4.5ex,
    line cap=round,
    line join=round,
    >={Stealth[length=3pt,width=4pt]},
]

\tikzset{
    midarrow/.style={
        thick,
        postaction={decorate},
        decoration={
            markings,
            mark=at position 0.53 with {\arrow{>}}
        }
    }
}

\begin{scope}[rotate around={180:(0,1)}]

% vertices
\coordinate (beta)  at (0,2.00);
\coordinate (alpha) at (-0.42,1.00);
\coordinate (nu)    at (0.42,1.00);
\coordinate (mu)    at (0,0);

% outer curve d, moved to the other side by sphericality
\draw[midarrow] (mu) .. controls (1.5,0.42) and (1.5,1.58) .. (beta);

% inner curved edges
\draw[midarrow] (beta) .. controls (-0.20,1.82) and (-0.44,1.58) .. (alpha); % b
\draw[midarrow] (beta) .. controls (0.20,1.82) and (0.44,1.58) .. (nu);      % c
\draw[midarrow] (alpha) .. controls (-0.44,0.42) and (-0.20,0.18) .. (mu);   % a
\draw[midarrow] (nu)    .. controls (0.44,0.42) and (0.20,0.18) .. (mu);     % e

% short central arc d
\draw[midarrow] (nu) .. controls (0.10,1.17) and (-0.10,1.17) .. (alpha);

% dots
\fill (beta)  circle (1.8pt);
\fill (alpha) circle (1.8pt);
\fill (nu)    circle (1.8pt);
\fill (mu)    circle (1.8pt);

% vertex labels
\node[above=-18pt] at (beta) {$\mu$};
\node[left=-15pt] at (alpha) {$\alpha$};
\node[right=-15pt] at (nu) {$\nu$};
\node[below=-18pt] at (mu) {$\beta$};

% edge labels
\node at (1.4,1.00) {$d^\vee$};
\node at (-0.6,1.50) {$c^\vee$};
\node at (0.6,1.45) {$i^\vee$};
\node at (-0.68,0.50) {$j^\vee$};
\node at (0,1.4) {$b^\vee$};
\node at (0.6,0.50) {$a^\vee$};

\end{scope}
\end{tikzpicture}
    \Biggr\rangle~.
\label{eq:F-half-turn}
\end{equation}

When the rotated vertices are expanded in the chosen standard bases, the corresponding
changes of multiplicity basis are represented by the matrices
$P_{\circlearrowleft}$ introduced above and their accumulated action is
denoted by $\mathcal P$.

The final equivalence in Eq.~\eqref{eq:F-half-turn} is a spherical
isotopy moving the external $d^\vee$-strand through the point at
infinity.  This deformation changes neither the object labels nor the
orientations of the physical strands and leaves the spherical
evaluation unchanged.

\item
\textbf{Braiding.}

The anticlockwise half-turn reverses the local ordering of the two tensor
factors at the transformed trivalent vertices.  We therefore apply
the corresponding braiding maps to restore the ordering appropriate
to the original $F$-symbol network.  Their action on the multiplicity
spaces is represented by the appropriate $R$-matrices.

Diagrammatically,

\begin{equation}
    \Biggl\langle
\begin{tikzpicture}[
    scale=0.95,
    baseline=4.5ex,
    line cap=round,
    line join=round,
    >={Stealth[length=3pt,width=4pt]},
]

\tikzset{
    midarrow/.style={
        thick,
        postaction={decorate},
        decoration={
            markings,
            mark=at position 0.53 with {\arrow{>}}
        }
    }
}

\begin{scope}[rotate around={180:(0,1)}]

% vertices
\coordinate (beta)  at (0,2.00);
\coordinate (alpha) at (-0.42,1.00);
\coordinate (nu)    at (0.42,1.00);
\coordinate (mu)    at (0,0);

% outer curve d, moved to the other side by sphericality
\draw[midarrow] (mu) .. controls (1.5,0.42) and (1.5,1.58) .. (beta);

% inner curved edges
\draw[midarrow] (beta) .. controls (-0.20,1.82) and (-0.44,1.58) .. (alpha); % b
\draw[midarrow] (beta) .. controls (0.20,1.82) and (0.44,1.58) .. (nu);      % c
\draw[midarrow] (alpha) .. controls (-0.44,0.42) and (-0.20,0.18) .. (mu);   % a
\draw[midarrow] (nu)    .. controls (0.44,0.42) and (0.20,0.18) .. (mu);     % e

% short central arc d
\draw[midarrow] (nu) .. controls (0.10,1.17) and (-0.10,1.17) .. (alpha);

% dots
\fill (beta)  circle (1.8pt);
\fill (alpha) circle (1.8pt);
\fill (nu)    circle (1.8pt);
\fill (mu)    circle (1.8pt);

% vertex labels
\node[above=-18pt] at (beta) {$\mu$};
\node[left=-15pt] at (alpha) {$\alpha$};
\node[right=-15pt] at (nu) {$\nu$};
\node[below=-18pt] at (mu) {$\beta$};

% edge labels
\node at (1.4,1.00) {$d^\vee$};
\node at (-0.6,1.50) {$c^\vee$};
\node at (0.6,1.45) {$i^\vee$};
\node at (-0.7,0.50) {$j^\vee$};
\node at (0,1.4) {$b^\vee$};
\node at (0.6,0.50) {$a^\vee$};

\end{scope}
\end{tikzpicture}
    \Biggr\rangle
    \quad
\xrightarrow{\;\mathcal R\;}
    \quad
    \Biggl\langle
\begin{tikzpicture}[
    scale=0.95,
    baseline=4.5ex,
    line cap=round,
    line join=round,
    >={Stealth[length=3pt,width=4pt]},
]

\tikzset{
    midarrow/.style={
        thick,
        postaction={decorate},
        decoration={
            markings,
            mark=at position 0.53 with {\arrow{>}}
        }
    },
    midarrowrev/.style={
        thick,
        postaction={decorate},
        decoration={
            markings,
            mark=at position 0.53 with {\arrow{>}}
        }
    }
}

\begin{scope}[rotate around={180:(0,1)}]

% vertices
\coordinate (beta)  at (0,2.00);
\coordinate (alpha) at (-0.42,1.00);
\coordinate (nu)    at (0.42,1.00);
\coordinate (mu)    at (0,0);

% outer curve d, moved to the other side by sphericality
\draw[midarrow] (mu) .. controls (1.5,0.42) and (1.5,1.58) .. (beta);

% inner curved edges
\draw[midarrow] (beta) .. controls (-0.20,1.82) and (-0.44,1.58) .. (alpha); % i
\draw[midarrow] (beta) .. controls (0.20,1.82) and (0.44,1.58) .. (nu);      % c
\draw[midarrow] (alpha) .. controls (-0.44,0.42) and (-0.20,0.18) .. (mu);   % a
\draw[midarrow] (nu)    .. controls (0.44,0.42) and (0.20,0.18) .. (mu);     % j

% short central arc b with reversed arrow
\draw[midarrow] (alpha) .. controls (-0.10,1.17) and (0.10,1.17) .. (nu);

% dots
\fill (beta)  circle (1.8pt);
\fill (alpha) circle (1.8pt);
\fill (nu)    circle (1.8pt);
\fill (mu)    circle (1.8pt);

% vertex labels
\node[above=-18pt] at (beta) {$\mu$};
\node[left=-15pt] at (alpha) {$\nu$};
\node[right=-15pt] at (nu) {$\alpha$};
\node[below=-18pt] at (mu) {$\beta$};

% edge labels
\node at (1.4,1.00) {$d^\vee$};
\node at (-0.6,1.50) {$i^\vee$};
\node at (0.6,1.45) {$c^\vee$};
\node at (-0.65,0.50) {$a^\vee$};
\node at (0,1.4) {$b^\vee$};
\node at (0.55,0.50) {$j^\vee$};

\end{scope}
\end{tikzpicture}
    \Biggr\rangle~.
\label{eq:F-braiding}
\end{equation}

In the present network the braiding exchanges the positions of the
$c^\vee$ and $i^\vee$ strands and, correspondingly, those of
$j^\vee$ and $a^\vee$, while the $b^\vee$ and $d^\vee$ strands are
unchanged.  Importantly, braiding does not introduce any additional
dualisation and does not reverse the physical orientation of the
strands.

We denote the accumulated action of the required braiding maps on the
multiplicity spaces by $\mathcal R$.  In multiplicity-free channels
this reduces to the corresponding product of scalar $R$-symbols.

\item
\textbf{Charge conjugation and pivotal identification.}

We next apply charge conjugation to each of the four transformed
trivalent vertices.  Charge conjugation acts on the object labels and
on their multiplicity spaces while preserving the standard graphical
orientation of the strands.  Its local action is represented by the
matrices $U$ introduced above.  Thus

\begin{equation}
    \Biggl\langle
\begin{tikzpicture}[
    scale=1.2,
    baseline=6ex,
    line cap=round,
    line join=round,
    >={Stealth[length=3pt,width=4pt]},
]

\tikzset{
    midarrow/.style={
        thick,
        postaction={decorate},
        decoration={
            markings,
            mark=at position 0.53 with {\arrow{>}}
        }
    },
    midarrowrev/.style={
        thick,
        postaction={decorate},
        decoration={
            markings,
            mark=at position 0.53 with {\arrow{>}}
        }
    }
}

\begin{scope}[rotate around={180:(0,1)}]

% vertices
\coordinate (beta)  at (0,2.00);
\coordinate (alpha) at (-0.42,1.00);
\coordinate (nu)    at (0.42,1.00);
\coordinate (mu)    at (0,0);

% outer curve d, moved to the other side by sphericality
\draw[midarrow] (mu) .. controls (1.5,0.42) and (1.5,1.58) .. (beta);

% inner curved edges
\draw[midarrow] (beta) .. controls (-0.20,1.82) and (-0.44,1.58) .. (alpha); % i
\draw[midarrow] (beta) .. controls (0.20,1.82) and (0.44,1.58) .. (nu);      % c
\draw[midarrow] (alpha) .. controls (-0.44,0.42) and (-0.20,0.18) .. (mu);   % a
\draw[midarrow] (nu)    .. controls (0.44,0.42) and (0.20,0.18) .. (mu);     % j

% short central arc b with reversed arrow
\draw[midarrow] (alpha) .. controls (-0.10,1.17) and (0.10,1.17) .. (nu);

% dots
\fill (beta)  circle (1.8pt);
\fill (alpha) circle (1.8pt);
\fill (nu)    circle (1.8pt);
\fill (mu)    circle (1.8pt);

% vertex labels
\node[above=-18pt] at (beta) {$\mu$};
\node[left=-15pt] at (alpha) {$\nu$};
\node[right=-15pt] at (nu) {$\alpha$};
\node[below=-18pt] at (mu) {$\beta$};

% edge labels
\node at (1.4,1.00) {$d^\vee$};
\node at (-0.6,1.50) {$i^\vee$};
\node at (0.6,1.45) {$c^\vee$};
\node at (-0.65,0.50) {$a^\vee$};
\node at (0,1.4) {$b^\vee$};
\node at (0.55,0.50) {$j^\vee$};

\end{scope}
\end{tikzpicture}
    \Biggr\rangle
    \quad   \xrightarrow{\;\mathcal U\;}
    \quad
    \Biggl\langle
\begin{tikzpicture}[
    scale=1.3,
    baseline=6.5ex,
    line cap=round,
    line join=round,
    >={Stealth[length=3pt,width=4pt]},
]

\tikzset{
    midarrow/.style={
        thick,
        postaction={decorate},
        decoration={
            markings,
            mark=at position 0.53 with {\arrow{>}}
        }
    },
    midarrowrev/.style={
        thick,
        postaction={decorate},
        decoration={
            markings,
            mark=at position 0.53 with {\arrow{>}}
        }
    }
}

\begin{scope}[rotate around={180:(0,1)}]

% vertices
\coordinate (beta)  at (0,2.00);
\coordinate (alpha) at (-0.42,1.00);
\coordinate (nu)    at (0.42,1.00);
\coordinate (mu)    at (0,0);

% outer curve d, moved to the other side by sphericality
\draw[midarrow] (mu) .. controls (1.5,0.42) and (1.5,1.58) .. (beta);

% inner curved edges
\draw[midarrow] (beta) .. controls (-0.20,1.82) and (-0.44,1.58) .. (alpha); % i
\draw[midarrow] (beta) .. controls (0.20,1.82) and (0.44,1.58) .. (nu);      % c
\draw[midarrow] (alpha) .. controls (-0.44,0.42) and (-0.20,0.18) .. (mu);   % a
\draw[midarrow] (nu)    .. controls (0.44,0.42) and (0.20,0.18) .. (mu);     % j

% short central arc b with reversed arrow
\draw[midarrow] (alpha) .. controls (-0.10,1.17) and (0.10,1.17) .. (nu);

% dots
\fill (beta)  circle (1.8pt);
\fill (alpha) circle (1.8pt);
\fill (nu)    circle (1.8pt);
\fill (mu)    circle (1.8pt);

% vertex labels
\node[above=-18pt] at (beta) {$\mu$};
\node[left=-15pt] at (alpha) {$\nu$};
\node[right=-15pt] at (nu) {$\alpha$};
\node[below=-18pt] at (mu) {$\beta$};

% edge labels
\node at (1.4,1.00) {$d^{\vee\vee}$};
\node at (-0.65,1.50) {$i^{\vee\vee}$};
\node at (0.65,1.40) {$c^{\vee\vee}$};
\node at (-0.65,0.50) {$a^{\vee\vee}$};
\node at (0,1.4) {$b^{\vee\vee}$};
\node at (0.6,0.50) {$j^{\vee\vee}$};

\end{scope}
\end{tikzpicture}
    \Biggr\rangle 
    \quad   \xrightarrow{\;\mathsf p\;}
    \quad
    \Biggr\langle 
\begin{tikzpicture}[
    scale=1.2,
    baseline=6ex,
    line cap=round,
    line join=round,
    >={Stealth[length=3pt,width=4pt]},
]

\tikzset{
    midarrow/.style={
        thick,
        postaction={decorate},
        decoration={
            markings,
            mark=at position 0.53 with {\arrow{>}}
        }
    },
    midarrowrev/.style={
        thick,
        postaction={decorate},
        decoration={
            markings,
            mark=at position 0.53 with {\arrow{>}}
        }
    }
}

\begin{scope}[rotate around={180:(0,1)}]

% vertices
\coordinate (beta)  at (0,2.00);
\coordinate (alpha) at (-0.42,1.00);
\coordinate (nu)    at (0.42,1.00);
\coordinate (mu)    at (0,0);

% outer curve d, moved to the other side by sphericality
\draw[midarrow] (mu) .. controls (1.5,0.42) and (1.5,1.58) .. (beta);

% inner curved edges
\draw[midarrow] (beta) .. controls (-0.20,1.82) and (-0.44,1.58) .. (alpha); % i
\draw[midarrow] (beta) .. controls (0.20,1.82) and (0.44,1.58) .. (nu);      % c
\draw[midarrow] (alpha) .. controls (-0.44,0.42) and (-0.20,0.18) .. (mu);   % a
\draw[midarrow] (nu)    .. controls (0.44,0.42) and (0.20,0.18) .. (mu);     % j

% short central arc b with reversed arrow
\draw[midarrow] (alpha) .. controls (-0.10,1.17) and (0.10,1.17) .. (nu);

% dots
\fill (beta)  circle (1.8pt);
\fill (alpha) circle (1.8pt);
\fill (nu)    circle (1.8pt);
\fill (mu)    circle (1.8pt);

% vertex labels
\node[above=-18pt] at (beta) {$\mu$};
\node[left=-15pt] at (alpha) {$\nu$};
\node[right=-15pt] at (nu) {$\alpha$};
\node[below=-18pt] at (mu) {$\beta$};

% edge labels
\node at (1.4,1.00) {$d$};
\node at (-0.6,1.50) {$i$};
\node at (0.6,1.50) {$c$};
\node at (-0.55,0.50) {$a$};
\node at (0,1.4) {$b$};
\node at (0.6,0.50) {$j$};

\end{scope}
\end{tikzpicture} 
    \Biggr\rangle~.
\label{eq:F-charge-conjugation}
\end{equation}

We denote the accumulated action of these four charge-conjugation
maps on the multiplicity spaces by $\mathcal U$.

After charge conjugation, all standard splitting vertices carry
double-dual labels.  The pivotal structure
$\mathsf p_x:x^{\vee\vee}\to x$ therefore supplies the final identification,
or equivalently the maps $\mathsf p$ on the corresponding splitting
spaces.
After the
pivotal identification, the resulting closed network is isotopic to
the original $F$-symbol network.

\end{itemize}

The preceding diagrams establish geometrically that the
complex-conjugated $F$-symbol differs from the original $F$-symbol
only by the local unitary changes of trivalent basis generated by the
linear part of the antiunitary transformation.  We now state this
relation precisely.

Using the two fusion-tree Hilbert spaces introduced in the preceding
subsection, let
\begin{equation}
    p=(i,\nu,\mu)~,
    \qquad
    q=(j,\alpha,\beta)~,
\end{equation}
so that
\begin{equation}
    \bigl[F_d^{abc}\bigr]_{pq}
    =
    {}_L\!\braket{
        L_p
        |
        \alpha_{a,b,c}
        |
        R_q
    }_R~.
\label{eq:F-overlap-LR}
\end{equation}
As usual in the graphical expressions, the associator
$\alpha_{a,b,c}$ is suppressed.

The transformation $A_a^{bc}=\mathsf M_a^{bc}K$ defined above induces
a single antiunitary map on each of the two fusion-tree Hilbert
spaces,
\begin{equation}
    A_L=\mathsf M_L K~,
    \qquad
    A_R=\mathsf M_R K~,
\end{equation}
whose unitary parts are
\begin{equation}
    \mathsf M_L
    =
    \bigoplus_i
    \mathsf M_i^{ab}\otimes\mathsf M_d^{ic}~,
    \qquad
    \mathsf M_R
    =
    \bigoplus_j
    \mathsf M_j^{bc}\otimes\mathsf M_d^{aj}~,
\label{eq:MLMR-local}
\end{equation}
with
\begin{equation}
    \mathsf M_a^{bc}
    =
    \mathsf p_a^{bc}
    U_{a^\vee}^{b^\vee c^\vee}
    R_{a^\vee}^{c^\vee b^\vee}
    P_{\circlearrowleft,a}^{bc} ~.
\label{eq:local-M}
\end{equation}

The diagrammatic construction above precisely expresses the
compatibility of the antiunitary transformations with the associator:
the transformed left- and right-associated fusion trees are related
by the same $F$-move.  Hence, using the antiunitarity of $A_R$,
\begin{equation}
\begin{aligned}
    \overline{
        \bigl[F_d^{abc}\bigr]}_{pq}
    &=
    {}_L\!\braket{
        A_L L_p
        |
        \alpha_{a,b,c}
        |
        A_R R_q
    }_R
    \\
    &=
    \sum_{p',q'}
    (\mathsf M_L)_{p'p}\,
    \bigl[F_d^{abc}\bigr]_{p'q'}\,
    \overline{(\mathsf M_R)}_{q'q}~.
\end{aligned}
\end{equation}
Written as a matrix identity, the second line above reads $\overline{F_d^{abc}}=\mathsf M_L^{T}\,F_d^{abc}\,\overline{\mathsf M_R}$. We can massage this expression as follows. First, conjugate it entrywise, which gives $F_d^{abc}=\mathsf M_L^{\dagger}\,\overline{F_d^{abc}}\,\mathsf M_R$, and substitute the relation back into itself:
\begin{equation}
    F_d^{abc}
    =
    \bigl(\mathsf M_L^{\dagger}\mathsf M_L^{T}\bigr)
    \,F_d^{abc}\,
    \bigl(\overline{\mathsf M_R}\,\mathsf M_R\bigr)~.
\end{equation}
Equivalently, since all the matrices involved are unitary and
$\bigl(\overline{\mathsf M_R}\,\mathsf M_R\bigr)^{-1}=\mathsf M_R^{\dagger}\mathsf M_R^{T}$,
\begin{equation}
    \bigl(\mathsf M_L^{\dagger}\mathsf M_L^{T}\bigr)\,F_d^{abc}
    =
    F_d^{abc}\,\bigl(\mathsf M_R^{\dagger}\mathsf M_R^{T}\bigr)~.
\end{equation}
Multiplying this on the left by $\mathsf M_L$ and on the right by
$\overline{\mathsf M_R}$, and using $\mathsf M_L\mathsf M_L^{\dagger}=\mathbbm 1$
together with
$\mathsf M_R^{T}\overline{\mathsf M_R}=\overline{\mathsf M_R^{\dagger}\mathsf M_R}=\mathbbm 1$,
we obtain
\begin{equation}
    \mathsf M_L^{T}\,F_d^{abc}\,\overline{\mathsf M_R}
    =
    \mathsf M_L\,F_d^{abc}\,\mathsf M_R^{\dagger}~.
\end{equation}
Therefore, for the complete $F$-matrix,
\begin{equation}
    \boxed{
        \overline{F_d^{abc}}
        =
        \mathsf M_L
        F_d^{abc}
        \mathsf M_R^{\dagger}={[F^{abc}_d]}^{M}
    ~.}
\label{eq:FstarMLFMR}
\end{equation}
This is precisely the standard trivalent gauge transformation with local gauge matrices $\Gamma_a^{bc} = \bigl(\mathsf M_a^{bc}\bigr)$. 
Importantly, no involutivity condition has been used in deriving it: at this stage we
have made no assumption that $A^2=1$, that the matrices
$\mathsf M_a^{bc}$ are symmetric, or that any Frobenius--Schur turning
operator is trivial.  Those additional conditions will enter only
when we ask whether the gauge transformation in
Eq.~\eqref{eq:FstarMLFMR} can itself be removed so as to obtain a
gauge in which $\overline F=F$.

\subsection{Choosing the real gauge}
\label{ssec:mechanismIfSymmetric}
The remaining question is whether the gauge transformation $\mathsf M_a^{bc}$ can be removed by a consistent change of trivalent fusion bases. In the following sections we will establish sufficient conditions under which the local unitary matrices are symmetric,
\begin{equation}
    \bigl(\mathsf M_a^{bc}\bigr)^T
    =
    \mathsf M_a^{bc}~.
\label{eq:M-local-symmetric}
\end{equation}
In the present subsection we assume this symmetry and show that it is
sufficient to construct a gauge in which
\begin{equation}
    \overline{F_d^{abc}}
    =
    F_d^{abc}~.
\end{equation}
The construction is particularly simple in the multiplicity-free
case, where the local matrices reduce to phases, and is generalised
to arbitrary fusion multiplicities by Takagi factorisation.

We first consider the multiplicity-free case, $N_{ab}^{c}=0,1$.
Each non-vanishing splitting space is then one-dimensional, and the
local unitary $\mathsf M_c^{ab}$ is simply a phase, which we denote by
\begin{equation}
    \mathsf m_c^{ab}\in U(1)~.
\end{equation}
Since a one-dimensional matrix is automatically symmetric, the
condition~\eqref{eq:M-local-symmetric} imposes no additional
restriction in this case.

Equation~\eqref{eq:FstarMLFMR} reduces componentwise to
\begin{equation}
\overline{
\bigl(F_d^{abc}\bigr)}_{ef}
=
\bigl(\mathsf m_e^{ab}\mathsf m_d^{ec}\bigr)
\bigl(\mathsf m_f^{bc}\mathsf m_d^{af}\bigr)^{-1}
\bigl(F_d^{abc}\bigr)_{ef}~.
\label{eq:Fbar-multiplicity-free}
\end{equation}
For every admissible trivalent vertex, choose a phase
$\mathsf w_c^{ab}$ satisfying
\begin{equation}
\mathsf m_c^{ab}
=
\bigl(\mathsf w_c^{ab}\bigr)^2~.
\label{eq:scalar-Takagi}
\end{equation}
We then perform the local gauge transformation
\begin{equation}
\Gamma_c^{ab}
=
\bigl(\mathsf w_c^{ab}\bigr)~.
\label{eq:scalar-real-gauge}
\end{equation}
Under a multiplicity-free gauge transformation,
\begin{equation}
\bigl(F_d^{abc}\bigr)_{ef}
\longmapsto
\Gamma_e^{ab}
\Gamma_d^{ec}
\bigl(\Gamma_f^{bc}\bigr)^{-1}
\bigl(\Gamma_d^{af}\bigr)^{-1}
\bigl(F_d^{abc}\bigr)_{ef}~.
\label{eq:multiplicityFreeGauge}
\end{equation}
Hence the transformed matrix element is
\begin{equation}
\bigl(F_d^{abc}\bigr)_{ef}^{\mathsf w}=
\bigl(\mathsf w_e^{ab}\mathsf w_d^{ec}\bigr)
\bigl(\mathsf w_f^{bc}\mathsf w_d^{af}\bigr)^{-1}
\bigl(F_d^{abc}\bigr)_{ef}~.
\label{eq:Fhat-multiplicity-free}
\end{equation}
Using Eq.~\eqref{eq:Fbar-multiplicity-free} and
$\mathsf m=(\mathsf w)^2$, we immediately obtain
\begin{equation}
\overline{
\bigl(F_d^{abc}\bigr)^{\mathsf w}}_{ef}
=
\bigl(F_d^{abc}\bigr)^{\mathsf w}_{ef}~.
\end{equation}
Thus, in the multiplicity-free case, the real gauge is obtained by
taking square roots of the local phases appearing in the unitary part
of the antiunitary operation. The resulting real $F$-symbols need not
be positive, as residual signs are compatible with a real orthogonal
$F$-matrix.
\subsubsection{Higher fusion multiplicities}
\label{sssec:higherMultDiagrammatic}
Now, we consider the case of gauge transformations with vertices of arbitrary multiplicity
\begin{align}
 \bigl[F_d^{abc}\bigr]_
 {(e,\alpha,\beta),(f,\mu,\nu)}^{\Gamma}=
 \sum_{\alpha',\beta',\mu',\nu'}
 \bigl[\Gamma_e^{ab}\bigr]_{\alpha\alpha'}
 \bigl[\Gamma_d^{ec}\bigr]_{\beta\beta'}
 \nonumber
 \bigl[F_d^{abc}\bigr]_
 {(e,\alpha',\beta'),(f,\mu',\nu')}
 \bigl[\bigl(\Gamma_f^{bc}\bigr)^{-1}\bigr]_{\mu'\mu}
 \bigl[\bigl(\Gamma_d^{af}\bigr)^{-1}\bigr]_{\nu'\nu}~.
 \label{eq:higherMultiplicityGauge}
\end{align}
Equivalently, defining
\begin{equation}
L_\Gamma
=
\bigoplus_e
\Gamma_e^{ab}\otimes\Gamma_d^{ec}~,
\qquad
R_\Gamma
=
\bigoplus_f
\Gamma_f^{bc}\otimes\Gamma_d^{af}~,
\end{equation}
we have
\begin{equation}
\left[F_d^{abc}\right]^{\Gamma}
=
L_\Gamma
F_d^{abc}
R_\Gamma^\dagger~.
\label{eq:blockGaugeTransformation}
\end{equation}
Suppose now that the local matrices $\mathsf M_c^{ab}$ appearing in
the antiunitary transformation are symmetric and unitary,
\begin{equation}
    \bigl(\mathsf M_c^{ab}\bigr)^T
    =
    \mathsf M_c^{ab}~.
\end{equation}
By Takagi factorisation, for every admissible trivalent splitting
space there exists a unitary matrix $\mathsf W_c^{ab}$ such that
\begin{equation}
    \mathsf M_c^{ab}
    =
    \mathsf W_c^{ab}
    \bigl(\mathsf W_c^{ab}\bigr)^T~.
\label{eq:localTakagi}
\end{equation}
For a general complex symmetric matrix the Takagi factorisation
contains a non-negative diagonal matrix between the two unitary
factors. Since $\mathsf M_c^{ab}$ is itself unitary, all Takagi
singular values are equal to one, so this diagonal matrix is the
identity.

We correspondingly define
\begin{equation}
    \mathsf W_L
    =
    \bigoplus_e
    \mathsf W_e^{ab}\otimes\mathsf W_d^{ec}~,
    \qquad
    \mathsf W_R
    =
    \bigoplus_f
    \mathsf W_f^{bc}\otimes\mathsf W_d^{af}~.
\label{eq:WLWR}
\end{equation}
Equation~\eqref{eq:localTakagi} then implies
\begin{equation}
    \mathsf M_L
    =
    \mathsf W_L\mathsf W_L^T,
    \qquad
    \mathsf M_R
    =
    \mathsf W_R\mathsf W_R^T~.
\label{eq:MLMR-Takagi}
\end{equation}

We now choose the local vertex-basis gauge transformation
\begin{equation}
    \Gamma_c^{ab}
    =
    \bigl(\mathsf W_c^{ab}\bigr)^T~.
\label{eq:Takagi-gauge}
\end{equation}
With our gauge-transformation convention,
Eq.~\eqref{eq:blockGaugeTransformation}, this gives
\begin{equation}
    \left[F_d^{abc}\right]^{\mathsf W}
    =
    \mathsf W_L^T
    F_d^{abc}
    \overline{\mathsf W_R}~.
\label{eq:realGaugeHigherMultiplicity}
\end{equation}
Here we have used
\begin{equation}
    \bigl(\mathsf W_R^T\bigr)^\dagger
    =
    \overline{\mathsf W_R}~.
\end{equation}

Taking the entrywise complex conjugate and using
Eq.~\eqref{eq:FstarMLFMR}, we obtain
\begin{align}
    \overline{
    \left[F_d^{abc}\right]^{\mathsf W}}
    &=
    \mathsf W_L^\dagger
    \overline{F_d^{abc}}
    \mathsf W_R
    =
    \mathsf W_L^\dagger
    \mathsf M_L
    F_d^{abc}
    \mathsf M_R^\dagger
    \mathsf W_R
    =
    \mathsf W_L^\dagger
    \bigl(\mathsf W_L\mathsf W_L^T\bigr)
    F_d^{abc}
    \bigl(\mathsf W_R\mathsf W_R^T\bigr)^\dagger
    \mathsf W_R
    \nonumber\\
    &=
    \mathsf W_L^T
    F_d^{abc}
    \overline{\mathsf W_R}
    =
    \left[F_d^{abc}\right]^{\mathsf W}~.
\end{align}
Thus the transformed $F$-symbols are real. 

The higher-multiplicity case is the direct matrix generalisation of the multiplicity-free construction: the square roots of scalar phases are replaced by Takagi factors of the symmetric unitary matrices $\mathsf M_c^{ab}$. Importantly, the Takagi factorisation is performed locally on each splitting space $V_c^{ab}$. The resulting gauge is therefore a single consistent global choice of trivalent bases, rather than an independent change of basis for each $F$-matrix. We have therefore shown that symmetry of the local unitary matrices $\mathsf M_a^{bc}$ is sufficient for the gauge equivalence between $\overline F$ and $F$ to be trivialised by a single consistent choice of trivalent fusion bases. In this gauge,
\begin{equation}
    \overline{F_d^{abc}}
    =
    F_d^{abc}~,
\end{equation}
for all admissible labels. Therefore, all $F$-symbols are real.

The problem of finding a real gauge has thus been reduced to
understanding when the local matrices $\mathsf M_a^{bc}$ are
symmetric.  In Sec. \ref{sec:sufficient} we will describe the categorical conditions under which this condition holds.

\section{Categorical description of gauge equivalences}\label{DefsPrelim}

In this section, we we will review unitary ribbon categories, carefully explaining the relationships between general ribbon categories and $F$-symbols, and between monoidal equivalences and the induced gauge transformations. 
Aside from establishing foundations, basic facts, and conventions, this section will have two main goals.
First, we will show that given a unitary ribbon category $\cC$, there exists a canonical braided equivalence $\cC\to\overline{\cC}^{\rev}$ which acts as a dual functor on objects.
Second, we will show that the gauge transformation associated with a monoidal equivalence $R:\cC\to\overline{\cC}$ can be gauged away provided that it is symmetric, and establish sufficient conditions for that to occur. 
In {Sec.}~\ref{sec:sufficient}, we will these tools to establish sufficient conditions for a gauge with real $F$-symbol.

The outline of this section is as follows.
In {Sec.}~\ref{ssec:URCBasics}, we will review the notion of unitary ribbon category and introduce related notions which will be useful later on.
{Sec.}~\ref{sssec:based} introduces the notion of \textit{based unitary ribbon category}, which is a unitary ribbon category equipped with properties and extra structures that function as a choice of gauge, by which categorical data such as the associator and braiding natural isomorphisms may be translated into the $F$-symbol and $R$-symbol of an anyon theory.
{Sec.}~\ref{sssec:2groupoids} describes translations between anyon theories, based unitary ribbon categories, and unitary ribbon categories, showing that the three notions are equivalent.
In {Sec.}~\ref{ssec:ribbonDefns}, we review the categorical formalization of the reflections and rotations described in Table~\ref{CatST}, as well as comparisons between them in the form of functors which package parts of the structure of a unitary ribbon category, such as a unitary dual functor and dagger functor.
We establish technical results about these functors which will eventually allow us to compute the crucial gauge symbols in {Sec.}~\ref{sec:sufficient}. 

\subsection{Unitary ribbon categories and algebraic theories of anyons}
\label{ssec:URCBasics}

It is well-known \cite{Moore:1991ks,Kitaev:2005hzj} that anyons in a (2+1)D topological phase are described by a unitary modular tensor category (UMTC).
For our purposes, it is helpful to study a slightly more general notion, that of a \textit{unitary ribbon category}, also known as \textit{unitary premodular category} \cite{MR1292673}.
\begin{defn}
    \label{defn:URC}
    A \textbf{unitary ribbon category}, also known as a \textit{unitary premodular category}, is a braided rigid $C^*$-tensor category which is finitely semisimple and has simple tensor unit (fusion rather than multifusion).\footnote{
     Philosophically, the adjective ``ribbon'' refers to the fact that strings in diagrammatic calculus behave like ribbons.
     This in no way implies that the unit object must be simple, so it might be better to not include simplicity of the tensor unit in this definition, and to call unitary ribbon categories with simple tensor unit ``unitary ribbon fusion categories.''
     However, we decline to do this because it would require either adding the word ``fusion'' in many places or generalizing calculations to account for the possibility of multifusion categories, which is straightforward to do, but orthogonal to the content of this paper.
    }
\end{defn}
\begin{remark}
    A ribbon category is a braided spherical fusion category, so one might expect the definition of unitary ribbon category to include the additional data of a choice of spherical pivotal structure.
    However, in the unitary case, rigidity ensures the existence of a canonical unitary dual functor with a spherical pivotal structure, unique up to a unique unitary monoidal natural isomorphism \cite{MR1444286,MR2091457,MR4133163}.
    Therefore, we do not include the choice of dual functor as additional data in Definition~\ref{defn:URC}, nor do we need to explicitly require its existence; see also \cite{Galindo_2014}.
    
    Once endowed with the canonical spherical structure, a unitary ribbon category is a unitary braided fusion category equipped with the additional structure necessary to close off braids into links.
    It is therefore the correct setting in which to define $S$- and $T$-matrices and interpret them as the evaluations of closed links in the graphical calculus.
\end{remark}
\begin{remark}
    The Drinfeld center of a unitary ribbon category is a UMTC, and the braiding on a ribbon category $\cC$ gives a full inclusion $\cC\hookrightarrow Z(\cC)$.
    It is therefore equivalent to say that unitary ribbon categories are the fusion subcategories of UMTCs.
    In other words, they are all obtained from UMTCs by restricting to subsets of the anyon labels closed under fusion. 
\end{remark}

In parts of the physics literature, including \cite{Kitaev:2005hzj,Barkeshli:2014cna}, unitary ribbon categories are packaged differently, as a finite list of nonisomorphic simple objects, a list of fusion multiplicities, and $F$- and $R$-symbols satisfying certain consistency conditions, or as equivalent diagrammatic data, sometimes called an \textit{anyon theory}.
The relationship between anyon theories and more notions such as unitary ribbon category or UMTC has been previously described in places such as \cite[Appendix~E.7]{Kitaev:2005hzj}, and is well-known to experts. 

For our purposes, it will be helpful to review this relationship and to adopt a slightly difference description of the correspondence from \cite{Kitaev:2005hzj}.
To that end, we introduce a decorated form of unitary ribbon category, called \textit{based unitary ribbon category}.
The additional data of a based unitary ribbon category will help to track the effects of the reflections previously discussed in {Sec.}~\ref{folding} on the choice of gauge, which is necessary to distinguish between situations where $F=\overline{F}$ and those where they are merely related by a gauge transformation.
In {Sec.}~\ref{sssec:2groupoids}, we will formalize a sense in which unitary ribbon categories, based unitary ribbon categories, and anyon theories are all equivalent notions by describing translations between them.

\subsubsection{Based unitary ribbon categories}
\label{sssec:based}
Next, we introduce a special class of unitary ribbon categories, designed to walk the line between strictness and skeletality, while excluding pathologies present in general unitary ribbon categories that obscure the relationship between gauge transformations of $F$-symbols and braided monoidal equivalences.
\begin{defn}
    \label{defn:standard}
    A \textbf{standard} set of simple objects for a unitary ribbon category $\cC$ is a set $S\subseteq\Irr(\cC)$ containing one simple object from each isomorphism class, such that
    \begin{enumerate}
        \item[(S1)] For each word $w=\prod_{j=1}^\ell s_j$ in the free monoid $S^*$ on $S$ (with unit $1\in S$ isomorphic to the tensor unit), the object $\ev(w)\cong\otimes_{j=1}^{\ell}s_j$, is independent of the bracketing of the tensor product
        \item[(S2)] For distinct words $v\neq w\in S^*$, the objects $\ev(v)$ and $\ev(w)$ are also distinct.
    \end{enumerate}
    A \textbf{standard} unitary ribbon category is a unitary ribbon category together with a choice of standard set of simple objects.
\end{defn}
(S1) requires that $\cC$ be strictly unital with the chosen tensor unit, and that the tensor product be strictly monoidal at the level of objects which are tensor-generated by $S$. 
(S2) excludes the possibility that, for $x,y,z,w\in S$ with $(x,y)\neq (z,w)$, we might have $x\otimes y=z\otimes w$ by coincidence.\footnote{
For example, a skeletal version of a Tambara-Yamigami category would violate (S2), since the non-Abelian simple object absorbs all the others.
}

Given a unitary ribbon category $\cC$, it is easy to find an equivalent $\cC'$ which satisfies (S1).
Indeed, one can take $\cC'$ to be the image of the embedding $\cC\hookrightarrow\End(\cC)$ given by $x\mapsto x\otimes\cdot$ (which is a usual strictification), after replacing $1\otimes\cdot$ with the isomorphic functor $\Id_{\cC}$ to ensure strict unitality.
To satisfy (S2), it may be necessary to add new objects.
We outline the details in the following construction.
\begin{const}
    \label{const:standardURC}
    Let $\cC$ be a unitary ribbon category.
    We will now construct a standard unitary ribbon category $(\widetilde\cC,S)$ where the underlying unitary ribbon category $\widetilde\cC$ is equivalent to $\cC$.
    
    As pointed out above, we may always assume that $\cC$ is strictly unital.
    Pick an arbitrary set of simple objects $S\subseteq\Irr(\cC)$ containing one object from each isomorphism class, with $1_\cC\in S$.
    We now define a new braided monoidal $\dag$-category $\cC'$:
    \begin{itemize}
        \item $\Obj(\cC')=S^*$, the free monoid on $S$ with unit $1_{\cC}$.
        \item $\cC'(v\to w):=\cC(\ev(v)\to\ev(w))$, where $\ev(v)$ is the right-bracketed tensor product of $v$.
        \item $v\otimes w:=vw$, the concatenation of words in $S^*$.
        \item The composition of morphisms, operation $(\cdot)^\dag$, tensor product of morphisms, associator, and braiding are chosen so that $\ev$ becomes a braided monoidal equivalence. That is:
        \begin{itemize}
            \item The composition and $\dag$ of morphisms are the same as in $\cC$.
            \item The tensorator $\mu:\ev(v\otimes w)\to\ev(v)\otimes\ev(w)$ is the rebracketing map obtained from the associator of $\cC$, which is unique since $\cC$ satisfies the pentagon equation.
            In particular, if $x,y\in S$, then $\mu_{x,y}=\id_{x\otimes y}$, and if $x,y,z\in S$, then $\mu_{xy,z}:=\alpha_{x,y,z}:\ev((xy)z)=x\otimes(y\otimes z)\to(x\otimes y)\otimes z)=\ev(xy)\otimes \ev(z)$, the associator from $\cC$. 
            \item The associator $\alpha_{v,w,x}:vwx=v\otimes (w\otimes x)\to (v\otimes w)\otimes x=vwx$ is given by $\id_{\ev(vwx)}$.
            This choice makes $\ev$ a monoidal functor, since the commuting square expressing this condition is a face of the associahedron of $\cC$.
            \item The braiding $\beta_{x,y}$ is the same as in $\cC$ for $x,y\in S$.
            In general, the value of the braiding on other words is forced by monoidality of $\ev$ and the hexagon equation of $\cC$.
        \end{itemize}
    \end{itemize}
    Now $\cC'$ is a braided monoidal category, and $\ev:\cC'\to\cC$ is a braided monoidal $\dag$-functor which is the inclusion of a full subcategory.
    
    Finally, we define $\widetilde\cC:=\Kar(\cC')$, the Karoubi completion of $\cC'$.
    The functor $\ev$ uniquely lifts to an equivalence $\ev:\widetilde{\cC}\to\cC$, which is a unitary braided monoidal equivalence.
    The set $S$ is manifestly a set of standard simple objects for $\widetilde{\cC}$.
\end{const}

The benefits of working with standard unitary ribbon categories become clear when considering functors between unitary ribbon categories, such as braided autoequivalences, which are symmetries of the corresponding anyon theory.
One would like to interpret an autoequivalence $\Phi:\cC\to\cC$ as a permutation of simple objects, together with a gauge transformation which preserves the $F$-symbol of $\cC$.
In general, such a $\Phi$ only permutes simple objects up to isomorphism, so to extract a gauge transformation, we also need to pick an arbitrary system of unitary isomorphisms $\Phi(x)\to x'$ for each $x\in S$, where $x'\in S$ is the unique element of $S$ isomorphic to $\Phi(x)$.
To keep track of these details, we introduce the following notion.
\begin{defn}
\label{defn:based}
 A \textbf{based unitary ribbon category} consists of
 \begin{itemize}
     \item A standard unitary ribbon category $(\cC,S_{\cC})$
     \item For each $a$, $b$, and $c\in S_{\cC}$, a choice of ordered orthonormal basis
     \[\left\{\gamma_a^{bc}[i]\right\}_{i=1}^{N_{bc}^a}=:B_a^{bc}\subseteq\cC(a\to bc)~,\]
     for the normalized spherical trace inner product \eqref{eq:tracialIP}.
     That is,
     \[\iprod{\gamma_a^{bc}[j]}{\gamma_a^{bc}[i]}:=\frac{1}{\sqrt{d_ad_bd_c}}\tr(\gamma_a^{bc}[j]^\dag\gamma_a^{bc}[i])=\delta_{i=j}
     \,\text{ and }\,
     \gamma_a^{bc}[j]^\dag\gamma_a^{bc}[i]=\delta_{i=j}\frac{\sqrt{d_bd_c}}{\sqrt{d_a}}\id_a~.\]
 \end{itemize}
\end{defn}

\begin{remark}
    In \S~\ref{sssec:mechanismR}, the calculations include an antiunitary operator $K$, which performs complex conjugation of coefficients with respect to a particular basis.
    The operator $K$ is necessarily basis dependent, though $K\mathsf{M}K=\overline{M}$ is not.
    The bases of trivalent vertex spaces in Definition~\ref{defn:based} will eventually play the role of the choice of basis implicit in the construction of $K$.
\end{remark}

\begin{nota}
    \label{nota:basedDiagrams}
    We can depict the basis $B_{\cC}$ of splitting channels diagrammatically as trivalent vertices, as follows.
    \[\gamma_a^{bc}[j]=\tikzmath{
     \draw[mid<] (0:0) -- (-90:.8);
     \node at (-90:1) {$a$};
     \draw[mid>] (0:0) -- (150:.8);
     \node at (150:1) {$b$};
     \draw[mid>] (0:0) -- (30:.8);
     \node at (30:1) {$c$};
     \node at (90:.3) {$\scriptstyle j$};
    }\]
\end{nota}

In this context, the role of condition (S2) from Definition~\ref{defn:standard} is to ensure that no two of the trivalent vertex spaces $\cC(x\to yz)$ are literally equal by coincidence.
In other words, we avoid choosing a category $\cC$ that secretly includes extraneous relations of the form $\gamma_a^{bc}=U_{a}\gamma_a^{de}$ for some $(b,c)\neq (d,e)$ and system of unitary matrices $(U_{a})_{a\in S_{\cC}}$, as these relations are evil and are not captured by diagrammatic calculus.

It is straightforward to see that a based unitary (braided) fusion category is one which is equipped with precisely the extra data needed to write down a specific $F$-symbol (and $R$-symbol).
As described in Remark~\ref{rem:goodInnerProduct}, the chosen bases $B_a^{bc}$ determine orthonormal bases of the tetravalent homs-spaces $\cC(d\to a(bc))$ and $\cC(d\to (ab)c)$.
Explicitly, if we abbreviate
\[R_d^{abc}[x,i,j]:=(\id_a\otimes\gamma_x^{bc}[j])\gamma_d^{ax}[i],\quad L_d^{abc}[y,k,\ell]:=(\gamma_y^{ab}[\ell]\otimes\id_c)\gamma_d^{yc}[k]~,\]
then $\{R_d^{abc}[x,i,j]\}$ and $\{L_d^{abc}[y,k,\ell]\}$ are the respective bases.  
The $F$-symbol is then given by
\begin{equation}
    \label{eq:FFromBasedF1C}
    F_d^{abc}[(y,k,\ell),(x,i,j)]:=\iprod{L_d^{abc}[y,k,\ell]}{\alpha_{a,b,c}R_d^{abc}[x,i,j]}~.
\end{equation}
One recovers the associator as
\begin{equation}
    \label{eq:alphaFromF}
    \alpha_{a,b,c}=\bigoplus_{d}\sum_{y,k,\ell}\sum_{x,i,j}\sqrt{\frac{d_d}{d_ad_bd_c}}L_d^{abc}[y,k,\ell]R_d^{abc}[x,i,j]^\dag~.
\end{equation}
Similarly, the $R$-symbol is given by
\begin{equation}
    \label{eq:RFromBasedF1C}
    R_x^{y,z}[k,j]:=\iprod{\gamma_x^{z,y}[k]}{\beta_{y,z}\circ\gamma_x^{y,z}[j]}~.
\end{equation}
Thus, given a based unitary ribbon category, we can recover the usual data of an anyon theory expressed in a particular gauge.
The pentagon and hexagon equations for $F$ and $R$ follow immediately from equations \eqref{eq:FFromBasedF1C} and \eqref{eq:RFromBasedF1C} together with the penatagon and hexagon equations for $\alpha$ and $\beta$.

To complete Definition~\ref{defn:based}, we now introduce the corresponding notion of unitary braided monoidal equivalence, which we call standard equivalence.\footnote{We leave the extension of this notion to monoidal functors which are not equivalences to future work, since they will not be needed here.}
\begin{defn}
    \label{defn:standardEquivalence}
    If $(\cC,S_{\cC})$ and $(\cD,S_{\cD})$ are standard unitary ribbon categories, a \textbf{standard equivalence} $\Phi:\cC\to\cD$ is a unitary braided monoidal equivalence with the following properties.
    \begin{enumerate}
        \item[(SF1)] For any $x\in S_{\cC}$, we have $\Phi(x)\in S_{\cD}$
        \item[(SF2)] For $x$ and $y\in S_{\cC}$, we have $\Phi(xy)=\Phi(x)\Phi(y)$
    \end{enumerate}
\end{defn}
If $\Phi$ is a standard equivalence, then we have $\cD(\Phi(x)\to\Phi(yz))=\cD(\sigma(x)\to\sigma(y)\sigma(z))$, where $\sigma:S_{\cC}\to S_{\cD}$ is the bijection induced by $\Phi$.
Therefore, we can define a gauge symbol $U_x^{yz}$ associated with the standard equivalence $\Phi$ with tensorator $(\theta_{x,y})_{x,y\in S_{\cC}}$ by the following formula.
\begin{equation}
    \label{eq:gaugeSymbolFromStdEq}
    U_x^{yz}[k,j]:=\iprod{(\gamma^{\cD})_x^{yz}[k]}{\theta_{y,z}\circ\Phi((\gamma^{\cC})_x^{yz}[j])}~.
\end{equation}
The tensorator satisfies the equation
\begin{equation}
    \alpha^{\cD}_{\Phi(A),\Phi(B),\Phi(C)}\circ(\id_{\Phi(A)}\otimes\theta_{B,C})\circ\theta_{A,BC} = (\theta_{A,B}\id_{\Phi(C)})\circ\theta_{AB,C}\circ\Phi(\alpha^{\cC}_{A,B,C})~.
\end{equation}
Writing this equation in the chosen gauges produces the equation
\begin{equation}
 \label{eq:gaugeTransformationMonoidal}
 \begin{aligned}
    &(F^{\cD})^{\Phi(a)\Phi(b)\Phi(c)}_{\Phi(d)}[(\Phi(e),i,j),(\Phi(f),k,\ell)] = \cdots
    \\ &\cdots 
    \sum_{i',j',k',\ell'}
    U_d^{ec}[i,i']U_e^{ab}[j,j']
    (F^{\cD})^{abc}_{d}[(e,i',j'),(f,k',\ell')]
    (U_d^{af})^\dag[k,k'](U_f^{bc})^\dag[\ell,\ell']~.
 \end{aligned}
\end{equation}
Similarly, that $\Phi$ is braided means that
\begin{equation}
    \beta_{\Phi(A),\Phi(B)}\theta_{A,B}=\theta_{B,A}\beta_{\Phi(B),\Phi(A)}~,
\end{equation}
which becomes the equation
\begin{equation}
 \label{eq:gaugeTransformationBraided}
 R_{\Phi(x)}^{\Phi(y)\Phi(z)}[i,j] = \sum_{i',j'} U_x^{zy}[i,i']R_x^{yz}[i',j'](U_x^{yz})^{\dag}[j,j']~.
\end{equation}
Equations~\eqref{eq:gaugeTransformationMonoidal} and \eqref{eq:gaugeTransformationBraided} are standard conditions on a gauge transformation found in sources such as \cite{Kitaev:2005hzj,BarkeshliBondersonChengWang2019}.

The following construction shows that every unitary braided monoidal equivalence of standard unitary ribbon categories is unitarily monoidally naturally isomorphic to a standard equivalence.
\begin{const}
    \label{const:standardEq}
    Let $(\cC,S_{\cC})$ and $(\cD,S_{\cD})$ be standard unitary ribbon categories, and let $\Phi:\cC\to\cD$ be a unitary braided monoidal equivalence with tensorator $(\theta_{y,z})$.
    We will proceed to define a standard equivalence $\Phi':(\cC,S_{\cC})\to(\cD,S_{\cD})$.
    For every $x\in S_{\cC}$, pick an arbitrary unitary $\iota_x\in\Phi(x)\to\sigma(x)$.
    For every $y$ and $z\in S_{\cC}$, pick an arbitrary unitary $\tau_{y,z}:\Phi(yz)\to\sigma(y)\sigma(z)$.
    
    We define $\Phi'$ on objects by 
    \begin{equation}
     \label{eq:stdFunctorObjs}
     \Phi'(C)=\begin{cases}
      \sigma(x)\text{ if }C=x\in S_{\cC}\\
      \sigma(y)\sigma(z)\text{ if }C=yz\text{ for }y,z\in S_{\cC}
      \Phi(C)\text{ otherwise}
     \end{cases}~.
    \end{equation}
    The condition (S2) from Definition~\ref{defn:standard} ensures that there $\Phi'$ is well-defined on objects.
    
    Now define a system of unitary isomorphisms $\beta_C:\Phi(C)\to\Phi'(C)$ by
    \begin{equation}
        \beta_C:=\begin{cases}
         \iota_x\text{ if }C=x\in S_{\cC}\\
         \tau_{y,z}\text{ if }C=yz\text{ for }y,z\in S_{\cC}\\
         \id_C\text{ otherwise}
        \end{cases}~.
    \end{equation}
    Again, condition (S2) means that $\beta_C$ is well-defined.
    
    For a morphism $f\in\cC(B\to C)$, define $\Phi'(f)=\beta_C\Phi(f)\beta_B^{\dag}$.
    It is immediate that $\Phi'$ is functorial and that $\beta:\Phi\to\Phi'$ is a unitary natural isomorphism.

    Finally, we define the tensorator $\theta'$ of $\Phi'$ by
    \begin{equation}
        \theta'_{y,z}:=(\iota_y\otimes\iota_z)\circ\theta_{y,z}\circ\tau_{y,z}^\dag~.
    \end{equation}
    It is immediate that $\theta'$ is a unitary natural isomorphism, that $\Phi'$ is a unitary braided monoidal functor, and that $\beta$ is a unitary monoidal natural isomorphism $\Phi\to\Phi'$.
    The means that $\Phi'$ is a unitar braided monoidal equivalence, and the choices \eqref{eq:stdFunctorObjs} ensure that $\Phi'$ is a standard equivalence. 
\end{const}
\begin{remark}
    The arbitrary choice of $\iota$ and $\tau$ in Construction~\ref{const:standardEq} shows that the association of a concrete gauge symbol to an equivalence is highly non-canonical, even after making choices of gauge on the source and target in the form of a based unitary ribbon category.
    This is no obstruction if we merely wish to determine when two equivalences are naturally isomorphic.
    However, when we care about the evil notion of equality of functors, we must take greater care.   
\end{remark}

A special case of equations \eqref{eq:gaugeTransformationMonoidal} and \eqref{eq:gaugeTransformationBraided} is that a standard equivalence with trivial associated gauge symbol is an \textit{equality} of $F$- and $R$-symbols.
This is of clear interest to us, because the complex conjugate $\overline{F}$ of the $F$-symbol of $\cC$ is the $F$-symbol of the complex conjugate unitary ribbon category; see {Sec.}~\ref{ssec:ribbonDefns} below.
\begin{remark}
    \label{rem:endIsTaut}
    Given a unitary braided equivalence $\Phi:\cC\to\cD$ of unitary ribbon categories, it is always possible to lift $\Phi$ to a unitary based equivalence.
    By Constructions~\ref{const:standardURC} and \ref{const:standardEq}, one can always assume that $\Phi$, $\cC$, and $\cD$ are standard. 
    One may then lift $\cC$ arbitraily to a based unitary ribbon category, set the gauge symbol associated with $\Phi$ to be $I$, and make $\cD$ based by choosing bases such that \eqref{eq:gaugeSymbolFromStdEq} holds.
    (This will be further explored in {Sec.}~\ref{sssec:2groupoids} below.)
    In other words, an equivalence $\cC\to\cD$ only means that $\cC$ and $\cD$ have the same fusion ring (up to a certain relabeling of simples) and the same $F$-symbol and $R$-symbol in some gauge.

    Autoequivalences $\cC\to\cC$ are a more interesting matter.
    When computing the gauge symbol for an autoequivalence, one uses a single basis for both the source and target $\cC$.
    Nontrivial soft autoequivalences $\cC\to\cC$ \cite{10.1063/1.4895764,Kobayashi:2025ykb} are precisely those which cannot be deformed to induce a trivial gauge symbol when using a single basis of $\cC$, but any gauge symbol could of course be trivialized if we permitted ourselves to use different row and column bases.
    For the same reasons, to show that $\cC$ has a real $F$-symbol, it will be necessary to construct a standard equivalence $\cC\to\overline{\cC}$ and trivialize the gauge symbol while maintaining the same bases for the trivalent vertex spaces of $\cC$ and $\overline{\cC}$.
    If we did not take care to preserve the alignment the bases of trivalent vertex spaces in $\cC$ and $\overline{\cC}$, we could only show that $F$ and $\overline{F}$ are related by a gauge transformation.
    Indeed, we will later find examples of categories where $F$ and $\overline{F}$ are gauge equivalent, but $F$ is not real in any gauge; see Remark~\ref{rem:FGaugeBarF}.
\end{remark}

\subsubsection{Anyon theories and unitary ribbon categories}
\label{sssec:2groupoids}
We now compile the preceding ideas to give a precise sense in which the notion of anyon theory in the sense of works like \cite{Kitaev:2005hzj,BarkeshliBondersonChengWang2019} is equivalent to the notion of unitary ribbon category (or non-degenerate unitary ribbon category), using based unitary ribbon categories as a bridge.
The overall equivalence between anyon theories and unitary ribbon categories is extremely well-understood by experts, and the details can already be found in \cite{Kitaev:2005hzj,BarkeshliBondersonChengWang2019,Moore:1988qv}.
We include this story to further illustrate how the notion of based unitary ribbon category relates to both, in preparation for the arguments of {Sec.}~\ref{sec:sufficient}.

\begin{defn}
    \label{defn:anyonTheory}
    An \textbf{anyon theory} $\mathfrak{C}$ is a commutative unital $\mathbb{Z}_+$ based ring $K_0(\mathfrak{C})$ \cite{EtingofNikshychOstrik2010,MR1976459} with distinguished basis $A_{\mathfrak{C}}$ and $1_{K_0(\mathfrak{C})}\in A_{\mathfrak{C}}$, together with an involution $(\cdot)^\vee$ and unitary $F$-symbol and $R$-symbols.
    We define the fusion coefficients $N_x^{yz}$ by
    \begin{equation}
     yz=\sum_{x\in A_{\mathfrak{C}}}N_x^{yz}x~,
    \end{equation}
    and require that $N_1^{xy}=\delta_{y=x^\vee}$.
    We also require that the $F$-symbol satisfies the pentagon equation, and that the $F$- and $R$-symbols satisfy the hexagon equation. 
\end{defn}
In many sources, an anyon theory would carry an additional assumption which ensures non-degeneracy of the $S$-matrix.
By not doing so, we allow anyon theories that describe a proper subset of anyon types which is closed under fusion, just as unitary ribbon categories arise as full fusion subcategories of unitary modular tensor categories.

The following two definitions, which mirror those of \cite[{Sec.}~III]{BarkeshliBondersonChengWang2019}, capture unitary braided monoidal equivalences and unitary monoidal natural transformations from the perspective of anyon theory.   
\begin{defn}
    A \textbf{gauge transformation} $U$ between anyon theories $\mathfrak{C}$ and $\mathfrak{D}$ is a bijection $\sigma_U:A_{\mathfrak{C}}\to A_{\mathfrak{D}}$ together with unitary matrices $U_{x}^{yz}$ for each $x,y,z\in A_{\mathfrak{C}}$ satisfying equations \eqref{eq:gaugeTransformationMonoidal} and \eqref{eq:gaugeTransformationBraided}.
\end{defn}
The composition of gauge transformations $V$ and $U$ is a gauge transformation given by
\begin{equation}
    \label{eq:ATGTComp}
    \sigma_{V\circ U}=\sigma_V\circ\sigma_U,\qquad
    [V\circ U]_x^{yz}=V_{\sigma_U(x)}^{\sigma_U(y)\sigma_U(z)}U_x^{yz}~.
\end{equation}
\begin{defn}
    \label{defn:ATNat}
    For gauge transformations $U,V:\mathfrak{C}\to\mathfrak{D}$, a \textbf{natural transformation} $\beta:U\to V$ is the equality $\sigma_U=\sigma_V$, together with a collection of phases $(\beta_x)_{x\in A_{\mathfrak{C}}}$ such that, for all $x$, $y$, and $z\in A_{\mathfrak{C}}$, 
    \begin{equation}
        \label{eq:ATNat}
        V_x^{yz}=\frac{\beta_y\beta_z}{\beta_x}U_x^{yz}~.
    \end{equation}
\end{defn}
Equation~\ref{eq:ATNat} appears as \cite[Eq.~(91)]{BarkeshliBondersonChengWang2019}\cite{Kitaev:2005hzj}.
\begin{defn}
    \label{defn:ATDag}
    For a unitary natural transformation $\beta:U\to V$ in the sense of Definition~\ref{defn:ATNat}, the \textbf{adjoint} natural transformation $\beta^\dag:V\to U$ is given by
    \begin{equation}
        (\beta^\dag)_x:=\overline{\beta_x}~.
    \end{equation}
\end{defn}
The vertical composition of natural transformations $U\xrightarrow{\beta}V\xrightarrow{\gamma}W$
is given by
\begin{equation}
    \label{eq:ATVerticalComp}
    [\gamma\circ\beta]_x=\gamma_x\beta_x~.
\end{equation}
The horizontal composition of natural transformations $\beta:U_1\to U_2$ and $\gamma:V_1\to V_2$ is given by
\begin{equation}
    \label{eq:ATHorizontalComp}
    [\gamma\cdot\beta]_x:=\gamma_{\sigma_U(x)}\beta_x~.
\end{equation}
Note that we may write $\sigma_U$ in equation~\eqref{eq:ATHorizontalComp} since the existence of $\beta$ implies that $\sigma_{U_1}=\sigma_{U_2}$.

\begin{defn}
    A \textbf{$\dag$-$2$-groupoid} is a $2$-category in which all $2$-morphisms are invertible, all $\Hom$ $1$-categories are $\dag$-categories, and all $1$-morphisms are invertible up to unitary $2$-morphism.
\end{defn}
\begin{defn}
    There is a $\dag$-$2$-groupoid $\AT$ whose objects are anyon theories, $1$-morphisms are gauge transformations, and $2$-morphisms are natural transformations in the sense of Definition~\ref{defn:ATNat}.
\end{defn}
It is straightforward to check that the compositions of gauge transformations and natural transformations from equations \eqref{eq:ATGTComp}-\eqref{eq:ATHorizontalComp} and the operation $\dag$ from Definition~\ref{defn:ATDag} assemble to form a $\dag$-$2$-category.

\begin{defn}
  There is a $\dag$-$2$-groupoid $\URC^*$ with the following description:
\begin{itemize}
    \item Objects are unitary ribbon categories
    \item $1$-morphisms are braided monoidal $\dag$-equivalences
    \item $2$-morphisms are unitary \textit{monoidal} natural transformations
\end{itemize}
\end{defn}
To carefully construct $\URC^*$, one first constructs a $\dag$-$2$-category $\URC$ including all braided monoidal $\dag$-functors and all natural transformations, and then restricts to $\URC^*$, which is essentially the unitary part of the core.

By similar means, one can construct a $2$-groupoid of based unitary ribbon categories $\URC^\based$.
\begin{defn}
    There is a $\dag$-$2$-groupoid $\URC^\based$ with the following description
    \begin{itemize}
        \item Objects are based unitary ribbon categories
        \item $1$-morphisms are standard equivalences
        \item $2$-morphisms are monoidal natural unitaries (with no extra data or conditions)
    \end{itemize}
\end{defn}

We will now show that these three $2$-groupoids are equivalent by describing equivalences from $\URC^{\based}$ to the other two.
\begin{nota}
    We denote by $\cU:\URC^\based\to\URC^*$ the forgetful functor, which forgets the additional data and properties that objects, $1$-morphisms, and $2$-morphisms are based and remembers only the underlying ribbon category, monoidal functor, or natural transformation respectively.
\end{nota}

\begin{defn}
    \label{defn:ex}
    We define a $\dag$-$2$-functor $\Ex:\URC^\based\to\AT$ as follows:
    \begin{itemize}
        \item Given a based unitary ribbon category $(\cC,S,(\gamma))$, we may extract an anyon theory $\Ex(\cC)$ with $A_{\Ex(\cC)}:=S$, fusion multiplicities $N_x^{yz}=\dim(\cC(x\to yz))$, and $F$- and $R$-symbols determined by equations \eqref{eq:FFromBasedF1C} and \eqref{eq:RFromBasedF1C}.
        \item Given a standard equivalence $(\Phi,\phi):\cC\to\cD$, we extract a gauge symbol $\Ex(\Phi)$ by \eqref{eq:gaugeSymbolFromStdEq}.
        \item Given a natural unitary $\eta:(\Phi,\phi)\to(\Psi,\psi)$, since both $\Phi$ and $\Psi$ are standard, we know that for $x\in S_{\cC}$, $\Phi(x)=\sigma(x)=\Psi(x)$.
        We therefore define phases $(\beta_x)_{x\in A_{\Ex(\cC)}}$ by
        \begin{equation}
            \beta_x\id_{\sigma(x)}:=\eta_x~.
        \end{equation}
        That these phases satisfy the correct equation follows from monoidality of $\eta$.
    \end{itemize}
\end{defn}
Checking that $\Ex$ is really a $2$-functor is a standard exercise in translating between coherence equations for unitary ribbon categories, evaluated at the chosen simples $S_{\cC}$, and the coherence equations for anyon theories.
\begin{prop}
    \label{prop:URCBased=AT}
    The $\dag$-$2$-functor $\Ex:\URC^{\based}\to\AT$ is an equivalence of $2$-groupoids.
\end{prop}
\begin{proof}[Proof sketch]
    To see that $\Ex$ is essentially surjective on objects, we observe that the data of an anyon theory $\mathfrak{C}$ is sufficient to define a tensor product and braiding on the $\dag$-category of $A_{\mathfrak{C}}$-graded vectorspaces, yielding a braided monoidal $\dag$-category $\cC$.
    Indeed, this is some of the content of \cite[Appenix~E.7]{Kitaev:2005hzj}.
    The involution condition of Definition~\ref{defn:anyonTheory} means that $\cC$ is rigid, which, together with unitarity, will make $\cC$ a unitary ribbon category when equipped with a braided balanced unitary dual functor and the associated canonical spherical structure.
    In general, $\cC$ will not be standard, but it may be standardized by means of Construction~\ref{const:standardURC}.
    
    Similarly, to see that $\Ex$ is an equivalence on $\Hom$ $1$-categories, we first see that $\Ex$ is essentially surjective on each.
    Given unitary based ribbon categories $\cC$ and $\cD$ and a gauge transformation $U:\Ex(\cC)\to\Ex(\cD)$, we need to define a corresponding standard equivalent $\Phi:\cC\to\cD$ such that $\Ex(\Phi)=U$.
    The permutation $\sigma_U$ determines the images of objects in $S_{\cC}^*$ under $\Phi$, and the gauge symbols $U_x^{yz}$ determine images of trivalent basis vectors for $\cC$ so that $\Phi$ becomes monoidal with trivial tensorator for objects in $S_{\cC}^*$.
    Since every object in $\cC$ is unitarily isomorphic to a subobject of some object in $S_{\cC}^*$, there is a unique extension of $\Phi$ to the objects and $\Hom$-spaces not in $S_{\cC}^*$ up to unique unitary monoidal natural isomorphism.
    
    Finally, a natural isomorphism of standard equivalences is entirely determined by the components at simple objects, which are the data given by a natural isomorphism in the sense of Definition~\ref{defn:ATNat}, so $\Ex$ is faithful on $2$-morphism spaces.
    To see that $\Ex$ is full, suppose that $\Phi$ and $\Psi$ are standard equivalences and $(\beta_x)$ is a natural isomorphism $\Ex(\Phi)\to\Ex(\Psi)$ with lift $\eta:\Phi\to\Psi$.
    Equation~\ref{eq:ATNat} precisely ensures that $\eta$ is indeed a monoidal natural isomorphism, so that $\eta$ really is a $2$-morphism of $\URC^{\based}$.
\end{proof}

\begin{prop}
    \label{prop:URCBased=URC}
    The forgetful $\dag$-$2$-functor $\cU:\URC^{\based}\to\URC$ is an equivalence of $2$-groupoids.
\end{prop}
\begin{proof}
    By Construction~\ref{const:standardURC}, we know that there is a standard unitary ribbon category equivalent to any unitary ribbon category, and we can always choose some based structure for this standard form.
    Therefore, $\cU$ is essentially surjective.

    Now let $\cC$ and $\cD$ be two based unitary ribbon categories.
    We need to check that $\cU_{\cC\to\cD}:\URC^{\based}(\cC\to\cD)\to\URC^*(\cU(\cC)\to\cU(\cD))$ is an equivalence of $1$-categories.
    Construction~\ref{const:standardEq} verifies that $\cU_{\cC\to\cD}$ is essentially surjective.
    The $2$-morphisms spaces in $\URC^*$ and $\URC^\based$ are the same by definition, so $\cU_{\cC\to\cD}$ is fully faithful on $\Hom$ $1$-categories.  
\end{proof}

\begin{cor}
    \label{cor:URC=AT}
    The $\dag$-$2$-groupoids $\AT$ and $\URC^*$ are equivalent.
\end{cor}
\begin{proof}
    Both are equivalent to $\URC^{\based}$.
\end{proof}

\subsection{Structural involutions of unitary ribbon categories and comparison functors}
\label{ssec:ribbonDefns}
In this subsection, we formalize the involutions on unitary ribbon categories listed in Table~\ref{CatST}.
We then package parts of the data of a unitary ribbon category as canonical braided equivalences between different rotations of $\cC$.\footnote{
Here, ``rotation'' means a composite of two of the reflections from Table~\ref{CatST}.
}
Finally, we will see that there is a unique such comparison between any two rotations of $\cC$, up to canonical monoidal natural unitary.
In particular, this gives us a canonical braided monoidal equivalence $\cC\to\overline{\cC}^{\rev}$, and also canonical braided monoidal equivalences implementing each arrow in \eqref{BraidedEquivs}.

As summarized in Table~\ref{CatST}, there are four natural ways to invert part of the structure of $\cC$, three of which have a natural interpretation as reflecting the graphical calculus of $\cC$, which we will now introduce.

\begin{defn}
 Given a unitary ribbon category $\cC$, the \emph{opposite category} $\cC^{\op}$ has the same objects, but $\cC^{\op}(x\to y):=\cC(y\to x)$, and the composition $\circ^{\op}$ of morphisms is given by $f\circ^{\op}g:=g\circ f$, where $\circ$ is the original composition of $\cC$.
 The tensor products of objects and morphisms, as well as the dagger structure are left unchanged.
 Coherence natural isomorphisms, including the associator, unitors, and braiding, are replaced with their inverses.
\end{defn}

\begin{remark}
 \label{rem:opBraidingChoice}
 The involution $(\cdot)^{\op}$ is a strict (symmetric monoidal) $2$-functor $\URC\to\URC^{2-\op}$, where $(\cdot)^{2-\op}$ means that the composition of $2$-morphisms is reversed. 
 Strictness means, in particular, that $(\cC^{\op})^{\op}=\cC$ for any unitary ribbon category $\cC$.
 A braided monoidal functor $F:\cC\to\cD$ gives a braided monoidal functor $F^{\op}:\cC^{\op}\to\cD^{\op}$, which acts the same on objects and morphisms.
 However, like the coherence natural isomorphisms for $\cC$ and $\cD$, the tensorator of $F^{\op}$ is inverse to the tensorator of $F$.
\end{remark}

\begin{remark}
    We choose the braiding $\beta^{\op}$ of $\cC^{\op}$ to be
    \[\beta^{\op}_{x,y}:=\beta_{x,y}^{-1},\]
    rather than $\beta_{y,x}$, so that $(\cdot)^{\op}$ is a reflection in graphical calculus. This is because, in both $\beta_{x,y}$ and $\beta_{x,y}^{-1}$, the same strand (the $x$ strand, in our conventions) is on top of the crossing.
    Some works instead adopt the convention $\beta^{\op}_{x,y}=\beta_{y,x}$.
    This has the effect of making $\op$ into a rotation of the graphical calculus, and exchanges $\cC^{\op}$ and $\cC^{\op,\rev}$ compared to our notation.
\end{remark}

\begin{defn}
    Given a unitary ribbon category $\cC$, the \emph{monoidal opposite} $\cC^{\mop}$ has the same underlying $\dag$-category $\cC$, but a different monoidal product $\otimes^{\mop}$ given by 
    \[x\otimes^{\mop}y:=y\otimes x~.\]
     The associator of $\otimes^{\mop}$ is $\alpha^\dag$, the inverse of the associator of $\alpha$, since reversing the monoidal product exchanges left and right bracketings.
     Left and right unitors are exchanged.
     We equip $\cC^{\mop}$ with the braiding $\beta^{\mop}_{x,y}:=\beta_{x,y}^{-1}$.
\end{defn}

\begin{remark}
    $(\cdot)^{\mop}$ is a strict (symmetric monoidal) $2$-functor $\URC\to\URC$.
    Given a braided monoidal functor $F:\cC\to\cD$ with tensorator $\mu$, the tensorator of $F^{\mop}$ is
    $\mu^{\mop}_{x,y}:=\mu_{y,x}$.
\end{remark}

\begin{remark}
    \label{rem:mopBraidingChoice}
    Similar to Remark~\ref{rem:opBraidingChoice}, some works instead adopt the convention $\beta^{\mop}_{x,y}:=\beta_{y,x}$.
    Our convention preserves the fact that the $x$ string is on top of the crossing, making $\mop$ into a reflection.
    Adopting the alternative convention would have the effect of exchanging the roles of $(\cdot)^{\mop}$ and $(\cdot)^{\mop,\rev}$.
\end{remark}

\begin{defn}
     Given a unitary ribbon category $\cC=(\cC,\beta)$, the \textbf{braided reverse} $\cC^{\rev}=(\cC,\beta^{\rev})$ is the same underlying unitary fusion category, but with the reversed braiding
     \[\beta^{\rev}_{x,y}:=\beta_{y,x}^\dag~.\]
\end{defn}

\begin{remark}
     As with $(\cdot)^{\mop}$, $(\cdot)^{\rev}$ is a strict involution on $\URC$.
     In particular, the braided reverse $F^{\rev}$ of a functor $F:\cC\to\cD$ is precisely the same monoidal $\dag$-functor, since being a \textit{braided} monoidal functor is only a property, not an additional structure.
     The effect of reversing the braiding on graphical calculus is to swap overcrossings and undercrossings in all string diagrams, which is the same as reflecting across a plane parallel to the viewing projection.
\end{remark}

\begin{defn}
    Given a unitary ribbon category $\cC$, the \textbf{complex conjugate} category $\overline{\cC}$ is the same $\mathbb{R}$-linear $\dag$ ribbon category, but with the conjugate scalar multiplication.
    That is, for $f\in\overline{\cC}(x\to y):=\cC(x\to y)$,
    \[i\cdot_{\overline{\cC}}f:=-i\cdot_{\cC}f~.\]
\end{defn}

\begin{remark}
    Taking complex conjugate is a strict (symmetric monoidal) $2$-functor $\URC\to\overline{\URC}$, the complex conjugate of $\URC$.
    This is because complex conjugating every component conjugates every natural isomorphism.
\end{remark}

\begin{remark}
    Although the associator and braiding coherators of $\overline{\cC}$ are the same as those from $\cC$, conjugating the scalar multiplication has the effect of conjugating their matrix entries in any basis.
    Thus, the $F$-symbols and $R$-symbols of $\overline{\cC}$ are the entry-wise complex conjugates of those from $\cC$.
    In particular, if $\cC$ and $\overline{\cC}$ have the same $F$-symbol in some common gauge, then that $F$-symbol is real.
\end{remark}

\begin{remark}
    The four involutions $(\cdot)^{\op}$, $(\cdot)^{\mop}$, $(\cdot)^{\rev}$,and $\overline{(\cdot)}$, are both strict and strictly commuting, meaning that for $I,J\in\{(\cdot)^{\op},(\cdot)^{\mop},(\cdot)^{\rev},\overline{(\cdot)}\}$ and $\cC\in\URC$,
    \[IJ(\cC)=JI(\cC)~,\]
    as unitary ribbon categories.
    We may therefore reorder the several involutions without ambiguity.
\end{remark}

\begin{remark}
 Each of the above involutions applies more generally than the setting of a unitary ribbon category.
 For example, $\op$, $\mop$, and $\overline{(\cdot)}$ all make sense for a unitary fusion category without a braiding.
 This is particularly relevant because the statement that $F$ is real in some gauge is only a statement about the monoidal structure and does not involve the braiding.
\end{remark}

Now that we have introduced the above rack of hats which a unitary ribbon category may wear, we introduce several accompanying \textbf{structural functors}, braided monoidal equivalences which compare two of the categories which may be obtained from $\cC$ via structural involutions.

\begin{defn}
    The \textbf{dagger functor} on a unitary ribbon category $\cC$ is the braided monoidal equivalence $(\cdot)^{\dag}:\cC\to\overline{\cC}^{\op}$ which acts as the identity on objects and maps a morphism $f\in\cC(X\to Y)$ to $f^\dag\in\overline{\cC}^{\op}(X\to Y):=\overline{\cC}(Y\to X)$.
\end{defn}

\begin{defn}
    A \textbf{balanced braided unitary right dual functor} on a unitary ribbon category $\cC$ is a braided monoidal equivalence $\pi:\cC\to\cC^{\op,\mop}$ such that for all objects $x\in\cC$, $\pi(x)\cong x^\vee$, the images of morphisms and tensorators for $\pi$ are determined using choices of duality data, as in \cite{MR4133163}, and for every $x\in\cC$,
    \begin{equation}
        \label{eq:balanced}
        \Tr_L(\id_x)=\Tr_R(\id_x)~,
    \end{equation}
    where $\Tr_L$ and $\Tr_R$ are the left and right pivotal traces determined by the canonical unitary pivotal structure.
    Note that to be a braided monoidal equivalence, $\pi$ must further satisfy the condition that, for all objects $x$ and $y\in\cC$,
    \begin{equation}
     \label{eq:braidedDualFunctor}
     \tikzmath{
        \draw[mid>] (.5,0) .. controls ++(0,.7) and ++(0,-.7) .. (0,1) arc (180:90:1);
        \draw[mid>] (2,0) -- (2,1) arc (0:90:1);
        \draw[knot,mid>] (0,0) .. controls ++(0,.7) and ++(0,-.7) .. (.5,1) arc (180:90:.5);
        \draw[mid>] (1.5,0) -- (1.5,1) arc (0:90:.5);
        \node at (-.15,-.2) {$\scriptstyle\pi(x)$};
        \node at (.65,-.2) {$\scriptstyle\pi(y)$};
        \node at (1.5,-.2) {$\scriptstyle x$};
        \node at (2,-.2) {$\scriptstyle y$};
     }
     \,\, 
     =
     \tikzmath{
        \draw[mid>] (2,0) .. controls ++(0,.7) and ++(0,-.7) .. (1.5,1) arc (0:90:.5);
        \draw[mid>] (.5,0) -- (.5,1) arc (180:90:.5);
        \draw[knot,mid>] (1.5,0) .. controls ++(0,.7) and ++(0,-.7) .. (2,1) arc (0:90:1);
        \draw[mid>] (0,0) -- (0,1) arc (180:90:1);
        \node at (-.15,-.2) {$\scriptstyle\pi(x)$};
        \node at (.65,-.2) {$\scriptstyle\pi(y)$};
        \node at (1.5,-.2) {$\scriptstyle x$};
        \node at (2,-.2) {$\scriptstyle y$};
     }
    \end{equation}
    where the four caps are part of the duality data associated with $\pi(x)$ and $\pi(y)$ as right duals.
\end{defn}
\begin{remark}
    The condition that $\pi$ be \emph{balanced} means that $\pi$ is $1$-balanced in the sense of \cite[{Sec.}~3.4]{MR4133163}; it is equivalent to asking that the canonical unitary pivotal structure associated with $\pi$ is the canonical spherical structure.
\end{remark}
\begin{nota}
    We denote the choice of duality data associated with a right dual functor $\pi$ by $e_x:\pi(x)x\to1$ and $c_x:1\to x\pi(x)$.
\end{nota}

\begin{defn}
    The \textbf{braiding functor} on a unitary ribbon category $\cC$ is the braided monoidal equivalence $B:\cC\to\cC^{\mop,\rev}$ which has underlying functor $B=\Id_{\cC}$ and the braiding $\beta_{x,y}:xy=B(xy)\to B(x)\otimes^{\mop}B(y)=yx$ as tensorator.
    The \textbf{reverse braiding functor} $\refB:\cC\to\cC^{\mop,\rev}$ instead has a tensorator $\mu$ given by the reversed braiding: $\mu_{x,y}=\beta_{y,x}^{-1}$.
\end{defn}

\begin{remark}
    The functors $B$ and $\refB$ are strict inverses, in the sense that
    \[\refB^{\mop,\rev} B=B^{\mop\rev}\refB=\Id_{\cC}~.\]
    We thus have $\mathbb{Z}$-many braided autoequivalences of $\cC$ from the braiding, namely $(B^{2n})_{n\in\mathbb{Z}}$, where $B^{-2n}:=\refB^{2n}$.
    There is a canonical equivalence $B^2\cong\Id_{\cC}$ of braided autoequivalences following from balancing \cite[Ch.~2]{MR1797619}.
\end{remark}

Each of these three structural equivalences above square to the identity, up to canonical unitary isomorphism.
That is, there is an equality
\begin{equation}
    \label{eq:doubleDagContracts}
    \overline{(\cdot)^\dag}^{\op}\circ(\cdot)^\dag=\Id_{\cC}~,
\end{equation}
a canonical braided monoidal equivalence
\begin{equation}
    \label{eq:doubleBraidContracts}
    \beta^2:B^{\mop,\rev}\circ B\to\Id_{\cC}~,
\end{equation}
given by the double braiding, and a canonical braided monoidal equivalence
\begin{equation}
    \label{eq:pivotalStructure}
    \rho:\pi^{\op,\mop}\circ\pi\to\Id_{\cC}~,
\end{equation}
namely the canonical unitary spherical structure.
Thus, although we can compose each structural equivalence with itself many times, different odd powers of $(\cdot)^\dag$, $B$, or $\pi$ only yield a single equivalence up to canonical unitary isomorphism.

\begin{remark}
    \label{rem:goodInnerProduct}
    Combining the $\dag$-category structure on $\cC$, a choice of unitary dual functor, and the canonical compatible spherical structure \cite{MR4133163}, we obtain a canonoical (unnormalized) spherical trace $\tr$ on endomorphism spaces of objects of $\cC$.
    From this, we obtain an inner product on all morphism spaces between tensor products of simple objects of $\cC$, making them into Hilbert spaces.
    For morphisms $f$ and $g\in\cC\left(\prod_{j=1}^nx_j,\prod_{k=1}^my_k\right)$, we define
    \begin{equation}
        \label{eq:tracialIP}
        \iprod{f}{g}:=\frac{1}{\sqrt{\prod_{j=1}^nd_{x_j}\prod_{k=1}^md_{y_k}}}\tr(f^\dag g)~.
    \end{equation}
    The benefits of this choice of inner product are several.
    First, it is rotationally invariant, as both the spherical trace and the scalar are rotationally invariant by construction.
    Second, the choice of normalization means that the composition of morphisms from orthonormal bases along a single edge (such as when using chosen bases of trivalent hom-spaces to produce bases of tetravalent hom-spaces in order to compute an $F$-symbol) again produces an orthonormal basis.
    Consequently, unitary morphisms are of norm $1$, and so are their rotations.
    This includes identity morphisms and duality data (cup and cap). 
\end{remark}

\begin{remark}
    The involutions $(\cdot)^{\op}$, $(\cdot)^{\mop}$, $(\cdot)^{\rev}$, and $\overline{(\cdot)}$ all make sense for based unitary ribbon categories.
    Standard sets of simple objects are preserved by each involution.
    If $\cC$ is equipped with the bases $B_x^{yz}=\{\gamma_x^{yz}[j]\}$, then $(B_x^{yz})_{x,y,z\in S_{\cC}}$ is also a system of orthonormal bases for $\cC^{\mop}$, $\cC^{\rev}$, and $\overline{\cC}$.
    For $\cC^{\op}$, we instead take $\{\gamma^\dag|\gamma\in B_x^{yz}\}$ as the basis for $\cC^{\op}(x\to yz)=\cC(yz\to x)$.

    The dagger functor then becomes a standard equivalence with trivial gauge symbol.
    Not every balanced braided unitary dual functor is a standard equivalence, but one may always choose a standard banaced braided unitary dual functor, which we henceforth call a \textit{standard dual functor} for short.
    
    The braiding functor $B:\cC\to\cC^{\mop,\rev}$ is not be a standard equivalence, because it is not monoidal on distinguished simple objects (as $\Id_{\cC}$ is not antimonoidal).
    However, it is straightforward to see that there is a standard equivalence $B'$ unitarily equivalent to $B$ where, for $x,y,z\in S_{\cC}$ and $\gamma\in \cC(x\to yz)$, we have
    \begin{equation}
        B'(yz)=zy,\qquad\text{ and }\quad B(\gamma)=\beta_{y,z}\circ\gamma~.
    \end{equation}
    The tensorator $B'_{y,z}$ is then necessarily trivial, so that the associated gauge symbol is the $R$-matrix $R_x^{yz}$.
    From the based perspective, this functor is behind the operations discussed in {Sec.}~\ref{sssec:diagrammaticBraidingMatrix}. 
\end{remark}

\subsubsection{The pivotal structure associated with a standard dual functor}
We now recall some important facts about balanced unitary dual functors and the canonical spherical structure, as expressed in terms of based unitary ribbon categories and anyon theories.
Given a standard unitary ribbon category $(\cC,S)$, we may take the balanced dual functor $\pi:\cC\to\cC^{\op,\mop}$ to be a standard equivalence.
Then $\pi(x)\in S$ is uniquely determined by $\pi(x)\cong x^\vee$, and $\pi(\pi(x))=x$ for all $x\in S$, so that $\rho_x:x\to x$ is an endomorphism of the simple $x$.
The following result tells us the possible values of $\rho$.
\begin{prop}
    \label{prop:standardBalancedDualFunctor}
    Let $(\cC,S)$ be a unitary ribbon category.
    Let $(\varkappa_x)_{x\in S}$ be a function $S\to U(1)$ so that, if $x$ is self-dual, then $\varkappa_x$ is the second Frobenius-Schur indicator of $x$.
    Then there is a standard balanced braided unitary dual functor $\pi:\cC\to\cC^{\op,\mop}$ so that the canonical spherical structure is given by
    \begin{equation}
    \label{eq:defnVarkappa}
     \rho_x=\varkappa_x\id_x~.
    \end{equation}
    Moreover, every standard balanced unitary dual functor is of this form.
\end{prop}
\begin{proof}
    By \cite[Cor.~3.10]{MR4133163}, we may compute $\rho_x$ with the following equation.
    \begin{equation}
        \label{eq:unitaryPivotalStructureShort}
        \rho_x=(\id_x\otimes e_x)\circ(e_{\pi(x)}^\dag\otimes\id_x)~.
    \end{equation}
    
    In case $x\in S$ is self-dual, then $\varkappa_x=\pm 1$ is the second Frobenius-Schur indicator of $x$ \cite{simon2022straighteningfrobeniusschurindicator}.
    Equation~\eqref{eq:unitaryPivotalStructureShort} becomes
    \begin{equation}
        \rho_x=(\id_x\otimes e_x)\circ(e_x^\dag\otimes\id_x)~,
    \end{equation}
    and the main point of \cite{simon2022straighteningfrobeniusschurindicator} is that the right-hand side is precisely the Frobenius-Schur indicator times $\id_x$.
    
    On the other hand, if $x$ is not self-dual, then the value of $\varkappa_x$ depends on the choice of unitary dual functor $\pi$.
    If $(e_x,c_x)$ is a choice of duality data expressing $\pi(x)$ as the right dual of $x$, then a valid choice of duality data expressing $\pi(\pi(x))=x$ as the right dual of $\pi(x)$ is $(c_x^\dag,e_x^\dag)$.
    For any scalar $\varkappa_x$, we may rescale the duality data to $(\varkappa_x c_x^\dag,\varkappa_x^{-1} e_x^\dag)$ while preserving the zig-zag equations.
    If and only if $\varkappa_x\in U(1)$, both the left and right traces remain equal, so that we still have a balanced dual functor.
    Equation~\eqref{eq:unitaryPivotalStructureShort} then becomes
    \begin{equation}
        \rho_{\pi(x)}=(\id_{\pi(x)}\otimes \varkappa_xc_x^\dag)\circ(e_x^\dag\otimes\id_{\pi(x)})=\varkappa_x\id_{\pi(x)}~,
    \end{equation}
    by the zig-zag equation.
\end{proof}
As foreshadowed in \eqref{eq:A-local-sequence}, the scalars $\varkappa_x$ of \eqref{eq:defnVarkappa} associated with the pivotal structure will contribute to the key gauge transformation.
Proposition~\ref{prop:standardBalancedDualFunctor} shows that non self-dual objects are a source of gauge freedom in deforming the gauge symbol.

\subsubsection{Commutativity of structural functors and autoequivalences}
\label{sssec:commutingSquares}
We now prove the existence of several commuting squares consisting of the structural functors associated with duals, braiding, and the $\dag$-structure and an arbitrary autoequivalence $J:\cC\to\cC$.
Some squares commute strictly, while others commute only up to a canonical monoidal unitary natural isomorphism.
The upshot is that if $\bm{a}$ and $\bm{b}$ are arbitrary composites of the structural involutions $(\cdot)^{\op}$, $(\cdot)^{\mop}$, $(\cdot)^{\rev}$, and $\overline{(\cdot)}$, then there is at most one comparison functor $\cC^{\bm{a}}\to\cC^{\bm{b}}$ obtained by composing the structural functors introduced in the previous section, up to canonical unitary isomorphism.
Consequently, an arbitrary composite of these structural functors and $J$ can be reordered (formally applying the involutions to one or more functors in the process) arbitrarily, with all the possible reorderings being unitarily equivalent.

\begin{lem}
    \label{lem:BDagSquare}
    For a unitary ribbon category $\cC$, the following square strictly commutes.
    \[\begin{tikzcd}
        \cC \arrow[r,"B"] \arrow[d,"(\cdot)^\dag"']
        & \cC^{\mop,\rev}
        \arrow[d,"(\cdot)^\dag"]
        \\ \overline{\cC}^{\op} \arrow[r,"B"]
        & \overline{\cC}^{\op,\mop,\rev}
    \end{tikzcd}\]
\end{lem}
\begin{proof}
    Each functor $B$ and $(\cdot)^\dag$ is the identity on objects, so the two sides of the diagram agree on all objects.
    Indeed, the braided monoidal functor $B$ has underlying functor $\Id_{\cC}$, so $B(f)=f$ for any morphism $f$, while $(f)^\dag=f^\dag$, so that the two sides agree on all morphisms.
    Finally, one checks that the tensorators of $(\cdot)^\dag\circ B$ and $B\circ(\cdot)^\dag$ agree.
    The tensorator of $B$ is just the braiding, while the tensorator of $(\cdot)^\dag$ is trivial, {i.e.} $(\cdot)^\dag$ is strictly monoidal.
    Since the braiding on $\overline{C}^{\op}$ is the inverse/adjoint of the braiding on $\cC$, the two sides again agree.
\end{proof}

\begin{lem}
    \label{lem:BJSquare}
    For a unitary ribbon category $\cC$ and autoequivalence $\Phi:\cC\to\cC$, the following square strictly commutes.
    \[\begin{tikzcd}
        \cC \arrow[r,"B"] \arrow[d,"\Phi"']
        & \cC^{\mop,\rev} \arrow[d,"\Phi^{\mop,\rev}"]
        \\ \cC \arrow[r,"B"]
        & \cC^{\mop,\rev}
    \end{tikzcd}\]
\end{lem}
\begin{proof}
    Follows immediately from the fact that $\Phi$ is a braided monoidal functor.
\end{proof}

\begin{lem}
    \label{lem:dagJSquare}
    For a unitary ribbon category $\cC$ and autoequivalence $\Phi:\cC\to\cC$, the following square strictly commutes.
    \[\begin{tikzcd}
        \cC \arrow[r,"(\cdot)^\dag"] \arrow[d,"\Phi"']
        & \overline{\cC}^{\op} \arrow[d,"\overline{\Phi}^{\op}"]
        \\ \cC \arrow[r,"(\cdot)^\dag"]
        & \overline{\cC}^{\op}
    \end{tikzcd}\]
\end{lem}
\begin{proof}
    Follows immediately from the fact that $\Phi$ is a $\dag$ functor.
\end{proof}

\begin{lem}
    \label{lem:BPiSquare}
    For a unitary ribbon category $\cC$ and braided unitary dual functor $\pi:\cC\to\cC^{\op,\mop}$ the following square strictly commutes.
    \[\begin{tikzcd}
        \cC \arrow[r,"\pi"] \arrow[d,"B"']
        & \cC^{\op,\mop} \arrow[d,"\refB^{\op,\mop}"]
        \\ \cC^{\mop,\rev} \arrow[r,"\pi^{\mop,\rev}"]
        & \cC^{\op,\rev}
    \end{tikzcd}\]
\end{lem}
\begin{proof}
    This follows immediately from the fact that $\pi$ is a braided dual functor, and specifically \eqref{eq:braidedDualFunctor}.
    Indeed, braidedness of $\pi$ and equality of the tensorators of $\refB^{\op,\mop}\circ\pi$ and $\pi^{\mop,\rev}$ reduce to precisely the same equation in $\cC$.
\end{proof}

\begin{lem}
    \label{lem:JPiSquare}
    For a unitary ribbon category $\cC$, braided unitary dual functor $\pi:\cC\to\cC^{\op,\mop}$, and autoequivalence $\Phi:\cC\to\cC$ with tensorator $(\theta_{xy}^z)$, there is a canonical commuting square
    \begin{equation}
        \label{eq:piJSquare}
        \begin{tikzcd}
        \cC \arrow[r,"\Phi"] \arrow[d,"\pi"']
        & \cC \arrow[d,"\pi"] \arrow[dl,Rightarrow,"\psi"']
        \\ \cC^{\op,\mop} \arrow[r,"\Phi^{\op,\mop}"']
        & \cC^{\op,\mop}
    \end{tikzcd}
    \end{equation}
    where (suppressing associators and unitors for brevity, and using the $\circ$ and $\otimes$ of $\cC$ rather than $\cC^{\op,\mop}$)
    \begin{equation}
        \label{eq:JPi2CellAlg}
        \psi_x=((\Phi(e_x)\circ\theta^\dag_{\pi(x),x})\otimes\id_{\pi_{\Phi(x)}})\circ(\id_{\Phi(\pi(x))}\otimes c_{\Phi(x)})~.
    \end{equation}
\end{lem}
\begin{proof}
The idea behind Lemma~\ref{lem:JPiSquare} is the well-known fact that the duality data determine a unique monoidal unitary natural isomorphism between unitary dual functors.
See, for example, \cite[{Sec.}~1]{MR4133163}.
By Proposition~\ref{prop:URCBased=URC}, without loss of generality, we may assume that $\cC$ and $\Phi$ are standard, with $S_{\cC}$ the set of standard simples.
In that case, there is a standard autoequivalence $\Psi:\cC\to\cC$ adjoint to $\Phi$, which is hence an object-wise section of $\Phi$ on $S_{\cC}^*$. 
Since $\Psi$ is an object-wise section of $\Phi$ on $S_{\cC}^*$, then $\pi$ and $\Phi\pi\Psi$ are both right dual functors for $\cC$.
Whiskering with $\Phi$ on the right gives the commuting square \eqref{eq:piJSquare} and the natural isomorphism $\psi$ between the two.
\end{proof}
    
    Diagrammatically, we may depict \eqref{eq:JPi2CellAlg} as
    \begin{equation}
        \label{eq:JPi2CellPic}
        \psi_x=
        \tikzmath{
         \draw (0,0) -- (0,1.8);
         \draw (.5,1.8) -- (.5,3.2);
         \draw (1,.4) -- (1,1.8);
         \draw (2,.4) -- (2,3.6);
         \roundNbox{fill=white}{(.5,1.8)}{.4}{.3}{.3}{$\scriptstyle\theta_{\pi(x),x}^\dag$}
         \roundNbox{fill=white}{(.5,3.2)}{.4}{0}{0}{$\scriptstyle\Phi(e_x)$}
         \roundNbox{fill=white}{(1.5,.4)}{.4}{.3}{.3}{$\scriptstyle c_{\Phi(x)}$}
         \node at (0,-.2) {$\scriptstyle\Phi\pi(x)$};
         \node at (2,3.8) {$\scriptstyle\pi\Phi(x)$};
        }
    \end{equation}

Finally, we turn to the task of producing a relationship between the dagger functor and dual functor.
Instead of phrasing this relationship as a commuting square, it will later be more convenient to have it in a slightly different form than the others.
Before stating the main results, we recollect some necessary facts about balanced unitary dual functors.
These facts all certainly appear in \cite{MR4133163,etingof2015tensor,Selinger2011} in certain forms, but we include them here for self-containedness and to set conventions.
\begin{lem}
    \label{lem:zetaSwitch}
    Let $X$ be an object in a unitary ribbon category $\cC$.
    Let $(Y_1,e_1,c_1)$ and $(Y_2,e_2,c_2)$ be two right duals of $X$, meaning that the duality data are $e_j:Y_jX\to 1$ and $c_j:1\to XY_j$.
    Then there is a unique isomorphism $\zeta:Y_1\to Y_2$ such that $e_1=e_2\circ \zeta\otimes\id_X$ and $c_2=\id_x\otimes\zeta$, namely 
    \begin{equation}
        \label{eq:dualInterchange}
        \zeta=(e_1\otimes\id_{Y_2})\circ(\id_{Y_1}\otimes c_2)~.
    \end{equation}
    Moreover, if $\zeta$ is unitary, then
    \begin{equation}
        \label{eq:zetaSwitch}
        \zeta^\dag=(e_2\otimes\id_{Y_1})\circ(\id_{Y_2}\otimes c_1)~.
    \end{equation}
\end{lem}
\begin{proof}
    That $\zeta$ is the unique morphism with the given property follows immediately from the zig-zag equations; see, for example, \cite[Rem.~2.3]{MR4133163}.
    
    If we exchange the roles of $Y_1$ and $Y_2$, we find that there is a unique isomorphism $\xi:Y_2\to Y_1$ such that
    $e_2=e_1\circ\zeta\otimes\id_X$ and $c_1=\id_x\otimes\xi$.
    On the one hand, by \eqref{eq:dualInterchange}, $\xi$ must be the right side of \eqref{eq:zetaSwitch}.
    On the other hand, it is straightforward to show that $\zeta^{-1}$ also satisfies the required properties, so $\xi=\zeta^{-1}$.
    When $\zeta$ is unitary, $\zeta^{-1}=\zeta^\dag$, completing the proof.
\end{proof}

As an example, given objects $X,Y\in\cC$, the dual functor $\pi$ gives us two right duals for $XY$.
One is $(\pi(XY),e_{XY},c_{XY})$, while another is $\left(\pi(Y)\pi(X),e_Y\circ(\id_{\pi(Y)}\otimes e_X\otimes\id_X),(\id_X\otimes c_Y\otimes\id_{\pi(X)})\circ c_X\right)$.
The tensorator of $\pi$ is just the comparison map \eqref{eq:dualInterchange} from the dual $\pi(Y)\pi(X)$ to the dual $\pi(XY)$ \cite[{Sec.}~3.2]{MR4133163}.
Note that the tensorator has the type it does because we write it as a morphism in $\cC$, while it is actually the tensorator of a functor whose target is $\cC^{\op,\mop}$.
In graphical calculus,
\begin{equation}
    \label{eq:piTensoratorGraphical}
    \theta_{x,y}:=
    \tikzmath[baseline]{
     \draw (0,0) -- (0,2.5);
     \node at (0,-.2) {$\scriptstyle\pi(Y)$};
     \draw (1,0) -- (1,1.5);
     \node at (1,-.2) {$\scriptstyle\pi(X)$};
     \draw (2,1.5) -- (2,.5);
     \draw (3,2.5) -- (3,.5);
     \draw (4,.5) -- (4,3);
     \node at (4,3.2) {$\scriptstyle\pi(XY)$};
     \roundNbox{fill=white}{(1.5,1.5)}{.3}{.4}{.4}{$\scriptstyle e_X$}
     \roundNbox{fill=white}{(1.5,2.5)}{.3}{1.4}{1.4}{$\scriptstyle e_Y$}
     \roundNbox{fill=white}{(3,.5)}{.3}{1}{1}{$\scriptstyle c_{XY}$}
    }
\end{equation}
When $\pi$ is a unitary dual functor, $\theta$ is unitary, so Lemma~\ref{lem:zetaSwitch} gives us the following alternative expression.
\begin{equation}
    \label{eq:piTensoratorSwapped}
    \theta_{x,y}=\tikzmath[baseline,xscale=-1]{
     \draw (0,0) -- (0,2.5);
     \node at (0,-.2) {$\scriptstyle\pi(Y)$};
     \draw (1,0) -- (1,1.5);
     \node at (1,-.2) {$\scriptstyle\pi(X)$};
     \draw (2,1.5) -- (2,.5);
     \draw (3,2.5) -- (3,.5);
     \draw (4,.5) -- (4,3);
     \node at (4,3.2) {$\scriptstyle\pi(XY)$};
     \roundNbox{fill=white}{(1.5,1.5)}{.3}{.4}{.4}{$\scriptstyle c_X^\dag$}
     \roundNbox{fill=white}{(1.5,2.5)}{.3}{1.4}{1.4}{$\scriptstyle c_Y^\dag$}
     \roundNbox{fill=white}{(3,.5)}{.3}{1}{1}{$\scriptstyle e_{XY}^\dag$}
    }
\end{equation}

We now have the tools to prove the following Lemma.
\begin{lem}
    \label{lem:dagPiLoop}
    Let $\cC$ be a unitary ribbon category, with $\pi:\cC\to\cC^{\op,\mop}$ a balanced braided unitary dual functor.
    Then the pivotal structure $\rho$ of $\pi$ is also a monoidal natural isomorphism
    \begin{equation}
        \label{eq:dualityCancellationType}
         \mu:\pi^{\op,\mop}\circ\overline{(\cdot)^\dag}^{\mop}\circ\overline{\pi}^{\op}\circ(\cdot)^\dag\xrightarrow{\rho}\Id_{\cC}~,
    \end{equation}
    in the sense that setting $\mu_X:=\rho_X$ determines a monoidal natural isomorphism of the correct type. 
\end{lem}
\begin{proof}
    For brevity, we denote the source of $\mu$ by $W:\cC\to\cC$, and we define
    \begin{equation}
        \label{eq:defGamma}
        \Gamma:=\overline{(\cdot)^\dag}^{\mop}\circ\overline{\pi}^{\op}\circ(\cdot)^\dag~,
    \end{equation}
    so that $W=\pi^{\op,\mop}\circ \Gamma$.
    Since the dagger functor is the identity on objects, $W$ and the double dual functor agree on objects.
    This shows that defining $\mu_x:=\rho_x$ for all objects $x\in\cC$ gives a family of morphisms of the correct types.
    We just need to check that $\mu$ is both natural and monoidal.
    Since these checks are largely routine, we outline the details.

    Although $W$ and $\pi^{\op,\mop}\circ\pi$ agree on objects, the two functors are conceptually different.
    To see why, observe that $\Gamma$ is actually a \textit{left} dual functor, specifically the one which expresses $\pi(X)$ as the left dual of $X$ using the duality data $c_X^{\dag}:X\pi(X)\to 1$, $e_X^\dag:1\to\pi(X)X$.
    By \cite[Cor.~3.10]{MR4133163},\footnote{Note that their pivotal structure is inverse to ours, as it is a map $\Id_{\cC}\to\pi^{\op,\mop}\circ\pi$. %\note{Peter}{Move this earlier, when we first cite Dave?}
    }
    the pivotal structure $\rho_x$ associated with $\pi$ is precisely given by
    \begin{equation}
        \label{eq:unitaryPivotalStructure}
        \rho_x=(e_{\pi(x)}\otimes\id_x)\circ(\id_{\pi^2(x)}\otimes e_x^\dag)=(\id_x\otimes c_{\pi(x)}^\dag)\circ(c_x\otimes\id_{\pi^2(x)})~.
    \end{equation}
    Evaluating $W$ on morphisms, we find that for any $f\in\cC(X\to Y)$,
    \begin{equation}
        \label{eq:Wf}
        W(f)=\rho_Y^\dag\circ f\circ \rho_X~.
    \end{equation}
    It immediately follows that 
    \begin{equation}
     \rho_Y W(f)=f\rho_X~,
    \end{equation}
    showing that $\mu$ is natural.

    To check that $\mu$ is monoidal, we first compute the tensorator of $W$, which is $\theta_{\pi(Y),\pi(X)}^\dag\circ\pi(\theta_{X,Y})$ since the dagger functor has trivial tensorator.
    Using \eqref{eq:piTensoratorSwapped} to expand both instances of $\theta$, we find that
    \begin{equation}
        \label{eq:WTensorator}
        \theta_{\pi(Y),\pi(X)}^\dag\circ\pi(\theta_{X,Y})=
        \tikzmath[baseline]{
         \draw (0,0) -- (0,2);
         \node at (0,-.2) {$\scriptstyle\pi^2(XY)$};
         \draw (1,2) -- (1,.5);
         \draw (2,.5) -- (2,2.5);
         \draw (3,.5) -- (3,1.5);
         \draw (4,1.5) -- (4,.5);
         \draw (5,2.5) -- (5,.5);
         \draw (6,.5) -- (6,2.5);
         \draw (7,2.5) -- (7,.5);
         \draw (8,2.5) -- (8,1.5);
         \draw (9,1.5) -- (9,3);
         \node at (9,3.2) {$\scriptstyle\pi^2(X)$};
         \draw (10,.5) -- (10,3);
         \node at (10,3.2) {$\scriptstyle\pi^2(Y)$};
         \roundNbox{fill=white}{(.5,2)}{.3}{.4}{.4}{$\scriptstyle e_{\pi(XY)}$}
         \roundNbox{fill=white}{(2,.5)}{.3}{.9}{.9}{$\scriptstyle e_{XY}^\dag$}
         \roundNbox{fill=white}{(3.5,2.5)}{.3}{1.4}{1.4}{$\scriptstyle c_X^\dag$}
         \roundNbox{fill=white}{(3.5,1.5)}{.3}{.4}{.4}{$\scriptstyle c_Y^\dag$}
         \roundNbox{fill=white}{(5,.5)}{.3}{.9}{.9}{$\scriptstyle c_{\pi(Y)\pi(X)}$}
         \roundNbox{fill=white}{(7,2.5)}{.3}{.9}{.9}{$\scriptstyle e_{\pi(Y)\pi(X)}$}
         \roundNbox{fill=white}{(8.5,.5)}{.3}{1.4}{1.4}{$\scriptstyle c_{\pi(Y)}$}
         \roundNbox{fill=white}{(8.5,1.5)}{.3}{.4}{.4}{$\scriptstyle c_{\pi(X)}$}
        }
    \end{equation}
    The middle instances of $c_{\pi(Y)\pi(X)}$ and $e_{\pi(Y)\pi(X)}$ vanish by the zig-zag equation, leaving us with
    \begin{equation}
        \label{eq:RhoWMonoidal}
        \theta_{\pi(Y),\pi(X)}^\dag\circ\pi(\theta_{X,Y})=(\rho_X^\dag\otimes\rho_Y^\dag)\circ\rho_{XY}~,
    \end{equation}
    which is exactly the condition that $\mu$ is monoidal.
\end{proof}
\begin{cor}
    \label{cor:pivotalContractionScalars}
    If $(\cC,S)$ is a standard unitary ribbon category and $\pi$ is a standard braided dual functor, then the monoidal natural isomorphism $\mu$ of \eqref{eq:dualityCancellationType} is given by
    \begin{equation}
        \mu_x=\varkappa_x\id_x~,
    \end{equation}
    for every $x\in S$,
    where $(\varkappa_x)_{x\in S}$ are the scalars from Proposition~\ref{prop:standardBalancedDualFunctor}. 
\end{cor}
\begin{remark}
    \label{rem:cancelledDualsAreDoubledDuals}
    In fact, Lemma~\ref{lem:dagPiLoop} shows that $\pi^{\op,\mop}\circ\overline{(\cdot)^\dag}^{\mop}\circ\overline{\pi}^{\op}\circ(\cdot)^\dag$ is equal to the double dual functor, although the left dual functor $\Gamma$ of \eqref{eq:defGamma} need not be equal to $\pi$.
    When $\pi$ is a standard dual functor for $(\cC,S)$, the equality $\Gamma=\pi$ holds exactly when $\varkappa_x=1$ for all non-self-dual $x\in S$.
\end{remark}

Restating Lemma~\ref{lem:dagPiLoop} in a form more similar to the previous lemmata gives the following. 
\begin{lem}
    \label{lem:dagPiSquare}
    For a unitary ribbon category $\cC$ and braided balanced unitary dual functor $\pi$, there is a canonical commuting square
        \begin{equation}
         \label{eq:dagPiSquare}
         \begin{tikzcd}
          \cC \arrow[r,"(\cdot)^\dag"] \arrow[d,"\pi"']
          & \overline{\cC}^{\op} \arrow[d,"\overline{\pi}^{\op}"] \arrow[dl,Rightarrow,"\mu'"']
          \\ \cC^{\op,\mop} \arrow[r,"((\cdot)^\dag)^{\op,\mop}"']
          & \overline{\cC}^{\mop}
         \end{tikzcd}
        \end{equation}
        where the monoidal natural unitary $\mu'$ is determined by the duality data of $\pi$.  
\end{lem}
The proof is entirely similar to Lemma~\ref{lem:dagPiLoop}. 

The commuting squares of Lemmata~\ref{lem:BDagSquare}, \ref{lem:BPiSquare}, and \ref{lem:dagPiSquare} and equations~\eqref{eq:doubleDagContracts}-\eqref{eq:pivotalStructure} show that, for any unitary ribbon category $\cC$ and sets of structural involutions $S,T\subseteq\{\overline{(\cdot)},(\cdot)^{\op},(\cdot)^{\mop},(\cdot)^{\rev}\}$ where $|S|=|T|\text{ mod }2$, there is a unique braided monoidal equivalence $\cC^S\to\cC^T$ generated by $\{\pi,(\cdot)^\dag,\pi\}$ up to canonical unitary natural monoidal equivalence.
In other words the $2$-groupoid generated by these data is generally the disjoint union of two contractible $2$-groupoids.\footnote{
For a general $\cC$, the $2$-groupoid need not actually be contractible, since some of the categories and functors could literally be equal.
For example, it is possible to construct Tanakian unitary ribbon categories where $B$ is equal to the identity functor.
A precise statement of the result should refer to a contractible $2$-groupoid within the free $\mathbb{C}$-linear braided monoidal $\dag$-category on some generators.
} 
In particular, there are unique such equivalences of each type appearing in~\eqref{BraidedEquivs}.
For example, we may choose
\begin{equation}
    \label{eq:BraidedEquivsImplementation}
    \overline{\cC}\xleftrightarrow{\overline{(\cdot)^{\dag}}}\cC^{\op}\xleftrightarrow{\pi^{\op}}\cC^{\mop}\xleftrightarrow{B^{\mop}}\cC^{\rev}~,
\end{equation}
with other choices differing by canonical natural unitaries.
There is also a unique such equivalence $\cC\to\overline{\cC}^{\rev}$, which can be obtained by composing $B$, $\pi$, and $(\cdot)^\dag$ in any order, {e.g.}
\begin{equation}
 \label{eq:RStructuralPart}
 \cC\xrightarrow{(\cdot)^\dag}\overline{\cC}^{\op}\xrightarrow{\overline{\pi}^{\op}}\overline{\cC}^{\mop}\xrightarrow{\overline{\refB}^{\mop}}\overline{\cC}^{\rev}~.    
\end{equation}
While this equivalence relates $\cC$ and $\overline{\cC}$, it is a dual functor on the level of objects and may induce a nontrivial gauge transformation on trivalent vertex spaces.

\section{Sufficient conditions for real \texorpdfstring{$F$}{F}}
\label{sec:sufficient}

In this section, we will state and prove Theorem~\ref{thm:main}, which gives sufficient conditions for the existence of a gauge in which the $F$-symbol of a unitary ribbon category $\cC$ is real.
Section~\ref{ssec:mainTheorem} will present our main results, while
{Sec.}~\ref{RealGauge} will recapitulate the proof in diagrammatic terms analogous to those of {Sec.}~\ref{sec:mechanism} in the special case where every object is self-dual, and the FS indicator is trivial.   

\subsection{Real F from a flat charge conjugation symmetry}
\label{ssec:mainTheorem}
In {Sec.}~\ref{ssec:ribbonDefns}, we saw that there is a canonical equivalence \eqref{eq:RStructuralPart} $\cC\to\overline{\cC}^{\rev}$ obtained from structural functors.
The assumptions of Theorem~\ref{thm:main} will allow us to build on this equivalence to produce an equivalence which preserves isomorphism classes of simple object and has trivial gauge symbol, becoming an equality of $F$- and $R$-symbols.

The first difficulty is that the structural equivalence $\cC\to\overline{\cC}^{\rev}$ of \eqref{eq:RStructuralPart} does not preserve isomorphism classes of simple objects, instead sending an object to its right dual.
In the event that $\cC$ has a charge-conjugation symmetry $J:\cC\to\cC$, {i.e.} a braided autoequivalence which also sends an object to its dual, we may compose \eqref{eq:RStructuralPart} with $J$ to obtain a functor $\cC\to\overline{\cC}^{\rev}$ which is (up to natural isomorphism / standardization) the identity on simple objects.
However, we will need to place requirements on $J$ to ensure that we can choose a gauge for $\cC$ which trivialize the corresponding gauge transformation. 
We incorporate some of these requirements into the following definition.
\begin{defn}
    \label{defn:CCSym}
    A \textbf{charge-conjugation symmetry} of a unitary ribbon category $\cC$ is a unitary braided monoidal equivalence $J:\cC\to\cC$ such that $J(x)\cong x^\vee$ for all $x\in\Obj(\cC)$
    together with a lift to a $\mathbb{Z}_2$-crossed braided extension of $\cC$.
\end{defn}
As a first example, if every object of $\cC$ is self-dual, then $\Id_{\cC}$ is a charge-conjugation symmetry.

We unpack the data of a lift of $J$ to a general $\mathbb{Z}_2$-crossed braided extension.
First, there must be a monoidal unitary natural isomorphism $\eta:J^2\to\Id_{\cC}$, so that we actually get a group homomorphism $\mathbb{Z}_2\to\Aut(\cC)$, the group of braided autoequivalences of $\cC$.
This rules out situations where $J^2$ is a nontrivial soft autoequivalence of $\cC$ \cite{10.1063/1.4895764,Kobayashi:2025ykb}.
Second, this $\mathbb{Z}_2$ action must be gaugeable, meaning that the group homomorphism $\mathbb{Z}_2\to\Aut(\cC)$ must lift to a $2$-group homomorphism $\mathbb{Z}_2\to\Pic(\cC)$ \cite[{Sec.}~7.8]{MR2677836} \cite{Barkeshli:2014cna}.\footnote{Physically, this statement corresponds to the fact that $\mathbb{Z}_2$ has vanishing 't Hooft anomaly (guaranteed by the fact that $H^4(\mathbb{Z}_2,U(1))$ is trivial) and that the charge conjugation does not form part of a non-split 2-group. It is in this sense that charge conjugation is gaugeable: it is non-anomalous and one is not forced to gauge a 1-form symmetry along with it. We emphasize that imposing gaugeability of charge conjugation in this sense is a priori a choice (although we are not aware of explicit examples where this is not true).}
We characterize this requirement in terms of the contribution of $\eta$ to the gauge symbol by the following Lemma.
\begin{lem}
    \label{lem:resolveCC}
    Suppose $J:\cC\to\cC$ is a charge-conjugation symmetry.
    Then there is a charge-conjugation symmetry $J':\cC\to\cC$ monoidally naturally unitary equivalent to $J$ which permutes $\Irr(\cC)$, such that $(J')^2$ has trivial gauge symbol.
\end{lem}
\begin{proof}
 By the results of {Sec.}~\ref{sssec:2groupoids}, we may assume without loss of generality that $\cC$ is a standard unitary ribbon category and that $J$ is a standard equivalence.
 The obstruction to lifting $J$ to a $\mathbb{Z}_2$-crossed braided extension is then described in \cite[{Sec.}~IV.D.1]{Barkeshli:2014cna}.
 If $\eta_x=\beta_x\id_x$ for $x\in S_{\cC}$, then the symmetry is gaugeable if and only if the obstruction $\Omega_x:=\frac{\beta_{J(x)}}{\beta_x}$ is cohomologically trivial \cite[Eq.~(191)]{Barkeshli:2014cna}.
 Altering $\Omega$ by a coboundary means replacing $\beta_x$ with $\beta'_x:=\beta_x\nu_x$ for $\nu\in\operatorname{Char}(\cU(\cC))$, the universal grading group of $\cC$ \cite[Eq.~(92)]{Barkeshli:2014cna} \cite{EtingofNikshychOstrik2010}.
 Since $\nu_x=\nu_{J(x)}^{-1}$ for such a character, the effect on $\Omega$ is to replace $\Omega_x$ with $\nu_x^2\Omega_x$.
 
 If $\Omega$ is cohomologically trivial, we may choose $\nu$ so that 
 \begin{equation}
    \label{eq:Omega1}
     \Omega_x':=\frac{\beta'_{J(x)}}{\beta'_x}=1~,
 \end{equation}
 for all $x\in S_{\cC}$.
 To implement this choice we
 keep the functor $J$ the same, while replacing the monoidal natural unitary $\eta:J^2\to\Id_{\cC}$ with $\eta'$ given by
 \begin{equation}
     \eta'_x:=\beta'_x\id_x~.
 \end{equation}
 The natural unitary $\eta'$ is still monoidal because it is the postcomposition of $\eta$ with is the monoidal natural automorphism of $\Id_{\cC}$ given by the character $(\nu_x)$.
 
 By equation \eqref{eq:Omega1}, $\beta'_x=\beta'_{J(x)}$,
 so may define $J'$ as the translation of $J$ by the natural isomorphism $\theta$ given by $\theta_x=(\beta'_x)^{1/2}\id_{J(x)}$.
 The whiskerings $J\cdot\theta$ and $\theta\cdot J$ translate the simple $x$ by $\beta_{J(x)}^{1/2}$ and $\beta_{x}^{1/2}$ respectively, but as these are the same phases, $J'^2$ is the translation of $J^2$ by $\eta$, which has trivial gauge symbol and is therefore $\Id_{\cC}$, as desired.
\end{proof}

Given a charge-conjugation symmetry $J$ in the sense of Definition~\ref{defn:CCSym}, we may write down a braided equivalence
\begin{equation}
    \label{eq:RDef}
    R:\cC\xrightarrow{(\cdot)^\dag}\overline{\cC}^{\op}\xrightarrow{\overline{\pi}^{\op}}\overline{\cC}^{\mop}\xrightarrow{\overline{\refB}^{\mop}}\overline{\cC}^{\rev}\xrightarrow{\overline{J}}\overline{\cC}^{\rev}~,
\end{equation}
which is, without loss of generality, the identity functor on objects.
Consequently, $R$ determines a gauge equivalence between $F$ and $\overline{F}$ (and between $R$ and $R^T$).
However, $R$ may still have a nontrivial gauge symbol, which we would wish to deform away by changing $R$ up to a natural isomorphism and/or choosing a different gauge for $\cC$.
As foreshadowed in {Sec.}~\ref{folding} and {Sec.}~\ref{sec:mechanism}, our strategy will be to arrange things so that $\overline{R}\circ R=\Id_{\cC}$, which will show that the gauge symbol for $R$ is symmetric, and then to apply Takagi factorization to find the appropriate change of gauge for $\cC$.

We begin by finding the gauge symbol for $\overline{R}\circ R$.
In order to do so, we write down a canonical natural isomorphism $\overline{R}\circ R\to\Id_{\cC}$ and then apply equation~\ref{eq:ATNat}.
\begin{prop}
    \label{prop:R2=1}
    Let $\cC$ be a unitary ribbon category with unitary braided right dual functor $\pi$ and canonical spherical structure $\mu:\pi^{\op,\mop}\circ\pi\to\Id_{\cC}$.
    Suppose that $J:\cC\to\cC$ is a charge-conjugation symmetry, with tensorator $(\theta_{x,y})$ and $J^2=\Id_{\cC}$.

    Define
    \begin{equation}
        \label{eq:defnR}
        R = \overline{J}^{\rev}\circ\overline{\refB}^{\mop}\circ\overline{\pi}^{\op}\circ(\cdot)^\dag:\cC\to\overline{\cC}^{\rev}~.
    \end{equation}
    Then we have $\overline{R}^{\rev}\circ R\cong\Id_{\cC}$, with the natural isomorphism given by 
    \begin{equation}
        \label{eq:R2Homotopy}
        (\id_{J^2}\cdot\mu)\circ(\id_J\cdot\psi^{\dag}\cdot\id_{\pi})\circ(\id_{J\pi(\cdot)^\dag J}\cdot\mu)~,
    \end{equation}
    where $\psi$ is the natural unitary $\pi\circ J\to J^{\op,\mop}\circ\pi$ from Lemma~\ref{lem:JPiSquare}
    and $\mu$ is the natural unitary from Lemma~\ref{lem:dagPiLoop}.
\end{prop}
\begin{proof}
    To simplify $\overline{R}\circ R$ (we henceforth omit the $(\cdot)^{\rev}$ in $\overline{R}^{\rev}$ as it is a property), we use the commuting squares constructed in {Sec.}~\ref{sssec:commutingSquares}.

    We first apply Lemmata~\ref{lem:BDagSquare}, \ref{lem:BJSquare}, and \ref{lem:BPiSquare} to bring the two instances of $\refB$ together.
    The application of Lemma~\ref{lem:BPiSquare} turns one $\refB$ into a $B$, and these cancel, as $\refB B=\Id_{\cC}$.
    We therefore obtain
    \begin{equation}
        \label{eq:R2NoB}
        \overline{R}\circ R=J\circ\pi^{\op,\mop}\circ\overline{(\cdot)^{\dag}}^{\mop}\circ \overline{J}^{\mop}\circ\overline{\pi}^{\op}\circ(\cdot)^\dag~.
    \end{equation}

    Next, we use Lemmata~\ref{lem:dagJSquare} and \ref{lem:JPiSquare} to reorder the functors so that the two applications of $J$ are consecutive, at the cost of applying the natural isomorphism $\psi$ from Lemma~\ref{lem:JPiSquare}.
    \begin{equation}
     \label{eq:R2PostPsi}
     \overline{R}\circ R=J\circ\pi^{\op,\mop}\circ\overline{J}^{\mop}\circ\overline{(\cdot)^{\dag}}^{\mop}\circ\overline{\pi}^{\mop}\circ(\cdot)^{\dag}\xrightarrow{\id_J\cdot\psi^\dag\cdot\id_{(\cdot)^\dag\pi(\cdot^\dag)}} J^2\circ\pi^{\op,\mop}\circ\left[{\overline{(\cdot)^\dag}}^{\mop}\circ\overline{\pi}^{\op}\circ(\cdot)^\dag\right]~.
    \end{equation}
    
    Finally, since $J^2=\Id_{\cC}$, the target of \eqref{eq:R2PostPsi} is unitarily equivalent to $\Id_{\cC}$ by the isomorphism $\mu$ from Lemma~\ref{lem:dagPiLoop}. 
\end{proof}

\begin{cor}
    \label{cor:R2Symbol}
    The gauge symbol associated with the autoequivalence $\overline{R}R:\cC\to\cC$ is given by
    \begin{equation}        \label{eq:R2HomotopySymbol}U_x^{yz}=\frac{\sigma_y\sigma_z}{\sigma_x}I~, \ \ \ \sigma_i:={\varkappa_i\over t_i}~,
    \end{equation}
    where $(\varkappa_x)_{x\in S_{\cC}}$ are the scalars associated with the dual functor $\pi$ by Proposition~\ref{prop:standardBalancedDualFunctor}
    and $t_x$ is the scalar such that $\psi_x=t_x\id_x$, where $\psi$ is the natural isomorphism from Lemma~\ref{lem:JPiSquare}.
\end{cor}
\begin{proof}
 By Corollary~\ref{cor:pivotalContractionScalars}, we know that the final natural isomorphism $\mu$ from Proposition~\ref{prop:R2=1} is given by $\mu_x=\varkappa_x\id_x$, just as $\psi_x$ is given by $t_x\id_x$.
 
 Passing the natural isomorphism \eqref{eq:R2Homotopy} into the language of anyon theories gives the following instance of \eqref{eq:ATNat}.  
 \begin{equation}
     U_x^{yz}=\frac{\sigma_y\sigma_z}{\sigma_x}I_x^{yz}~,
 \end{equation}
 where $I_x^{yz}$ is the gauge symbol associated with $\Id_{\cC}$, {i.e.} the identity matrix.
\end{proof}

In view of equation~\ref{eq:R2HomotopySymbol}, we define the following additional property of charge conjugation symmetries, and use it to state Theorem~\ref{thm:main}, which is the precise version of Claim~\ref{claim:main}.
\begin{defn}
    A charge-conjugation symmetry is \textbf{flat} when there exists a choice of unitary dual functor such that the gauge symbol \eqref{eq:R2HomotopySymbol} is trivial.
\end{defn}
Note that flatness is necessarily invariant under translating $J$ by unitary natural isomorphisms.
\begin{thm}
    \label{thm:main}
    Let $\cC$ be a unitary ribbon category with a flat charge conjugation symmetry $J:\cC\to\cC$.
    Then there is a gauge for which $\cC$ has real $F$-symbol and symmetric $R$-symbol.
\end{thm}
\begin{proof}
    By Proposition~\ref{prop:URCBased=URC}, we can work in the based setting.
    That is, it is enough to consider a based unitary ribbon category $(\cC,S,(B_x^{yz})_{x,y,z\in S})$ and for which $J$ is a standard autoequivalence.
    Let $B_x^{yz}=\{\gamma_x^{yz}[j]\}_{j=1}^{N_x^{yz}}$.
    By Proposition~\ref{prop:R2=1} and Corollary~\ref{cor:R2Symbol}, that $J$ is a flat charge conjugation symmetry implies that there is a standard equivalence $R:\cC\to\overline{\cC}^{\rev}$ such that $R(x)=x$ for all $x\in S$ and $\overline{R}\circ R:\cC\to\cC$ has trivial gauge symbol.
    Since $\overline{R}\circ R$ is a standard equivalence, we actually have $\overline{R}\circ R=\Id_{\cC}$.

    Let $M_x^{yz}$ denote the gauge transformation associated with $R$, {i.e.}~$M=\Ex(R)$, where $\Ex$ is the functor from Definition~\ref{defn:ex}.
    Then the gauge transformation associated with $\overline{R}$ is $\overline{M}$, so the equation $\overline{R}\circ R=\Id_{\cC}$ becomes
    \begin{equation}
        \overline{M}_x^{yz}\circ M_x^{yz}=I_x^{yz}~.
    \end{equation}
    Since $M_x^{yz}$ is unitary, this is equivalent to
    \begin{align}
        \overline{M_x^{yz}} &= (M_x^{yz})^\dag\\
        \label{eq:BasedMSymmetric}
        M_x^{yz} &= (M_x^{yz})^T~.
    \end{align}
    In other words, each $M_x^{yz}$ is a symmetric unitary matrix.

    Armed with \eqref{eq:BasedMSymmetric}, we proceed as in {Sec.}~\ref{sssec:higherMultDiagrammatic}.
    By Takagi factorization, for every $x,y,z\in S$, there is a unitary $W_x^{yz}$ such that
    \begin{equation}
        \label{eq:BasedDefnW}
        M_x^{yz}=(W_x^{yz})^TW_x^{yz}~.
    \end{equation}
    Consider the based unitary fusion category $\cC':=(\cC,S,((W_x^{yz})^\dag C_x^{yz}))$, where the basis $C_x^{yz}$ is the translation of $B_x^{yz}$ by $(W_x^{yz})^\dag$.
    That is, $C_x^{yz}=\{\xi_x^{yz}[j]\}_{j=1}^{N_x^{yz}}$, where
    \begin{equation}
        \xi_x^{yz}[k] := \sum_{j=1}^n (W_x^{yz})^\dag[k,j]\gamma_x^{yz}[j]~.
    \end{equation}
    We then have $\overline{\cC'}=(\overline{\cC},S,C_x^{yz})$.
    
    The functor $R$ is also a standard equivalence $\cC'\to\overline{\cC'}^{\rev}$, since $\cC'$ has the same underlying standard unitary ribbon category as $\cC$.
    Note, however, that the basis $C_x^{yz}$ for $\overline{\cC}^{\rev}(x\to yz)$ is the translation of $B_x^{yz}$ not by $(W_x^{yz})^\dag$, but rather by $\overline{(W_x^{yz})^\dag}=(W_x^{yz})^T$, since the action of scalars on morphisms in $\overline{\cC}$ is complex-conjugated.
    This means that the gauge symbol associated with the standard equivalence $R:\cC'\to\overline{\cC'}^{\rev}$ is
    \begin{equation}
        ((W_x^{yz})^T)^\dag M_x^{yz}(W_x^{yz})^\dag = I_x^{yz}~,
    \end{equation}
    where the simplification to the identity matrix follows from \eqref{eq:BasedDefnW}.
    
    Since $R:\cC'\to\overline{\cC'}^{\rev}$ is a standard equivalence with $R(x)=x$ for $x\in S$ and trivial associated gauge symbol, the $F$- and $R$-symbols of $\cC'$ and $\overline{\cC'}^{\rev}$ are equal.
    Equality of the $F$-symbols means that $F_C=\overline{F_C}$, where $F_C$ is the $F$-symbol of $\cC$ computed in the gauge $(C_x^{yz})$.
    Equality of the $R$-symbols means that $R_C=\overline{R_C}^\dag=R_C^T$, where $R_C$ is the $R$-symbol of $\cC$ computed in the gauge $(C_x^{yz})$.
    Hence, $(C_x^{yz})$ is a choice of gauge for $\cC$ in which $F$ is entrywise real and each $R$-matrix is symmetric.
\end{proof}

In the case where every object of $\cC$ is self-dual, we may simplify Theorem~\ref{thm:main} as follows.
\begin{cor}
    \label{cor:mainSelfDual}
    If $\cC$ is a unitary ribbon category, every object of $\cC$ is self-dual, and the second Frobenius-Schur indicator is a $\mathbb{Z}_2$ grading of $\cC$, then there is a gauge for which $\cC$ has real $F$-symbol and symmetric $R$-symbol. 
\end{cor}
\begin{proof}
    Since every object of $\cC$ is self-dual, we may pick $J=\Id_{\cC}$.
    This ensures that $t_x=1$ for all $x$, so the flatness condition $U_x^{yz}=I$ reduces to
    \begin{equation}
        \label{eq:FS2Triv}
        \frac{\varkappa_y\varkappa_z}{\varkappa_x}=1~,
    \end{equation}
    for every $x$, $y$, and $z$ with $N_x^{yz}\neq0$.
    This is precisely the statement that the Frobenius-Schur indicator is a $\mathbb{Z}_2$ grading on $\cC$.
\end{proof}
In {Sec.}~\ref{ssec:FS2ObstructionExample}, we exhibit a group $G$ where every representation of $G$ is self-dual but $\Rep(G)$ has inherently complex $F$-symbol, showing that the condition on Frobenius-Schur indicator is not entirely superfluous.

Even when $\cC$ has non-self-dual objects, the conditions of Theorem~\ref{thm:main} are often satisfied.
We provide several positive examples in {Sec.}~\ref{implications}.
More generally, one might hope to interpret the trivialization of the $U_x^{yz}$ as the statement that there is a grading of $\cC$ given by a certain twisted Frobenius-Schur indicator, with the $t$-phases doing the twisting.
The appropriate notion of twisted Frobenius-Schur indicator was introduced in the case $\cC=\Rep(G)$ in \cite{kawanaka1990twisted}.
We translate between their notion and ours in Appendix~\ref{KMAppendix}.
We leave a full interpretation of our gauge symbol as a twisted Frobenius Schur indicator in unitary ribbon categories which are not Tannakian as a topic for future work.

\begin{remark}
    \label{rem:FGaugeBarF}
    In cases where there is a braided autoequivalence $J:\cC\to\cC$ which is a dual functor at the level of objects, but $J$ is not a flat charge conjugation symmetry, either because $J^2$ is a nontrivial soft autoequivalence, because $J$ is not gaugeable, or because there is no extension of the corresponding twisted Frobenius-Schur indicator to a $\mathbb{Z}_2$-grading, there is still a gauge equivalence between $F$ and $\overline{F}$, given by the functor $R$ from \eqref{eq:defnR}.
    Nevertheless, it is still possible that $F$ is not real in any gauge.
    An example of a unitary ribbon category with a charge conjugation symmetry (which fails to be flat) and inherently complex $F$-symbol is discussed in {Sec.}~\ref{ssec:FS2ObstructionExample}.
\end{remark}

\subsection{Diagrammatic interpretation in the self-dual, trivial Frobenius-Schur case}
\label{RealGauge}

We conclude this section by giving a diagrammatic interpretation of the symmetry condition derived above. We restrict here to the particularly transparent case in which every simple object is self-dual and has trivial second Frobenius-Schur indicator,
\begin{equation}
    a^\vee\simeq a,
    \qquad
    \nu_2(a)=+1~.
\end{equation}
In this case the charge-conjugation autoequivalence may be taken to be $J=\Id_{\cC}$.

We emphasise that the argument below is intended as a diagrammatic interpretation of the categorical proof of Theorem~\ref{thm:main}, rather than as an independent proof. A completely intrinsic graphical treatment of rotations and straightening of trivalent vertices requires the flagging conventions of \cite{Kitaev:2005hzj,simon2022straighteningfrobeniusschurindicator}, which distinguish the Frobenius-Schur factors generated by different ways of bending and straightening strands. Since all second Frobenius-Schur indicators are $+1$ in the present case, these factors are trivial, and we suppress the corresponding flags in the diagrams below.

In Sec.~\ref{sec:mechanism}, the transformation of a trivalent splitting space was represented by the local antiunitary
\begin{equation}
    A_a^{bc}=\mathsf M_a^{bc}K~,
\end{equation}
where $K$ denotes complex conjugation in the chosen trivalent basis. The categorical proof above shows, in the present specialization, that the corresponding local gauge symbol satisfies
\begin{equation}
    \overline{\mathsf M_a^{bc}}\,
    \mathsf M_a^{bc}
    =
    \mathbbm 1_{V_a^{bc}}~.
\label{eq:MbarM-diagrammatic}
\end{equation}
Equivalently,
\begin{equation}
    \bigl(A_a^{bc}\bigr)^2
    =
    \mathbbm 1_{V_a^{bc}}~.
\end{equation}

Diagrammatically, the second application of $A_a^{bc}$ reflects the half-turn and the braiding produced by the first. The crossing is therefore paired with its inverse. The remaining local operation is a zig-zag, or $S$-shaped turning, of the incident world lines.  In a fully-flagged graphical calculus this is precisely the step at which Frobenius-Schur factors are recorded \cite{Kitaev:2005hzj, simon2022straighteningfrobeniusschurindicator}. Under the present assumption $\nu_2(a)=+1$ for every simple object, these factors are trivial, so the corresponding unflagged diagram may schematically be straightened to the original trivalent vertex.

The corresponding world-line straightening is
\begin{equation}
\label{eq:SStraightenFSTriv}
\ket{
\begin{tikzpicture}[
    scale=0.7,
    baseline=-0.8ex,
    line cap=round,
    line join=round,
    >={Stealth[length=3pt,width=4pt]},
    sstrand/.style={
        thick,
        postaction={decorate},
        decoration={
            markings,
            mark=at position .58 with {\arrow{>}}
        }
    }
]

\draw[sstrand] (0,-1) .. controls ++(-.5,.5) and ++(.5,-.5) .. (0,0);
\draw[sstrand] (0:0) .. controls (75:.71) and (-15:.71) ..  (30:1);
\draw[sstrand] (0:0) .. controls (195:.71) and (105:.71) .. (150:1);
\node at (-60:.6) {$\scriptstyle a$};
\node at (60:.6) {$\scriptstyle c$};
\node at (180:.6) {$\scriptstyle b$};
\node at (-15:.4) {$\scriptstyle\lambda$};
\end{tikzpicture}
}
\xrightarrow[\nu_2=+1]{\text{straightening}}
\ket{
\begin{tikzpicture}[
    scale=0.35,
    baseline=-0.8ex,
    line cap=round,
    line join=round,
    >=stealth
]

\coordinate (O) at (0,0);
\coordinate (A) at (0,-1.3);
\coordinate (B) at (-1.1,1.1);
\coordinate (C) at (1.1,1.1);

\draw[thick] (A) -- (O);
\draw[thick] (O) -- (B);
\draw[thick] (O) -- (C);

\draw[->,thick] (0,-0.85) -- (0,-0.45);
\draw[->,thick] (-0.45,0.45) -- (-0.7,0.7);
\draw[->,thick] (0.45,0.45) -- (0.7,0.7);

\node at (0.6,-0.8) {$\scriptstyle a$};
\node at (-1.2,0.4) {$\scriptstyle b$};
\node at (1.2,0.4) {$\scriptstyle c$};
\node at (-0.55,-0.25) {$\scriptstyle\lambda$};

\end{tikzpicture}
}~.
\end{equation}

Since $\mathsf M_a^{bc}$ is unitary, Eq.~\eqref{eq:MbarM-diagrammatic} gives
\begin{equation}
    \overline{\mathsf M_a^{bc}}
    =
    \bigl(\mathsf M_a^{bc}\bigr)^\dagger~,
\end{equation}
and therefore
\begin{equation}
    \boxed{
    \bigl(\mathsf M_a^{bc}\bigr)^T
    =
    \mathsf M_a^{bc}~.
    }
\end{equation}
Thus the diagrammatic picture reproduces the symmetry condition obtained in the categorical proof. As shown in Sec.~\ref{sec:mechanism}, Takagi factorisation then produces a trivalent-basis gauge in which the $F$-symbols are real.

\section{Examples of models with real \texorpdfstring{$F$}{F}-symbols and implications of our results}
\label{implications}

We now discuss several families of unitary ribbon categories that illustrate the general claims proved in the previous section, some of which were already mentioned briefly in the introduction. We begin with familiar classes of anyon models for which gauges with real $F$-symbols are already known explicitly. For each such class, we show that the hypotheses of our theorem are satisfied, thereby explaining from a unified categorical perspective why their $F$-symbols can be made real.

We then turn to self-dual Chern-Simons TQFTs with compact simply connected gauge groups in Sec.~\ref{ChernSimons}. In these theories, the $F$-symbols are typically difficult to compute explicitly, and no manifestly real gauge is generally known. Nevertheless, we argue that these theories possess a flat charge-conjugation symmetry and therefore satisfy the assumptions of our theorem. Our result consequently guarantees the existence of a real gauge even in cases where it is not apparent from any known presentation of the $F$-symbols and where such reality had not previously been expected. In this sense, the Chern-Simons examples represent a new application of our framework, extending the phenomenon of real $F$-symbols beyond the standard models in which it can be verified directly.

By cataloguing both types of examples, we show that our theorem not only unifies many known instances of real $F$-symbols, but also predicts their existence in a broader class of anyon theories.

\subsection{Abelian theories}
Abelian unitary ribbon categories consist of anyons with fusion rules characterized by an Abelian group
\begin{equation}
G\cong\prod_{i=1}^N\mathbb{Z}_{n_i}~,\ \ \ \mathbb{Z}_{n_i}\cong\langle g_i\rangle~.
\end{equation}
In these categories, simple objects admit the unique presentation $a=\prod_ig_i^{a_i}$. Such a category also has a quadratic form, $q$, that specifies the topological spins via $\theta(a)=\exp(2\pi iq(a))$. Equivalently, the spin function $\theta:G\to U(1)$ is itself a ($U(1)$-valued) quadratic form, $\theta(na)=\theta(a)^{n^{2}}$, with associated bicharacter $b_{\theta}(a,b)=\theta(a+b)/\theta(a)\theta(b)$. We pass freely between the two, writing $O(G,q)=O(G,\theta)$ and denoting the category by $\CC(G,\theta)$.

As a group, the invertible 0-form symmetries / braided automorphisms of $\CC(G,\theta)$ are isomorphic to $O(G,\theta)$. By construction, these automorphisms preserve the quadratic form. It is trivial to check that charge conjugation (which sends $a\to-a\cong a^{\vee}$) is in $O(G,\theta)$ (the only cases of Abelian theories with trivial charge conjugation are the self-dual theories with $G\cong\mathbb{Z}_2^{\oplus n}$). Since these theories have no fusion multiplicity, the results of Sec. \ref{folding} imply these theories have real $F$-symbols when they are modular. In fact, unitary ribbon Abelian categories are full braided fusion subcategories of Abelian UMTCs\footnote{Indeed, as discussed above (and as argued in \cite{street1993braided,Quinn:1998un}), such a category, $\CC(G,\theta)$, is determined by a finite Abelian group, $G$, with quadratic form $\theta$ (non-degeneracy of $\theta$ is not assumed). Now, define $D:=G\oplus\widehat G$ and $\Theta(a,\chi):=\theta(a)\,\chi(a)$ where $\chi\in \widehat G$. Then $\Theta(-x)=\Theta(x)$, and the associated bilinear form is
\begin{equation}
B\bigl((a,\chi),(b,\psi)\bigr)=\frac{\Theta(a+b,\chi\psi)}{\Theta(a,\chi)\Theta(b,\psi)}=b_{\theta}(a,b)\,\chi(b)\,\psi(a)~.
\end{equation}
Since $b_\theta$ is the bicharacter built from $\theta$, we see that $\Theta$ is a quadratic form. If $(a,\chi)$ lies in the radical of $B$, then setting $b=0$ gives $\psi(a)=1$ for every $\psi\in\widehat G$, hence $a=0$. Now, setting $a=0$ gives $\chi(b)=1$ for every $b\in G$, hence $\chi=1$. Thus $\Theta$ is non-degenerate and $\CC(D,\Theta)$ is an Abelian UMTC. Finally $\Theta(a,1)=\theta(a)$, so $a\mapsto(a,1)$ is an isometry of quadratic forms and realizes $\CC(G,\theta)$ as a full braided fusion subcategory of $\CC(D,\Theta)$, compatibly with the ribbon structures.} and so the results of Sec. \ref{folding} apply more generally by restriction.

This conclusion can be explicitly verified by going to the so-called Quinn gauge \cite{Quinn:1998un} (see also \cite{Lee:2018eqa}), where the $F$-symbols are real
\begin{eqnarray}\label{FAbelian}
F(a,b,c)=\prod_i\begin{cases}
1 & {\rm if}\ b_i+c_i<n_i~,\\
\theta(g_i)^{n_ia_i}=\pm1& {\rm otherwise}~.
\end{cases}
\end{eqnarray}
Here $\theta(g_i)$ is the topological spin / self-statistics corresponding to $g_i$, and we have used the fact that the topological spins are $n^{\rm th}$ or $2n^{\rm th}$ roots of unity.

\subsection{Metaplectic UMTCs}\label{metaplectic}
Metaplectic UMTCs arise as the categories of Wilson lines in ${\rm Spin}(2p+1)_2$ Chern-Simons theories. We will discuss a more general class of Chern-Simons theories in Sec. \ref{ChernSimons} that includes the Metaplectic theories as a special case.

The main point is that ${\rm Spin}(2p+1)_2$ theories are self-dual. Moreover, in Sec. \ref{ChernSimons}, we will show that they respect an FS grading (this claim has also been argued in \cite{simon2022straighteningfrobeniusschurindicator}). Therefore, our theorem implies these theories have a gauge with real $F$-symbols. Indeed, the metaplectic UMTCs have real $F$-symbols in the gauge discussed in \cite{ardonne2021classification}.

\subsection{\texorpdfstring{$SU(2)_k$}{SU(2)k} Chern-Simons theories}
\label{SU2k}
$SU(2)_k$ Chern-Simons theories are self-dual: this property follows from the fact that all representations of $SU(2)$ are self-dual. As in the case of the metaplectic UMTCs, the $SU(2)_k$ theories have an FS grading (see Sec. \ref{ChernSimons} and also the argument in \cite{simon2022straighteningfrobeniusschurindicator}). Therefore, our results imply there should be a gauge in which the $F$-symbols are real. Indeed, $SU(2)_k$ $F$-symbols are known to be real in the standard gauge chosen in \cite{kirillov1989representations}. 

\subsection{Braided TY categories}
Tambara-Yamagami (TY) categories are captured by a finite Abelian group, $G$, an $\mathcal{R}$-valued symmetric non-degenerate bi-character, $\chi(a,b)$, and a square root, $\tau$, of $1/|G|$.\footnote{Here, $\mathcal{R}$ is a ring that contains roots of unity of order $8|G|$.} They have fusion rules
\begin{equation}\label{TYfusion}
X\otimes X=\sum_{g\in G}g~, \ \ \ g\times X=X\times g = X~, \ \ \ g\times h=g+h~,
\end{equation}
where $g$ and $h$ are any elements in $G$, and $X$ is a single non-invertible simple object.

In order to admit a braiding, $G\cong\mathbb{Z}_2^{\oplus m}$, and so braided TY categories are self-dual. In all of these cases, there is an FS grading and so there is a flat charge conjugation symmetry corresponding to the trivial symmetry. Indeed, the topological spins of the Abelian objects are all $\theta_a=\chi(a,a)=\pm1$ \cite{siehler2000braided} and therefore these objects have trivial FS indicator. Since $X\not\in X\otimes X$, no matter the FS indicator of $X$ (given by $\nu_2(X)={\rm sign}(\tau)$), the FS grading is respected by the fusion rules.

Our results therefore imply the existence of a gauge in which the $F$-symbols are real. Indeed, an analysis of hexagon equations shows that $\chi(a,b)=\pm1$ \cite{siehler2000braided}. In the gauge discussed in \cite{siehler2000braided}, this fact suffices to ensure that the $F$-symbols are real.

\subsection{Non-symmetric near group categories}\label{NSNGC}
There are also complete results in braided non-symmetric near group categories generalizing the TY examples of the previous section. Here, non-symmetric means that the $S$-matrix is not completely degenerate, {i.e.}~rank $1$.
(We will return to the infinitely many categories with completely degenerate braiding in {Sec.}~\ref{Complex}.)
These categories are completely classified \cite{thornton2011braided} (see also \cite{siehler2003near}), and their fusion rules are fixed by
\begin{equation}\label{NGfus}
X\otimes X=\sum_{g\in G}g+\alpha X~, \ \ \ \alpha\in\mathbb{Z}_{\ge0}~,
\end{equation}
where $G$ is an Abelian group with $g$ running over the corresponding invertible simples. As in the TY discussion, $X$ is again the only non-invertible object.

The above fusion rules are summarized in the shorthand $G+\alpha$, and there are sometimes multiple categories sharing the same fusion rules.\footnote{In this notation, TY corresponds to $G+0$. As we have discussed above, TY categories are specified by additional data (i.e., $\chi$ and $\tau$).} Indeed, there are two unitary non-symmetric braided categories with $\mathbb{Z}_1+1$ fusion (they are UMTCs corresponding to $(G_2)_1$ and $(F_4)_1$ Chern-Simons theory).\footnote{There are two further non-unitary braided categories with $\mathbb{Z}_1+1$ fusion rules: the Lee-Yang MTC and its conjugate. These also turn out to have real $F$-symbols in an appropriate gauge.} Both these categories have trivial FS indicators and therefore have flat (albeit trivial) charge conjugation symmetries.\footnote{Triviality of the FS indicator of the non-invertible simple object in each category follows from inspecting $\varkappa_X=d_X\left(F^{XXX}_X\right)_{11}=1$ (e.g., see \cite{Feiguin:2006ydp}, and note that our $X$ is their $\tau$).}

Next, there are two non-symmetric near group categories with $\mathbb{Z}_2+1$ fusion: they come from equivariantizing the two different UMTCs with $\mathbb{Z}_3$ fusion rules with respect to the charge conjugation automorphism of $\mathbb{Z}_3$.\footnote{Physically, these theories come from gauging charge conjugation in either $SU(3)_1$ Chern-Simons or $(E_6)_1$ Chern-Simons and throwing out twisted sectors.} They again have trivial FS indicators and therefore have flat (and again trivial) charge conjugation symmetries.\footnote{Triviality of the FS indicator of the non-invertible simple object in each category again follows from inspecting $\varkappa_X=d_X\left(F^{XXX}_X\right)_{11}=1$ in \cite{siehler2003near}. The triviality of the FS indicator of the non-trivial invertible simple object follows from the M\"uger center discussion in \cite{thornton2011braided}.}

Finally, there is one such category with $\mathbb{Z}_3+2$ fusion. It arises from equivariantizing the UMTC with $\mathbb{Z}_2^{\oplus2}$ fusion rules where all non-trivial objects are fermions.\footnote{Physically, this category arises from gauging the $\mathbb{Z}_3<S_3$ exchange symmetry of the non-trivial lines in ${\rm Spin}(8)_1$ Chern-Simons theory and throwing out twisted sectors.} It corresponds to a subcategory of $SU(3)_3$ Chern-Simons theory \cite{Barkeshli:2014cna}.

All the categories described above except the one with $\mathbb{Z}_3+2$ fusion rules have all objects self-dual and an FS grading. Therefore, all the self-dual non-symmetric near group categories have gauges with real $F$-symbols.\footnote{There are no other non-symmetric braided near group categories, and the three categories with $|G|>1$ are non-modular (as follows from the discussion above, their M\"uger centers are ${\rm Vec}_G$; physically this is the case because we have discarded twisted sectors in the gauging of the underlying Abelian Chern-Simons theories) \cite{thornton2011braided}.}  A conjectured generalization of a result we present in Sec. \ref{ChernSimons} (see the discussion around \eqref{Gnsd}) implies that the above category with $\mathbb{Z}_3+2$ fusion rules also has a gauge with real $F$-symbols via its connection with $SU(3)_3$ Chern-Simons.

Indeed, one can check that the gauges in  \cite{siehler2003near,Feiguin:2006ydp} are real for all categories except the one with $\mathbb{Z}_3+2$ fusion rules. However, one can make the following transformation 
\begin{equation}
\Gamma^{XX}_X=\frac{1}{\sqrt{2}}\begin{pmatrix}
1 &1\\
i&-i
\end{pmatrix}~,
\end{equation}
where all other $\Gamma$'s are taken to be trivial to get real $F$-symbols in this final case. Note that, In our notation, $m$ of \cite{siehler2003near} is denoted $X$.

\subsection{\texorpdfstring{${\rm Rep}(G)$}{Rep(G)} with \texorpdfstring{$|G|<20$}{|G|<20}}
In this subsection we discuss representation categories of groups with order $|G|<20$. These examples include cases like ${\rm Rep}(S_3)$, ${\rm Rep}(D_8)$, ${\rm Rep}(A_4)$, and beyond. Using the relevant formula for the FS indicator
\begin{equation}\label{eq:fs1}
\nu_2(\chi)\;=\;\frac{1}{|G|}\sum_{g\in G}\chi(g^2)
\;=\;\begin{cases}
+1~, & \text{real (orthogonal) type}~,\\
\phantom{+}0~, & \text{complex type, } \chi\neq\overline{\chi}~,\\
-1~, & \text{quaternionic (symplectic) type}~,
\end{cases}
\end{equation}
it is possible to check that all examples of self-dual ${\rm Rep}(G)$ categories in our range have an FS grading. Therefore, by our results, such examples have a gauge with real $F$-symbols.

In the non-self-dual ${\rm Rep}(G)$ cases in our range (e.g., ${\rm Rep}(A_4)$), it turns out that the twisted FS indicator can be computed through the Kawanaka-Matsuyama formula \cite{kawanaka1990twisted}
\begin{equation}\label{eq:KM1}
  \nu_\tau(\chi)\;=\;\frac{1}{|G|}\sum_{g\in G}\chi\big(g\,\tau(g)\big)~.
\end{equation}
Here $\tau\in{\rm Aut}(G)$ is a class-inverting involution and can be understood via the split extension $E=G\rtimes_\tau\langle t\rangle$ (we have checked that every $G$ in the range $|G|<20$ has such a split extension). Note that $tgt^{-1}=\tau(g)$ and $t^2=1$. When $\tau={\rm Id}$, we recover the case of ambivalent $G$ and the FS formula in \eqref{eq:fs1}. See Appendix \ref{KMAppendix} for details.

In every case with $|G|<20$, the KM formula furnishes a grading, and so all such categories have a gauge with real $F$-symbols. 

\subsection{Implications for Chern-Simons theories with compact simply connected gauge groups}\label{ChernSimons}
For general Chern-Simons theories, it is difficult to compute the $F$-symbols and explicitly check whether real gauges exist. {\it Therefore, these theories present an opportunity to apply our results.}  As motivation, we note the results discussed in Secs. \ref{metaplectic} and \ref{SU2k} show that ${\rm Spin}(2p+1)_2$ and $SU(2)_k$ Chern-Simons theories have real $F$-symbols via the explicit gauges in \cite{ardonne2021classification,kirillov1989representations}. We have also seen in Sec. \ref{NSNGC} that $(G_2)_1$ and $(F_4)_1$ Chern-Simons theories have real $F$-symbols.

In this section, we will argue these facts are more generally true by studying Chern-Simons theories with connections valued in the simple Lie algebras $\mathfrak{sl}_{2}$, $\mathfrak{so}_{2r+1}$, $\mathfrak{sp}_{2r}$, $\mathfrak{so}_{4m}$, $\mathfrak{g}_2$, $\mathfrak{f}_4$, $\mathfrak{e}_7$, and $\mathfrak{e}_8$. For simplicity, we focus on the Chern-Simons theories built from the corresponding simply connected groups (here we mean that $\pi_0(G)$ and $\pi_1(G)$ are trivial)\footnote{Throughout, $r\ge2$ in
${\rm Spin}(2r+1)$, $r\ge3$ in ${\rm USp}(2r)$, and $m\ge2$ in ${\rm Spin}(4m)$. The reason for working in these ranges is because $\mathfrak{so}_3\cong\mathfrak{sp}_2
\cong\mathfrak{sl}_2$, $\mathfrak{sp}_4\cong\mathfrak{so}_5$, and $\mathfrak{so}_4$ is not simple.}
\begin{equation}\label{SdualCS}
G_{\rm sd}=\left\{SU(2)~,\ {\rm Spin}(2r+1)~, \  {\rm USp}(2r)~,\  {\rm Spin}(4m)~,\ G_2~, \ F_4~,\ E_7~, \ E_8\right\}~.
\end{equation}

In these cases, the theory is specified by choosing an integer level (and also a framing)\footnote{More generally, for Chern-Simons theories with non-simply connected groups (e.g., those arising from anyon condensation of lines in the center of $G$) one should include sums over different topological sectors (e.g., see \cite{Cordova:2017vab}).}
\begin{equation}\label{CSaction}
S[A]=\frac{k}{4\pi}\int_M\left(\langle A\wedge dA\rangle_{\rm long}+\frac{1}{3}\langle A\wedge [A\wedge A]\rangle_{\rm long}\right)~,
\end{equation}
where $A$ is the connection, $M$ is an oriented three-manifold, and $\langle\cdot,\cdot\rangle_{\rm long}$ is the invariant form on $\mathfrak{g}$ normalized so that long roots have squared length 2. We will refer to the theory with action \eqref{CSaction} as $G_k$ Chern-Simons theory (or more verbosely as Chern-Simons theory with gauge group $G$ at level $k$).

Note that the $G_k$ Chern-Simons theories for $G\in G_{\rm sd}$ are self-dual. Therefore, by our main theorem, we get a gauge with real $F$-symbols in such theories provided there is a grading for the FS indicator. In fact, the arguments in \cite{simon2022straighteningfrobeniusschurindicator} strongly suggest that such a grading exists. The remainder of this section can be understood as spelling out steps suggested in \cite{simon2022straighteningfrobeniusschurindicator} that were left implicit. The upshot is that Chern-Simons theories with gauge groups in \eqref{SdualCS} have gauges with real $F$-symbols. We will comment on generalizations to non-self dual Chern-Simons theory with compact simply connected gauge group $G\in\left\{SU(r+1),\ {\rm Spin}(4m+2),\ E_6\right\}$ ($r>1$) at the end of this section.

To describe the simple Wilson lines in $G_k$ Chern-Simons theory, we first recall some basics of Lie theory. To that end, let $P\supset Q$ denote the weight and root lattices, $h^\vee$ the dual Coxeter number, $W$ the Weyl group with $w_0\in W$ its longest element, and $\theta^\vee$ the coroot of the highest root, $\theta$, of $\mathfrak g$
\begin{equation}
\theta^\vee:=\frac{2\theta}{\langle\theta,\theta\rangle}~.
\end{equation}
The Wilson lines of $G_k$ are then labelled by the dominant
weights in the level-$k$ alcove,
\begin{equation}\label{WilsonGk}
\Lambda_k=\{\lambda\in P_+:\langle\lambda+\rho,\theta^\vee\rangle<k+h^\vee\}~,
\end{equation}
and we write $V_\lambda$ for the corresponding objects of the underlying UMTC, $\CC(\mathfrak g,k)$. The theory is ribbon, with each $V_\lambda$ carrying a twist, $\theta_\lambda$, and pairs of lines carrying a braiding, $c_{x,y}$.

Duality acts on labels by $V_\lambda^\vee=V_{-w_0\lambda}$, and $-w_0$ is an automorphism of the Dynkin diagram. Every Wilson line is self-dual (and therefore, the corresponding $G_k$ TQFT is self-dual) when $-w_0={\rm Id}$ or, in other words, when $-1\in W$. This happens exactly for the cases in \eqref{SdualCS}.

As the final piece of Lie-theoretic data that will be useful for us below, let us introduce a distinguished central element. To that end, set $\rho^\vee:=\sum_i\omega^\vee_i$ to be half the sum of the positive coroots, and define
\begin{equation}\label{zcentral}
z:=\exp(2\pi i\rho^\vee)~.
\end{equation}
It is easy to see that $z$ lies in the center of $G$, and that $z^{2}=1$. Consequently, $\lambda\mapsto\lambda(z)$ is a $\{\pm1\}$-valued character of the weight lattice that is trivial on the root lattice, $Q$.\footnote{To understand this statement, note that for a character, $\lambda$, of the maximal torus, $\lambda(\exp(2\pi iH))=e^{2\pi i\langle\lambda,H\rangle}$. Since $\langle\alpha_i,\rho^\vee\rangle=1$ for every simple root, $\alpha(z)=1$ for every root. Therefore, $z$ is central, because it commutes with every root vector and with the torus. Since $2\rho^\vee=\sum_{\alpha>0}\alpha^\vee$ lies in the coroot lattice, $\langle\lambda,2\rho^\vee\rangle\in\mathbb Z$ for every $\lambda\in P$, so $\lambda(z^{2})=1$. Since $G$ is simply connected, $P$ is the full character lattice and characters separate points, giving $z^{2}=1$. Finally, triviality on $Q$ is $\alpha(z)=1$.} For $G=SU(2)$, $z=-I$ and $\lambda(z)=(-1)^{2j}$ distinguishes half-integer from integer spin. On the other hand, for $G$ with trivial center, as in the case of  $G_2,F_4,E_8$, $z=1$ and every $\lambda(z)=1$.

The main object of interest in this section is the FS indicator, $\varkappa$. We have already encountered it as a component of the pivotal structure, $p_x=\varkappa_x\in\left\{\pm1\right\}$. In this section it will be more useful for us to use the braiding formula for the FS indicator (e.g., see \cite{Kitaev:2005hzj})
\begin{equation}\label{Rkappa}
\varkappa_x=\theta_xR^{xx}_{\mathbf 1}\in\left\{\pm1\right\}~.
\end{equation}
In particular, thinking of $\CC(\mathfrak{g},k)$ and ${\rm Rep}(G)$ as two \lq\lq specializations" of a single underlying algebra that depends on a formal variable $q$, we will argue that this formula allows us to equate the indicator in $\CC(\mathfrak{g},k)$  with the FS indicator in ${\rm Rep}(G)$. Using this identification then allows us to establish an FS grading. As we will see, the utility of \eqref{Rkappa} lies in the fact that it involves no antilinear ingredient (unlike the pivotal definition) and both the twists and braiding are defined over the coefficient ring specified by $q$, as we now make precise.

The underlying algebra is defined over Laurent polynomials in a fractional power of $q$, $\CA=\mathbb Z[q^{\pm1/L}]$.\footnote{The algebra in question is Sawin's ribbon integral form, $U^{\dagger}_{\CA}(\mathfrak g)$, and not the restricted integral form, $U^{\rm res}_{\CA}(\mathfrak g)$, that we use below to define the $\CA$-forms of the Weyl modules. Indeed, by \cite[Thm.~3]{sawin2006quantum} the $U^{\dagger}_{\CA}(\mathfrak g)$-modules that are free of finite rank over $\CA$ constitute an abelian ribbon category enriched over $\CA$, so that the braiding and the twist are defined over $\CA$. Note that the restricted form alone is not a ribbon Hopf algebra, since writing the $R$-matrix requires both fractional powers of $q$ and a completion of the Cartan part \cite[Sec. 1]{sawin2006quantum}. Our $\CA$ is the ring denoted $\CA'=\mathbb Z[s^{\pm1}]$ in \cite{sawin2006quantum}, with $s=q^{1/L}$.} The powers of $q$ are fractional because the braiding involves exponents like $q^{\langle\lambda,\mu\rangle}$ that we want to be in $\CA$ \cite{sawin2006quantum}.\footnote{In other words, $L$ is the least integer such that $L\langle\lambda,\mu\rangle\in\mathbb{Z}$ {for all $\lambda,\mu\in P$. Here, and in all $q$-exponents in the remainder of this section, $\langle\cdot,\cdot\rangle$ denotes the invariant form normalized so that short roots have squared length $2$, as in \cite{sawin2006quantum}. It equals $D$ times the form appearing in \eqref{CSaction}, where $D$ is the lacing number (the ratio of long to short squared root lengths). With this convention, $L$ agrees with the values tabulated in \cite{sawin2006quantum}, and the level $k$ theory corresponds to the specialization $q=e^{i\pi/D(k+h^\vee)}$.} In particular, $L$ depends on $\mathfrak{g}$ and not on $k$.} In this context, specializing means setting $q=1$ (to obtain observables in ${\rm Rep}(G)$) or to an appropriate root of unity (and performing a quotient we will describe below) to obtain observables in $\CC(\mathfrak{g},k)$. Mathematically, this maneuver amounts to applying a ring homomorphism $\CA\to\mathbb C$ sending $q$ to its specified value (note that no limit is taken).

Two elementary facts about $\CA$ are crucial. First, its invertible elements are $\pm q^{m/L}$ with $m\in\mathbb Z$ arbitrary.\footnote{Indeed, if $fg=1$, then comparing highest and lowest terms forces $f$ and $g$ to be single monomials whose coefficients are units of $\mathbb Z$.} Every unit therefore factors uniquely as $\mathrm{sgn}(f)\,q^{m/L}$ with $\mathrm{sgn}(f)\in\{\pm1\}$. The second fact is that ring homomorphisms fix $1$ and $-1$. We emphasise that $\mathrm{sgn}$ is an invariant of an element of $\CA$, not of its image in $\mathbb C$.\footnote{Indeed, the specialisation maps are far from injective, and units of opposite sign can have equal images. For example, if $\phi(q^{1/L})=\zeta$ has even order $2n$, as it does at every level-$k$ specialisation below, then $\phi(-q^{n/L})=-\zeta^{n}=1=\phi(1)$. In particular, a unit $\pm q^{m/L}$ with $m\neq0$ specialises to different complex numbers, of no well-defined sign, at $q=1$ and at a root of unity.} What we shall use instead is the sharper statement that, since $\CA$ is an integral domain, a unit $f\in \CA^{\times}$ satisfies $f^{2}=1$ if and only if $f=\pm1$ (i.e., if and only if $m=0$). Hence a unit of $\CA$ that squares to $1$ takes one and the same value, namely its own sign, under every specialisation.

The braiding acts invertibly on $\Hom(\mathbf 1,x\otimes x)$, and we want it to act by a unit of $\CA$. Indeed, the braiding eigenvalue alone does not survive specialisation, but Lemma~\ref{lem:kappa} shows that its product with the twist is an element of $\CA$ squaring to $1$, and such elements do. For this procedure to work, the Hom space must be free of rank one over $\CA$\footnote{An $\CA$-module $M$ is free of rank one if $M\cong \CA$ as an $\CA$-module. In this case, there is a generator $\gamma\in M$ such that every element of $M$ is $r\gamma$ for a unique $r\in \CA$. Both the existence of $r$ and its uniqueness are used below.} with generator $\gamma$. The action of braiding is then $\gamma\mapsto u\gamma$ for a single unit $u\in \CA$ (invertibility holds because braiding is invertible). 

The argument also requires that a generator over $\CA$ remains a generator after both specialisations. This requirement does not follow from freeness. Instead, in Lemma~\ref{lem:free} we show that $\gamma$ is primitive. This property suffices to establish that $\gamma$ remains a generator after both specializations.

A final subtlety is that, when $q$ is a root of unity, the representation theory degenerates. In particular, modules that were irreducible become indecomposable but reducible, and the category stops being semisimple. Tilting modules are the part that survives. They are defined as the modules admitting both a Weyl filtration and a dual Weyl filtration \cite[Def.~4]{sawin2006quantum}. They are closed under tensor product \cite[Cor.~4]{sawin2006quantum}. To get to the Chern-Simons category, $\CC(\mathfrak{g},k)$, we need to take the quotient of the category of tilting modules by the \lq\lq negligible" objects (i.e., for indecomposables, the objects of quantum dimension zero \cite[Lem.~7]{sawin2006quantum}) \cite[Thm.~5]{sawin2006quantum}. 

Now we establish the necessary results discussed above. To do this, we must first distinguish between the \lq\lq $\CA$-form" $V_{\lambda,\CA}$ and $V_\lambda$. The latter is the Wilson line and results from specializing $V_{\lambda,\CA}$ to $q$ a root of unity and quotienting by negligible objects. As a further useful disambiguation, we write $V_{\lambda,\zeta}:=V_{\lambda,\CA}\otimes_{\CA,\phi}\mathbb{C}$ for the specialization of the $\CA$-form at a root of unity $\zeta=\phi(q^{1/L})$, so that $V_\lambda=\Pi(V_{\lambda,\zeta})$ with $\Pi$ the quotient functor introduced below. We write $V^{\rm cl}_\lambda$ for the specialisation of $V_{\lambda,\CA}$ at $q^{1/L}\mapsto1$, the irreducible representation of $G$ with highest weight $\lambda$. Given this background, we prove the following:

\begin{lem}\label{lem:free}
Specialize to the self-dual cases discussed above. Let $\lambda$ lie in the level-$k$ alcove and set $H=\Hom(\mathbf 1,V_{\lambda,\CA}\otimes V_{\lambda,\CA})$.
Then
\begin{enumerate}
\item[\bf(a)] $H$ is free of rank one over $\CA$, with generator $\gamma$
\item[\bf(b)] $\gamma$ is primitive, hence has nonzero image under any ring homomorphism $\CA\to\mathbb C$ sending $q^{1/L}$ to a root of unity (including the case $q\mapsto1$). These are the only specializations used in this section.
\item[\bf(c)] the image of $H$ in $\cC(\mathfrak g,k)$ is again one dimensional, spanned by the image
of $\gamma$
\end{enumerate}
\end{lem}

See Appendix \ref{CSApp} for a proof.

To ensure the utility of the formula \eqref{Rkappa} in relating the FS indicator in the Chern-Simons and ${\rm Rep}(G)$ categories, we need the following technical result:

\begin{lem}\label{lem:kappa}
Let $\lambda$ be as in Lemma~\ref{lem:free}, and let $x=V_\lambda$ be the corresponding Wilson line with $x$ simple and self-dual. Then $\varkappa_x=\theta_x\,R^{xx}_{\mathbf 1}$, and $(\theta_xR^{xx}_{\mathbf 1})^2=1$. Moreover, the second identity has a counterpart at the level of $\CA$. With $u\in \CA$ the unit by which the braiding acts on the generator $\gamma$ of Lemma~\ref{lem:free}, and $\theta_\lambda^\CA\in \CA$ the scalar by which the twist acts on $V_{\lambda,\CA}$, one has $(\theta_\lambda^\CA u)^{2}=1$. Hence $\theta_\lambda^\CA u\in\{\pm1\}$.
\end{lem}

\begin{proof}
We have already explained the identity $\varkappa_x=\theta_x\,R^{xx}_{\mathbf 1}$ for self-dual $x$ around \eqref{Rkappa}. For the second claim, evaluate the balancing axiom $\theta_{a\otimes b}=c_{b,a}c_{a,b}(\theta_a\otimes\theta_b)$, with $a=b=x$, on the vacuum channel $\gamma$. Naturality of the twist gives $\theta_{x\otimes x}\circ\gamma=\gamma\circ\theta_{\mathbf 1}$, so that $(R^{xx}_{\mathbf 1})^2=\theta_{\mathbf 1}\theta_x^{-2}$. Since $\theta_{\mathbf{1}}=1$, we have $(\theta_xR^{xx}_{\mathbf 1})^2=1$.

For the final assertion, note that, by \cite[Thm.~3]{sawin2006quantum}, the $U^{\dagger}_{\CA}(\mathfrak g)$-modules that are free of finite rank over $\CA$ form an abelian ribbon category enriched over $\CA$, so the balancing axiom applies before any specialization. The twist is an endomorphism of $V_{\lambda,\CA}$, so it commutes with the action of $U^{\rm res}_\CA$ and preserves weight spaces. The highest-weight space is free of rank one over $\CA$, so the twist sends $v\mapsto\theta_\lambda^\CA v$ for some $\theta_\lambda^\CA\in\CA$, and $V_{\lambda,\CA}=U^{\rm res}_\CA(\mathfrak g)\cdot v$ then forces it to be $\theta_\lambda^\CA$ times the identity. Evaluating the same balancing axiom on $\gamma$, now with $a=b=V_{\lambda,\CA}$, gives $(\theta_\lambda^\CA u)^{2}\gamma=\gamma$. Since $V_{\lambda,\CA}\otimes V_{\lambda,\CA}$ is free over $\CA$, the element $\gamma$ is not torsion, and therefore $(\theta_\lambda^\CA u)^{2}=1$ in $\CA$. Since $\CA$ is an integral domain, $\theta_\lambda^\CA u=\pm1$.
\end{proof}

We are now in a position to relate the FS indicator in our two categories of interest:

\begin{lem}\label{lem:rigid}
Let $\lambda$ be as in Lemma~\ref{lem:free}, and let $x=V_\lambda$ be the corresponding Wilson line with $x$ simple and self-dual. Then $\varkappa_x$ takes the same value in $\cC(\mathfrak g,k)$ as it does at $q=1$, where it is the classical Frobenius--Schur indicator of the corresponding representation of $G$.
\end{lem}

\begin{proof}
By Lemma~\ref{lem:free}, the braiding acts on the generator $\gamma$ of a free rank-one module and is invertible, so $c_{V_{\lambda,\CA},V_{\lambda,\CA}}\gamma=u\gamma$ for a unit $u\in \CA$. On its own, this discussion does not pin down something that survives specialisation. Indeed, a unit $\pm q^{m/L}$ takes different values, of different signs, under different ring homomorphisms $\CA\to\mathbb C$. By the final assertion of Lemma~\ref{lem:kappa}, however, $\varepsilon:=\theta_\lambda^\CA u$ is the constant $\pm1$, and ring homomorphisms fix $\pm1$. So $\varepsilon$, unlike $u$ itself, takes the same value after either specialization. At the root of unity, Lemma~\ref{lem:free}(b,c) shows that the image of $\gamma$ spans the vacuum channel of $V_\lambda\otimes V_\lambda$ in $\CC(\mathfrak g,k)$, so $u$ specializes there to $R^{xx}_{\mathbf 1}$, while $\theta^\CA_\lambda$ specializes to $\theta_x$. As a result, $\varepsilon=\theta_xR^{xx}_{\mathbf 1}=\varkappa_x$ by the first identity of Lemma~\ref{lem:kappa}.

At $q=1$ the twist specializes to $1$ and the braiding degenerates to the flip of tensor factors. Regarding the twist, note that $\theta_\lambda^\CA=\pm u^{-1}$ by Lemma~\ref{lem:kappa} and so $\theta_\lambda^\CA$ is a unit of $\CA$. Therefore, $\theta_\lambda^\CA=\mathrm{sgn}\bigl(\theta_\lambda^\CA\bigr)\,q^{m/L}$ for some integer $m$.\footnote{Neither $m$ nor the sign depends on the level, since $V_{\lambda,\CA}$ and its twist are defined over $\CA$. The level enters only through the choice of specialisation.} Fix $\lambda$, which lies in the level-$k$ alcove for all large $k$, and apply the level-$k$ specialisation for each such $k$. This maneuver sends $q^{1/L}\mapsto e^{i\pi/DL(k+h^\vee)}$, and it sends $\theta^\CA_\lambda$ to the topological spin of the Wilson line $V_\lambda=x$ in $\CC(\mathfrak g,k)$
\begin{equation}\label{eq:spinCS}
\theta_{x}=e^{2\pi ih_\lambda}~,    \qquad h_\lambda=\frac{\langle\lambda,\lambda+2\rho\rangle_{\rm long}}{2(k+h^\vee)}~,
\end{equation}
where $\langle\cdot,\cdot\rangle_{\rm long}$ is the form of \eqref{CSaction}. Comparing the two expressions for this image gives
\begin{equation}\label{eq:signEqn}
\mathrm{sgn}\bigl(\theta_\lambda^\CA\bigr)=e^{i\pi A/(k+h^\vee)}~,\qquad A:=\langle\lambda,\lambda+2\rho\rangle_{\rm long}-\frac{m}{DL}~,
\end{equation}
where the rational number $A$ depends on $\lambda$ but not on $k$. Squaring \eqref{eq:signEqn} gives $e^{2\pi iA/(k+h^\vee)}=1$, so that $A/(k+h^\vee)\in\mathbb Z$ for every admissible $k$. Were $A$ nonzero, we could take $k+h^\vee>|A|$ and obtain $0<|A/(k+h^\vee)|<1$, which is not an integer. Hence $A=0$, and \eqref{eq:signEqn} gives $\mathrm{sgn}(\theta^\CA_\lambda)=+1$. In particular, $\theta^\CA_\lambda$ specializes to $1$ under any $\phi$ with $\phi(q^{1/L})=1$.\footnote{Note that no limit is taken here. Each specialisation is a ring homomorphism, and the conclusion rests only on the integrality of $A/(k+h^\vee)$. As a by product, $A=0$ identifies the exponent, $m=DL\,\langle\lambda,\lambda+2\rho\rangle_{\rm long}$.} 
 
Regarding the braiding, note that $\check R=\tau_{\rm flip}\circ R$ \cite[Sec. 3.2.4]{rowell2005family}. At $q=1$, the $R$-matrix reduces to the identity ${\rm id}_{V^{\rm cl}_\lambda\otimes V^{\rm cl}_\lambda}$, since in the explicit form \cite[eq.~(1)]{sawin2006quantum} the Cartan prefactor $q^{\sum_i\check\alpha_i\otimes\lambda_i}$ becomes $1$, the leading term of the sum, with all $t_r=0$, is $1\otimes1$,\footnote{That the $t_r=0$ summand is present is forced by quasitriangularity: applying the counit to the first factor of $R$ must give $1$, whereas the counit annihilates $E^{(t)}_{\beta_r}$ for every $t\ge1$. Equivalently, without that summand $R$ would not be invertible.} while every remaining term carries a factor $(1-q_{\beta_r}^{-2})^{t_r}$ with $t_r\ge1$, which vanishes there. Therefore, after this specialization, $\check R$ equals the flip.
 
 Hence, by Lemma~\ref{lem:free}(b), the image of $\gamma$ spans the invariant line in $V^{\rm cl}_\lambda\otimes V^{\rm cl}_\lambda$, and $\varepsilon$ is the sign by which the flip acts on it (recall from above that the twist specializes to $1$ there). That sign is $+1$ if the invariant bilinear form on $V^{\rm cl}_\lambda$ is symmetric and $-1$ if alternating. In other words, it is the classical Frobenius-Schur indicator.
\end{proof}

Now that we have shown that the Chern-Simons Frobenius-Schur indicator is equal to the classical ${\rm Rep}(G)$ indicator, we can take advantage of known properties of the classical indicator to deduce an FS grading. The input we need is the following classical fact:

\begin{lem} \cite[Lem.~5.2]{adams2014real}\label{lem:adams}\footnote{See also references therein.} Let $G$ be a connected reductive complex group and $\pi$ an irreducible finite-dimensional self-dual representation with central character $\chi_\pi$. Then the Frobenius-Schur indicator of $\pi$ is $\varkappa^{\rm cl}(\pi)=\chi_\pi(z)$, where $z$ is defined as in \eqref{zcentral}.
\end{lem}

\noindent
To apply this result to $\pi=V^{\rm cl}_\lambda$, two points need to be checked. First, that the lemma covers our compact $G$ and second that its right-hand side is $\lambda(z)$ (its self-duality hypothesis is automatic here).

Note that Lemma~\ref{lem:adams} is stated in \cite{adams2014real} for a connected reductive complex group, whereas our $G$ is compact. The two settings agree, because one may pass to the complexification. For $G$ compact and simply connected, $G_{\mathbb C}$ is connected and semisimple, hence reductive, so Lemma~\ref{lem:adams} applies to it verbatim. Restriction along $G\hookrightarrow G_{\mathbb C}$ is an equivalence from the finite-dimensional holomorphic representations of $G_{\mathbb C}$ onto the finite-dimensional representations of $G$ (Weyl's unitarian trick). Being a restriction along a group homomorphism, this equivalence is symmetric monoidal, so it matches tensor products and duals. Being an equivalence, it is also fully faithful, so for every representation $V$, it identifies $\Hom_{G_{\mathbb C}}(V\otimes V,\mathbb C)$ with $\Hom_{G}(V\otimes V,\mathbb C)$, compatibly with the flip of the two tensor factors. Therefore, it matches invariant bilinear forms together with their symmetry type. Since $\varkappa^{\rm cl}$ is defined by that symmetry type, it is preserved under the restriction. Finally, the element $z$ of \eqref{zcentral} lies in $G$ itself, as \cite[Sec. 5]{adams2014real} mentions for every real form, so $\chi_\pi(z)$ denotes the same number in both categories.

It remains to identify $\chi_{V_\lambda^{\rm cl}}(z)$ with $\lambda(z)$. As $z$ is central, it acts on the irreducible $V_\lambda^{\rm cl}$ by a scalar, and that scalar may be read off on the highest weight space, where it is $\lambda(z)$. Since $\langle\lambda,2\rho^\vee\rangle\in\mathbb Z$, as noted below \eqref{zcentral}, we obtain $\lambda(z)=e^{2\pi i\langle\lambda,\rho^\vee\rangle}=(-1)^{\langle\lambda,2\rho^\vee\rangle}$. Then, in our notation
\begin{equation}
\varkappa^{\rm cl}(V_\lambda^{\rm cl})=\lambda(z)=(-1)^{\langle\lambda,2\rho^\vee\rangle}~.
\end{equation}

We are now ready to establish the existence of the FS grading (i.e., of multiplicativity of the FS indicator). To that end, recall that $G_k$ Wilson line fusion is given by the Kac-Walton formula
\begin{equation}\label{eq:KW}
  N^{(k)\,\nu}_{\lambda\mu}=\sum_{w\in W_\ell,\ w\cdot\nu\in P_+}\det(w)\,m^{w\cdot\nu}_{\lambda\mu}~,
\end{equation}
where $m$ are the classical multiplicities, $\ell:=k+h^\vee$, $\det(w)$ is the determinant of the finite Weyl part of $w$, and the sum is over $W_\ell=W\ltimes\ell Q^\vee$. This group acts via the shifted action $w\cdot\nu=w(\nu+\rho)-\rho$. Here $Q^\vee$ is regarded as a lattice of weights by means of the invariant form of \eqref{CSaction}, normalized so that long roots have squared length $2$.

Given this formula, we can determine the following useful fact about the fusion coefficients:

\begin{lem}\label{lem:mult}
If $N^{(k)\,\nu}_{\lambda\mu}\neq0$ then $\nu\equiv\lambda+\mu\pmod Q$.
\end{lem}

\begin{proof}
Classically, all weights of $V_\lambda^{\rm cl}\otimes V_\mu^{\rm cl}$ lie in $\lambda+\mu+Q$. The shifted action preserves $Q$-classes: for $w\in W$ one has $w(\nu)-\nu\in Q$ and $w(\rho)-\rho$ is a sum of roots. The translations also preserve $Q$-classes. Indeed, with the identification just fixed, a coroot $\alpha^\vee$ is sent to $\alpha$ when $\alpha$ is long and to $D\alpha$ when $\alpha$ is short, $D$ being the lacing number. In either case, the image lies in $Q$, so $Q^\vee\subseteq Q$, and translation by $\ell Q^\vee$ acts trivially on $P/Q$. Hence every term surviving in \eqref{eq:KW}, whatever its sign, sits in the class $\lambda+\mu+Q$.
\end{proof}

Now we are ready for the final result:

\begin{thm}\label{thm:A}
Let $\mathfrak g$ be simple with $-1\in W$, and $k\ge1$. Then for every Wilson line,
\begin{equation}\label{kappafinal}
\varkappa(V_\lambda)=\lambda(z)=(-1)^{\langle\lambda,2\rho^\vee\rangle}~,
\end{equation}
and $\varkappa$ is multiplicative on level-$k$ fusion.
\end{thm}

\begin{proof}
Lemmas~\ref{lem:kappa}, \ref{lem:free} and \ref{lem:rigid} identify $\varkappa(V_\lambda)$ with the classical indicator, which is $\lambda(z)$ by Lemma~\ref{lem:adams}. By our discussion around \eqref{zcentral}, that function is trivial on $Q$, and by Lemma~\ref{lem:mult} level-$k$ fusion preserves $Q$-classes, so $\varkappa_\lambda\varkappa_\mu=\varkappa_\nu$ whenever $N^{(k)\nu}_{\lambda\mu}\neq0$.
\end{proof}

Therefore, we have filled in the details to establish the result of \cite{simon2022straighteningfrobeniusschurindicator} for our TQFTs of interest: Chern-Simons theories with gauge groups $G\in G_{\rm sd}$ have an FS grading. As a result, by Corollary~\ref{cor:mainSelfDual}

\bigskip
\noindent
{\bf Implication 1.} Chern-Simons theories based on the gauge groups given in \eqref{SdualCS} all have gauges in which the corresponding $F$-symbols are real and the $R$-matrices are symmetric.

\bigskip
\noindent
More generally, for non-self-dual Chern-Simons theories based on the remaining compact simply connected gauge groups 
\begin{equation}\label{Gnsd}
G_{\rm nsd}=\left\{SU(r+1)~,\ {\rm Spin}(4m+2)~,\ E_6\right\}~,\ \ \ r>1~,
\end{equation}
we expect an analogous result to hold and for there to be gauges with real $F$-symbols and symmetric $R$-matrices. In fact, gauges with real $F$-symbols do indeed
exist in these cases, as established in \cite{Cesar} by different methods. We believe one route to obtaining a generalization of our result to $G_{\rm nsd}$ would be to define a generalization of the Kawanaka-Matsuyama indicator for compact simply connected Lie groups in place of the FS indicator and to appropriately generalize our above analysis.

\section{Examples of models with inherently complex \texorpdfstring{$F$}{F}-symbols}
\label{Complex}
We have seen the crucial roles that braiding and a flat $\mathbb{Z}_2$ charge conjugation symmetry play in arriving at unitary ribbon fusion categories that have gauges with real $F$-symbols.\footnote{Recall that we define a $\mathbb{Z}_2$ charge conjugation symmetry to be \lq\lq flat" if it is braided and satisfies $\lambda_a^{bc}={\sigma_b\sigma_c\over\sigma_a}=1$ at each fusion vertex.} In the case of a self-dual unitary ribbon category, the charge conjugation symmetry can be taken to be the trivial symmetry (i.e., to act as the identity functor), and flatness reduces to the existence of a grading for the Frobenius-Schur indicator, $\varkappa$: we have ${\varkappa_b\varkappa_c\over\varkappa_a}=1$ at each fusion vertex. 

If we consider categories that either lack a braiding or a flat charge conjugation, we can have inherently complex $F$-symbols. In this section we consider these two scenarios in turn. We refer the interested reader to our recent paper \cite{Paper1} for a more detailed analysis of two minimal ${\rm Rep}(G)$ categories that lack a braided $\mathbb{Z}_2$ charge conjugation and have inherently complex $F$-symbols. By definition, they also lack a flat charge conjugation symmetry. It is more subtle to find examples of categories with braided charge conjugation symmetries that are not flat and have inherently complex $F$-symbols. In this section, we find a minimal example exhibiting such a phenomenon, ${\rm Rep}(\mathbb{F}_3^2\rtimes Q_8)$. This category is self-dual (i.e., it has trivial charge conjugation) and no FS grading.

Let us first understand how relaxing braiding allows for complex $F$-symbols. To that end, consider the case of a general Abelian fusion category, ${\rm Vec}_G^{\omega}$, with fusion group $G$ and associator $[\omega]\in H^3(G,U(1))$. In such theories, the $F$-symbols are determined by the three input lines in Fig. \ref{fig:fr-moves}, and the pentagon equations imply that the $F$-symbols are three-cocycles, $\omega$. In this language, gauge transformations are shifts of $\omega$ by three-coboundaries, and the $F$-symbols are inherently complex if the following conditions are satisfied in group cohomology
\begin{equation}
[\omega^2]\ne[1]\ \Leftrightarrow\ [\omega]\ne\overline{[\omega]}~.
\end{equation}
Using the hexagon relations, one can show that if ${\rm Vec}_G^{\omega}$ admits a braiding, then $[\omega]=\overline{[\omega]}$ (see \cite{Paper1} for more details). Equivalently, if the $F$-symbols are inherently complex, then ${\rm Vec}_G^{\omega}$ does not admit a braiding.\footnote{Physically, this scenario should be thought of as describing non-genuine lines that bound topological surfaces implementing the $G$ 0-form symmetry.} In particular, we see that generic Abelian fusion categories are of this type.

A more interesting question is whether there are unitary ribbon fusion categories with inherently complex $F$-symbols. By our main result in this paper, such categories can have inherently complex $F$-symbols only if they lack a flat charge conjugation symmetry. While we have not shown that the absence of such a symmetry implies inherently complex $F$-symbols (this would be a direct converse of our main result), its absence is clearly a necessary condition.

One way to explore this question is to apply a result of Davydov \cite{MR3421083}: given a finite group, $G$, its category of irreps, ${\rm Rep}(G)$, has a braided charge conjugation automorphism if and only if $G$ has a class inverting automorphism.
Moreover, all braided charge conjugation automorphisms arise from class inverting automorphisms. Using this result, we studied the two simplest ${\rm Rep}(G)$ examples, by rank and categorical dimension, that lack a charge conjugation braided autoequivalence
\begin{equation}\label{Cexamples}
{\rm Rep}(\mathbb{Z}_7\rtimes\mathbb{Z}_3)~,\ \ \ {\rm Rep}(\mathbb{Z}_5\rtimes\mathbb{Z}_4)~,
\end{equation}
and explicitly showed that the corresponding $F$-symbols are inherently complex \cite{Paper1}.\footnote{Both categories have rank five, and their total categorical dimensions are twenty-one and twenty respectively.}\footnote{In fact, \cite{Cesar} substantially generalizes these results, showing that ${\rm Rep}(G)$ categories admit no gauges with real $F$-symbols for any non-abelian $G$ of odd order and determining criteria for inherently complex $F$-symbols in the case of ${\rm Rep}(\mathbb{F}_q^{+}\rtimes\mathbb{Z}_m)$.} Note that the second example above is a symmetric near group category (therefore, unlike the non-symmetric cases discussed in Sec. \ref{NSNGC}, this family of categories admits inherently complex $F$-symbols). By restriction, the same must be true for any of the (twisted) Drinfeld centers based on these groups: $\CZ^{\alpha}(\mathbb{Z}_7\rtimes\mathbb{Z}_3)$ and $\CZ^{\alpha'}(\mathbb{Z}_5\rtimes\mathbb{Z}_4)$, where $\alpha\in H^3(\mathbb{Z}_7\rtimes\mathbb{Z}_3, U(1))$ and $\alpha'\in H^3(\mathbb{Z}_5\rtimes\mathbb{Z}_4, U(1))$.

In the remainder of this section, we study an example of a unitary ribbon category, ${\rm Rep}(\mathbb{F}_3^2\rtimes Q_8)$, that has a braided charge conjugation symmetry but nonetheless lacks a flat charge conjugation symmetry: it is self-dual but lacks an FS grading.\footnote{As far as we are aware, this fact was first discussed in \cite{Ladisch}. We have explicitly verified it.} We will show ${\rm Rep}(\mathbb{F}_3^2\rtimes Q_8)$ has inherently complex $F$.

In fact, it turns out that ${\rm Rep}(\mathbb{F}_3^2\rtimes Q_8)$ has rank six and is a minimal ${\rm Rep}(G)$ example with inherently complex $F$ that has a braided but non-flat charge conjugation symmetry.\footnote{As in \cite{Paper1}, we define minimality by rank.} Our claim for minimality rests on the following facts:\footnote{Note that we do not rule out the possibility of smaller (by rank) fusion categories that are not of ${\rm Rep}(G)$ type.}\footnote{The same argument we discuss below, combined with the fact that all $G$ with $|G|<20$ have class-inverting automorphisms, also implies that the examples in \cite{Paper1} are the smallest-rank ${\rm Rep}(G)$-type categories with inherently complex $F$-symbols.}
\begin{enumerate}
\item There are no self-dual ${\rm Rep}(G)$ examples with inherently complex $F$-symbols and rank smaller than six. Indeed, such categories necessarily arise from \lq\lq ambivalent" groups: any $g\in G$ is conjugate to $g^{-1}$. However, using the relevant formula for the FS indicator,
\begin{equation}\label{eq:fs}
\nu_2(\chi)\;=\;\frac{1}{|G|}\sum_{g\in G}\chi(g^2)
\;=\;\begin{cases}
+1~, & \text{real (orthogonal) type}~,\\
\phantom{+}0~, & \text{complex type, } \chi\neq\overline{\chi}~,\\
-1~, & \text{quaternionic (symplectic) type}~.
\end{cases}
\end{equation}
one can explicitly check that all ambivalent groups with $|G|\le60$ have FS-graded ${\rm Rep}(G)$ categories. Moreover, all $G$ with $|G|>60$ have ${\rm Rep}(G)$ with rank at least six \cite{burnside1911theory,miller1910groups}.
\item All ${\rm Rep}(G)$ categories with $|G|\le60$ that admit a class-inverting automorphism and a charge conjugation acting non-trivially on objects have a flat $\mathbb{Z}_2$ charge conjugation with trivial Postnikov class and twisted FS indicator computed by the Kawanaka-Matsuyama formula \cite{kawanaka1990twisted}
\begin{equation}\label{eq:KM}
  \nu_\tau(\chi)\;=\;\frac{1}{|G|}\sum_{g\in G}\chi\big(g\,\tau(g)\big)~.
\end{equation}
Recall from our discussion around \eqref{eq:KM1} that $\tau\in{\rm Aut}(G)$ is a class-inverting involution and can be understood via the split extension $E=G\rtimes_\tau\langle t\rangle$ (we have checked that every $G$ in the range described has such a split extension). Note that $tgt^{-1}=\tau(g)$ and $t^2=1$. When $\tau={\rm Id}$, we recover the case of ambivalent $G$ and the FS formula in \eqref{eq:fs}. See Appendix \ref{KMAppendix} for details.
\end{enumerate}

Before discussing the example of ${\rm Rep}(\mathbb{F}_3^2\rtimes Q_8)$ in more detail, let us establish a few general results that will be useful in our analysis. First, we know from the discussion in Sec. \ref{RealGauge} that for a self-dual unitary ribbon category
\begin{equation}
F^{M_A}=\overline{F}~,
\end{equation}
for any $F$-symbol in the category. In particular, any $F$-symbol is gauge-equivalent to its complex conjugate. {\it As a result, unlike the case of the examples in \cite{Paper1} which lack a charge conjugation symmetry, no gauge-invariant combination of $F$-symbols with real coefficients can detect inherent complexity.}\footnote{Recall from \cite{Paper1} that, in the case of ${\rm Rep}(\mathbb{Z}_7\rtimes\mathbb{Z}_3)$ and ${\rm Rep}(\mathbb{Z}_5\rtimes\mathbb{Z}_4)$, $(F^{\chi\rho\rho}_\rho)_{\rho\rho}$ and $(F^{\chi\rho\rho}_{\rho})(F^{\rho\rho\chi}_{\rho})^{-1}$ are respective gauge-invariant detectors of inherent complexity.}

Instead, we will use the following fact:
\begin{lem}\label{StabLem}
Let $\CC$ be a unitary ribbon category in which every simple object is self-dual, so that the charge conjugation symmetry may be taken to be $J=\Id_{\CC}$. Fix a gauge, and let $M_A$ be the unitary vertex-basis gauge transformation of \eqref{eq:FstarMLFMR}, with components $(M_A)^{XY}_Z=\mathsf M^{XY}_Z$ determined by the local antiunitaries $A^{XY}_Z=\mathsf M^{XY}_ZK$, so that
\begin{equation}
F^{M_A}=\overline F~,\ \ \ (M_A)^{XY}_Z\,\overline{(M_A)^{XY}_Z}=\lambda^{XY}_Z\cdot\Id_{V^{XY}_Z}~.
\end{equation}
The second relation is the gauge symbol of $\overline R\circ R$ from Corollary~\ref{cor:R2Symbol}. Let $T\subseteq\Irr(\CC)^4$, and define $\CF:=\{F^{XYZ}_W:(X,Y,Z,W)\in T\}$ to be the corresponding collection of blocks of the $F$-symbol in this gauge. Finally, let ${\rm Stab}(\CF):=\{\Gamma: F^\Gamma=F\ \forall F\in\CF\}$ be its stabiliser inside the group of unitary vertex-basis gauge transformations. Suppose there is an admissible triple $(X,Y;Z)$ with
\begin{enumerate}[label=(\roman*)]
\item $\lambda^{XY}_Z=-1$, and
\item every $\Gamma\in{\rm Stab}(\CF)$ restricts to a scalar on $V^{XY}_Z$, i.e.\
$\Gamma^{XY}_Z=c(\Gamma)\cdot\Id_{V^{XY}_Z}$ with $c(\Gamma)\in U(1)$.
\end{enumerate}
Then there is no gauge transformation $H$ for which $F^H$ is real for every $F\in\CF$. In particular, $\CC$ has inherently complex $F$-symbols.
\end{lem}
\begin{proof}
To understand this statement, suppose to the contrary that $\CF$ can be made simultaneously real. That is, suppose there exists a single gauge transformation, $H$, such that $F^H$ is real for every $F\in\CF$. Then, for each such $F$
\begin{equation}\label{gaugeTrans}
F^H=\overline{F^H}=\overline F^{\overline H}=(F^{M_A})^{\overline H}=F^{\overline H M_A}~.
\end{equation}
Now, applying the $H^{-1}$ gauge transformation produces
\begin{equation}
(F^{H})^{H^{-1}}=F^{H^{-1}H}=F~,
\end{equation}
on the lefthand side of \eqref{gaugeTrans}. On the righthand side of that same equation, we find 
\begin{equation}
(F^{\overline H M_A})^{H^{-1}}=F^{H^{-1}\overline H M_A}~,
\end{equation}
and so
\begin{equation}
H^{-1}\overline H M_A\in{\rm Stab}(\CF)~.
\end{equation}
By hypothesis, we then have $H^{-1}\overline H M_A=c\cdot{\rm Id}_{V^{XY}_Z}$ when restricted to $V^{XY}_Z$. As a result, $\overline H=cHM_A^{-1}$. Then conjugating and substituting yields
\begin{eqnarray}
H&=&\bar c\cdot \overline H\overline{M_A^{-1}}=\bar c(cHM_A^{-1})\overline{M_A^{-1}}=\bar c(cHM_A^{-1})(\lambda^{XY}_Z)^{-1}{M_A}=|c|^2(\lambda^{XY}_Z)^{-1}H\cr&=&-H~,
\end{eqnarray}
where we have used the fact that $M_A\overline{M_A}=\lambda^{XY}_Z\cdot{\rm Id}_{V^{XY}_Z}$ and that $c$ is a phase as $H^{-1}\overline{H}M_A$ is unitary. We therefore arrive at a contradiction since $H\ne0$ is invertible. Therefore, $\CF$ cannot be made simultaneously real.
\end{proof}
As we will see, by judiciously choosing $\CF$ and showing that ${\rm Stab}|_{V^{XY}_Z}(\CF)$ reduces to a scalar, we will be able to prove inherent complexity in the case of ${\rm Rep}(\mathbb{F}_3^2\rtimes Q_8)$.

\subsection{A self-dual example: \texorpdfstring{${\rm Rep}(\mathbb{F}_3^2\rtimes Q_8)$}{F3xF3x|Q8}}
\label{ssec:FS2ObstructionExample}
The above discussion motivates a search for ambivalent groups whose representation categories lack an FS grading. Several examples of such groups were discussed in \cite{Ladisch}, namely an order-seventy-two group, $G:=\mathbb{F}_3^2\rtimes Q_8$, along with four order-sixty-four groups.\footnote{These latter groups have GAP IDs [64, 218], [64, 224], [64, 243], and [64, 245].} However, ${\rm Rep}(G)$ has rank six, while the other representation categories have rank at least nineteen. Since there are no ambivalent groups with order less than sixty-four that lack an FS grading, and since all groups of order larger than sixty have at least six irreps (see \cite{burnside1911theory,miller1910groups}), ${\rm Rep}(G)$ has minimal rank. We will now argue it has inherently complex $F$-symbols.\footnote{An infinite family of related examples, $(\mathbb{F}_p^2)^{+}\rtimes Q_8$ with $p>8$, is treated independently in \cite{Cesar} using other methods.}

To understand this statement, let us first construct $G$. The starting point is the observation that the quaternions, $Q_8=\{\pm1,\pm i,\pm j,\pm k\}$, embed in $SL(2,3)$ via
\begin{equation}
i\longmapsto\begin{pmatrix}0&-1\\1&0\end{pmatrix}~,\ \ \ 
j\longmapsto\begin{pmatrix}1&1\\1&-1\end{pmatrix}~,\ \ \
k=ij\longmapsto\begin{pmatrix}-1&1\\1&1\end{pmatrix}~ \ \ \ {\rm mod\ 3}~.
\end{equation}
Therefore, we have an action of $Q_8$ on $A=\mathbb{F}_3^2$, because each $v\in A$ can be written as a two-component column vector acted on from the left by the above matrices. Using this action, we define an order-seventy-two group, $\mathbb{F}_3^2\rtimes Q_8$, with multiplication 
\begin{equation}\label{sdp}
(v,q)(v',q'):=(v+q\,v',\;qq')~, \ \ \ v\in A~, \ \ \ q\in Q_8~.
\end{equation}

\begin{table}[t]
\begin{center}
\begin{tabular}{c|cccccc}
& $C_e$ & $C_v$ & $C_{-1}$ & $C_i$ & $C_j$ & $C_k$\\
\hline
$|C_x|$ & $1$ & $8$ & $9$ & $18$ & $18$ & $18$\\
\hline
$\chi_1$ & $1$ & $1$ & $1$ & $1$ & $1$ & $1$\\
$\chi_i$ & $1$ & $1$ & $1$ & $1$ & $-1$ & $-1$\\
$\chi_j$ & $1$ & $1$ & $1$ & $-1$ & $1$ & $-1$\\
$\chi_k$ & $1$ & $1$ & $1$ & $-1$ & $-1$ & $1$\\
$\sigma$ & $2$ & $2$ & $-2$ & $0$ & $0$ & $0$\\
$\rho$ & $8$ & $-1$ & $0$ & $0$ & $0$ & $0$\\
\end{tabular}
\end{center}
\caption{Character table for $G:=\mathbb{F}_3^2\rtimes Q_8$, with $C_v$ the conjugacy class of any $(v,1)$, and $C_{q\ne1}$ the conjugacy class of $(0,q)$. The group is ambivalent, and hence the characters are real.}
\label{chars}
\end{table}

From \eqref{sdp}, it is straightforward to see that we have six conjugacy classes (see Table \ref{chars} for details). The corresponding six irreducible representations can be understood as follows. Five of the 6 irreps are lifts from $Q_8$ via pullback along $\pi:G\to G/A\cong Q_8$: these are the four one-dimensional irreps---$\chi_1:=1$, $\chi_i$, $\chi_j$, and $\chi_k$---along with the two-dimensional irrep, $\sigma$. This latter representation is fixed by
\begin{equation}\label{eq:sigmamat}
\sigma(v,i)=\begin{pmatrix}\mathrm i&0\\0&-\mathrm i\end{pmatrix}~,\ \
\sigma(v,j)=\begin{pmatrix}0&1\\-1&0\end{pmatrix}~,\ \
\sigma(v,k)=\begin{pmatrix}0&\mathrm i\\ \mathrm i&0\end{pmatrix}~,\ \
\sigma(v,-1)=-\mathds{1}~.
\end{equation}
Note that these matrices are independent of $v$, and so we will sometimes be cavalier with this index. Note also that $\sigma(\pm1)$ and $\sigma(\pm j)$ are real, while $\sigma(\pm i)$ and $\sigma(\pm k)$ are purely imaginary, so that
\begin{equation}\label{eq:sigmareal}
\overline{\sigma(q)}=\pm\,\sigma(q)~,\qquad q\in Q_8~.
\end{equation}
In particular, $q\mapsto\overline{\sigma(q)}$ is again a faithful (pseudoreal) two-dimensional representation of $Q_8$, and the sign in \eqref{eq:sigmareal} cancels between the two outer factors of a conjugation,
\begin{equation}\label{eq:conjsame}
\sigma(p)\,X\,\sigma(p)^{-1}=\overline{\sigma(p)}\,X\,\overline{\sigma(p)}^{-1}~,\qquad
p\in Q_8~,\ \ X\in M_2(\mathbb{C})~.
\end{equation}
The one-dimensional irreps appear via
\begin{equation}\label{eq:pairing}
\sigma(g)^{-1}\,\sigma(0,q)\,\sigma(g)=\chi_q(g)\,\sigma(0,q)~,\qquad
\chi_q(v,r)=\begin{cases}+1~,&rq=qr~,\\-1~,&\text{else~,}\end{cases}
\end{equation}
where $\chi_q(v,r)$ is likewise independent of $v\in A$, with $q\in\left\{1,i,j,k\right\}$ and $g=(v,r)\in G$.

The final irrep, $\rho$, is eight-dimensional and is induced from $A$ (it can be induced from any of the non-trivial $A$ irreps). Its restriction to $A$ is
\begin{equation}\label{eq:restriction}
\rho|_{A}\cong\bigoplus_{q\in Q_8}\hat\chi_{w_q}~,
\end{equation}
where $\hat\chi_w(v):=\zeta_3^{\langle w,v\rangle}$, $\zeta_3:=\exp(2\pi i/3)$, and $\langle v,w\rangle:=v^Tw$ is the inner product on $A$. The $A$-weights are defined as follows
\begin{equation}\label{Aweights}
w_q:=(q^{-1})^{T}\begin{pmatrix}1\\0\end{pmatrix}~,
\end{equation}
and so we have
\begin{center}
\renewcommand{\arraystretch}{.75}
\begin{tabular}{c|cccccccc}
$q$ & $1$ & $-1$ & $i$ & $-i$ & $j$ & $-j$ & $k$ & $-k$\\
\hline
$w_q$
& $\begin{pmatrix}1\\0\end{pmatrix}$
& $\begin{pmatrix}2\\0\end{pmatrix}$
& $\begin{pmatrix}0\\1\end{pmatrix}$
& $\begin{pmatrix}0\\2\end{pmatrix}$
& $\begin{pmatrix}2\\2\end{pmatrix}$
& $\begin{pmatrix}1\\1\end{pmatrix}$
& $\begin{pmatrix}1\\2\end{pmatrix}$
& $\begin{pmatrix}2\\1\end{pmatrix}$
\end{tabular}~.\label{Aweighttable}
\end{center}
In what follows, we will find that it is useful to define the following map
\begin{equation}\label{eq:mu}
\mathsf{m}:Q_8\setminus\{-1\}\longrightarrow Q_8\setminus\{1\}~,\ \ \ w_{\mathsf{m}(u)}=w_1+w_u~.
\end{equation}
This map is well-defined and bijective since $w_1+w_u=0$ only for $u=-1$, and $w_1+w_u\ne w_1$. From the table of $A$-weights above, we see that
\begin{equation}\label{eq:mutable}
\mathsf{m}(1)=-1~,\qquad (i,\,-j,\,-k)\ \circlearrowleft~,\qquad
(-i,\,k,\,j)\ \circlearrowleft~.
\end{equation}
The space, $V_\rho$, on which $\rho$ acts has orthonormal basis $\{\psi_q\}_{q\in Q_8}$ and
\begin{equation}\label{eq:rho}
\rho(v,1)\psi_q=\hat\chi_{w_q}(v) \psi_q\ \ \ ~,\ \ \  \rho(0,r)\psi_q=\psi_{rq}~.
\end{equation}
The six irreps are summarized in Table \ref{chars}, and their fusion rules are given in Table \ref{fusion}.

\begin{table}[t]
\begin{center}
\begin{tabular}{c|cccccc}
$\otimes$ & $\chi_1$ & $\chi_i$ & $\chi_j$ & $\chi_k$ & $\sigma$ & $\rho$\\
\hline
$\chi_1$ & $\chi_1$ & $\chi_i$ & $\chi_j$ & $\chi_k$ & $\sigma$ & $\rho$\\
$\chi_i$ & $\chi_i$ & $\chi_1$ & $\chi_k$ & $\chi_j$ & $\sigma$ & $\rho$\\
$\chi_j$ & $\chi_j$ & $\chi_k$ & $\chi_1$ & $\chi_i$ & $\sigma$ & $\rho$\\
$\chi_k$ & $\chi_k$ & $\chi_j$ & $\chi_i$ & $\chi_1$ & $\sigma$ & $\rho$\\
$\sigma$ & $\sigma$ & $\sigma$ & $\sigma$ & $\sigma$ &
$\chi_1\oplus\chi_i\oplus\chi_j\oplus\chi_k$ & $2\rho$\\
$\rho$ & $\rho$ & $\rho$ & $\rho$ & $\rho$ & $2\rho$ &
$\chi_1\oplus\chi_i\oplus\chi_j\oplus\chi_k\oplus2\sigma\oplus7\rho$\\
\end{tabular}
\end{center}
\caption{Fusion rules for the irreducible representations of $G:=\mathbb{F}_3^2\rtimes Q_8$.}
\label{fusion}
\end{table}

Given the above discussion, we can compute the FS indicators via \eqref{eq:fs}. From the data in Table \ref{chars}, we see that the FS indicators are
\begin{equation}
\nu_2(1)=\nu_2(\chi_i)=\nu_2(\chi_j)=\nu_2(\chi_k)=\nu_2(\rho)=1~, \ \ \ \nu_2(\sigma)=-1~,
\end{equation}
and, from the fusion rules in Table \ref{fusion}, we see that ${\rm Rep}(G)$ lacks an FS grading because for $V^{\sigma\rho}_{\rho}$, $V^{\rho\sigma}_{\rho}$, and $V^{\rho\rho}_{\sigma}$ we have
\begin{equation}
\lambda^{\sigma\rho}_{\rho}=\lambda^{\rho\sigma}_{\rho}=\lambda^{\rho\rho}_{\sigma}=-1~.
\end{equation}
Therefore, in order to find $F$-symbols that are inherently complex, it is natural to study cases involving these vertices.

As we have discussed at the beginning of this section, unlike in the examples of \cite{Paper1}, in the case of ${\rm Rep}(\mathbb{F}_3^2\rtimes Q_8)$ we are not able to construct gauge-invariant inherent complexity detectors from expressions involving the $F$-symbols and real coefficients. Instead, we will study the following collection of $F$-symbols involving $V^{\sigma\rho}_{\rho}$ and $V^{\rho\sigma}_{\rho}$
\begin{equation}
\label{eq:exampleFSet}
\CF:=\left\{F^{\rho\sigma\rho}_{\chi}~,\ F^{\sigma\sigma\rho}_{\rho}~,\ F^{\sigma\rho\rho}_{\rho}~,\ F^{\rho\sigma\rho}_{\rho}\right\}~,
\end{equation}
where $\chi$ is any of the one-dimensional irreps of $G$, and argue that every $\Gamma\in{\rm Stab}(\CF)$ acts as a scalar on $\mathcal{K}:=V^{\sigma\rho}_{\rho}$. In other words,
\begin{equation}
{\rm Stab}(\CF)|_{\mathcal{K}}\subseteq U(1)\cdot\mathds{1}_2<U(2)~, \ \text{where}\ {\rm Stab}(\CF)|_{\mathcal{K}}:=\left\{\Gamma^{\sigma\rho}_{\rho}:\Gamma\in{\rm Stab}(\CF)\right\}~.
\end{equation}
By Lemma \ref{StabLem}, this result then implies that data in $\CF$ cannot be made simultaneously real.

Let us first discuss $F^{\rho\sigma\rho}_{\chi}$. The intertwiners we construct live in the fusion spaces $V_{\sigma\rho}^{\rho}$, $V_{\rho\sigma}^{\rho}$, and $V_{\rho\rho}^{\chi}$. However, $F^{\rho\sigma\rho}_{\chi}$ is an overlap involving the adjoint splitting spaces $V^{\sigma\rho}_{\rho}$, $V^{\rho\sigma}_{\rho}$, and $V^{\rho\rho}_{\chi}$, as in \eqref{eq:Ftrace} below.

We write a basis of intertwiners as
\begin{equation}
\varphi_t\in V_{\sigma\rho}^{\rho}~, \ \ \ \varphi'_s\in V_{\rho\sigma}^{\rho}~, \ \ \ \Phi_\chi\in V_{\rho\rho}^{\chi}~,
\end{equation}
where $t,s\in\mathbb{C}^2$ are vectors labelling the possible $\varphi_t$ and $\varphi'_s$ (recall that $N_{\sigma\rho}^{\rho}=N_{\rho\sigma}^{\rho}=2$ and that $V_{\rho\rho}^\chi$ is one dimensional). In particular, we have
\begin{equation}
\varphi_t:=\sum_{\alpha}t_{\alpha}\varphi_{\alpha}~, \ \ \ \varphi'_s:=\sum_{\beta}s_{\beta}\varphi'_{\beta}~,
\end{equation}
where $\varphi_{\alpha}$ and $\varphi'_{\beta}$ are orthonormal.

Given this discussion, we can think of $F^{\rho\sigma\rho}_{\chi}$ as an overlap between left and right splitting paths $\CL_{(\rho,\alpha)}, \CR_{(\rho,\mu)}\in{\rm Hom}(V_{\chi},V_\rho\otimes V_\sigma\otimes V_\rho)$
\begin{equation}\label{eq:Ftrace}
\left[F^{\rho\sigma\rho}_{\chi}\right]_{(\rho,\alpha)(\rho,\mu)}=\bigl\langle \mathcal R_{(\rho,\mu)}\big|\mathcal L_{(\rho,\alpha)}\bigr\rangle=\bigl\langle L_{(\rho,\alpha)}\big|R_{(\rho,\mu)}\bigr\rangle=\frac{1}{\dim V_\chi}\,{\rm Tr}_{V_\rho\otimes V_\sigma\otimes V_\rho}\Bigl(L_{(\rho,\alpha)}^{\dagger}\,R_{(\rho,\mu)}\Bigr)~,
\end{equation}
where $\mathcal R_{(\rho,\mu)}=R_{(\rho,\mu)}^\dagger$ and $\mathcal L_{(\rho,\alpha)}=L_{(\rho,\alpha)}^\dagger$ are built in terms of the daggers of the intertwiners (since the $F$-symbol is defined in terms of splitting spaces), and
\begin{equation}\label{LR}
L_{(\rho,\alpha)}(\psi_q\otimes x_a\otimes\psi_{q'})=\Phi_\chi\bigl(\varphi'_{\alpha}(\psi_q\otimes x_a)\otimes\psi_{q'}\bigr)~, \ \ \  R_{(\rho,\mu)}(\psi_q\otimes x_a\otimes\psi_{q'})=\Phi_\chi\bigl(\psi_q\otimes\varphi_\mu(x_a\otimes\psi_{q'})\bigr)~.
\end{equation}
Here $x_a\in V_\sigma$ and $\psi_q,\psi_{q'}\in V_\rho$ are orthonormal basis vectors. The second equality in \eqref{eq:Ftrace} uses the fact that $\dagger$ is antiunitary.

To evaluate the $F$-symbols, we need to fix the corresponding intertwiners via $G$-equivariance
\begin{equation}\label{Gequiv}
T\circ(\pi_X(g)\otimes\pi_Y(g))=\pi_Z(g)\circ T~,
\end{equation}
for $T\in V_{XY}^Z$. It suffices to impose this condition for $g\in A, Q_8<G$, since these elements generate $G$.

We first consider $\varphi_t\in V_{\sigma\rho}^{\rho}$.  A priori,
\begin{equation}
\varphi_t(x_a\otimes\psi_q)=\sum_{q''}c_{q''}(x_a)\,\psi_{q''}~,
\end{equation}
with each $c_{q''}$ linear in $x_a$. Now, impose \eqref{Gequiv} for $g=(v,1)\in A<G$, using $A\subseteq\ker\sigma$ on the left and \eqref{eq:rho} on both sides:
\begin{equation}
\varphi_t\bigl(x_a\otimes\rho(v)\psi_q\bigr)=\hat\chi_{w_q}(v)\sum_{q''}c_{q''}(x_a)\,\psi_{q''}~, \ \ \rho(v)\,\varphi_t(x_a\otimes\psi_q)=\sum_{q''}c_{q''}(x_a)\,\hat\chi_{w_{q''}}(v)\,\psi_{q''}~.
\end{equation}
Comparing coefficients of $\psi_{q''}$ gives
\begin{equation}\label{eq:separate}
c_{q''}(x)\,\left(\hat\chi_{w_q}(v)-\hat\chi_{w_{q''}}(v)\right)=0~,\ \ \ \forall\ v\in A~.
\end{equation}
By \eqref{eq:restriction} the eight $A$-weights, $w_q$, described in \eqref{Aweights} are pairwise distinct, so for $q''\ne q$ the characters $\hat\chi_{w_q}$ and $\hat\chi_{w_{q''}}$ differ, and some
$v$ makes the bracket nonzero. Therefore, $c_{q''}=0$ for every $q''\ne q$. What remains is
\begin{equation}
\varphi_t(x_a\otimes\psi_q)=c_q(x_a)\,\psi_q=\langle t(q),x_a\rangle\,\psi_q~,
\end{equation}
where $t(q)\in V_\sigma^{*}$ is a functional we will fix. To that end, for $(0,r)\in Q_8<G$, evaluate
both sides of \eqref{Gequiv} on $x_a\otimes\psi_q$:
\begin{equation}
\varphi_t\bigl(\sigma(0,r)x_a\otimes\psi_{rq}\bigr)=\langle t(rq),\sigma(0,r)x_a\rangle\,\psi_{rq}~,\ \ \ \rho(r)\varphi_t(x_a\otimes\psi_q)=\langle t(q),x_a\rangle\,\psi_{rq}~.
\end{equation}
As $\langle t(rq),\sigma(r)x_a\rangle=\langle t(rq)\sigma(r),x_a\rangle$, and $\left\{x_a\right\}$ is a basis,
\begin{equation}\label{eq:transportEquivariance}
t(rq)\,\sigma(r)=t(q)~.
\end{equation}
Now, setting $q=1$ in \eqref{eq:transportEquivariance} gives $t(r)=t\sigma(r)^{-1}$ with $t:=t(1)$. This relation satisfies \eqref{eq:transportEquivariance} for all $q$, since $t\sigma(rq)^{-1}\sigma(r)=t\sigma(q)^{-1}$. The solution space is $V_\sigma^{*}\cong\mathbb{C}^{2}$, which is consistent with $N^{\rho}_{\sigma\rho}=2$ (see Table \ref{fusion}). After appropriately normalizing, we obtain
\begin{equation}\label{varphifinal}
\varphi_t(x_a\otimes \psi_q)=\langle t\sigma(q)^{-1},x_a\rangle\,\psi_q~.
\end{equation}
Similar logic shows that 
\begin{equation}\label{varphiPfinal}
\varphi'_s(\psi_q\otimes x_a)=\langle s\sigma(q)^{-1},x_a\rangle\,\psi_q~.
\end{equation}

Finally, consider $\Phi_\chi$. Here the target is one-dimensional, $V_\chi\cong\mathbb{C}$, and, a priori we have
\begin{equation}\label{PhiStart}
\Phi_\chi(\psi_q\otimes\psi_{q'})=c_{q,q'}\in\mathbb{C}~,\ \ \  q~,\ q'\in Q_8~.
\end{equation}
Again, we impose \eqref{Gequiv} for $g=(v,1)\in A<G$ and evaluate the two sides on $\psi_q\otimes\psi_{q'}$, using \eqref{eq:rho} on the source and $A\subseteq\ker\chi$ on the target
\begin{equation}
\Phi_\chi\bigl(\rho(v)\psi_q\otimes\rho(v)\psi_{q'}\bigr)=\hat\chi_{w_q}(v)\,\hat\chi_{w_{q'}}(v)\,c_{q,q'} =\hat\chi_{w_q+w_{q'}}(v)\,c_{q,q'}~,\ \ \ \chi(v)\,\Phi_\chi\bigl(\psi_q\otimes\psi_{q'}\bigr)=c_{q,q'}~.
\end{equation}
Equating the two results above gives
\begin{equation}\label{eq:separate2}
c_{q,q'}\bigl[\hat\chi_{w_q+w_{q'}}(v)-1\bigr]=0~,\ \ \ \forall\ v\in A~.
\end{equation}
The bracket vanishes for every $v$ exactly when $\hat\chi_{w_q+w_{q'}}$ is trivial, so $w_q+w_{q'}=0$ (since the pairing is nondegenerate). Therefore
\begin{equation}
c_{q,q'}=c_q\,\delta_{q',-q}~, \ \ \ c_q:=c_{q,-q}~,
\end{equation}
and the sixty-four unknowns in \eqref{PhiStart} drop to eight. To constrain things further, impose \eqref{Gequiv} for $g=(0,r)\in Q_8$ on a surviving basis vector $\psi_q\otimes\psi_{-q}$. Using $\rho(r)\psi_q=\psi_{rq}$, and $\rho(r)\psi_{-q}=\psi_{-rq}$, the two sides are $c_{rq}$ and $\chi(r)c_q$, so
\begin{equation}\label{eq:cocycle}
c_{rq}=\chi(r)\,c_q~.
\end{equation}
Now, setting $q=1$ gives $c_r=\chi(r)c_1$, and this satisfies \eqref{eq:cocycle} for all $q$ because $\chi$ is a homomorphism: $c_{rq}=\chi(rq)c_1=\chi(r)\chi(q)c_1=\chi(r)c_q$. One free constant remains, which is consistent with $N^{\chi}_{\rho\rho}=1$. Normalizing $\Phi_\chi$ gives
\begin{equation}\label{ChiFinal}
\Phi_\chi(\psi_q\otimes\psi_{q'})=8^{-1/2}\delta_{q',-q}\chi(q)~.
\end{equation}

We are now ready to compute $F^{\rho\sigma\rho}_{\chi}$. To that end, plugging \eqref{varphifinal}, \eqref{varphiPfinal}, and \eqref{ChiFinal} into \eqref{LR} we obtain
\begin{eqnarray}
L_{(\rho,\alpha)}(\psi_q\otimes x_a\otimes\psi_{q'})&=&8^{-1/2}\,\delta_{q',-q}\,\chi(q)\,[\sigma(q)^{-1}]_{\alpha a}~,\cr
R_{(\rho,\mu)}(\psi_q\otimes x_a\otimes\psi_{q'}) &=&8^{-1/2}\,\delta_{q',-q}\,\chi(q)\,[\sigma(-q)^{-1}]_{\mu a}=-8^{-1/2}\,\delta_{q',-q}\,\chi(q)\,[\sigma(q)^{-1}]_{\mu a}~, \ \ \ \ \
\end{eqnarray}
Finally, substituting into \eqref{eq:Ftrace} we obtain the $F$-symbol
\begin{eqnarray}\label{Frsrc}
\left[F^{\rho\sigma\rho}_{\chi}\right]_{(\rho,\alpha)(\rho,\mu)}
&=&\sum_{q,a,q'}
\overline{L_{(\rho,\alpha)}\bigl(\psi_q\otimes x_a\otimes\psi_{q'}\bigr)}\;R_{(\rho,\mu)}\bigl(\psi_q\otimes x_a\otimes\psi_{q'}\bigr)\cr&=&-\frac18\sum_{q\in Q_8}\sum_{a=1}^{2} \overline{[\sigma(q)^{-1}]_{\alpha a}}\;[\sigma(q)^{-1}]_{\mu a} =-\frac18\cdot8\,\delta_{\alpha\mu}=-\delta_{\alpha\mu}~.
\end{eqnarray}

Before moving on to the remaining $F$-symbols in $\CF$, let us understand how \eqref{Frsrc} constrains ${\rm Stab}(\CF)$. To that end, the four gauge factors acting on $F^{\rho\sigma\rho}_{\chi}$ are
\begin{equation}
P'=\Gamma^{\rho\sigma}_{\rho}~, \ \ \ g=\Gamma^{\rho\rho}_{\chi}~,\ \ \ P=\Gamma^{\sigma\rho}_{\rho}~,
\end{equation}
with $g$ and $g^{-1}$ appearing and so
\begin{equation}
\left[F^{\rho\sigma\rho}_{\chi}\right]^\Gamma_{(\rho,\alpha)(\rho,\mu)} =\sum_{\alpha',\mu'}P'_{\alpha\alpha'}\;g\;\left[F^{\rho\sigma\rho}_{\chi}\right]_{(\rho,\alpha')(\rho,\mu')} \bigl[P^{-1}\bigr]_{\mu'\mu}\;g^{-1}=\bigl(P'F^{\rho\sigma\rho}_\chi P^{-1}\bigr)_{\alpha\mu}=-\left[P'P^{-1}\right]_{\alpha\mu}~.
\end{equation}
The stability condition becomes
\begin{equation}\label{PprP}
P'=P~.
\end{equation}

To obtain further constraints, let us consider $F^{\sigma\sigma\rho}_\rho$. The intertwiners we construct live in the fusion spaces $V_{\sigma\sigma}^\chi$, $V_{\chi\rho}^\rho$, and $V_{\sigma\rho}^\rho$. However, $F^{\sigma\sigma\rho}_{\rho}$ is an overlap involving the corresponding adjoint splitting spaces.

We write a basis of intertwiners on the first two spaces as
\begin{equation}
\phi_\chi\in V_{\sigma\sigma}^\chi~, \ \ \ \eta_\chi\in V_{\chi\rho}^\rho~,
\end{equation}
and recall the result in \eqref{varphifinal} for $\varphi_t\in V_{\sigma\rho}^\rho$.

Given this discussion, we can think of $F^{\sigma\sigma\rho}_{\rho}$ as a comparison between the splitting paths $\CL_\chi,\CR_{(\rho,\mu,\nu)}\in{\rm Hom}(V_\rho,V_\sigma\otimes V_\sigma\otimes V_\rho)$
\begin{equation}\label{eq:Ftrace2}
\bigl[F^{\sigma\sigma\rho}_{\rho}\bigr]_{(\chi)(\rho,\mu,\nu)}=\langle\CR_{(\rho,\mu,\nu)}|\CL_\chi\rangle=\bigl\langle L_{\chi}\big|\,R_{(\rho,\mu,\nu)}\bigr\rangle=\frac{1}{\dim V_\rho}{\rm Tr}_{V_\sigma\otimes V_\sigma\otimes V_\rho}\Bigl(L_{\chi}^\dagger R_{(\rho,\mu,\nu)}\Bigr)~,
\end{equation}
where $\CL_\chi=L_{\chi}^\dagger$, $\CR_{(\rho,\mu,\nu)}=R_{(\rho,\mu,\nu)}^\dagger$, and
\begin{eqnarray}\label{LR2}
L_{\chi}\bigl(x_a\otimes x_b\otimes\psi_q\bigr)=\eta_\chi\bigl(\phi_\chi(x_a\otimes x_b)\otimes\psi_q\bigr)~,\cr
R_{(\rho,\mu,\nu)}\bigl(x_a\otimes x_b\otimes\psi_q\bigr)=\varphi_\nu\bigl(x_a\otimes\varphi_\mu(x_b\otimes\psi_q)\bigr)~,
\end{eqnarray}
where $x_a, x_b\in V_\sigma$ and $\psi_q\in V_\rho$ are orthonormal basis vectors.

We first consider $\phi_\chi$. A priori, we have
\begin{equation}
\phi_\chi(x_a\otimes x_b)=\Psi_{ab}\in\mathbb{C}~,\qquad a,b\in\{1,2\}~,
\end{equation}
corresponding to four unknowns. Now, impose \eqref{Gequiv} for $g=(v,1)\in A<G$ and evaluate both sides on $x_a\otimes x_b$, using $A\subseteq\ker\sigma$ on the source and $A\subseteq\ker\chi$ on the target:
\begin{equation}
\phi_\chi\bigl(\sigma(v,1)x_a\otimes\sigma(v,1)x_b\bigr)=\phi_\chi(x_a\otimes x_b)=\Psi_{ab}~, \ \ \ \chi(v)\,\phi_\chi(x_a\otimes x_b)=\Psi_{ab}~.
\end{equation}
The two outcomes agree identically, so the $A$-condition is vacuous. Let us now take $g=(0,r)\in Q_8<G$. Since $\sigma(0,r)x_a=\sum_c\sigma(0,r)_{ca}x_c$, the two sides on $x_a\otimes x_b$ are
\begin{eqnarray}
\phi_\chi\bigl(\sigma(0,r)x_a\otimes\sigma(0,r)x_b\bigr)&=&\sum_{c,d}\sigma(0,r)_{ca}\,\sigma(0,r)_{db}\,\Psi_{cd}=\bigl[\sigma(0,r)^{T}\Psi\,\sigma(0,r)\bigr]_{ab}~,\cr \chi(r)\,\phi_\chi(x_a\otimes x_b)&=&\chi(r)\,\Psi_{ab}~,
\end{eqnarray}
and so
\begin{equation}\label{eq:formcond}
\sigma(r)^{T}\Psi\,\sigma(r)=\chi(r)\,\Psi~, \ \ \ r\in Q_8~.
\end{equation}

Define
\begin{equation}
\widehat\Psi:=\Omega^{-1}\Psi~,\ \ \ \Omega:=\begin{pmatrix}0&1\\-1&0\end{pmatrix}~,
\end{equation}
which is a bijection of $M_2(\mathbb{C})$. With $\Psi=\Omega\widehat\Psi$, the left side of \eqref{eq:formcond} becomes
$\sigma(r)^{T}\Omega\widehat\Psi\sigma(r)
=\Omega\,\sigma(r)^{-1}\widehat\Psi\sigma(r)$, so \eqref{eq:formcond} is equivalent to
\begin{equation}\label{conjSubst}
\sigma(r)^{-1}\widehat\Psi\,\sigma(r)=\chi(r)\,\widehat\Psi~.
\end{equation}
The point of this substitution is that it converts a congruence $\Psi\mapsto\sigma^{T}\Psi\sigma$ into a conjugation $\widehat\Psi\mapsto\sigma^{-1}\widehat\Psi\sigma$ which \eqref{eq:pairing} diagonalises.

Now, the four matrices, $\sigma(0,q)$, $q\in\{1,i,j,k\}$, are orthogonal under the Hilbert--Schmidt product. Indeed, for $q\ne q'$ in that set, ${\rm Tr}\bigl(\sigma(0,q)^{\dagger}\sigma(0,q')\bigr)={\rm Tr}\bigl(\sigma(0,q^{-1}q')\bigr)=0$, since $q^{-1}q'\notin\{\pm1\}$. The $\sigma(0,q)$ are therefore a basis for $M_2(\mathbb{C})$, and we may expand $\widehat\Psi=\sum_q m_q\,\sigma(q)$. By \eqref{eq:pairing} conjugation acts diagonally in this basis, $\sigma(r)^{-1}\sigma(q)\sigma(r)=\chi_q(r)\sigma(q)$, and therefore \eqref{conjSubst} becomes
\begin{equation}\label{eq:separate3}
m_q\bigl[\chi_q(r)-\chi(r)\bigr]=0~, \ \ \ \forall\ r\in Q_8~,
\end{equation}
the exact analogue of \eqref{eq:separate} and \eqref{eq:separate2}. The four characters, $\chi_q$, exhaust the one-dimensional irreps, so $m_q=0$ unless $\chi_q=\chi$ (i.e., unless $q=q_\chi$). Therefore, $\widehat\Psi\in\mathbb{C}\,\sigma(q_\chi)$, and there is one free constant (consistent with $N^{\chi}_{\sigma\sigma}=1$). Finally, after normalizing $\phi_\chi$ one obtains
\begin{equation}\label{phifinal}
\phi_{\chi}(x_a\otimes x_b)=2^{-1/2}\left(\Omega\,\sigma(q_\chi)\right)_{ab}~.
\end{equation}

Next, consider $\eta_{\chi}$. Let $e_\chi\in V_\chi$. Then, a priori, we have
\begin{equation}
\eta_\chi(e_\chi\otimes\psi_q)=\sum_{q''}m_{q'',q}\,\psi_{q''}~,\ \ \  q,q''\in Q_8~,
\end{equation}
which has sixty-four unknowns. Impose \eqref{Gequiv} for $g=(v,1)$ and use $A\subseteq\ker\chi$
\begin{equation}
\eta_\chi\bigl(\chi(v)e_\chi\otimes\rho(v)\psi_q\bigr)=\hat\chi_{w_q}(v)\sum_{q''}m_{q'',q}\,\psi_{q''}~,\ \ \ \rho(v)\,\eta_\chi(e_\chi\otimes\psi_q)=\sum_{q''}m_{q'',q}\,\hat\chi_{w_{q''}}(v)\,\psi_{q''}~.
\end{equation}
Comparing coefficients of $\psi_{q''}$ gives $m_{q'',q}\bigl[\hat\chi_{w_q}(v)-\hat\chi_{w_{q''}}(v)\bigr]=0$ for all $v\in A$. By \eqref{eq:restriction} the $A$-weights are pairwise distinct, so $m_{q'',q}=0$ for $q''\ne q$. Writing $m_q:=m_{q,q}$, the sixty-four unknowns drop to eight, and we find
\begin{equation}
\eta_\chi(e_\chi\otimes\psi_q)=m_q\,\psi_q~.
\end{equation}
Next impose equivariance for $g=(0,r)$
\begin{equation}
\eta_\chi\bigl(\chi(r)e_\chi\otimes\rho(r)\psi_q\bigr)=\chi(r)\,m_{rq}\,\psi_{rq}~,\ \ \ \rho(r)\,\eta_\chi(e_\chi\otimes\psi_q)=m_q\,\psi_{rq}~,
\end{equation}
and so $\chi(r)\,m_{rq}=m_q$. Therefore, $m_{rq}=\chi(r)^{-1}m_q=\chi(r)\,m_q$, with the last equality holding because $\chi$ takes values $\pm1$. Setting $q=1$ gives $m_r=\chi(r)m_1$, which satisfies the relation for all $q$, since $\chi$ is a homomorphism. One free constant remains, which is consistent with $N_{\rho}^{\chi\rho}=1$. Normalizing then gives
\begin{equation}\label{etafinal}
\eta_\chi(e_\chi\otimes\psi_q)=\chi(q)\psi_q~.
\end{equation}

To compute $F^{\sigma\sigma\rho}_{\rho}$ we need to compare the two paths in \eqref{LR2} by substituting in \eqref{phifinal}, \eqref{etafinal}, and \eqref{varphifinal}
\begin{eqnarray}\label{eq:LRssrcomp}
L_{\chi}\bigl(x_a\otimes x_b\otimes\psi_q\bigr)
&=&2^{-1/2}\,\bigl[\Omega\sigma(q_\chi)\bigr]_{ab}\,\chi(q)\,\psi_q~,\cr
R_{(\rho,\mu,\nu)}\bigl(x_a\otimes x_b\otimes\psi_q\bigr)
&=&[\sigma(q)^{-1}]_{\mu b}\,[\sigma(q)^{-1}]_{\nu a}\,\psi_q~.
\end{eqnarray}
Both are multiples of $\psi_q$, so \eqref{eq:Ftrace2} is a sum of products of coefficients. Writing $K_\chi:=\Omega\sigma(q_\chi)$, the $(a,b)$-sum in \eqref{eq:Ftrace2} is
\begin{equation}
\sum_{a,b}[\sigma(q)^{-1}]_{\nu a}\,\overline{[K_\chi]_{ab}}\,[\sigma(q)^{-1}]_{\mu b}=\Bigl[\sigma(q)^{-1}\,\overline{K_\chi}\,({\sigma(q)^{-1}})^{\,T}\Bigr]_{\nu\mu}=\bigl[\sigma(q)^{\dagger}\overline{K_\chi}\,\overline{\sigma(q)}\bigr]_{\nu\mu}=\overline{\chi(q)\,[K_\chi]_{\nu\mu}}~,
\end{equation}
using $\sigma(q)^{-1}=\sigma(q)^{\dagger}$ in the middle step and, in the last step, the complex conjugate of the equivariance condition, $\sigma(r)^{T}K_\chi\sigma(r)=\chi(r)K_\chi$, that fixed $\phi_\chi$. Hence
\begin{equation}\label{eq:Fssr}
\bigl[F^{\sigma\sigma\rho}_{\rho}\bigr]_{(\chi)(\rho,\mu,\nu)}
=\frac{2^{-1/2}}{8}\sum_{q\in Q_8}\overline{\chi(q)^{2}\,[K_\chi]_{\nu\mu}}
=2^{-1/2}\overline{\bigl[\Omega\sigma(q_\chi)\bigr]_{\nu\mu}}~.
\end{equation}
We can then write the $F$-symbol as a non-real matrix via
\begin{equation}
F^{\sigma\sigma\rho}_{\rho}=\frac{1}{\sqrt{2}}\begin{pmatrix}0&1&-1&0\\0&\mathrm i&\mathrm i&0\\-1&0&0&-1\\ -\mathrm i&0&0&\mathrm i\end{pmatrix}~,
\end{equation}
where rows and columns are ordered $\chi=\chi_1,\chi_i,\chi_j,\chi_k$ and $(\nu,\mu)=(1,1),(1,2),(2,1),(2,2)$ respectively so that the $\chi$-th row is the $2\times 2$ matrix $\overline{\Omega\sigma(q_\chi)}$.

Now let us understand the constraints on ${\rm Stab}(\CF)$. The four gauge factors acting on $F^{\sigma\sigma\rho}_\rho$ are
\begin{equation}\label{gaugeFact}
\Gamma^{\sigma\sigma}_{\chi}~,\ \ \ \Gamma^{\chi\rho}_{\rho}~,\ \ \ \Gamma^{\sigma\rho}_{\rho}=P~,\ \ \ \Gamma^{\sigma\rho}_{\rho}=P~.
\end{equation}
Setting $z_\chi:=\Gamma^{\sigma\sigma}_{\chi}\Gamma^{\chi\rho}_{\rho}$, the relevant gauge transformation becomes
\begin{equation}\label{FssrrT}
\bigl[F^{\sigma\sigma\rho}_{\rho}\bigr]^{\Gamma}_{(\chi)(\rho,\mu,\nu)}
=z_\chi\sum_{\mu',\nu'}
\bigl[F^{\sigma\sigma\rho}_{\rho}\bigr]_{(\chi)(\rho,\mu',\nu')}
\bigl[P^{-1}\bigr]_{\mu'\mu}\bigl[P^{-1}\bigr]_{\nu'\nu}
=z_\chi\,2^{-1/2}\Bigl[\bigl(P^{-1}\bigr)^{T}\overline{K_\chi} P^{-1}\Bigr]_{\nu\mu}~.
\end{equation}
Unlike the first case of $F^{\rho\sigma\rho}_\chi$, the two scalars do not cancel: the vertices $V_{\chi}^{\sigma\sigma}$ and $V_{\rho}^{\chi\rho}$ are different spaces, and their product, $z_\chi$, appears. Stability therefore amounts to the requirement that
\begin{equation}\label{eq:pauliconstraint}
P^{T}\bigl(\Omega\overline{\sigma(q_\chi)}\bigr)P=z_\chi\,\Omega\overline{\sigma(q_\chi)}~, \ \ \ \forall\ \chi~.
\end{equation}
Now, using the fact that $P^{T}\Omega=\det(P)\,\Omega P^{-1}$ the stability condition becomes
\begin{equation}\label{conjrel}
P^{-1}\overline{\sigma(q_\chi)}P=s_{q_\chi}\, \overline{\sigma(q_\chi)}~,\qquad s_{q_\chi}:=z_{\chi}/\det P~.
\end{equation}
Taking the determinant of both sides gives $s_{q_\chi}^2=1$, so $s_{q_\chi}=\pm1$. Applying \eqref{conjrel} to $\sigma(k)=\sigma(i)\sigma(j)$ gives $s_k=s_is_j$. Therefore, 
\begin{equation}\label{signP}
(s_i,s_j,s_k)\in\left\{(+,+,+),(+,-,-),(-,+,-),(-,-,+)\right\}~.
\end{equation}
Therefore, \eqref{conjrel} is the action of the following inner automorphism of $Q_8$
\begin{equation}\label{inner1}
P^{-1}\overline{\sigma(q_\chi)}P=\overline{\sigma(\kappa_p(q_\chi))}=\sigma(p)^{-1}\overline{\sigma(q_\chi)}\sigma(p)~,\ \ \ q_\chi\in Q_8~,
\end{equation}
where $\kappa_p(q)=p^{-1}qp$, and the choices of $p\in\left\{1,i,j,k\right\}$ are in one-to-one correspondence with the choices in \eqref{signP}. The second equality in \eqref{inner1} is \eqref{eq:conjsame} together with the fact that $\overline{\sigma(\,\cdot\,)}$ is a representation. Equation \eqref{inner1} then says says that $P\sigma(p)^{-1}$ commutes with $\overline{\sigma(q)}$ for every $q\in\left\{1,i,j,k\right\}$. By \eqref{eq:sigmareal} these differ from the $\sigma(q)$ only by signs, so they span $M_2(\mathbb{C})$, and $P\sigma(p)^{-1}$ is therefore a scalar. As a result,
\begin{equation}\label{Psigma}
P=z\sigma(p)~,
\end{equation}
for some $z\in U(1)$. If we can show $p=1$, then we are done ($p=-1$ also works, because it amounts to sending $z\to-z$).

Let us next consider $F^{\sigma\rho\rho}_\rho$.  The intertwiners we construct live in the fusion spaces $V_{\sigma\rho}^\rho$ and $V_{\rho\rho}^\rho$. However, $F^{\sigma\rho\rho}_\rho$ is an overlap involving the corresponding adjoint splitting spaces.

We have already constrained $\varphi_t\in V_{\sigma\rho}^\rho$ in \eqref{varphifinal}. It only remains to constrain the intertwiner
\begin{equation}
\nu_u\in V_{\rho\rho}^\rho~.
\end{equation}
A priori we have
\begin{equation}
\nu_u(\psi_q\otimes\psi_{q'})=\sum_{q''}\nu^{q''}_{q,q'}\,\psi_{q''}~,\ \ \ q,q',q''\in Q_8~.
\end{equation}
Now impose  \eqref{Gequiv} for $g=(v,1)$, using \eqref{eq:rho}
\begin{equation}
\nu_u\bigl(\rho(v,1)\psi_q\otimes\rho(v,1)\psi_{q'}\bigr)=\hat\chi_{w_q+w_{q'}}(v)\sum_{q''}\nu^{q''}_{q,q'}\,\psi_{q''}~,\ \ \  \rho(v)\,\nu_u(\psi_q\otimes\psi_{q'})=\sum_{q''}\nu^{q''}_{q,q'}\,\hat\chi_{w_{q''}}(v)\,\psi_{q''}~.
\end{equation}
Comparing coefficients of $\psi_{q''}$ gives the constraint
\begin{equation}\label{eq:separate4}
\nu^{q''}_{q,q'}\bigl[\hat\chi_{w_q+w_{q'}}(v)-\hat\chi_{w_{q''}}(v)\bigr]=0~, \ \ \ \forall v\in A~.
\end{equation}
Therefore, $\nu^{q''}_{q,q'}=0$ unless
\begin{equation}\label{eq:selection}
w_q+w_{q'}=w_{q''}~.
\end{equation}
This selection rule has two consequences. First, the pair $(q,q')$ must satisfy $q'\ne-q$, since $w_q+w_{-q}=0$ is not an $A$-weight. Writing $q'=qu$ with $u\in Q_8\setminus\{-1\}$, the surviving pairs are parametrised by $q$ and $u$. Second, $q''$ is then determined. Indeed, recalling \eqref{Aweights} and the definition in \eqref{eq:mu}, we see that
\begin{equation}
w_q+w_{qu}=(q^{-1})^{T}\bigl(w_1+w_u\bigr)
=(q^{-1})^{T}w_{\mathsf{m}(u)}=w_{q\,\mathsf{m}(u)}~, \ \ \ q''=q\,\mathsf{m}(u)~,
\end{equation}
and so
\begin{equation}
\nu_u(\psi_q\otimes\psi_{qu})=c_{q,u}\,\psi_{q\,\mathsf{m}(u)}~,\ \ \ \nu_u(\psi_q\otimes\psi_{-q})=0~.
\end{equation}

Next impose equivariance with respect to $g=(0,r)$
\begin{equation}
\nu_u\bigl(\rho(0,r)\psi_q\otimes\rho(0,r)\psi_{qu}\bigr)=\nu_u\bigl(\psi_{rq}\otimes\psi_{(rq)u}\bigr)=c_{rq,u}\,\psi_{rq\,\mathsf{m}(u)}~, \ \ \ \rho(r)\,\nu_u(\psi_q\otimes\psi_{qu})=c_{q,u}\,\psi_{rq\,\mathsf{m}(u)}~.
\end{equation}
Hence
\begin{equation}\label{eq:constancy}
c_{rq,u}=c_{q,u}=c_u~, \ \ \ \forall r\in Q_8~,
\end{equation}
is independent of $q$.

Seven free $c_u$ remain, one for each $u\in Q_8\setminus\{-1\}$, and so $N_{\rho}^{\rho\rho}=7$ as expected. Taking $c_u=1$ and
$c_{u'}=0$ for $u'\ne u$ gives
\begin{equation}\label{nufinal}
\nu_u(\psi_q\otimes\psi_{q'})=\delta_{q',qu}\,\psi_{q\,\mathsf{m}(u)}~,
\end{equation}
which is already properly normalized.

Given this discussion, we can think of $F^{\sigma\rho\rho}_{\rho}$ as the overlap of the two splitting paths $\CL_{(\rho,\alpha,u)},\CR_{(\rho,u',\nu)}\in{\rm Hom}(V_\rho,V_\sigma\otimes V_\rho\otimes V_\rho)$
\begin{eqnarray}\label{Fsrrs}
\left[F^{\sigma\rho\rho}_{\rho}\right]_{(\rho,\alpha,u)(\rho,u',\nu)}&=&\langle \CR_{(\rho,u',\nu)}|\CL_{(\rho,\alpha,u)}\rangle=\langle L_{(\rho,\alpha,u)}|R_{(\rho,u',\nu)}\rangle\cr&=&{1\over \dim V_\rho}{\rm Tr}_{V_\sigma\otimes V_\rho\otimes V_\rho}\left(L_{(\rho,\alpha,u)}^\dagger R_{(\rho,u',\nu)}\right)~,
\end{eqnarray}
where $\CL_{(\rho,\alpha,u)}=L_{(\rho,\alpha,u)}^\dagger$, $\CR_{(\rho,u',\nu)}=R_{(\rho,u',\nu)}^\dagger$, and
\begin{equation}
L_{(\rho,\alpha,u)}\bigl(x_a\otimes\psi_q\otimes\psi_{q'}\bigr)=\nu_{u}\bigl(\varphi_\alpha(x_a\otimes\psi_q)\otimes\psi_{q'}\bigr)~, \ \ R_{(\rho,u',\nu)}\bigl(x_a\otimes\psi_q\otimes\psi_{q'}\bigr)=\varphi_\nu\bigl(x_a\otimes\nu_{u'}(\psi_q\otimes\psi_{q'})\bigr)~.
\end{equation}
Using our results in \eqref{varphifinal} and \eqref{nufinal}, these quantities become
\begin{eqnarray}\label{eq:LRsrrcomp}
L_{(\rho,\alpha,u)}\bigl(x_a\otimes\psi_q\otimes\psi_{q'}\bigr)
&=&[\sigma(q)^{-1}]_{\alpha a}\;\delta_{q',qu}\;\psi_{q\mathsf{m}(u)}~,\cr
R_{(\rho,u',\nu)}\bigl(x_a\otimes\psi_q\otimes\psi_{q'}\bigr)
&=&\delta_{q',qu'}\;[\sigma(q\mathsf{m}(u'))^{-1}]_{\nu a}\;\psi_{q\mathsf{m}(u')}~.
\end{eqnarray}
Both are multiples of $\psi_{q\mathsf{m}(u)}$ on the common support, so \eqref{Fsrrs} is again a sum of products of coefficients. The two Kronecker deltas give $\sum_{q'}\delta_{q',qu}\delta_{q',qu'}=\delta_{uu'}$, and, on $u=u'$, the $a$-sum is
\begin{eqnarray}
\sum_{a}[\sigma(q\mathsf{m}(u))^{-1}]_{\nu a}\;\overline{[\sigma(q)^{-1}]_{\alpha a}}&=&\overline{\sum_{a}\overline{[\sigma(q\mathsf{m}(u))^{-1}]_{\nu a}}\;[\sigma(q)^{-1}]_{\alpha a}}=\overline{\bigl[\sigma(q)^{-1}\sigma(q\mathsf{m}(u))\bigr]_{\alpha\nu}}\cr&=&\overline{\bigl[\sigma(\mathsf{m}(u))\bigr]_{\alpha\nu}}~.
\end{eqnarray}
In arriving at this result, we pulled the complex conjugate out of the sum, used $\overline{\sigma(q)^{-1}}^{\,T}=\sigma(q)$ in the first step, and $\sigma(q\mathsf{m}(u))=\sigma(q)\sigma(\mathsf{m}(u))$ in the second step. The $q$-sum then contributes a factor eight, cancelling the $1/\dim V_\rho$, and so
\begin{equation}\label{eq:Fsrr}
\bigl[F^{\sigma\rho\rho}_{\rho}\bigr]_{(\rho,\alpha,u)(\rho,u',\nu)}=\delta_{uu'}\overline{\bigl[\sigma(\mathsf{m}(u))\bigr]_{\alpha\nu}}~.
\end{equation}

We can give a similar derivation for $F^{\rho\sigma\rho}_{\rho}$ using the intertwiners we have already found. The two paths in this case are daggers of the following
\begin{equation}
L_{(\rho,\alpha,u)}\bigl(\psi_q\otimes x_a\otimes\psi_{q'}\bigr)
=\nu_{u}\bigl(\varphi'_\alpha(\psi_q\otimes x_a)\otimes\psi_{q'}\bigr)~, \ \
R_{(\rho,\mu,u')}\bigl(\psi_q\otimes x_a\otimes\psi_{q'}\bigr)
=\nu_{u'}\bigl(\psi_q\otimes\varphi_\mu(x_a\otimes\psi_{q'})\bigr)~,
\end{equation}
and, by \eqref{varphifinal}, \eqref{varphiPfinal} and \eqref{nufinal},
\begin{eqnarray}\label{eq:LRrsrcomp}
L_{(\rho,\alpha,u)}\bigl(\psi_q\otimes x_a\otimes\psi_{q'}\bigr)
&=&[\sigma(q)^{-1}]_{\alpha a}\;\delta_{q',qu}\;\psi_{q\mathsf{m}(u)}~,\cr
R_{(\rho,\mu,u')}\bigl(\psi_q\otimes x_a\otimes\psi_{q'}\bigr)
&=&\delta_{q',qu'}\;[\sigma(q')^{-1}]_{\mu a}\;\psi_{q\mathsf{m}(u')}~.
\end{eqnarray}
Again the two Kronecker deltas give $\sum_{q'}\delta_{q',qu}\delta_{q',qu'}=\delta_{uu'}$. Therefore, the $a$-sum is
\begin{equation}
\sum_{a}[\sigma(qu)^{-1}]_{\mu a}\;\overline{[\sigma(q)^{-1}]_{\alpha a}}
=\overline{\sum_{a}\overline{[\sigma(qu)^{-1}]_{\mu a}}\;[\sigma(q)^{-1}]_{\alpha a}}
=\overline{\bigl[\sigma(q)^{-1}\sigma(qu)\bigr]_{\alpha\mu}}
=\overline{\bigl[\sigma(u)\bigr]_{\alpha\mu}}~.
\end{equation}
We have pulled the complex conjugation outside the sum, used $\overline{\sigma(g)^{-1}}^{\,T}=\sigma(g)$ in the first step, and $\sigma(qu)=\sigma(q)\sigma(u)$ in the second. Again, the $q$-sum contributes eight, cancelling the $1/\dim V_\rho$, and therefore
\begin{equation}\label{eq:Frsr}
\bigl[F^{\rho\sigma\rho}_{\rho}\bigr]_{(\rho,\alpha,u)(\rho,\mu,u')}=\delta_{uu'}\overline{\bigl[\sigma(u)\bigr]_{\alpha\mu}}~.
\end{equation}

We now conclude with the final constraints on ${\rm Stab}(\CF)$. First consider constraints arising from $F^{\sigma\rho\rho}_\rho$. The four gauge factors are $P=\Gamma^{\sigma\rho}_\rho$ on $\alpha$ and $\nu$ and 
\begin{equation}
\Xi:=\Gamma^{\rho\rho}_\rho~,
\end{equation}
on $u$ and $u'$. Substituting \eqref{eq:Fsrr} into the transformation, performing the sums forced by the two Kronecker deltas, and multiplying on the right by $P$ and $\Xi$, $F^{\Gamma}=F$ becomes 
\begin{equation}
\Xi[u,u']\,\bigl[P\,\overline{\sigma(\mathsf{m}(u'))}\bigr]_{\alpha\nu}=\Xi[u,u']\,\bigl[\overline{\sigma(\mathsf{m}(u))}\,P\bigr]_{\alpha\nu}~,\ \ \ \forall\alpha,\nu,u,u'~,
\end{equation}
so $\Xi[u,u']\ne0$ forces $P\overline{\sigma(\mathsf{m}(u'))}P^{-1}=\overline{\sigma(\mathsf{m}(u))}$. Using \eqref{Psigma}, this fact implies
\begin{equation}\label{FsrrrC}
\Xi[u,u']\ne0\Longrightarrow \sigma(p)\overline{\sigma(\mathsf{m}(u'))}\sigma(p)^{-1}=\overline{\sigma(\mathsf{m}(u))}~. 
\end{equation}

Finally, consider the constraints arising from $F^{\rho\sigma\rho}_\rho$. The four gauge factors are now $P'=P$ on $\alpha$ and $\mu$ and $\Xi$ on $u'$ and $u$. By logic similar to the above, we conclude that 
\begin{equation}\label{FrsrrC}
\Xi[u,u']\ne0\Longrightarrow \sigma(p)\overline{\sigma(u')}\sigma(p)^{-1}=\overline{\sigma(u)}~. 
\end{equation}

We now wish to marshal our constraints in \eqref{PprP}, \eqref{Psigma}, \eqref{FsrrrC}, and \eqref{FrsrrC} to prove that $P=z\sigma(p)$ is a scalar. In other words, we wish to show that $p\in\left\{\pm1\right\}$. By Lemma \ref{StabLem}, this conclusion is enough to show that the complex $F$-symbols we have found in $\CF$ are inherently complex (i.e., $\CF$ cannot be made simultaneously real).

To that end, suppose that $p\ne\pm1$. Since $\Xi$ is invertible (it is a unitary gauge transformation), it must have a non-zero entry in row $p$. Now, by \eqref{FrsrrC}, that entry must be in a column with $u'$ satisfying $\sigma(p)\overline{\sigma(u')}\sigma(p)^{-1}=\overline{\sigma(p)}$. By \eqref{eq:conjsame} this reads $\overline{\sigma(p\,u'\,p^{-1})}=\overline{\sigma(p)}$, and $\overline{\sigma(\,\cdot\,)}$ is faithful, so $u'=p$, and hence $\Xi[p,p]\ne0$. Therefore, \eqref{FsrrrC} says that $\sigma(p)\overline{\sigma(\mathsf{m}(p))}\sigma(p)^{-1}=\overline{\sigma(\mathsf{m}(p))}$, which by the same manipulation requires $\mathsf{m}(p)$ to lie in the centralizer $C_{Q_8}(p)=\left\{1,-1,p,-p\right\}$. From \eqref{eq:mutable}, we see that this does not happen for $p\ne\pm1$. Therefore, we see that any $\Gamma\in{\rm Stab}(\CF)$ has
\begin{equation}
\Gamma^{\sigma\rho}_\rho=c\cdot \mathds{1}~.
\end{equation}
As a result, by Lemma \ref{StabLem} there is no gauge in which $\CF$ can be made simultaneously real. We conclude that ${\rm Rep}(\mathbb{F}_3^2\rtimes Q_8)$ has inherently complex $F$-symbols.

\section{Conclusion and outlook}

We have identified a general mechanism underlying the existence of gauges with real $F$-symbols in unitary ribbon categories. A braided charge-conjugation symmetry relates the $F$-symbols to their complex conjugates by local gauge transformations on the trivalent fusion spaces. Then, the flatness condition on the charge conjugation symmetry (or, equivalently, the corresponding condition on the twisted Frobenius-Schur data) ensures that these local transformations can be chosen symmetric and hence trivialised by Takagi factorisation. Our main result is that a unitary ribbon category with a flat charge-conjugation symmetry admits a gauge in which all $F$-symbols are real and, simultaneously, all $R$-matrices are symmetric. In special cases, such as multiplicity-free categories, the conditions simplify further, and a braided $\mathbb{Z}_2$ charge conjugation symmetry is sufficient for the existence of a real $F$-symbol gauge. From an experimental perspective, this simultaneously reduces the complexity of reconstructing the fusion and braiding data: the $F$-symbols require only real amplitudes, while symmetric $R$-matrices contain fewer independent matrix elements.

This work suggests various interesting future directions:
\begin{itemize}
\item We have provided sufficient conditions for a gauge with real $F$-symbols. It would be interesting to establish necessary and sufficient conditions.
It seems likely that our sufficient conditions are stronger than necessary, because they also produce symmetry of the $R$-symbol. 
\item In all the Chern-Simons theories we encountered, there is a gauge with real $F$-symbols. Indeed, this statement holds in all self-dual Chern-Simons theories with compact and simply connected gauge groups (see Sec. \ref{ChernSimons}) as well as all the $K$-matrix theories (these are all Abelian; see Sec. \ref{implications}). On the other hand, we have seen, in addition to the inherently complex $F$-symbols discussed in \cite{Paper1}, the inherently complex case of the ${\rm Rep}(\mathbb{F}_3^2\rtimes Q_8)$ $F$-symbols in Sec. \ref{ssec:FS2ObstructionExample}. It would be interesting to understand if a real $F$-symbol gauge is a general property of Chern-Simons theories. If so, inherent complexity of $F$-symbols would be an invariant way to argue that a particular TQFT does not admit a Chern-Simons description (such questions are related to the existence of \lq\lq exotic" symmetries \cite{Hong:2007ty}).
\item Clearly, a gauge with real $F$-symbols is not invariant under 0-form gauging. More mathematically, it is not an invariant of the Witt group. To see this, note that the trivial theory, ${\rm Vec}$, has real $F$-symbols ($F^{111}_1=1$), but gauging a $G=\mathbb{F}_3^2\rtimes Q_8$ symmetry leads to a $G$ discrete gauge theory with complex $F$-symbols (by the results of Sec. \ref{ssec:FS2ObstructionExample}, it has a ${\rm Rep}(G)$ subcategory with inherently complex $F$-symbols). It would be interesting to understand if there is a refinement of the Witt group that has invariants that capture inherent complexity of $F$-symbols.
\item In Sec. \ref{ssec:FS2ObstructionExample} and App. \ref{KMAppendix}, we saw the role the Kawanaka-Matsuyama indicator played: it is the twisted FS index for ${\rm Rep}(G)$-type categories. Flatness then leads to a gauge with real $F$-symbols while non-flatness can lead to inherently complex $F$-symbols. It would be interesting to understand a generalized Kawanaka-Matsuyama indicator that applies beyond these classes of categories.
\item More generally, it would be interesting to understand if there are universal gauge invariants that can tell us if a theory has a gauge with real $F$-symbols. One possibility might be the Whitehead link (see \cite{Bonderson:2018ryx} for an appearance of this link in a potentially related context).
\end{itemize}

\ack{We thank L.~Byles, J.F.~Martins, D.~Nikshych, B.~Rayhaun, E.~Rowell, J.~Slingerland, K.~Walker, Z.~Wang, and X.~Yang for many helpful discussions. We also thank C.~Galindo for helpful discussions about the relation between \cite{Cesar} and the present work. We are grateful to the Isaac Newton Institute for Mathematical Sciences at Cambridge University for support and hospitality during the programme \lq\lq Quantum Field Theory With Boundaries, Impurities, and Defects,” where this project was initiated. We were partly supported by EPSRC grant {no.}~EP/Z000580/1. {M.B.}'s work was partly supported by STFC under the grant “Amplitudes, Strings and Duality” and by a Simons Fellowship. {P.H.}~was supported by EPSRC Programme Grant {no.}~EP/W007509/1. {J.K.P.} was partially supported by EPSRC grants nos.~EP/W524372/1 and UKRI1337:Anyons24. No new data were generated or analysed in this study.}

\newpage

\begin{appendix}
\section{The Kawanaka-Matsuyama formula and \texorpdfstring{$\sigma$}{sigma}}\label{KMAppendix}
In this appendix, we restrict to the case $\CC={\rm Rep}(G)$ in order to better understand \eqref{eq:KM} and its relation to $\sigma$. To that end, let us choose representatives, $(V_\chi,\rho_\chi)$, with orthonormal bases.\footnote{As an aside on notation, we will often name objects in ${\rm Rep}(G)$ by their underlying spaces (so we will use $V$ as shorthand for $(V,\rho)$).} In particular, for $\tau$ a class-inverting involution, we claim that
\begin{equation}\label{sigmaApp}
\sigma_\chi=p_\chi\psi_{\chi^{\vee}}^{-1}=\nu_\tau(\chi)~,
\end{equation}
where $\sigma_\chi$ is the data appearing in Corollary \ref{cor:R2Symbol} when the Postnikov class is trivial, and $\nu_\tau$ is given by the Kawanaka-Matsuyama (KM) formula \cite{kawanaka1990twisted}
\begin{equation}\label{eq:KMApp}
\nu_\tau(\chi)\;=\;\frac{1}{|G|}\sum_{g\in G}\chi\big(g\,\tau(g)\big)~.
\end{equation}
We take our categorical charge conjugation to be $J=\tau^*$ and so $J^2(V_\chi,\rho_\chi)=(V_\chi,\rho_\chi\circ\tau\circ\tau)=(V_\chi,\rho_\chi)$, which means that $J^2={\rm Id}$.  This fact implies that we can use a symmetry gauge transformation to set $\eta={\rm Id}$. The remaining factor $\beta_\chi$ in $\sigma_\chi=\beta_\chi p_\chi\psi^{-1}_{\chi^{\vee}}$ records the components of $\eta$ after transport along the identifications $J(\chi)\cong\chi^{\vee}$, which we will describe in greater detail below. Normalising those identifications sets $\beta_\chi=1$, which is why \eqref{sigmaApp} contains only two factors (see the discussion in Sec. \ref{ssec:mainTheorem} for further details).

To prove \eqref{sigmaApp}, we first do some groundwork and fix our right duality functor, $\pi$, as follows
\begin{equation}
\pi(V_\chi)=V_\chi^\vee:=V_\chi^*~,\ \ \ \rho_\chi^\vee(g)=\overline{\rho_\chi(g)}~,
\end{equation}
in the dual basis. The duality data is $e_{V_\chi}(f\otimes v)=f(v)$ and $c_{V_\chi}(1)=\sum_iv_i\otimes v^i$, which satisfy the zig-zags and are balanced.\footnote{Explicitly, the zig-zag identities are the two relations
\begin{equation}\label{zigzag}
(\mathrm{id}_V\otimes e_V)\circ(c_V\otimes\mathrm{id}_V)=\mathrm{id}_V~, \ \ \  (e_V\otimes\mathrm{id}_{V^\vee})\circ(\mathrm{id}_{V^\vee}\otimes c_V)=\mathrm{id}_{V^\vee}~.
\end{equation}
For the choice above, they hold on the nose: writing $c_V(1)=\sum_i v_i\otimes v^i$ and $e_V(f\otimes v)=f(v)$, the first sends $v\mapsto\sum_i v_i\otimes v^i\otimes v\mapsto\sum_i v_i\,v^i(v)=v$, and analogously for the second relation. The data are called \lq\lq balanced" when the two ways of closing a loop agree with the quantum dimension,
\begin{equation}
  e_V\circ e_V^{\dagger}=d_V=c_V^{\dagger}\circ c_V~.
\end{equation}
Here $d_V=\dim V$, which fixes the relative scale of $e_V$ and $c_V$, since the zig-zags alone only determine the pair up to $(e_V,c_V)\mapsto(\alpha e_V,\alpha^{-1}c_V)$.}

The following useful facts hold for every object of ${\rm Rep}(G)$ (in particular, they do not require a choice of identification $J(\chi)\cong\chi^{\vee}$): {\bf (a)} $J(V^\vee,\rho^\vee)=J(V,\rho)^\vee$, {\bf(b)} the duality comparison of Lemma \ref{lem:JPiSquare}, $\psi_V=\big((J(e_V)\circ\theta^{\dagger}_{\pi V,V})\otimes\mathrm{id}\big)\circ(\mathrm{id}\otimes c_{J(V)})=\mathrm{id}_{V^\vee}$, and {\bf(c)} $p_V=1$. 

Let us first establish {\bf(a)}. To that end, both $J(V^\vee,\rho^\vee)$ and $J(V,\rho)^\vee$ have underlying space $V^{*}$, and so we need only compare the actions
\begin{equation}
J(V^\vee,\rho^\vee)=\big(V^{*},\,\rho^{\vee}\!\circ\tau\big)~,\ \ \ J(V,\rho)^{\vee}=\big(V^{*},\,(\rho\circ\tau)^{\vee}\big)~.
\end{equation}
For every $g\in G$ and $f\in V^{*}$,
\begin{equation}
\big(\rho^{\vee}\!\circ\tau\big)(g)\,f\;=\;f\circ\rho\big(\tau(g)^{-1}\big)
\;=\;f\circ\rho\big(\tau(g^{-1})\big)\;=\;\big((\rho\circ\tau)^{\vee}\big)(g)\,f~,
\end{equation}
using only the fact that $\tau$ is a homomorphism. Hence $J(V^\vee,\rho^\vee)=J(V,\rho)^{\vee}$ as objects.

Next let us establish {\bf (b)}. To that end, by {\bf(a)} the components of the 2-cell of Lemma \ref{lem:JPiSquare}, $\psi_V:J(\pi V)\to\pi(JV)$, has source and target the same object, so a zig-zag can close on the nose. Three substitutions reduce the displayed composite. Strict monoidality of $\tau^{*}$ gives
\begin{equation}
  \theta_{\pi V,V}=\id\ \Longrightarrow\ \theta^{\dagger}_{\pi V,V}=\id~.
\end{equation}
Next, $J$ is the identity on underlying spaces and linear maps. Moreover, for all $g\in G$
\begin{equation}
e_V\circ\big(\rho^{\vee}(\tau(g))\otimes\rho(\tau(g))\big)=e_V~,
\end{equation}
because $e_V$ is equivariant at every group element and, in particular, at $\tau(g)$. So $J(e_V)=e_V$ as a linear map, and likewise $c_{J(V)}=c_V$. Substituting all three,
\[
  \psi_V
  =\big((J(e_V)\circ\theta^{\dagger}_{\pi V,V})\otimes\id\big)\circ\big(\id\otimes c_{J(V)}\big)
  =\big(e_V\otimes\id_{V^{\vee}}\big)\circ\big(\id_{V^{\vee}}\otimes c_V\big)
  =\id_{V^{\vee}}~,
\]
the last step being the second zig-zag identity in \eqref{zigzag}. So the square commutes strictly.

Finally, let us establish {\bf(c)}. Recall that the canonical pivotal $\varphi:\Id_{\CC}\Rightarrow(\cdot)^{\vee\vee}$ is our $\mu$ read backwards, $\varphi_V=\mu_V^{-1}$, and is given by
\begin{equation}\label{eq:pivotalApp}
\varphi_V=\big(c_V^{\dagger}\otimes\mathrm{id}_{V^{\vee\vee}}\big)\circ
\big(\mathrm{id}_V\otimes c_{V^{\vee}}\big)~.
\end{equation}
We evaluate this on the duality data fixed above. Write $\{v_i\}$ for the chosen orthonormal basis of $V$, $\{v^i\}$ for the dual basis of $V^{*}$, and $\{v_i^{**}\}$ for the double-dual basis of $V^{**}$, so that $v_i^{**}(v^j)=\delta_i^j$. Then
\begin{equation}
c_V(1)=\sum_i v_i\otimes v^i~,\ \ \ c_{V^{\vee}}(1)=\sum_j v^j\otimes v_j^{**}~,
\end{equation}
and, since $\{v_i\}$ is orthonormal, the adjoint $c_V^{\dagger}:V\otimes V^{*}\to\mathbb{C}$ acts by $c_V^{\dagger}(v_k\otimes v^l)=\delta_{kl}$. Feeding $v_k$ into \eqref{eq:pivotalApp},
\begin{equation}
v_k\;\longmapsto\;\sum_j v_k\otimes v^j\otimes v_j^{**}
\;\longmapsto\;\sum_j c_V^{\dagger}\big(v_k\otimes v^j\big)\,v_j^{**}
\;=\;\sum_j\delta_{kj}\,v_j^{**}\;=\;v_k^{**}~.
\end{equation}
Hence $\varphi_V$ is the canonical map $v\mapsto\big(f\mapsto f(v)\big)$, and $\mu_V=\varphi_V^{-1}$ is its inverse. But this is precisely the identification of $V^{**}$ with $V$ to which the choice $\pi(V)=V^{*}$ already commits us, so $\mu$ introduces no additional scalar, and $p_V=1$.

To make explicit contact between $\sigma_\chi$ and group-theoretical quantities, we first need to prove one more general result. To that end, let $F,\widetilde F:\CC\to\CD$ be unital monoidal functors and let $m:F\Rightarrow \widetilde F$ be a monoidal natural isomorphism (soon we will apply this to the case $F=J$). Let $\psi_x^F:F(x^\vee)\to F(x)^\vee$ be the canonical duality comparison 2-cell (the unique morphism with $e_{F(x)}\circ(\psi^F_x\otimes\mathrm{id}_{Fx})=F(e_x)\circ\theta^F_{x^\vee,x}$). Then
\begin{equation}\label{eq:transport}
\psi^{\widetilde F}_x\circ m_{x^\vee}\;=\;(m_x^\vee)^{-1}\circ\psi^F_x~.
\end{equation}
To understand this statement, note that both sides of \eqref{eq:transport} are morphisms $F(x^\vee)\to\widetilde F(x)^\vee$. Since $e_{\widetilde F(x)}$ is non-degenerate,\footnote{To understand this point, recall that the assignment $k\mapsto e_{\widetilde F(x)}\circ(k\otimes\id)$ is a bijection, $\Hom\big(F(x^\vee),\widetilde F(x)^\vee\big)\to\Hom\big(F(x^\vee)\otimes\widetilde F(x),\mathds 1\big)$, with inverse $h\mapsto(h\otimes\id)\circ(\id\otimes c_{\widetilde F(x)})$ by the zig-zag identities in \eqref{zigzag}.} to prove \eqref{eq:transport} it suffices to compare
\begin{equation}\label{compare}
e_{\widetilde F(x)}\circ\Big(\big(\psi^{\widetilde F}_x\circ m_{x^\vee}\big)\otimes\id\Big)\stackrel{?}{=}e_{\widetilde F(x)}\circ\Big(\big((m^\vee_x)^{-1}\circ\psi^F_x\big)\otimes\id\Big)~.
\end{equation}
To that end, first consider the lefthand side
\begin{equation}
e_{\widetilde F(x)}\circ\Big(\big(\psi^{\widetilde F}_x\circ m_{x^\vee}\big)\otimes\id\Big)=\Big[e_{\widetilde F(x)}\circ\big(\psi^{\widetilde F}_x\otimes\id\big)\Big]\circ\big(m_{x^\vee}\otimes\id\big)=\widetilde F(e_x)\circ\theta^{\widetilde F}_{x^\vee,x}\circ\big(m_{x^\vee}\otimes\id_{\widetilde F(x)}\big)~.
\end{equation}
Now insert $\id_{\widetilde F(x)}=m_x\circ m_x^{-1}$, and monoidality of $m$, $\theta^{\widetilde F}\circ(m_u\otimes m_v)=m_{u\otimes v}\circ\theta^F$, can be applied
\begin{equation}
\theta^{\widetilde F}\circ\big(m_{x^\vee}\otimes\id\big)=\theta^{\widetilde F}\circ\big(m_{x^\vee}\otimes m_x\big)\circ\big(\id\otimes m_x^{-1}\big)=m_{x^\vee\otimes x}\circ\theta^F_{x^\vee,x}\circ\big(\id\otimes m_x^{-1}\big)~.
\end{equation}
Finally, naturality of $m$ at the morphism $e_x:x^\vee\otimes x\to\mathds 1$ gives $\widetilde F(e_x)\circ m_{x^\vee\otimes x}=m_{\mathds 1}\circ F(e_x)$, and $m_{\mathds 1}=\id$ by unitality. Hence
\begin{equation}\label{eq:LHSvalue}
e_{\widetilde F(x)}\circ\Big(\big(\psi^{\widetilde F}_x\circ m_{x^\vee}\big)\otimes\id\Big)=F(e_x)\circ\theta^F_{x^\vee,x}\circ\big(\id\otimes m_x^{-1}\big)~.
\end{equation}
Now consider the righthand side of \eqref{compare},
\begin{eqnarray}
e_{\widetilde F(x)}\circ\Big(\big((m^\vee_x)^{-1}\circ\psi^F_x\big)\otimes\id\Big)&=&\Big[e_{\widetilde F(x)}\circ\big((m^\vee_x)^{-1}\otimes\id\big)\Big]\circ\big(\psi^F_x\otimes\id\big)\cr&=&e_{F(x)}\circ\big(\id\otimes m_x^{-1}\big)\circ\big(\psi^F_x\otimes\id\big)~,
\end{eqnarray}
where, in the final equality, we have used the defining property of the dual morphism with inverses substituted.

The $\id\otimes m_x^{-1}$ and $\psi^F_x\otimes\id$ factors act on different tensor legs, so they may be combined and then commuted
\begin{equation}
\big(\id\otimes m_x^{-1}\big)\circ\big(\psi^F_x\otimes\id\big)=\psi^F_x\otimes m_x^{-1}=\big(\psi^F_x\otimes\id\big)\circ\big(\id\otimes m_x^{-1}\big)~,
\end{equation}
and so
\begin{equation}\label{eq:RHSvalue}
e_{\widetilde F(x)}\circ\Big(\big((m^\vee_x)^{-1}\circ\psi^F_x\big)\otimes\id\Big)=F(e_x)\circ\theta^F_{x^\vee,x}\circ\big(\id\otimes m_x^{-1}\big)~,
\end{equation}
and we have established \eqref{compare} and hence \eqref{eq:transport} by the fact that $e_{\widetilde F(x)}$ is non-degenerate.

At this point, the reader may be under the misapprehension that $\sigma_\chi=1$. This is not so: facts {\bf(a)}-{\bf(c)} above concern the functor $J$, whereas the factors appearing in \eqref{sigmaApp} belong to the \lq\lq transported" functor $\widetilde J$ to be introduced momentarily. The two differ precisely by the identifications $J(\chi)\cong\chi^\vee$. To actually compute $\sigma_\chi\ne1$, we need to choose a normalization $j_\chi\in{\rm Hom}(J(V_\chi),V_\chi^\vee)$.\footnote{Since $J(V_\chi)$ and $V_\chi^\vee$ are simple and have the same character $\overline\chi$, this space is one-dimensional by Schur's lemma, and requiring $j_\chi$ to be unitary fixes it up to a phase.}

In particular, we define the \lq\lq transported" quantities
\begin{equation}
\widetilde J(x)=x^\vee~,\ \ \ \widetilde J(f):=j_y\circ J(f)\circ j_x^{-1}~,\ \ \
\widetilde\theta_{x,y}:=j_{x\otimes y}\circ\theta^J_{x,y}\circ\big(j_x\otimes j_y\big)^{-1}~,
\end{equation}
where $f\in{\rm Hom}(x,y)$. The second and third definitions make $j$ a monoidal natural isomorphism $J\Rightarrow\widetilde J$. We can also define a transported contracting map $\widetilde\eta:\widetilde J^2\Rightarrow{\rm Id}$ via
\begin{equation}
\widetilde\eta:=\eta\circ(j\cdot j)^{-1}=(j\cdot j)^{-1}~, \ \ \ (j\cdot j)_x:=j_{x^\vee}\circ J(j_x)~,
\end{equation}
where, in the definition of $\widetilde\eta$, we have noted that $\eta={\rm Id}$ (as discussed above). In our fixed coordinate basis, we write $U_\chi^{\dagger}$ as the matrix corresponding to the morphism $j_\chi$. By equivariance of $j_\chi$, this matrix satisfies
\begin{equation}\label{eq:Udef}
\rho_\chi(\tau(g))=U_\chi\,\overline{\rho_\chi(g)}\,U_\chi^{\dagger}~,\ \ \ \forall\ g\in G~.
\end{equation}
In particular, $\tau$ is class-inverting iff for every $\chi$ there is a unitary $U_\chi$ satisfying \eqref{eq:Udef}, and $U_\chi$ is unique up to a phase by Schur's lemma.\footnote{To understand this claim in more detail, note $\rho_\chi\circ\tau$ and $\overline{\rho_\chi}$ are irreducible with characters $\chi\circ\tau$ and $\overline\chi$ respectively. They are equivalent for all $\chi$ iff $\chi\circ\tau=\overline\chi$. Both being unitary, the intertwiner may be taken unitary, and Schur gives uniqueness up to $U(1)$.}

We are now ready to relate our abstract quantities entering $\sigma_\chi$ with group theoretical data. In particular, we claim
\begin{equation}\label{eq:sigmainvariant}
\beta_\chi=\big(U_{\bar\chi}^{\dagger}U_\chi^{\dagger}\big)^{-1}~,\ \ \ p_\chi=1~,\ \ \ \widetilde\psi^{-1}_{\chi^\vee}=U_\chi^{\dagger}\,\overline{U_{\bar\chi}}~,
\end{equation}
and that their product is therefore
\begin{equation}\label{eq:sigmaproduct}
\sigma_\chi=\beta_\chi\,p_\chi\,\widetilde\psi^{-1}_{\chi^\vee}=U_\chi\overline{U_\chi}=:\varepsilon_\chi~,
\end{equation}
independently of how $j_{\chi^\vee}$ is fixed relative to $j_\chi$.\footnote{Ultimately, we will show that $\varepsilon_\chi$ is just the KM indicator evaluated on $\chi$.}

First note that $\beta_\chi=(U_{\bar\chi}^{\dagger}U_\chi^{\dagger})^{-1}$ follows directly from the definition of $\widetilde\eta$ and the fact that $J$ acts as the identity on linear maps. Next consider $\widetilde\psi$. Here we apply our result in \eqref{eq:transport} at the object $\chi^\vee$, with $F=J$, $\widetilde F=\widetilde J$, and $m=j$.\footnote{The hypothesis that $j$ be monoidal holds because $\widetilde\theta$ was \emph{defined} by conjugating $\theta$.} Using $\chi^{\vee\vee}=\chi$ and
$\psi^J={\rm Id}$, we obtain
\begin{equation}
\widetilde\psi_{\chi^\vee}=(j^\vee_{\chi^\vee})^{-1}\circ\psi^J_{\chi^\vee}\circ j_\chi^{-1}=(j^\vee_{\chi^\vee})^{-1}\circ j_\chi^{-1}~,\ \ \ \Rightarrow\ \widetilde\psi^{-1}_{\chi^\vee}=j_\chi\circ j^\vee_{\chi^\vee}~.
\end{equation}
Therefore, $\widetilde\psi^{-1}_{\chi^\vee}=U_\chi^{\dagger}\overline{U_{\bar\chi}}$. Note the evaluation point: the $\psi$-stage is right-whiskered by $\pi$, so it is read at $\chi^\vee$ and the second label map appearing is $U_{\bar\chi}$, not $U_\chi$.

Before evaluating the product, we derive two properties of the $U_\chi$ matrices that follow from Schur's lemma. Applying \eqref{eq:Udef} twice and using $\tau^2={\rm id}$ yields
\begin{equation}
\rho(g)=\rho(\tau^2(g))=U\overline{\rho(\tau(g))}U^{\dagger}=U\overline U\,\rho(g)\,\overline{U^{\dagger}}U^{\dagger}=\big(U\overline U\big)\,\rho(g)\,\big(U\overline U\big)^{-1}~,
\end{equation}
and so $U\overline U$ commutes with the irreducible $\rho$. Schur's lemma then gives $U_\chi\overline{U_\chi}=\varepsilon_\chi\mathds 1$ with $|\varepsilon_\chi|=1$ by unitarity. Note that $\varepsilon_\chi$ is unchanged under the residual phase freedom, $U_\chi\mapsto cU_\chi$, so it is well defined. Moreover $\overline U=\varepsilon U^{\dagger}$. Now, conjugating yields $U=\overline\varepsilon\,U^{T}$, and transposing gives
\begin{equation}\label{eq:Usym}
U_\chi^{T}=\overline\varepsilon_\chi\,U_\chi~.
\end{equation}
Substituting \eqref{eq:Usym} back into $U=\overline\varepsilon\,U^{T}$ gives $U=\overline\varepsilon^{\,2}U$, so that
\begin{equation}\label{eq:epssign}
\varepsilon_\chi\in\{\pm1\}~,\ \ \ U_\chi^{T}=\varepsilon_\chi U_\chi~.
\end{equation}
In other words, $\varepsilon_\chi=+1$ or $-1$ according to whether $U_\chi$ is symmetric or antisymmetric. This discussion is consistent with the abstract result of the main text that $\sigma_x\in\{\pm1\}$. What remains for us to do is to associate $\varepsilon_\chi$ with the KM indicator.

With \eqref{eq:epssign} in hand, the form of the product in \eqref{eq:sigmaproduct} follows, and two natural normalisations give the same answer. If $U_{\bar\chi}=U_\chi^{\dagger}$, then $\beta_\chi=(UU^{\dagger})^{-1}=1$ and
\begin{equation}
\widetilde\psi^{-1}_{\chi^\vee}=U^{\dagger}\overline{U^{\dagger}}=U^{\dagger}U^{T}=\varepsilon_\chi\,U^{\dagger}U=\varepsilon_\chi~.
\end{equation}
If instead $U_{\bar\chi}=\overline{U_\chi}$, then $\widetilde\psi^{-1}_{\chi^\vee}=U^{\dagger}U=1$ and
\begin{equation}
\beta_\chi=\big(U^{T}U^{\dagger}\big)^{-1}=\big(\varepsilon_\chi\,UU^{\dagger}\big)^{-1}=\varepsilon_\chi^{-1}=\varepsilon_\chi~.
\end{equation}
The last equality holds because $\varepsilon_\chi=\pm1$. Either way, 
\begin{equation}\label{sigmaRes}
\sigma_\chi=\varepsilon_\chi=U_\chi\overline{U_\chi}~.
\end{equation}
The two normalisations differ by $U_{\bar\chi}\mapsto\varepsilon_\chi U_{\bar\chi}$ and move the indicator between $\beta$ and $\widetilde\psi$ without changing the product; only $\sigma_\chi$ is convention-independent. The first is the choice that sets $\beta_\chi=1$, which is what allows us to display \eqref{sigmaApp} with only two factors.

We are now ready to prove our main result: that $\varepsilon_\chi=\nu_\tau(\chi)$ is the KM indicator in \eqref{eq:KMApp}. To that end, insert \eqref{eq:Udef} into \eqref{eq:KMApp} and expand
\begin{equation}
\nu_\tau(\chi)=\frac1{|G|}\sum_g\operatorname{tr}\!\big(\rho(g)U\overline{\rho(g)}U^{\dagger}\big)=\frac1{|G|}\sum_g\sum_{i,j,k,l}\rho(g)_{ij}U_{jk}\overline{\rho(g)_{kl}}\,U^{\dagger}_{li}~.
\end{equation}
Schur orthogonality, $|G|^{-1}\sum_g\rho(g)_{ij}\overline{\rho(g)_{kl}}=n^{-1}\delta_{ik}\delta_{jl}$ for $\rho$ irreducible unitary of degree $n=n_\chi$, collapses this expression to
\begin{equation}
\nu_\tau(\chi)=\frac1n\sum_{i,j}U_{ji}U^{\dagger}_{ji}=\frac1n\operatorname{tr}\big(U^{T}U^{\dagger}\big)=\frac1n{\rm Tr}\big(\varepsilon_\chi\,UU^{\dagger}\big)=\varepsilon_\chi~,
\end{equation}
as desired. In particular, from \eqref{sigmaRes}, {\it we have proven \eqref{sigmaApp} and equated the KM index with our abstract $\sigma_\chi$ for the case of $\CC={\rm Rep}(G)$.}

\section{Chern-Simons theory proof}\label{CSApp}
In this appendix we prove an important structural result for our discussion in Sec. \ref{ChernSimons} on Chern-Simons theories. In particular, we show that

\smallskip
\noindent
{\bf Lemma \ref{lem:free}.} {\it Specialize to the self-dual cases discussed above. Let $\lambda$ lie in the level-$k$ alcove and set $H=\Hom(\mathbf 1,V_{\lambda,\CA}\otimes V_{\lambda,\CA})$.
Then
\begin{enumerate}
\item[\bf(a)] $H$ is free of rank one over $\CA$, with generator $\gamma$
\item[\bf(b)] $\gamma$ is primitive, hence has nonzero image under any ring homomorphism $\CA\to\mathbb C$ sending $q^{1/L}$ to a root of unity (including the case $q\mapsto1$). These are the only specializations used in this section.
\item[\bf(c)] the image of $H$ in $\cC(\mathfrak g,k)$ is again one dimensional, spanned by the image
of $\gamma$
\end{enumerate}}

\begin{proof}
First, we define $M:=V_{\lambda,\CA}\otimes V_{\lambda,\CA}$ and take $H:=\Hom(\mathbf 1,M)\subseteq M$ to be identified with the submodule of invariants.\footnote{Two remarks on which algebra we are working over. To begin with, $V_{\lambda,\CA}$ is an object of the ribbon category of \cite[Thm.~3]{sawin2006quantum}. More precisely, it is a $U^{\dagger}_{\CA}$-module that is free of finite rank over $\CA$, being a finite direct sum of free weight spaces \cite[Sec. 1]{sawin2006quantum}. Therefore, so too is $M$. Next, we compute $H$ below as the $U^{\rm res}_{\CA}$-invariants of $M$, and this agrees with the $U^{\dagger}_{\CA}$-invariants, so that the braiding and the twist do restrict to $H$. Indeed, $U^{\dagger}_{\CA}$ is generated, densely, by $U^{\rm res}_{\CA}$ together with the weight projections $\delta_\nu$, which act on a vector of weight $\nu'$ by $1$ if $\nu=\nu'$ and by $0$ otherwise, and whose counit is $1$ if $\nu=0$ and $0$ otherwise \cite[Sec. 1]{sawin2006quantum}. Since $K_iv=v$ over $\CA$ already forces $v$ to have weight $0$ (the variable $q^{1/L}$ is free), each $\delta_\nu$ acts on $H$ by its counit.}

Crucially in what follows, $M$ is a free $\CA$-module. To understand this point, note that a Weyl module over the integral form, $U^{\rm res}_\CA(\mathfrak g)$, is obtained by applying the algebra to a highest weight vector, $V_{\lambda,\CA}:=U^{\rm res}_\CA(\mathfrak g)\cdot v$, and is the direct sum of its weight spaces. Each such space is free and of finite rank over $\CA$, with the same dimensions as in the classical Weyl module \cite[Sec. 1]{sawin2006quantum}. Hence $V_{\lambda,\CA}$ is free over $\CA$, of rank the classical dimension (this is the precise sense in which its character does not depend on $q$). Since a tensor product of free modules is free, $M$ is also a free $\CA$-module.

\noindent
{\bf(a)} {\it Freeness.} $H$ is the joint kernel of the operators $E^{(n)}_i$, $F^{(n)}_i$ and $K_i-1$, of which only finitely many act non-vacuously, since $E^{(n)}_i$ shifts weight by $n\alpha_i$, and $M$ has finitely many weight spaces. Each is $\CA$-linear, so $H$ is the kernel of a single $\CA$-linear map $M\to M^{\oplus r}$.

Next note that $H$ is saturated in $M$. Indeed, if $fm\in H$ with $f\in \CA$ nonzero, then for each defining operator, $T$, we have $0=T(fm)=fT(m)$ by $\CA$-linearity. Now, since $M$ is torsion-free (being free over a domain), $T(m)=0$. Therefore, $m\in H$.

Next we claim that $\dim_K(H\otimes_\CA K)=1$, where $K:=\mathbb{Q}(q^{1/L})$. Since $M$ is torsion-free, $H\otimes_\CA K$ is the $K$-span of $H$ inside $M_K:=M\otimes_\CA K$, and this span is precisely the space of invariants of $M_K$. Any element of $M_K$ can be written as $m/f$ with $m\in M$ and $0\neq f\in \CA$, and $T(m)/f=0$ forces $T(m)=0$ for each defining operator $T$. To understand this point, note that over $K$ the category is semisimple, and the tensor product of Weyl modules decomposes with the same multiplicities as in the classical case \cite[Sec. 1]{sawin2006quantum}. So the multiplicity of the trivial module in $V_{\lambda,K}\otimes V_{\lambda,K}$ equals the classical one. Recalling that we have defined $V^{\rm cl}_\lambda$ for the irreducible representation of $G$ of highest weight $\lambda$, that classical multiplicity is
\begin{equation}
\dim\Hom_G(\mathbf 1,\,V^{\rm cl}_\lambda\otimes V^{\rm cl}_\lambda)=\dim\Hom_G\bigl((V^{\rm cl}_\lambda)^{*},V^{\rm cl}_\lambda\bigr)=1~.
\end{equation}
The first equality is duality, and the second holds because $V^{\rm cl}_\lambda$ is self-dual. Therefore, the right-hand side is $\End_G(V^{\rm cl}_\lambda)$ and has dimension one by Schur's lemma.

Now fix a basis $e_1,\dots,e_N$ of the free module $M$ and choose any nonzero $h\in H$, say $h=\sum_i a_ie_i$ with $a_i\in \CA$. Since $\CA=\mathbb Z[q^{\pm1/L}]$ is a localisation of $\mathbb Z[t]$, it is a unique factorisation domain, so the $a_i$ have a greatest common divisor $d$. Define
\begin{equation}
  \gamma:=\tfrac1d\,h=\sum_i(a_i/d)\,e_i\ \in M~,
\end{equation}
whose coordinates have greatest common divisor $1$. Then $\gamma\in H$, because $d\gamma=h\in H$ and $H$ is saturated. We claim $H=\CA\gamma$. To see this take $h'\in H$. As $H$ has rank one, $h'$ and $\gamma$ are proportional over $K$, say $h'=(p/r)\gamma$ with $p,r\in \CA$ coprime. Then $rh'=p\gamma$, so $r$ divides $p\,(a_i/d)$ for every $i$. Since $r$ is coprime to $p$, it divides every $a_i/d$, and these have greatest common divisor $1$. Therefore, $r$ is a unit and $h'\in \CA\gamma$. Hence $H$ is free of rank one with generator $\gamma$.

\smallskip\noindent
{\bf(b)} {\it Primitivity and non-vanishing.} By construction, the coordinates of $\gamma$ have greatest common divisor $1$. Let $\phi:\CA\to\mathbb C$ be either of our specialisations, so $\phi(q^{1/L})=\zeta$ for a root of unity $\zeta$, of order $m$ (the case $q\mapsto1$ is $m=1$). Then $\ker\phi$ is principal, generated by the $m$-th cyclotomic polynomial $\Phi_m$. Indeed $\Phi_m$ is irreducible over $\mathbb Q$, so it is the minimal polynomial of $\zeta$, and any $f$ vanishing at $\zeta$ is divisible by it there. Moreover, $\Phi_m$ is monic, so division by it stays within $\mathbb Z[q^{\pm1/L}]$, and the quotient has coefficients in $\CA$. If $\phi(\gamma)$ were zero, every coordinate of $\gamma$ would lie in $\ker\phi$, so $\Phi_m$ would divide all of them, contradicting the fact that their greatest common divisor is $1$. Hence $\phi(\gamma)\neq0$, and the image of $\gamma$ spans the specialized channel, which is one dimensional for either specialization. Indeed, for a general module, the dimension of the invariant subspace can jump under specialization, but here it cannot exceed one. The reason is that, at $q=1$, the module $V^{\rm cl}_\lambda$ is irreducible and self-dual. At the non-trivial root of unity, the specialized Weyl module, $V_{\lambda,\zeta}$, is simple for $\lambda$ in the (open) level-$k$ alcove by the linkage principle \cite[Cor.~6]{sawin2006quantum},\footnote{This fact is not immediate from \cite[Cor.~6]{sawin2006quantum}, which asserts only that a composition factor, $L^{\mu}$, of the Weyl module of highest weight $\lambda$ satisfies $\mu\le\lambda$ and $\mu=\sigma(\lambda+\rho)-\rho$ for some $\sigma$ in the affine Weyl group at level $k$. To deduce simplicity, note that such a $\mu$ is dominant, and that $\mu\le\lambda$ forces $\langle\mu+\rho,\theta^\vee\rangle\le\langle\lambda+\rho,\theta^\vee\rangle<k+h^\vee$, because $\langle\alpha,\theta^\vee\rangle\ge0$ for every positive root $\alpha$. Hence $\mu$ lies in the level-$k$ alcove $\Lambda_k$ of \eqref{WilsonGk}, as does $\lambda$. Now the weights of $\Lambda_k$ lie in the \emph{interior} of the principal alcove, since the inequality in \eqref{WilsonGk} is strict and $\langle\lambda+\rho,\alpha_i^\vee\rangle\ge1$ for $\lambda\in P_+$. That alcove is a fundamental domain for the action \cite[Lem.~1]{sawin2006quantum}, and the affine Weyl group permutes the alcoves simply transitively \cite[Sec. 4]{sawin2006quantum}. Two interior weights in one orbit therefore coincide, so $\mu=\lambda$, and $L^{\lambda}$ is the only composition factor. We conclude that $V_{\lambda,\zeta}$ is simple.} and is isomorphic to its dual (the dual is simple with highest weight $-w_0\lambda=\lambda$). In either case, rigidity gives $\Hom(\mathbf 1,V\otimes V)\cong\Hom(V^{\vee},V)$, which has dimension at most one by Schur's lemma. Since $\phi(\gamma)$ is a nonzero invariant vector, the dimension is exactly one.

\smallskip\noindent
{\bf(c)} {\it Survival of the channel.} Let $\Pi$ denote the quotient functor from the category of tilting modules to $\CC(\mathfrak g,k)$. By \cite[Thm.~5]{sawin2006quantum}, the functor $\Pi$ is \emph{full}, the simple objects of $\CC(\mathfrak g,k)$ are the images of the tilting modules with highest weight in the alcove, and each such image is non-null (since $\dim_q V_{\lambda,\zeta}\neq0$ there).

Two facts follow. First, $V_\lambda=\Pi(V_{\lambda,\zeta})$ is a simple object, and it is self-dual by hypothesis. Therefore,
\begin{equation}
\Hom(\mathbf 1,V_\lambda\otimes V_\lambda)\cong\End(V_\lambda)=\mathbb C~,
\end{equation}
is one dimensional by Schur's lemma. Second, fullness means $\Pi$ is surjective on morphism spaces, so the specialized channel, one dimensional and spanned by the image of $\gamma$ by (b), maps onto it. A surjection between one-dimensional spaces is an isomorphism, so the image of $\gamma$ is nonzero and spans $\Hom(\mathbf 1,V_\lambda\otimes V_\lambda)$.
\end{proof}

\end{appendix}

\newpage
\bibliography{chetdocbib}

@article{byles2026,
      title={Reconstructing non-Abelian braiding and fusion without anyon transport},
      author={Lucy Byles and Matthew D. Horner and Benjamin T. H. Varcoe and Jiannis K. Pachos},
      year={2026},
      eprint={2608.04103},
      archivePrefix={arXiv},
      primaryClass={quant-ph},
      url={https://arxiv.org/abs/2608.04103},
}

@article{Cesar,
    author = "Galindo, Cesar",
    title = "{On Split Forms of Fusion Categories (to Appear)}",
}

@article{Schafer-Nameki:2023jdn,
    author = "Schafer-Nameki, Sakura",
    title = "{ICTP lectures on (non-)invertible generalized symmetries}",
    eprint = "2305.18296",
    archivePrefix = "arXiv",
    primaryClass = "hep-th",
    doi = "10.1016/j.physrep.2024.01.007",
    journal = "Phys. Rept.",
    volume = "1063",
    pages = "1--55",
    year = "2024"
}

@article{Moore:1988qv,
    author = "Moore, Gregory W. and Seiberg, Nathan",
    title = "{Classical and Quantum Conformal Field Theory}",
    reportNumber = "IASSNS-HEP-88-39",
    doi = "10.1007/BF01238857",
    journal = "Commun. Math. Phys.",
    volume = "123",
    pages = "177",
    year = "1989"
}

@article{Bonderson:2018ryx,
    author = "Bonderson, Parsa and Delaney, Colleen and Galindo, C{\'e}sar and Rowell, Eric C. and Tran, Alan and Wang, Zhenghan",
    title = "{On invariants of Modular categories beyond modular data}",
    eprint = "1805.05736",
    archivePrefix = "arXiv",
    primaryClass = "math.QA",
    doi = "10.1016/j.jpaa.2018.12.017",
    journal = "J. Pure Appl. Algebra",
    volume = "223",
    pages = "4065--4088",
    year = "2019"
}

@article{Hong:2007ty,
    author = "Hong, Seung-moon and Rowell, Eric and Wang, Zhenghan",
    title = "{On exotic modular tensor categories}",
    eprint = "0710.5761",
    archivePrefix = "arXiv",
    primaryClass = "math.GT",
    doi = "10.1142/S0219199708003162",
    journal = "Commun. Contemp. Math.",
    volume = "10",
    number = "supp01",
    pages = "1049--1074",
    year = "2008"
}

@article{Moore:1991ks,
    author = "Moore, Gregory W. and Read, N.",
    title = "{Nonabelions in the fractional quantum Hall effect}",
    doi = "10.1016/0550-3213(91)90407-O",
    journal = "Nucl. Phys. B",
    volume = "360",
    pages = "362--396",
    year = "1991"
}

@article{adams2014real,
  title={The real Chevalley involution},
  author={Adams, Jeffrey},
  journal={Compositio Mathematica},
  volume={150},
  number={12},
  pages={2127--2142},
  year={2014},
  publisher={London Mathematical Society}
}

@article{rowell2005family,
  title={On a family of non-unitarizable ribbon categories},
  author={Rowell, Eric C},
  journal={Mathematische Zeitschrift},
  volume={250},
  number={4},
  pages={745--774},
  year={2005},
  publisher={Springer}
}

@article{sawin2006quantum,
  title={Quantum groups at roots of unity and modularity},
  author={Sawin, Stephen F},
  journal={Journal of Knot Theory and Its Ramifications},
  volume={15},
  number={10},
  pages={1245--1277},
  year={2006},
  publisher={World Scientific}
}

@article{street1993braided,
  title={Braided tensor categories},
  author= {Joyal, A and Street, Ross},
  journal={Advances in Math},
  volume={102},
  number={1},
  pages={20--78},
  year={1993}
}

@article{Feiguin:2006ydp,
    author = "Feiguin, Adrian and Trebst, Simon and Ludwig, Andreas W. W. and Troyer, Matthias and Kitaev, Alexei and Wang, Zhenghan and Freedman, Michael H.",
    title = "{Interacting anyons in topological quantum liquids: The golden chain}",
    eprint = "cond-mat/0612341",
    archivePrefix = "arXiv",
    doi = "10.1103/PhysRevLett.98.160409",
    journal = "Phys. Rev. Lett.",
    volume = "98",
    pages = "160409",
    year = "2007"
}

@article{kawanaka1990twisted,
  title={A twisted version of the Frobenius-Schur indicator and multiplicity-free permutation representations},
  author={Kawanaka, Noriaki and Matsuyama, Hiroshi},
  journal={Hokkaido Mathematical Journal},
  volume={19},
  number={3},
  pages={495--508},
  year={1990},
  publisher={Hokkaido University, Department of Mathematics}
}

@article{Ladisch,
  author       = {Frieder Ladisch},
  title        = {Finite groups in which every character has real values: grading the representations},
  howpublished = {MathOverflow},
  year         = {2011},
  note         = {\href{https://mathoverflow.net/questions/53126/finite-groups-in-which-every-character-has-real-values-grading-the-representati}{MathOverflow}}
}

@book{burnside1911theory,
  title     = {Theory of Groups of Finite Order},
  author    = {Burnside, William},
  edition   = {2nd},
  publisher = {Cambridge University Press},
  year      = {1911},
}

@article{miller1910groups,
  title={Groups involving only a small number of sets of conjugate operators},
  author={Miller, George Abram},
  journal={Arch. Math. und Phys.},
  volume={17},
  pages={199--204},
  year={1911}
}

@article{BarkeshliBondersonChengWang2019,
  author  = {Maissam Barkeshli and Parsa Bonderson and Meng Cheng
             and Zhenghan Wang},
  title   = {Symmetry Fractionalization, Defects, and Gauging of
             Topological Phases},
  journal = {Physical Review B},
  volume  = {100},
  pages   = {115147},
  year    = {2019},
  doi     = {10.1103/PhysRevB.100.115147}
}

@article{EtingofNikshychOstrik2010,
  author  = {Pavel Etingof and Dmitri Nikshych and Victor Ostrik},
  title   = {Fusion Categories and Homotopy Theory},
  journal = {Quantum Topology},
  volume  = {1},
  pages   = {209--273},
  year    = {2010}
}

@article{simon2022straighteningfrobeniusschurindicator,
      title={Straightening Out the Frobenius-Schur Indicator},
      author={Steven H. Simon and Joost K. Slingerland},
      year={2022},
      eprint={2208.14500},
      archivePrefix={arXiv},
      primaryClass={hep-th},
      url={https://arxiv.org/abs/2208.14500},
}

@article{Paper1,
    author = "Buican, Matthew and Huston, Peter and Pachos, Jiannis K.",
    title = "{Anyons and Inherently Complex F-symbols}",
    eprint = "2607.10181",
    archivePrefix = "arXiv",
    primaryClass = "cond-mat.str-el",
    month = "7",
    year = "2026"
}

@article{KNBalasubramanian:2025vum,
    author = "K. N. Balasubramanian, Mahesh and Buican, Matthew and Delcamp, Clement and Radhakrishnan, Rajath",
    title = "{Gauging non-invertible symmetries in (2+1)d topological orders}",
    eprint = "2507.01142",
    archivePrefix = "arXiv",
    primaryClass = "hep-th",
    doi = "10.1088/1751-8121/ae82d0",
    journal = "J. Phys. A",
    volume = "59",
    number = "28",
    pages = "285401",
    year = "2026"
}

@book{etingof2015tensor,
  title={Tensor categories},
  author={Etingof, Pavel and Gelaki, Shlomo and Nikshych, Dmitri and Ostrik, Victor},
  volume={205},
  year={2015},
  publisher={American Mathematical Soc.}
}

@article{Davydov:2010kfz,
    author = "Davydov, Alexei and Mueger, Michael and Nikshych, Dmitri and Ostrik, Victor",
    title = "{The Witt group of non-degenerate braided fusion categories}",
    eprint = "1009.2117",
    archivePrefix = "arXiv",
    primaryClass = "math.QA",
    month = "9",
    year = "2010"
}

@article{Kong:2024ykr,
    author = "Kong, Liang and Zhang, Zhi-Hao and Zhao, Jiaheng and Zheng, Hao",
    title = "{Higher condensation theory}",
    eprint = "2403.07813",
    archivePrefix = "arXiv",
    primaryClass = "cond-mat.str-el",
    month = "3",
    year = "2024"
}

@article{Gaiotto:2014kfa,
    author = "Gaiotto, Davide and Kapustin, Anton and Seiberg, Nathan and Willett, Brian",
    title = "{Generalized Global Symmetries}",
    eprint = "1412.5148",
    archivePrefix = "arXiv",
    primaryClass = "hep-th",
    doi = "10.1007/JHEP02(2015)172",
    journal = "JHEP",
    volume = "02",
    pages = "172",
    year = "2015"
}

@article {MR2312110,
    AUTHOR = {Freedman, Michael H. and Wang, Zhenghan},
     TITLE = {Large quantum {F}ourier transforms are never exactly realized
              by braiding conformal blocks},
   JOURNAL = {Phys. Rev. A (3)},
  FJOURNAL = {Physical Review. A. Third Series},
    VOLUME = {75},
      YEAR = {2007},
    NUMBER = {3},
     PAGES = {032322, 5},
      ISSN = {1050-2947,1094-1622},
   MRCLASS = {81P68 (81T45)},
  MRNUMBER = {2312110},
       DOI = {10.1103/PhysRevA.75.032322},
       URL = {https://doi.org/10.1103/PhysRevA.75.032322},
}

@article{Leinaas:1977fm,
    author = "Leinaas, J. M. and Myrheim, J.",
    title = "{On the theory of identical particles}",
    doi = "10.1007/BF02727953",
    journal = "Nuovo Cim. B",
    volume = "37",
    pages = "1--23",
    year = "1977"
}

@article{Wilczek:1982wy,
    author = "Wilczek, Frank",
    title = "{Quantum Mechanics of Fractional Spin Particles}",
    reportNumber = "NSF-ITP-82-56",
    doi = "10.1103/PhysRevLett.49.957",
    journal = "Phys. Rev. Lett.",
    volume = "49",
    pages = "957--959",
    year = "1982"
}

@article{Kitaev:2005hzj,
    author = "Kitaev, Alexei",
    title = "{Anyons in an exactly solved model and beyond}",
    eprint = "cond-mat/0506438",
    archivePrefix = "arXiv",
    doi = "10.1016/j.aop.2005.10.005",
    journal = "Annals Phys.",
    volume = "321",
    number = "1",
    pages = "2--111",
    year = "2006"
}

@article{siehler2000braided,
  title={Braided near-group categories},
  author={Siehler, Jacob A},
  year={2000},
  month = "nov",
  eid = {arXiv:math/0011037},
  doi = {10.48550/arXiv.math/0011037},
  archivePrefix = {arXiv},
  eprint = {math/0011037}
}

@article{siehler2003near,
  title={Near-group categories},
  author={Siehler, Jacob},
  journal={Algebraic \& Geometric Topology},
  volume={3},
  number={2},
  pages={719--775},
  year={2003},
  publisher={Mathematical Sciences Publishers}
}

@article{ng2024recovering,
       author = {{Ng}, Siu-Hung and {Rowell}, Eric C and {Wen}, Xiao-Gang},
        title = "{Recovering R-symbols from modular data}",
         year = 2024,
        month = "aug",
          eid = {arXiv:2408.02748},
          doi = {10.48550/arXiv.2408.02748},
archivePrefix = {arXiv},
       eprint = {2408.02748},
}

@article{Cordova:2017vab,
    author = "Cordova, Clay and Hsin, Po-Shen and Seiberg, Nathan",
    title = "{Global Symmetries, Counterterms, and Duality in Chern-Simons Matter Theories with Orthogonal Gauge Groups}",
    eprint = "1711.10008",
    archivePrefix = "arXiv",
    primaryClass = "hep-th",
    doi = "10.21468/SciPostPhys.4.4.021",
    journal = "SciPost Phys.",
    volume = "4",
    number = "4",
    pages = "021",
    year = "2018"
}

@article{BylesForbesPachos2024,
doi = {10.1088/1367-2630/ae4aca},
url = {https://doi.org/10.1088/1367-2630/ae4aca},
year = {2026},
month = {apr},
publisher = {IOP Publishing},
volume = {28},
number = {4},
pages = {044501},
author = {Byles, Lucy and Forbes, Ewan and Pachos, Jiannis K},
title = {Demonstration of magic state power of {$D(S_3)$} anyons with two qudits},
journal = {New Journal of Physics}
}

@article{thornton2011braided,
  title={On braided near-group categories},
  author={Thornton, Josiah},
eprint       = {1102.4640},
  archivePrefix= {arXiv},
  year={2011}
}

@article{Buican:2020who,
    author = "Buican, Matthew and Li, Linfeng and Radhakrishnan, Rajath",
    title = "{$a\times b=c$ in $2+1$D TQFT}",
    eprint = "2012.14689",
    archivePrefix = "arXiv",
    primaryClass = "hep-th",
    reportNumber = "QMUL-PH-20-37",
    doi = "10.22331/q-2021-06-04-468",
    journal = "Quantum",
    volume = "5",
    pages = "468",
    year = "2021"
}

@article{Lee:2018eqa,
    author = "Lee, Yasunori and Tachikawa, Yuji",
    title = "{A study of time reversal symmetry of abelian anyons}",
    eprint = "1805.02738",
    archivePrefix = "arXiv",
    primaryClass = "hep-th",
    reportNumber = "IPMU-18-0074",
    doi = "10.1007/JHEP07(2018)090",
    journal = "JHEP",
    volume = "07",
    pages = "090",
    year = "2018"
}

@article{Quinn:1998un,
    author = "Quinn, Frank",
    editor = "Hass, Joel and Scharlemann, Martin",
    title = "{Group categories and their field theories}",
    eprint = "math/9811047",
    archivePrefix = "arXiv",
    doi = "10.2140/gtm.1999.2.407",
    journal = "Geom. Topol. Monographs",
    volume = "2",
    pages = "407--453",
    year = "1999"
}

@article{Kapustin:2010hk,
    author = "Kapustin, Anton and Saulina, Natalia",
    title = "{Topological boundary conditions in abelian Chern-Simons theory}",
    eprint = "1008.0654",
    archivePrefix = "arXiv",
    primaryClass = "hep-th",
    doi = "10.1016/j.nuclphysb.2010.12.017",
    journal = "Nucl. Phys. B",
    volume = "845",
    pages = "393--435",
    year = "2011"
}

@article{ardonne2021classification,
  title={Classification of metaplectic fusion categories},
  author={Ardonne, Eddy and Finch, Peter E and Titsworth, Matthew},
  journal={Symmetry},
  volume={13},
  number={11},
  pages={2102},
  year={2021},
  publisher={MDPI}
}

@article{Kaidi:2021gbs,
    author = "Kaidi, Justin and Komargodski, Zohar and Ohmori, Kantaro and Seifnashri, Sahand and Shao, Shu-Heng",
    title = "{Higher central charges and topological boundaries in 2+1-dimensional TQFTs}",
    eprint = "2107.13091",
    archivePrefix = "arXiv",
    primaryClass = "hep-th",
    doi = "10.21468/SciPostPhys.13.3.067",
    journal = "SciPost Phys.",
    volume = "13",
    number = "3",
    pages = "067",
    year = "2022"
}

@article{morrison2012non,
  title={Non-cyclotomic fusion categories},
  author={Morrison, Scott and Snyder, Noah},
  journal={Transactions of the American Mathematical Society},
  volume={364},
  number={9},
  pages={4713--4733},
  year={2012}
}

@article{kirillov1989representations,
  title={Representations of the algebra Uq (sl (2)), q-orthogonal polynomials and invariants of links},
  author={Kirillov, AN and Reshetikhin, N Yu},
  journal={Infinite dimensional Lie algebras and groups},
  volume={7},
  pages={285--339},
  year={1989},
  publisher={World Scientific Singapore}
}

@article{Barkeshli:2014cna,
    author = "Barkeshli, Maissam and Bonderson, Parsa and Cheng, Meng and Wang, Zhenghan",
    title = "{Symmetry Fractionalization, Defects, and Gauging of Topological Phases}",
    eprint = "1410.4540",
    archivePrefix = "arXiv",
    primaryClass = "cond-mat.str-el",
    doi = "10.1103/PhysRevB.100.115147",
    journal = "Phys. Rev. B",
    volume = "100",
    number = "11",
    pages = "115147",
    year = "2019"
}

@article{Kobayashi:2025ykb,
    author = "Kobayashi, Ryohei and Barkeshli, Maissam",
    title = "{Soft symmetries of topological orders}",
    eprint = "2501.03314",
    archivePrefix = "arXiv",
    primaryClass = "cond-mat.str-el",
    month = "1",
    year = "2025"
}

@article{MR3421083,
    AUTHOR = {Davydov, Alexei},
     TITLE = {Unphysical diagonal modular invariants},
   JOURNAL = {J. Algebra},
  FJOURNAL = {Journal of Algebra},
    VOLUME = {446},
      YEAR = {2016},
     PAGES = {1--18},
      ISSN = {0021-8693,1090-266X},
   MRCLASS = {81R10},
  MRNUMBER = {3421083},
MRREVIEWER = {Sorin\ D\u asc\u alescu},
       DOI = {10.1016/j.jalgebra.2015.09.007},
       URL = {https://doi.org/10.1016/j.jalgebra.2015.09.007}
}

@article{mignard2021modular,
  title={Modular categories are not determined by their modular data},
  author={Mignard, Micha{\"e}l and Schauenburg, Peter},
  journal={Letters in Mathematical Physics},
  volume={111},
  number={3},
  pages={60},
  year={2021},
  publisher={Springer}
}

@article{MR4133163,
	author = {Penneys, David},
	fjournal = {Higher Structures},
	journal = {High. Struct.},
	mrclass = {18M05 (46L37)},
	mrnumber = {4133163},
	note = {\mathscinet{MR4133163} \arxiv{1808.00323}},
	number = {2},
	pages = {22--56},
	title = {Unitary dual functors for unitary multitensor categories},
	volume = {4},
	year = {2020}}

@article{MR2677836,
	author = {Etingof, Pavel and Nikshych, Dmitri and Ostrik, Victor},
	doi = {10.4171/QT/6},
	fjournal = {Quantum Topology},
	issn = {1663-487X},
	journal = {Quantum Topol.},
	mrclass = {18D10 (55S35)},
	mrnumber = {2677836 (2011h:18007)},
	mrreviewer = {Juan Mart{\'{\i}}n Mombelli},
	note = {With an appendix by Ehud Meir, \mathscinet{MR2677836} \doi{10.4171/QT/6} \arxiv{0909.3140}},
	number = {3},
	pages = {209--273},
	title = {Fusion categories and homotopy theory},
	url = {http://dx.doi.org/10.4171/QT/6},
	volume = {1},
	year = {2010}}

@article{MR1976459,
	author = {Ostrik, Victor},
	doi = {10.1007/s00031-003-0515-6},
	fjournal = {Transformation Groups},
	issn = {1083-4362},
	journal = {Transform. Groups},
	mrclass = {18D10 (16D90 17B10 18E10)},
	mrnumber = {MR1976459 (2004h:18006)},
	mrreviewer = {Volodymyr V. Lyubashenko},
	note = {\mathscinet{MR1976459} \arXiv{math/0111139}},
	number = {2},
	pages = {177--206},
	title = {Module categories, weak {H}opf algebras and modular invariants},
	url = {http://dx.doi.org/10.1007/s00031-003-0515-6},
	volume = {8},
	year = {2003}}

@article{Selinger2011,
    author  = "Selinger, P.",
    title   = "{A Survey of Graphical Languages for Monoidal Categories}",
    journal = "Lect. Notes Phys.",
    volume  = "813",
    pages   = "289--355",
    year    = "2011",
    eprint  = "0908.3347",
    archivePrefix = "arXiv",
    primaryClass  = "math.CT",
    doi     = "10.1007/978-3-642-12821-9_4"
}

@book{MR1797619,
	address = {Providence, RI},
	author = {Bakalov, Bojko and Kirillov, Jr., Alexander},
	isbn = {0-8218-2686-7},
	mrclass = {18D10 (17B37 57M27 57R56 81R50 81T40 81T45)},
	mrnumber = {MR1797619 (2002d:18003)},
	mrreviewer = {J. Stasheff},
	note = {\mathscinet{MR1797619}},
	publisher = {American Mathematical Society},
	series = {University Lecture Series},
	title = {Lectures on tensor categories and modular functors},
	volume = {21},
	year = {2001}}

@article{MR1444286,
	author = {Longo, R. and Roberts, J. E.},
	coden = {KTHEEO},
	doi = {10.1023/A:1007714415067},
	fjournal = {$K$-Theory. An Interdisciplinary Journal for the Development, Application, and Influence of $K$-Theory in the Mathematical Sciences},
	issn = {0920-3036},
	journal = {$K$-Theory},
	mrclass = {46L37 (18D10 46M15 81T05)},
	mrnumber = {1444286},
	mrreviewer = {Sze-Kai Tsui},
	note = {\mathscinet{MR1444286} \doi{10.1023/A:1007714415067} \arxiv{funct-an/9604008}},
	number = {2},
	pages = {103--159},
	title = {A theory of dimension},
	url = {http://dx.doi.org/10.1023/A:1007714415067},
	volume = {11},
	year = {1997}}

@article{MR2091457,
	author = {Yamagami, Shigeru},
	fjournal = {Journal of Operator Theory},
	issn = {0379-4024},
	journal = {J. Operator Theory},
	mrclass = {46L05 (18D15 46L37)},
	mrnumber = {2091457 (2005f:46109)},
	mrreviewer = {Teodor Banica},
	note = {\mathscinet{MR2091457}},
	number = {1},
	pages = {3--20},
	title = {Frobenius duality in {$C^*$}-tensor categories},
	volume = {52},
	year = {2004}}

@article{Galindo_2014,
    title={On Braided and Ribbon Unitary Fusion Categories},
    volume={57},
    DOI={10.4153/CMB-2013-017-5},
    number={3},
    journal={Canadian Mathematical Bulletin},
    author={Galindo, César},
    year={2014},
    pages={506–510}}

@book{MR1292673,
	address = {Berlin},
	author = {Turaev, Vladimir G.},
	isbn = {3-11-013704-6},
	mrclass = {57M25 (16W30 17B37 57N10 81R50)},
	note = {\mathscinet{MR1292673}},
	publisher = {Walter de Gruyter \& Co.},
	series = {de Gruyter Studies in Mathematics},
	title = {Quantum invariants of knots and 3-manifolds},
	volume = {18},
	year = {1994}}

@article{10.1063/1.4895764,
    author = {Davydov, Alexei},
    title = {Bogomolov multiplier, double class-preserving automorphisms, and modular invariants for orbifolds},
    journal = {Journal of Mathematical Physics},
    volume = {55},
    number = {9},
    pages = {092305},
    year = {2014},
    month = {09},
    issn = {0022-2488},
    doi = {10.1063/1.4895764},
    url = {https://doi.org/10.1063/1.4895764},
    note = {\doi{10.1063/1.4895764} \arxiv{1312.7466}}
    }

\end{document}